\pdfoutput=1
\documentclass[runningheads,orivec,envcountsame,envcountsect,appendix]{llncs}
\usepackage[T1]{fontenc}
\usepackage{graphicx}

\usepackage[numbers,sort&compress]{natbib}
\usepackage{hyperref}
\usepackage[prologue]{xcolor}

\usepackage{tikz}
\usepackage{xspace}
\usepackage{xargs}
\usetikzlibrary{calc,positioning}
\usetikzlibrary{arrows,automata,shapes,matrix,patterns,patterns.meta}
\usepackage{style/notation}
\usepackage{style/hmsc-tikz}
\usepackage{newunicodechar}
\newunicodechar{∃}{\ensuremath{\exists}}
\newunicodechar{∀}{\ensuremath{\forall}}
\newunicodechar{θ}{\ensuremath{\theta}}
\newunicodechar{τ}{\ensuremath{\tau}}
\newunicodechar{φ}{\ensuremath{\varphi}}
\newunicodechar{ξ}{\ensuremath{\xi}}
\newunicodechar{ζ}{\ensuremath{\zeta}}
\newunicodechar{ψ}{\ensuremath{\psi}}
\newunicodechar{π}{\ensuremath{\pi}}
\newunicodechar{α}{\ensuremath{\alpha}}
\newunicodechar{β}{\ensuremath{\beta}}
\newunicodechar{γ}{\ensuremath{\gamma}}
\newunicodechar{δ}{\ensuremath{\delta}}
\newunicodechar{ε}{\ensuremath{\varepsilon}}
\newunicodechar{κ}{\ensuremath{\kappa}}
\newunicodechar{λ}{\ensuremath{\lambda}}
\newunicodechar{μ}{\ensuremath{\mu}}
\newunicodechar{ρ}{\ensuremath{\rho}}
\newunicodechar{σ}{\ensuremath{\sigma}}
\newunicodechar{ω}{\ensuremath{\omega}}
\newunicodechar{Γ}{\ensuremath{\Gamma}}
\newunicodechar{Φ}{\ensuremath{\Phi}}
\newunicodechar{Δ}{\ensuremath{\Delta}}
\newunicodechar{Σ}{\ensuremath{\Sigma}}
\newunicodechar{Π}{\ensuremath{\Pi}}
\newunicodechar{∑}{\ensuremath{\Sigma}}
\newunicodechar{∏}{\ensuremath{\Pi}}
\newunicodechar{Θ}{\ensuremath{\Theta}}
\newunicodechar{Ω}{\ensuremath{\Omega}}
\newunicodechar{⇒}{\ensuremath{\Rightarrow}}
\newunicodechar{⇐}{\ensuremath{\Leftarrow}}
\newunicodechar{⇔}{\ensuremath{\Leftrightarrow}}
\newunicodechar{→}{\ensuremath{\rightarrow}}
\newunicodechar{←}{\ensuremath{\leftarrow}}
\newunicodechar{↔}{\ensuremath{\leftrightarrow}}
\newunicodechar{¬}{\ensuremath{\neg}}
\newunicodechar{∧}{\ensuremath{\land}}
\newunicodechar{∨}{\ensuremath{\lor}}
\newunicodechar{≠}{\ensuremath{\neq}}
\newunicodechar{≡}{\ensuremath{\equiv}}
\newunicodechar{∼}{\ensuremath{\sim}}
\newunicodechar{≈}{\ensuremath{\approx}}
\newunicodechar{≥}{\ensuremath{\geq}}
\newunicodechar{≤}{\ensuremath{\leq}}
\newunicodechar{≫}{\ensuremath{\gg}}
\newunicodechar{≪}{\ensuremath{\ll}}
\newunicodechar{∅}{\ensuremath{\emptyset}}
\newunicodechar{⊆}{\ensuremath{\subseteq}}
\newunicodechar{⊂}{\ensuremath{\subset}}
\newunicodechar{∩}{\ensuremath{\cap}}
\newunicodechar{⋂}{\ensuremath{\cap}}
\newunicodechar{∪}{\ensuremath{\cup}}
\newunicodechar{⋃}{\ensuremath{\cup}}
\newunicodechar{⊎}{\ensuremath{\uplus}}
\newunicodechar{∈}{\ensuremath{\in}}
\newunicodechar{∉}{\ensuremath{\not\in}}
\newunicodechar{⊤}{\ensuremath{\top}}
\newunicodechar{⊥}{\ensuremath{\bot}}
\newunicodechar{₀}{\ensuremath{_0}}
\newunicodechar{₁}{\ensuremath{_1}}
\newunicodechar{₂}{\ensuremath{_2}}
\newunicodechar{₃}{\ensuremath{_3}}
\newunicodechar{₄}{\ensuremath{_4}}
\newunicodechar{₅}{\ensuremath{_5}}
\newunicodechar{₆}{\ensuremath{_6}}
\newunicodechar{₇}{\ensuremath{_7}}
\newunicodechar{₈}{\ensuremath{_8}}
\newunicodechar{₉}{\ensuremath{_9}}
\newunicodechar{⁰}{\ensuremath{^0}}
\newunicodechar{¹}{\ensuremath{^1}}
\newunicodechar{²}{\ensuremath{^2}}
\newunicodechar{³}{\ensuremath{^3}}
\newunicodechar{⁴}{\ensuremath{^4}}
\newunicodechar{⁵}{\ensuremath{^5}}
\newunicodechar{⁶}{\ensuremath{^6}}
\newunicodechar{⁷}{\ensuremath{^7}}
\newunicodechar{⁸}{\ensuremath{^8}}
\newunicodechar{⁹}{\ensuremath{^9}}
\newunicodechar{𝔹}{\ensuremath{\mathbb{B}}}
\newunicodechar{ℝ}{\ensuremath{\mathbb{R}}}
\newunicodechar{ℕ}{\ensuremath{\mathbb{N}}}
\newunicodechar{ℂ}{\ensuremath{\mathbb{C}}}
\newunicodechar{ℚ}{\ensuremath{\mathbb{Q}}}
\newunicodechar{𝕋}{\ensuremath{\mathbb{T}}}
\newunicodechar{𝕏}{\ensuremath{\mathbb{X}}}
\newunicodechar{ℤ}{\ensuremath{\mathbb{Z}}}
\newunicodechar{✓}{\checkmark}
\newunicodechar{✗}{\ensuremath{\times}}
\newunicodechar{◊}{\ensuremath{\lozenge}}
\newunicodechar{□}{\ensuremath{\square}}
\newunicodechar{𝓐}{\ensuremath{\mathcal{A}}}
\newunicodechar{𝓑}{\ensuremath{\mathcal{B}}}
\newunicodechar{𝓒}{\ensuremath{\mathcal{C}}}
\newunicodechar{𝓓}{\ensuremath{\mathcal{D}}}
\newunicodechar{𝓔}{\ensuremath{\mathcal{E}}}
\newunicodechar{𝓕}{\ensuremath{\mathcal{F}}}
\newunicodechar{𝓖}{\ensuremath{\mathcal{G}}}
\newunicodechar{𝓗}{\ensuremath{\mathcal{H}}}
\newunicodechar{𝓘}{\ensuremath{\mathcal{I}}}
\newunicodechar{𝓙}{\ensuremath{\mathcal{J}}}
\newunicodechar{𝓚}{\ensuremath{\mathcal{K}}}
\newunicodechar{𝓛}{\ensuremath{\mathcal{L}}}
\newunicodechar{𝓜}{\ensuremath{\mathcal{M}}}
\newunicodechar{𝓝}{\ensuremath{\mathcal{N}}}
\newunicodechar{𝓞}{\ensuremath{\mathcal{O}}}
\newunicodechar{𝓟}{\ensuremath{\mathcal{P}}}
\newunicodechar{𝓠}{\ensuremath{\mathcal{Q}}}
\newunicodechar{𝓡}{\ensuremath{\mathcal{R}}}
\newunicodechar{𝓢}{\ensuremath{\mathcal{S}}}
\newunicodechar{𝓣}{\ensuremath{\mathcal{T}}}
\newunicodechar{𝓤}{\ensuremath{\mathcal{U}}}
\newunicodechar{𝓥}{\ensuremath{\mathcal{V}}}
\newunicodechar{𝓦}{\ensuremath{\mathcal{W}}}
\newunicodechar{𝓧}{\ensuremath{\mathcal{X}}}
\newunicodechar{𝓨}{\ensuremath{\mathcal{Y}}}
\newunicodechar{𝓩}{\ensuremath{\mathcal{Z}}}
\newunicodechar{…}{\ensuremath{\ldots}}
\newunicodechar{∗}{\ensuremath{\ast}}
\newunicodechar{⊢}{\ensuremath{\vdash}}
\newunicodechar{⊧}{\ensuremath{\models}}
\newunicodechar{′}{\ensuremath{'}}
\newunicodechar{″}{\ensuremath{''}}
\newunicodechar{‴}{\ensuremath{'''}}
\newunicodechar{∥}{\ensuremath{\|}}
\newunicodechar{⊕}{\ensuremath{\oplus}}
\newunicodechar{⁺}{\ensuremath{^+}}
\newunicodechar{⊇}{\ensuremath{\supseteq}}
\newunicodechar{∘}{\ensuremath{\circ}}
\newunicodechar{∙}{\ensuremath{\cdot}}
\newunicodechar{⋅}{\ensuremath{\cdot}}
\newunicodechar{≈}{\ensuremath{\approx}}
\newunicodechar{×}{\ensuremath{\times}}
\newunicodechar{∞}{\ensuremath{\infty}}
\newunicodechar{⊑}{\ensuremath{\sqsubseteq}}
 \usepackage{mathpartir}
\usepackage{booktabs}
\usepackage[shortlabels,inline]{enumitem}
\usepackage{subcaption}
\usepackage{wrapfig}
\usepackage{savefnmark}
\usepackage{amssymb}
\usepackage{thm-restate}
\usepackage[capitalise]{cleveref}
\usepackage{crossreftools}
\usepackage{mathrsfs}

\usepackage[title]{style/appendix}
\additionalmaterial[The extended version~\cite{fullversion}]{\cite{fullversion}}

\ifappendix
\definecolor[named]{links}{cmyk}{0.55,1,0,0.15}
  \definecolor[named]{urls}{cmyk}{1,0.58,0,0.21}
  \hypersetup{
    bookmarksdepth=2,colorlinks,
    linkcolor=links,
    citecolor=links,
    urlcolor=urls,
    filecolor=urls
  }
\fi

\definecolor{colorblind1}{RGB}{216, 27, 96} \definecolor{colorblind2}{RGB}{30, 136, 229} \definecolor{colorblind3}{RGB}{255, 193, 7} \definecolor{colorblind4}{RGB}{0, 77, 64} 

\colorlet{MST}{colorblind1}
\colorlet{AMP}{colorblind2}

\newif\ifdraft
\drafttrue
\newtoggle{arxiv}
\toggletrue{arxiv}

\setlist[itemize]{topsep=3pt,partopsep=3pt}
\setlist[enumerate]{topsep=3pt,partopsep=3pt}
\renewlist{enumerate*}{enumerate*}{1}
\setlist[enumerate*]{label=(\roman*)}
\makeatletter
\renewcommand\paragraph{\@startsection{paragraph}{4}{\z@}{-1\p@ \@plus -4\p@ \@minus 4\p@}{-0.5em \@plus -0.22em \@minus 0.1em}{\normalfont\normalsize\itshape}}
\makeatother

\begin{document}

\title{An Automata-theoretic Basis for Specification and Type Checking of Multiparty Protocols}

\author{ Felix Stutz\thanks{corresponding author}\inst{1,2}\orcidID{0000-0003-3638-4096} \and \\
Emanuele D'Osualdo \inst{1,3}\orcidID{0000-0002-9179-5827}
\\ \email{felix.stutz@uni.lu, emanuele.dosualdo@uni-konstanz.de}
}
\institute{
MPI-SWS, Germany \and
University of Luxembourg, Luxembourg \and
University of Konstanz, Germany
}

\authorrunning{F.\ Stutz and E.\ D'Osualdo}
\titlerunning{AMP: An Automata-theoretic Basis for Multiparty Protocols} \maketitle

\begin{abstract}
We propose the Automata-based Multiparty Protocols framework~(AMP)
for top-down protocol development. 
The framework features a new very general formalism for global protocol specifications called Protocol State Machines~(PSMs),
Communicating State Machines~(CSMs) as specifications for local participants,
and a type system to check a \picalc\ with session interleaving and delegation
against the CSM specification.
Moreover, we define a large class of PSMs, called ``tame'',
for which we provide a sound and complete PSPACE projection operation that computes a CSM describing the same protocol as a given PSM if one exists. 
We propose these components as a backwards-compatible new backend for
frameworks in the style of Multiparty Session Types. In comparison to the latter, AMP offers a considerable improvement in expressivity,
decoupling of the various components (\eg projection and typing),
and robustness (thanks to the complete projection).   \end{abstract}

\section{Introduction}

Designing correct distributed communication protocols is an important and hard problem.
Consider a finite set of protocol \emph{participants}
(\ie independent processes)
whose only means of interaction between each other is asynchronous message passing through reliable FIFO channels.
The goal is to program each participant so that some global emergent behaviour is achieved, \eg a~leader is~elected.
Unfortunately, even when each participant is finite-state, the presence of unbounded delays (\ie unbounded communication channels) makes any non-trivial property of the emergent global behaviour undecidable \cite{DBLP:journals/jacm/BrandZ83}.

The top-down protocol design approach proposes to work around this issue by a reversal in the methodology:
instead of first programming the participants and then checking that their global behaviour is what we desired it to be, we first specify the desired global behaviour and then synthesize each participant's local specification
so that local behaviour gives rise to the correct global behaviour by construction.
Each participant's concrete implementation is then checked against its local specification, which
\begin{enumerate*}[(a)]
\item can be achieved by static means like type systems, and
\item makes the verification of the implementation local and~modular.
\end{enumerate*}

\acronym{Multiparty Session Types}{MSTs} \cite{DBLP:conf/popl/HondaYC08}
is one of the most prominent and extensively studied formalisms
supporting this top-down design methodology.
The key components of the framework,
depicted in \cref{fig:framework}, are:
\begin{enumerate*}[(1)]
  \item \textbf{Global Types:} a~dedicated language to specify correct global behaviour;\item \textbf{Local Types:} a~dedicated language to specify each participant's actions in the protocol;
  \item \textbf{Programs:} a~programming language (typically a \picalc) to express concrete implementations of each participant of the protocol.
\end{enumerate*}

Imagine, as a simple example,
we want to specify a centralised leader election protocol,
where an arbiter~\p{a} selects a leader among~\p{p} and~\p{q}
and the selected participant communicates the win to the other.
A possible global type representing the protocol is
$
  G =
  (\act{a->p:\lbl{sel}} \seq \act{p->q:\lbl{win}})
  +
  (\act{a->q:\lbl{sel}} \seq \act{q->p:\lbl{win}})
$
where
$\act{a->q:\lbl{sel}}$ says that
$\p{a}$ sends $\lbl{sel}$ and $\p{q}$ receives it and $+$ denotes branching.
The local type of \p{p} would then be
$ (\ev<{a->p:\lbl{sel}} \seq \ev>{p->q:\lbl{win}}) + \ev<{q->p:\lbl{win}} $
where $\ev<{a->p:\lbl{sel}}$ means ``$\p{p}$~receives message $\lbl{sel}$ from $\p{a}$'' and
$\ev<{p->q:\lbl{win}}$ means ``$\p{p}$ sends message $\lbl{win}$ to $\p{q}$''.
Thus, \p{p} is supposed to listen for a message from \p{a} or from \p{q};
in the former case it would then communicate the win to \p{q},
in the latter, just concede.
A program implementation may consist of a process for each participant;
the process for the arbiter \p{a} may implement any specific policy for selecting the
leader (\eg always choose~\p{p}), as long as the communications
follow the protocol.

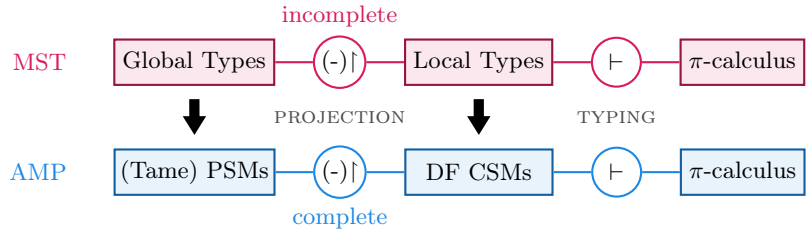
\begin{figure}[t]
  \centering
  \begin{tikzpicture}[
  shade/.code={
    \colorlet{fw@col}{#1}
    \tikzset{draw=#1}
  },
  link/.style={
    circle,
    inner sep=1pt,
minimum width = 2em,
    draw,
    fill=white,
  },
  spec/.style={
strut sized,
draw=fw@col!70!black,
    fill=fw@col!10!white,
  },
  fat ->/.style={
    -{[scale=.5]Triangle},shorten=5pt,line width=1ex,
  },
  thick,
  node distance=1cm and .5cm,
]

\begin{scope}[start chain=mst,shade=MST]
  \node[on chain,MST] {MST};
  \node[on chain,spec] {Global Types};
  \node[on chain,link] {$ (\hole)\proj $};
  \node[on chain,spec]  {Local Types};
  \node[on chain,link] {$ \vdash $};
  \node[on chain,spec] {\picalc};

  \draw
    (mst-2) edge (mst-3)
    (mst-3) edge (mst-4)
    (mst-4) edge (mst-5)
    (mst-5) edge (mst-6)
  ;
\end{scope}

\begin{scope}[start chain=amp,shade=AMP]
  \node[on chain,AMP, below=of mst-1,xshift=-.6pt] {AMP};
  \node[on chain,spec] {\phantom{Global Types}} node at (\tikzlastnode) {(Tame) PSMs};
  \node[on chain,link] {$ (\hole)\proj $};
  \node[on chain,spec]  {\phantom{Local Types}} node at (\tikzlastnode) {DF CSMs};
  \node[on chain,link] {$ \vdash $};
  \node[on chain,spec] {\picalc};

  \draw
    (amp-2) edge (amp-3)
    (amp-3) edge (amp-4)
    (amp-4) edge (amp-5)
    (amp-5) edge (amp-6)
  ;
\end{scope}

\node[above=0 of mst-3, font=\footnotesize, MST] {incomplete};
\node[below=0 of amp-3, font=\footnotesize, AMP] {complete};
\path (mst-3) --node[fill=white,text=black!70]{\textsc{projection}} (amp-3);
\path (mst-5) --node[fill=white,text=black!70]{\textsc{typing}} (amp-5);

\draw[fat ->] (mst-2) -- (amp-2);
\draw[fat ->] (mst-4) -- (amp-4);

\end{tikzpicture}
   \vspace{-1.5ex}
  \caption{The components of top-down frameworks.}
  \label{fig:framework}
  \vspace{-3ex}
\end{figure}

The relationship between the three representations of the protocol, \ie global types, local types, and programs, is delicate.
First, the global type and the local types should give rise to the same behaviour;
however it is not always possible to capture the behaviour of a global type with local types.
Suppose, for instance,
that
we modified the leader election protocol $G$
to \mbox{$
  G' =
  (\act{a->p:\lbl{sel}} \seq \act{q->p:\lbl{lose}})
  \, + $} $
  (\act{a->q:\lbl{sel}} \seq \act{p->q:\lbl{lose}}) .
$
While, from a global perspective, it is possible to insist on the losing participant
informing the winner that they lost, locally, the losing participant has no way to determine whether they won or not.
Therefore~$G'$ is not realisable by local processes: we say it is not projectable.
Second, local types are a more abstract representation of the system than programs, but we still want to show that, when implementation details are omitted, a program only performs communications that adhere to the local specification.

In MSTs, the relationship between the three layers of the framework
are enforced through two  procedures:
\begin{enumerate*}[label=(\roman*)]
  \item \textbf{Projection},
    which (when possible) extracts, from a global type~$G$,
    some local types that are guaranteed to behave as described by $G$;
    and
  \item \textbf{Type Checking},
    which checks that the program implementation of each participant adheres to the behaviour specified in its local type. \end{enumerate*}

\begin{samepage}
\noindent
In a perfect world, a framework for top-down protocol design
should be:
\begin{enumerate}[topsep=3pt]
\item \textbf{Expressive:}
    It should support as many protocols as possible.
\item \textbf{Decoupled:}
    Its components
    (global/local specifications, programs, projection, type checking)
    should depend on each other as little as possible,
    and be specified independently of their algorithmic implementation,
    to allow for reuse and modularity.
\item \textbf{Robust:}
    It should only reject a global specification if there is a genuine
    issue with it (\ie no false positives).
\end{enumerate}

Unfortunately, the MST frameworks in the literature leave something to be desired against this ideal picture.
They all suffer from:
\begin{itemize}[topsep=3pt]
  \item \textbf{Expressivity Limitations:}
    Many legitimate protocols are rejected either because the global specification syntax is too restricted, or because the projection algorithm cannot handle them.
    For example,
    every MST framework we are aware of
    can only handle global types with \emph{directed choice},
    \ie where every branching point involves exactly one sender and one receiver.
    This immediately rules out
    our example leader election protocol~$G$
    because the branches involve different receivers.
  \item \textbf{Decoupling Limitations:}
    In MSTs, the syntax of global types directly influences the definition of the projection algorithm and the syntax of the local types,
    which in turn influence the type system design.
    Typically, changing one of the framework's components requires adapting (and reproving correctness of) all the others.
Furthermore, many MST frameworks solely give the intended relation between global and local types through the projection algorithm and do not give a declarative definition.
\item \textbf{Robustness Limitations:}
    The heuristic nature of the projection algorithms makes it very hard
    to predict if a global type will be handled or not by an MST framework,
    even in the case where the behaviour specified by the global type is unproblematic.
\end{itemize}
\end{samepage}

In this paper, we propose a new foundation for top-down protocol design machinery,
dubbed AMP~(Automata-based Multiparty Protocols),
that achieves the expressivity, decoupling and robustness goals.

\paragraph{Expressivity of \framework.}
\Cref{fig:framework} shows the components of AMP.
To achieve expressivity, we propose a new general formalism for (finite-control) global protocol specifications, which we call \emph{Protocol State Machines}~(PSMs).
The formalism is based on automata which are given semantics in terms
of (finite and infinite) sets of words, over an alphabet of send ($\snd{\procA}{\procB}{\val}$) and receive ($\rcv{\procA}{\procB}{\val}$) actions.
PSMs remove many of the restrictions of global types, while retaining their character:
they specify the expected behaviour from a global perspective, and satisfy some basic correctness properties by construction (e.g. every send is eventually received, no type mismatches, etc).
Owing to their generality, PSMs can represent any global type, but can go well beyond them: 
they also strictly subsume High-Level Message Sequence Charts~(HMSCs).
For maximizing expressivity at the local level, we adopt Communicating State Machines~(CSMs) as the formalism for local protocol specifications.
They are a canonical representation for decentralised asynchronous communication
and as such do not impose constraints over what can be represented.
Finally, to maximise expressivity at the program level, we consider a \picalc\ with session interleaving and delegation.

\paragraph{Decoupling of \framework.}
In \framework, decoupling is achieved through its handling of projection and type checking.
For projection, the framework merely specifies the semantic requirements that a correct projection algorithm needs to satisfy:
essentially that it produces a deadlock-free CSM which represents the same language as the input PSM.
This limits the impact of projection on global and local specifications, and leaves open any algorithmic/manual strategy to achieve the projection goal. (We discuss how \framework\ proposes to actually implement projection in the discussion of robustness).
For example, scenarios in which the user provides both the PSM and the CSM and a proof that they represent the same language and the CSM satisfies desirable properties (like deadlock freedom) 
or where an algorithm infers a PSM from a CSM, are both compatible with our framework thanks to this declarative approach to projection.
This treatment is in line with some MST works where the fundamental property of projection is expressed in terms of some behavioural equivalence between
local and global types.
For type checking, decoupling is achieved by defining the type system by
depending exclusively on programs and CSMs.
The standard guarantees of subject reduction,
communication safety and session fidelity
are proven by appealing to properties of CSMs.
This demonstrates how effective CSMs are in providing a clean decoupled interface between projection and type checking.

\paragraph{Robustness of \framework.}
Finally, we demonstrate how robustness can be achieved in \framework,
by identifying a large class of PSMs, called \emph{Tame PSMs}, for which we provide a decidable, sound and complete projection operation.
Tame PSMs extend the reach of sound and complete projection
beyond global types and can handle a large class of HMSCs as well as protocols that cannot be expressed as either global types nor HMSCs.
The main constraint that makes a PSM tame is what we call \emph{sender-driven choice}: that at any branching point, the sender in all the branches is the same participant and takes distinct actions in the branches.
Our projection algorithm builds on a recently proposed complete projection for sender-driven global types~\cite{DBLP:conf/cav/LiSWZ23}.
Thanks to a surprising reduction,
we manage to extend the algorithm to tame PSMs while keeping the complexity in PSPACE.
Due to the fact that our projection operation is complete,
only protocols that do not admit any valid projection will be rejected:
those are protocols which simply cannot be implemented by local processes.
We also show that our class is in a sense ``maximally robust'':
lifting the sender-driven restriction makes projection undecidable,
even for global types.
\framework\xspace is also robust in the sense that one can select the desired
guarantees of the type system and check whether they can be enforced
by checking (syntactic) properties of the global protocol, pinpointing exactly which guarantee is provided by a PSM.
Finally, we show that the framework is backwards-compatible with MSTs:
not only can we encode global types into PSMs and project them, we also pinpoint the (simple) conditions under which our projection
yields CSMs which are equivalent to local~types.

\paragraph{Contributions.}
\begin{samepage}
In summary:
\begin{itemize}
  \item
    We propose PSMs as an expressive general formalism for
    (finite-control) global protocol specifications.
  \item
    We propose CSMs as a canonical model for local protocol specifications
    and specify their desired relationship with PSMs declaratively.
  \item
    We define the first session type system based on CSMs,
    pinpointing exactly the properties of the CSM that are needed
    to provide each of the desired guarantees;
    these properties can be enforced by construction
    by ensuring the PSMs conform to some simple checks.
  \item
    We define Tame PSMs 
    (encompassing all directed and sender-driven global types)
    and give a sound and complete projection algorithm for them.
  \item
    We show that sender-driven choice is a necessary restriction
    even for global types: projection is undecidable otherwise.
  \item
    We characterise which class of PSMs corresponds to global types,
    and which CSMs correspond to local types,
giving us full backward-compatibility with standard MST theory.
\end{itemize}
\end{samepage}

We think of AMP as a backend for top-down protocol design tools with the following workflow.
Any specific tool, e.g.\ Scribble \cite{DBLP:conf/tgc/YoshidaHNN13}, provides a dedicated syntax for types and processes. Then, a global specification is compiled to a PSM (where the compiler guarantees its tameness, which would be trivial for global types) and invokes the projection of AMP, producing a CSM. 
This could be re-translated for user consumption, but also be used to drive typing using AMP. 
Failure of projection can be directly translated by the frontend to an explanation of why the protocol is flawed and must be repaired. Given the generality of PSMs, it should also be easier to experiment with extensions of the frontend language.

\ifappendix The appendix \else \cite{fullversion} \fi contains all proofs, omitted details, and additional examples.

 \section{Motivation and Key Ideas}
\label{sec:key-ideas}

In this section, we give an informal overview of the key ideas behind \framework\xspace
before proceeding with the formal development from \cref{sec:automata-based-protocols}.

\subsection{Global Specifications via Protocol State Machines}
\label{sec:global-specs-via-PSMs}
Our first goal is to define an expressive formalism for specifying global protocols, that is also constrained enough to make it tractable for top-down protocol development.
One of the most accomplished such formalisms, used in MSTs, is \emph{global types}.
\Cref{fig:example1-mst} shows an example of a global type,
represented in \Cref{fig:example1-hmsc} as an HMSC. 

\begin{figure}[b]
  \adjustfigure[\footnotesize]\begin{minipage}[b]{.49\textwidth}
    \[
      μX.
      {+} \begin{cases}
            \act{p->q:m_1} \seq
            \act{q->r:1} \seq
            \act{r->p:v_1} \seq \zero
            \\
            \act{p->q:m_2} \seq
            \act{q->r:1} \seq
            \act{r->p:v_2} \seq \zero
            \\
            \act{p->q:m_3} \seq
            \act{q->r:3} \seq
            \act{p->r:v_3} \seq X
          \end{cases}
    \]
    \vspace{-1ex}
    \caption{Example global type.}
    \label{fig:example1-mst}
  \end{minipage}\begin{minipage}[b]{.49\textwidth}
    \centering \begin{tikzpicture}[hmsc,baseline,event sep=0.6em]

\begin{scope}[msc=s1,final]
  \begin{scope}[participant=p]
    \node[head];
    \node[event] {};
    \node[no event] {};
    \node[event] {};
  \end{scope}

  \begin{scope}[participant=q]
    \node[head];
    \node[event] {};
    \node[event] {};
    \node[no event] {};
  \end{scope}

  \begin{scope}[participant=r]
    \node[head];
    \node[no event] {};
    \node[event] {};
    \node[event] {};
  \end{scope}

  \draw[messages]
    (p-2) edge node[msg,pos=.4]{$m_1$} (q-2)
    (q-3) edge node[msg]{$1$} (r-3)
    (r-4) edge node[msg,below]{$v_1$} (p-4)
  ;
\end{scope}

\begin{scope}[msc=s2,final]
  \begin{scope}[participant=p]
    \node[head,right=of s1r];
    \node[event] {};
    \node[no event] {};
    \node[event] {};
  \end{scope}

  \begin{scope}[participant=q]
    \node[head];
    \node[event] {};
    \node[event] {};
    \node[no event] {};
  \end{scope}

  \begin{scope}[participant=r]
    \node[head];
    \node[no event] {};
    \node[event] {};
    \node[event] {};
  \end{scope}

  \draw[messages]
    (p-2) edge node[msg,pos=.4]{$m_2$} (q-2)
    (q-3) edge node[msg]{$1$} (r-3)
    (r-4) edge node[msg,below]{$v_2$} (p-4)
  ;
\end{scope}

\begin{scope}[msc=s3]
  \begin{scope}[participant=p]
    \node[head,right=of s2r];
    \node[event] {};
    \node[no event] {};
    \node[event] {};
  \end{scope}

  \begin{scope}[participant=q]
    \node[head];
    \node[event] {};
    \node[event] {};
    \node[no event] {};
  \end{scope}

  \begin{scope}[participant=r]
    \node[head];
    \node[no event] {};
    \node[event] {};
    \node[event] {};
  \end{scope}

  \draw[messages]
    (p-2) edge node[msg,pos=.4]{$m_3$} (q-2)
    (q-3) edge node[msg]{$3$} (r-3)
    (p-4) edge node[msg,below]{$v_3$} (r-4)
  ;
\end{scope}
\useasboundingbox +(0,.3cm); \node[initial,empty block,name=s0,above=.8 of s2q];
\draw[trans]
  (s0) edge[bend right=10pt] (s1)
       edge (s2)
       edge[bend left=10pt] (s3)
  (s3) edge[bend right] (s0);

\end{tikzpicture}
     \vspace{-1ex}
    \caption{A protocol as an HMSC.}
    \label{fig:example1-hmsc}
  \end{minipage}\end{figure}
The term
$ \act{p->q:m_1} $
indicates the transmission of message $m_1$ from $\procA$ to $\procB$.
The symbol~$\zero$ denotes termination of the protocol.
Recursion can be specified by binding a recursion variable~$X$ with~$μX$ and using $X$ subsequently to jump back to where~$X$ was bound.
Branching is denoted by~$+$.
In the example, $\procA$ can pick between three branches by sending different messages to~$\procB$.
Subsequently, $\procB$ sends messages to~$\procC$ in all branches:
$1$ in the top and middle branch and $3$ in the bottom branch.
Participant~$\procC$ is supposed to send messages $v_1$ or $v_2$ in the top and middle branch while it receives from~$\procA$ in the bottom branch, which also recurses using~$X$.

What makes these formalisms tractable?
Their first key characteristic is that,
as a specification tool, they allow the user to
(a) adopt a global point of view,
describing what coordination between all the participants is induced by the protocol;
(b) express this coordination without enumerating all possible interleavings
of the send and receive events that can happen due to the asynchronous nature of communication, \eg
$ \act{p->q:m_1} $ indicates the send of the message immediately followed by its receipt, although in any asynchronous implementation, the receipt might happen at a much later point, after other independent events took place.
In \cref{fig:example1-mst},
$\procC$ may lag behind arbitrarily while $\procA$ and $\procB$ keep sending messages.
The second key characteristic of global types and HMSCs is that
they are \emph{finite-control}:
their control structure can be described using a finite graph.
This makes it possible to algorithmically manipulate them,
\eg for verifying they satisfy some desirable properties,
    or for extracting local protocol specifications.

Our aim is to distil these two characterising features and remove any other restriction that is not necessary, to obtain a more expressive global specification formalism.
To do this, we take a language-theoretic view of protocols,
where a protocol is seen as the set of sequences of send and receive events
that are considered compliant with it.
More precisely, a send event $\snd{\procA}{\procB}{\val}$
  records that~$\procA$ sent the message~$\val$ to~$\procB$;
a receive event $\rcv{\procA}{\procB}{\val}$
  records that~$\procB$ received message~$\val$ from~$\procA$.
A protocol specification is the language of desired finite or infinite words
of events.
For the purpose of this section, we will focus on finite words, but the technical development considers both finite and infinite words. 

Not all languages over these events are meaningful in the context of protocols.
First, the sequences of events might not be feasible when using FIFO channels
(\eg $\snd{\procA}{\procB}{1}\cat\rcv{\procA}{\procB}{2}$ is not FIFO);
we write $\FIFO$ for the language of all words that satisfy FIFO order.
Second, if 
\mbox{ $
  \snd{\procA}{\procB}{\val_1} \cat
  \snd{\procC}{\procB}{\val_2} \, \cdot \negmedspace $ } $
  \rcv{\procA}{\procB}{\val_1} \cat
  \rcv{\procC}{\procB}{\val_2}
$
is accepted by a procotol,
it ought to also accept
$
  \snd{\procC}{\procB}{\val_2}\cat
  \snd{\procA}{\procB}{\val_1}\cat
  \rcv{\procA}{\procB}{\val_1}\cat
  \rcv{\procC}{\procB}{\val_2}
$
as this kind of reorderings are induced by the scheduling of participants and network delays which are out of the control of participants.
We write $\close(L)$ for the closure
of the language~$L$ under such reorderings.
Thus a language~$L \subs \FIFO$ represents the global interaction patterns of the protocols;
moreover $L$ can specify only some of these interactions and get all the ones
that should also be possible under the asynchronous semantics
by declaring the full set of acceptable words to be $\close(L)$.

Now, to obtain a finite-control formalism,
we propose to express such a ``core'' language~$L$ for the protocol $\close(L)$ using a finite state machine~$M$
with $\lang(M)=L$.
Since $\FIFO$ is not regular,
the only feasible way of ensuring $ \lang(M) \subs \FIFO $
is by requiring~$M$ to keep track of which sent messages are still pending,
which a finite-state machine can only do up to some maximum capacity for the send buffers.
We thus arrive at the requirement that $\lang(M) \subs \FIFO_B$,
where $ \FIFO_B $ is the set of words respecting FIFO but where the number of pending sends never exceed $B\in\Nat$ at any point in time.
Note that $\FIFO_B$ is regular.

Building on these observations, we define a \emph{Protocol State Machine}~(PSM)
to be a finite state machine~$M$ recognising words of send and receive events,
with $\lang(M) \subs \FIFO_B$ for some~$B\in\Nat$. \cref{fig:example1-psm} shows the protocol of \cref{fig:example1-mst} as a PSM.
\begin{figure}[tb]
  \centering
  \begin{tikzpicture}[psm, node distance=2em and 4em, on grid,baseline=-.5ex]
  \node[state,initial] (s0) {};

  \begin{scope}[start chain=up going right, state/.append style={on chain}]
    \node[state,above right=of s0] {};
    \node[state] {};
    \node[state] {};
    \node[state] {};
    \node[state] {};
    \node[state,final] {};
  \end{scope}
  \begin{scope}[start chain=mid going right, state/.append style={on chain}]
    \node[state,right=of s0] {};
    \node[state] {};
    \node[state] {};
    \node[state] {};
    \node[state] {};
    \node[state,final] {};
  \end{scope}
  \begin{scope}[start chain=down going right, state/.append style={on chain}]
    \node[state,below right=of s0] {};
    \node[state] {};
    \node[state] {};
    \node[state] {};
    \node[state] {};
    \node[state] {};
  \end{scope}
  \draw
    (s0) edge[out=north,in=west] node[send'=p->q:m_1,straight,above]{} (up-1)
         edge node[send'=p->q:m_2,straight]{} (mid-1)
         edge[out=south,in=west] node[send'=p->q:m_3,straight,below]{} (down-1)
    (up-1) edge node[recv'=p->q:m_1]{} (up-2)
    (up-2) edge node[send'=q->r:1]{} (up-3)
    (up-3) edge node[recv'=q->r:1]{} (up-4)
    (up-4) edge node[send'=r->p:v_1]{} (up-5)
    (up-5) edge node[recv'=r->p:v_1]{} (up-6)
    (mid-1) edge node[recv'=p->q:m_2]{} (mid-2)
    (mid-2) edge node[send'=q->r:1]{} (mid-3)
    (mid-3) edge node[recv'=q->r:1]{} (mid-4)
    (mid-4) edge node[send'=r->p:v_2]{} (mid-5)
    (mid-5) edge node[recv'=r->p:v_2]{} (mid-6)
    (down-1) edge node[recv'=p->q:m_3]{} (down-2)
    (down-2) edge node[send'=q->r:3]{} (down-3)
    (down-3) edge node[recv'=q->r:3]{} (down-4)
    (down-4) edge node[send'=p->r:v_3]{} (down-5)
    (down-5) edge node[recv'=p->r:v_3]{} (down-6)
    (down-6) edge[out=90,in=-35,looseness=.2] (s0)
  ;

\end{tikzpicture}   \vspace{-1.5ex}
  \caption{A PSM encoding for the protocol of \cref{fig:example1-mst}.}
  \label{fig:example1-psm}
  \vspace{-1ex}
\end{figure}
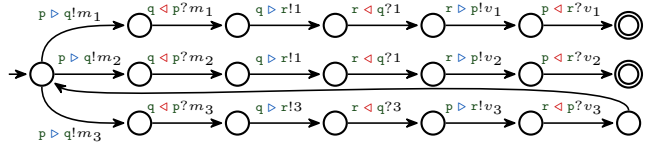
Interpreted as a mere automaton~$M$,
it recognises a language $\lang(M)$ of words with at most one pending send at all time.
(We call~$M$ \sumOnePSM because the total number of messages in flight is at most~$1$;
if we allowed~$1$ message \emph{per channel}, it would be called a \pre1-PSM.) As a PSM, however, $M$ denotes the language $ \close(\lang(M)) $,
which admits words with unbounded channel behaviours and is not even regular in general.
For instance, $ \close(\lang(M)) $
includes words starting with
$
  (
    \ev*{p->q:m_3} \cat
    \ev>{q->r:3}   \cat
    \ev>{p->r:v_3}
  )^n \cat (
    \ev<{q->r:3}   \cat
    \ev<{p->r:v_3}
  )^n
  \dots
$
where~$\procC$ is running at a lower rate than the other participants,
and leaves~$n$ pending sends from~$\procA$ and from~$\procB$
before it consumes them.

PSMs achieve a substantial gain in expressivity while retaining the
key characteristics of global types.
In terms of expressivity, every global type can be encoded as a \sumOnePSM;
furthermore PSMs can be used to encode HMSCs, which strictly subsume global types because the latter cannot specify simultaneous message exchanges between a pair of participants 
\cite{DBLP:journals/corr/abs-2209-10328}. 
PSMs can even represent protocols that are outside the reach of HMSCs.
Consider, for example, the PSM in \cref{fig:psm-no-hmsc-possible}.
In that protocol, \p{p} commits to some integer (abstracted as the label~$\lbl{int}$) at the beginning by sending
it to~\p{r} and sends a $\lbl{go}$ signal to~\p{q}.
Note that here we use the paired send and receive notation
$ \act{p->q:\lbl{ok}} $ to emit the two events in sequence.
Then \p{q} and \p{r} engage in some negotiation of arbitrary length
until \p{q} decides to exit the loop, at which point \p{r} is finally
allowed to receive the message sent by \p{p}.
No HMSC can represent such protocol:
the matching events $\snd{\procA}{\procC}{\lbl{int}}$ and $\rcv{\procA}{\procC}{\lbl{int}}$
are separated by an arbitrary number of events
(with no opportunity for reordering up to~$\close(\hole)$);
since matching events in HMSCs need to belong to the same basic block,
such block would also need to contain the arbitrarily many events in between,
which is~impossible.

\begin{figure}[tb]
  \centering \begin{tikzpicture}[psm, node distance=1.5em and 3em,baseline=-.5ex]
  \begin{scope}[
    start chain=seq going right,
    state/.append style={on chain,join=by trans},
  ]
  \node[state,init]{};
  \node[state] {};
  \node[state] {};
  \node[state] {};
  \node[state,final] {};
  \end{scope}
  \path
    (seq-1) -- node[send=p to r:\lbl{int}]{}
    (seq-2) -- node[event]{$\act{p->q:\lbl{go}}$}
    (seq-3) -- node[event]{$\act{q->r:\lbl{ok}}$}
    (seq-4) -- node[recv=r from p:\lbl{int}]{}
    (seq-5)
  ;
  \node[state,below=of seq-3] (qr) {};
  \draw
    (seq-3) edge[bend right] node[event,straight,left]{$\act{q->r:\lbl{int}}$} (qr)
    (qr) edge[bend right] node[event,straight,right]{$\act{r->q:\lbl{int}}$} (seq-3)
  ;
\end{tikzpicture}
   \caption{A~protocol not expressible as an HMSC.
    Transitions labelled with~$\act{p->q:m}$
    should be interpreted as emitting the sequence
    $ \ev*{p->q:m} $.
  }
  \label{fig:psm-no-hmsc-possible}
\end{figure}
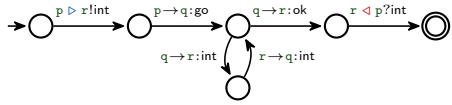

Of course, this level of generality would be pointless if we were not able to
provide for it in the other components of top-down protocol design.
We start by studying the first crucial component: projection.

\subsection{From Global to Local Specifications: Projection}

When considering projection, our first concern is the goal of decoupling:
we want to define a general interface for projection,
such that both different algorithmic implementations of projection can be used
without altering the design of the rest of the framework;
and such that typing is not dependent on global specifications (nor projection details).

In \framework, the key to decoupling is in choosing \emph{Communicating State Machines}~(CSMs) as the formalism for local specifications.
A CSM $\CSM{A}$ associates a finite state automaton $ A_{\procA} $ to each participant~$\procA \in \Procs$,
where transitions can either send or receive on the channels of~$\procA$;
the semantics of $\CSM{A}$ is defined on configurations that include the local states for each participant and an (unbounded) FIFO buffer for each channel.
They induce a FIFO language $ \lang(\CSM{A}) $ over send/receive events, by considering as final
the configurations where all the participants are in final local states and
all the buffers are empty.
CSMs thus represent a canonical general model of finite-control asynchronous
protocol implementations.

\begin{figure}[b]
  \centering
  \begin{tikzpicture}[psm, node distance=2em and 4em, on grid,baseline=-.5ex]
  \node[state,initial,label={above left:\normalsize$A_\procA$}] (s0) {};

  \begin{scope}[start chain=up going right, state/.append style={on chain}]
    \node[state,above right=of s0] {};
    \node[state,final] {};
  \end{scope}
  \begin{scope}[start chain=mid going right, state/.append style={on chain}]
    \node[state,right=of s0] {};
    \node[state,final] {};
  \end{scope}
  \begin{scope}[start chain=down going right, state/.append style={on chain}]
    \node[state,below right=of s0] {};
    \node[state] {};
  \end{scope}
  \draw
    (s0) edge[out=north,in=west] node[send'=p->q:m_1,straight,above]{} (up-1)
         edge node[send'=p->q:m_2,straight]{} (mid-1)
         edge[out=south,in=west] node[send'=p->q:m_3,straight,below]{} (down-1)
    (up-1) edge node[recv'=r->p:v_1]{} (up-2)
    (mid-1) edge node[recv'=r->p:v_2]{} (mid-2)
    (down-1) edge node[send'=p->r:v_3,below]{} (down-2)
    (down-2) edge[out=130,in=-35,looseness=.4] (s0)
  ;

\end{tikzpicture}
\quad
\begin{tikzpicture}[psm, node distance=2em and 4em, on grid,baseline=-.5ex]
  \node[state,initial,label={above left:\normalsize$A_\procB$}] (s0) {};

  \begin{scope}[start chain=up going right, state/.append style={on chain}]
    \node[state,above right=of s0] {};
    \node[state,final] {};
  \end{scope}
  \begin{scope}[start chain=down going right, state/.append style={on chain}]
    \node[state,below right=of s0] {};
    \node[state] {};
  \end{scope}
  \draw
    (s0) edge[out=north,in=west] node[recv'=p->q:m_1,straight,above]{} (up-1)
         edge[out=east,in=south] node[recv'=p->q:m_2,straight,above,pos=.3]{} (up-1)
         edge[out=south,in=west] node[recv'=p->q:m_3,straight,below]{} (down-1)
    (up-1) edge node[send'=q->r:1]{} (up-2)
    (down-1) edge node[send'=q->r:3,below]{} (down-2)
    (down-2) edge[out=130,in=-35,looseness=.4] (s0)
  ;

\end{tikzpicture}
\quad
\begin{tikzpicture}[psm, node distance=2em and 4em, on grid,baseline=-.5ex]
  \node[state,initial,label={above left:$A_\procC$}] (s0) {};

  \begin{scope}[start chain=up going right, state/.append style={on chain}]
    \node[state,above right=of s0] {};
    \node[state,final] {};
  \end{scope}
  \begin{scope}[start chain=down going right, state/.append style={on chain}]
    \node[state,below right=of s0] {};
    \node[state] {};
  \end{scope}
  \draw
    (s0) edge[out=north,in=west] node[recv'=q->r:1,straight,above]{} (up-1)
edge[out=south,in=west] node[recv'=q->r:3,straight,below]{} (down-1)
    (up-1) edge[out=10,in=170] node[send'=r->p:v_1,above,straight]{} (up-2)
    (up-1) edge[out=-10,in=190] node[send'=r->p:v_2,below,straight]{} (up-2)
    (down-1) edge node[recv'=r->p:v_3,below]{} (down-2)
    (down-2) edge[out=130,in=-35,looseness=.4] (s0)
  ;

\end{tikzpicture}   \vspace{-1ex}
  \caption{Example CSM.}
  \label{fig:example1-csm}
\end{figure}
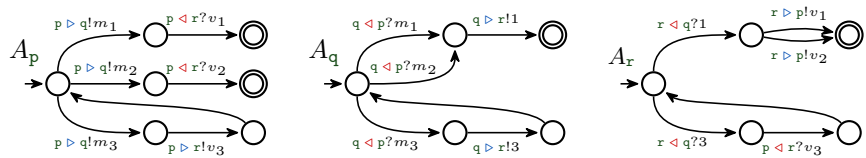

Per se, this is not a particularly original choice:
MST's local types have been linked to CSMs of a certain shape before \cite{DBLP:conf/esop/DenielouY12,DBLP:conf/ecoop/Stutz23}, and HMSC-based work used them as local specifications.
What \framework\ demonstrates is that it is possible to build the entire
top-down methodology around CSMs (with fewer restrictions),
including a session type system,
gaining both in expressivity and in decoupling.

Having fixed our model for local behaviour, we can ask when it defines behaviour consistent with a global specification.
We say a CSM~$\CSM{A}$ \emph{is a projection of} a PSM~$M$ if
$\CSM{A}$ is deadlock-free 
and $ \lang(\CSM{A}) = \close(\lang(M)) $.
We focus on the \emph{(projection) synthesis problem}, producing a CSM as result. 
The corresponding decision problem is the \emph{projectability problem}, which simply asks if there exists such a CSM.
Notably, projectability can have lower complexity. 

Even for simple protocols, projection can be tricky.
Take the example of \cref{fig:example1-mst}:
$\procC$ can never distinguish between the top two branches,
as its only observations would be to have received~$1$ from $\p{q}$.
The instance of the protocol with $m_1 \ne m_2$ and $v_1 \ne v_2$
would thus not be projectable.
If $m_2 = m_3$ then $\p{q}$ would not be able to send the appropriate message to~$\p{r}$.
Therefore, the only projectable instances with no redundant branches
are the ones where
$m_1$, $m_2$, and $m_3$ are pair-wise distinct and $v_1 = v_2$.
\Cref{fig:example1-csm} shows a candidate projection of the PSM in \cref{fig:example1-mst}.
If $m_2 = m_3$ or $v_1 \ne v_2$, the PSM is not projectable,
and in fact the CSM can reach a deadlock.

Given CSMs are Turing-complete models, it is unsurprising that
checking if a given CSM is a projection of a given PSM is undecidable. The key advantage of the top-down approach boils down to the fact that
it is nevertheless often possible to efficiently \emph{compute}
a valid projection from a global specification.
This is precisely the goal of the \emph{projection} operation.
A projection operation $(\hole)\proj$ is a function taking a PSMs as input
and returning either $\bot$ or a CSM;
it is a \emph{sound projection} if for all PSM~$M$,
if $ M\proj = \CSM{A} $ then $ \CSM{A} $ is a projection of~$M$;
it is a \emph{complete projection} if for all projectable $M$, $M\proj \ne \bot$.

The MST literature proposed a number of sound but incomplete
projection algorithms for global types.
Incompleteness makes MST frameworks lack robustness:
a projectable global type might still be rejected by the framework
because the projection is unable to handle it;
this leaves the user in the awkward position of having to
build a mental model for the projection algorithm to be able to
design viable global types.
\citet{DBLP:conf/cav/LiSWZ23} proposed the first sound and complete projection
algorithm for sender-driven global types.
Its PSPACE complexity stems from the need for determinisation.
Their evaluation, though, showed that these corner cases will likely not occur in reality. 
This provides initial evidence that robustness is achievable without compromising efficiency.

As is to be expected, the jump in expressivity by adopting PSMs cannot come for free: the problem of computing a sound and complete projection for PSMs is in general undecidable, a fact inherited from being able to encode \mbox{HMSCs}. This does not defeat us, however:
one of our main positive results is the definition of a very large class of PSMs, called \emph{Tame PSMs}, that enjoys sound and complete PSPACE projection. A PSM is \emph{tame} if it satisfies three constraints:
\begin{enumerate*}[(a)]
  \item \label{cond:tame-PSM-1} \negmedspace
  a technical refinement of the notion of the bound~$B$ for buffers,
  \item \label{cond:tame-PSM-2} \negmedspace
  that final states have no outgoing transitions, and
  \item \label{cond:tame-PSM-3} \negmedspace
  \emph{sender-driven choice}: at each branching point, there is a single sender taking distinct actions.
\end{enumerate*}

Our proof works by reducing the problem to an instance of projectability of MSTs with sender-driven choice,
which was proven to be decidable in PSPACE~\cite{DBLP:conf/cav/LiSWZ23}.
Our reduction is surprising because it produces a transformed protocol
which is different from the original one:
the encoded protocol language is different
and involves additional participants and additional message exchanges;
and yet its synthesized local specifications can be transformed back to local specifications for the original protocol.
Due to the mismatch in expressivity between PSMs and global types,
it is necessary that the reduction modifies the protocol semantics.
Furthermore, we show the reduction preserves the complexity class,
giving us a PSPACE algorithm for projectability
of sender-driven~PSMs.

Despite being a restriction,
Tame PSMs are still much more general than global types:
every sender-driven global type gives rise to a Tame PSM;
moreover, every example given so far is tame.
While the first two constraints
\ref{cond:tame-PSM-1}
and
\ref{cond:tame-PSM-2}
are not severe,
the third condition
\ref{cond:tame-PSM-3}
imposes a genuine restriction on expressivity.

In fact our main negative result is that sender-driven choice
is in a sense ``minimal'':
we prove that projectability is undecidable
for
global types (the most primitive kind of PSMs)
with general choice (aka ``mixed choice'').

\subsection{Processes and Typing}

To complete the top-down toolkit, we
provide a mean to check
that a program correctly implements a protocol specified as a CSM.
We achieve this by defining a CSM-based session type system for an expressive variant of \picalc\ with session interleaving and delegation.
The process calculus is adapted from~\cite{DBLP:journals/pacmpl/ScalasY19}
which represents a feature-rich modern presentation of multiparty session typing.

The type system's main soundness argument hinges, as is standard, on a subject reduction result: if a typable program can take a step, it remains typable.
From this, we derive two main safety correctness guarantees:
typable programs cannot produce type mismatches (\ie receiving a message that the process is not expecting) and terminated sessions do not leave orphan messages behind. 
We further prove a progress property under standard restrictions:
roughly speaking, if the process contains only one session, then, if the type of the session is not final, the process can take a step
(among the ones allowed by the type).
Global progress in the presence of session interleaving is out of scope
of this paper, but it may be attainable by adapting the (orthogonal) analysis employed in~\cite{CoppoDYP16,Kobayashi06}.

In line with our decoupling goal, the guarantees of the type system
are derived from the key properties of CSMs produced by projection
(\eg deadlock freedom).
This makes it even compatible with the bottom-up methodology of
\cite{DBLP:journals/pacmpl/ScalasY19}
which forgoes global types and proposes to check key properties
on local types directly.
If a CSM satisfying the desired properties is provided to our type system,
the corresponding guarantees apply to typable processes regardless
of the existence of a PSM representing the protocol.
This also liberates the type system completely from the choice of representation
for global protocols.

Overall we obtain an expressive, decoupled and robust backend
for top-down protocol development.
Finally, we also show that this backend is backwards-compatible with MSTs:
not only every sender-driven global type can be encoded as Tame PSM, but we also prove that,
when there exists a local type that is a projection of a global type,
our projection produces a CSM that can be translated back to a local type.
This shows under which conditions PSMs and global types as well as CSMs and local types are equivalent, despite their structural differences. 

 \section{Automata-based Protocol Specifications}
\label{sec:automata-based-protocols}

We start our technical development by introducing a language-theoretic
view of protocol specifications.
We define protocols as special languages of words,
and use CSMs as our local specifications of such languages.
Finally, we introduce PSMs as global protocol specifications.

\subsection{State Machines and Protocol Languages}

Let $\StdAlphabet$ be a finite alphabet.
The set of finite words over~$\StdAlphabet$ is denoted by $\StdAlphabet^*$, the set of infinite words by $\StdAlphabet^\omega\negthinspace$, and their union by~$\StdAlphabet^\infty\negthinspace$.
We write $\emptystring$ for the empty word.
For two strings $u \in \StdAlphabet^*$ and $v \in \StdAlphabet^\infty$,
their concatenation is $u \cat w$, and
we say that $u$ is a \emph{prefix} of $v$,
written $u \leq v$,
if there is some $w\in\StdAlphabet^\infty$ such that $u \cat w = v$;
$\pref(v)$ denotes all prefixes of~$v$ and is lifted to languages as expected.
For a language $L \subseteq \StdAlphabet^{\infty}$, we distinguish between the language of finite words $L_{\fin} \is L \inters \StdAlphabet^*$ and the language of infinite words $L_{\inf} \is L \inters \StdAlphabet^\omega$.

\begin{definition}[State machines]
A \emph{state machine} $A = (Q, \StdAlphabet, \delta, q_{0}, F)$ is a $5$-tuple with
a finite set of states $Q$,
an alphabet $\StdAlphabet$,
a transition relation
  $\delta \subseteq Q \times (\StdAlphabet \union \set{\emptystring}) \times Q$,
an initial state $q_{0}\in Q$ from the set of states, and
a set of final states $F$ with $F \subseteq Q$.
If $(q, a, q')\in\delta$, we also write $q \xrightarrow{a} q'\negthinspace$.
A \emph{run} is a finite or infinite sequence
$q_0\xrightarrow{a_0} q_1 \xrightarrow{a_1} \ldots$,
with $q_i~\in~Q$ and $a_i\in \StdAlphabet \union \set{\emptystring}$ for $i\geq 0$,
such that $q_0$ is the initial state, and for each $i\geq 0$, it holds that
$(q_i, a_i, q_{i+1})\in\delta$.
The \emph{trace} of such run is the word
$a_0\cat a_1 \cat \ldots\in \StdAlphabet^\infty\negthinspace$.
A run is \emph{maximal} if it ends in a final state or is infinite.
The \emph{(core) language} $\lang(A)$ of $A$ is the set of traces of all maximal runs.
If~$Q$ is finite, we say~$A$ is a \emph{finite state machine} (FSM).
A state machine is \emph{dense} if
for every $q \xrightarrow{x} q' \in \delta$,
the transition label $x$ is $\emptystring$ implies that $q$ has only one outgoing transition.
A state machine is \emph{deterministic} if
  $ \A (q, a, q') \in \delta. a \neq \emptystring $ and
  $ \A (q, a, q'), (q, a, q'') \in \delta. q' = q'' $.
We call a dense state machine \emph{deterministic} if
$ \A (q, a, q'), (q, a, q'') \in \delta. q' = q'' $.
A state $q \in Q$ is called a \emph{sink state}
if it has no outgoing transitions,
\ie $\A a \in \StdAlphabet \union \set{\emptystring}, q' \in Q. (q, a, q') \notin \delta$.
We say a state machine is \emph{\sinkfinal} if, for every state, it is final iff it is a sink.
\end{definition}

\paragraph{A language-theoretic view of protocols.}
Let $\val\in\MsgVals$ be a finite set of messages 
and $\procA,\procB,\ldots\in\Procs$ be a finite set of participants.
The alphabet of $\procA$'s send and receive events 
is the set
$
  \AlphAsync_{\procA} \is
  \Union_{\procB \in \Procs, \val \in \MsgVals}
  \set{
      \snd{\procA}{\procB}{\val}, \, \rcv{\procB}{\procA}{\val}
}
$.
A~send event $\snd{\procA}{\procB}{\val}$
  records that~$\procA$ sent the message~$\val$ to~$\procB$;
a receive event $\rcv{\procB}{\procA}{\val}$
  records that~$\procA$ received message~$\val$ from~$\procB$.
The alphabet of all events is the set
$\AlphAsync_{\Procs} \is \Union_{\procA \in \Procs} \AlphAsync_{\procA}$.
A paired event is a send event and its corresponding receive event:
$
  \msgFromTo{\procA}{\procB}{\val} \is
    \snd{\procA}{\procB}{\val}
    \cat
    \rcv{\procA}{\procB}{\val}
$.
We define the alphabet of paired events as
$
  \AlphSync_\Procs \is \set{
    \msgFromTo{\procA}{\procB}{\val}
      \mid \procA,\procB ∈ \Procs \text{ and } \val ∈ \MsgVals}
$.
For the remainder of the paper,
we fix an arbitrary set of
participants $\Procs$ and messages $\MsgVals$,
and often write $\AlphAsync$ for~$\AlphAsync_\Procs$ and $\AlphSync$ for $\AlphSync_\Procs$.
Given a word, we can project it to all letters of a certain shape:
for instance, $w\wproj_{\snd{\procA}{\procB}{\_}}$ is the subword of $w$ with all of its send events where $\procA$ sends any message to~$\procB$.
We write $\MsgVals(w)$ for the sequence of values in $w$ (in the same order). In 
$w = w_1 \ldots \in \Alphabet^\infty$,
a send event $w_i = \snd{\procA}{\procB}{\val}$ is \emph{matched} by a receive event $w_j = \rcv{\procA}{\procB}{\val}$ if $i < j$ and
$\MsgVals((w_1 \ldots w_i) \wproj_{\snd{\procA}{\procB}{}})$
=
$\MsgVals((w_1 \ldots w_j) \wproj_{\rcv{\procA}{\procB}{}})$.
A~send event $w_i$ is \emph{unmatched} if there is no such receive event~$w_j$. A language $L \subseteq \AlphAsync^\infty$ satisfies \emph{feasible eventual reception} if for every finite word $w \is w_1 \ldots w_n \in \pref(L)$ with an unmatched send event $w_i$, 
there is an extension $w \preforder w' \in L$ 
such that $w_i$ is matched in $w'$\negthinspace.

A sequence of send and receive events shall describe the execution of a protocol.
We define when such a sequence uses channels in FIFO manner.

\begin{definition}[FIFO Language]
A word $w \in \AlphAsync^\infty$ is \channelcompliant if
    for each prefix $w'$ of $w$, it holds that
    $\MsgVals(w'\wproj_{\rcv{\procA}{\procB}{\_}})$ is a prefix of $\MsgVals(w'\wproj_{\snd{\procA}{\procB}{\_}})$,
    for every $\procA,\procB \in \Procs$.
We denote the set of all infinite \channelcompliant words by $\FIFOlang_{\inf}$.
For finite words, we require that all send events are matched.
Thus,
$\FIFOlang_{\fin} \is \set{w \mid w \text{ is \channelcompliant and }
    \MsgVals(w \wproj_{\snd{\procA}{\procB}{\_}}) = \MsgVals(w \wproj_{\rcv{\procA}{\procB}{\_}}) \; \forall \procA, \procB \in \Procs}$.
We denote the (non-regular) set of all FIFO words by $\FIFOlang = \FIFOlang_{\inf} \dunion \FIFOlang_{\fin}$.
A language $L \subseteq \FIFOlang$ is a called a \emph{FIFO language}.
\end{definition}

As the model of distributed implementation of a protocol,
we adopt communicating state machines: parallel compositions of finite-control processes communicating asynchronously
via point-to-point FIFO channels.

\begin{definition}[Communicating state machines]
We call $\CSMabb{A} = \CSM{A}$ a \emph{communicating state machine} (CSM)  over $\Procs$ and~$\MsgVals$ if
${A}_\procA$
is a finite state machine
with alphabet~$\AlphAsync_{\negthinspace\procA}$ for every $\procA\in\Procs\negthinspace$. The semantics of a CSM~$\CSMabb{A}$
is the language $\lang(\CSMabb{A}) \subs \FIFOlang$
whose definition is standard (see \appendixref{app:csms}).
Roughly,
for each pair of distinct participants $\procA, \procB\in \Procs$ there are
two FIFO channels
$
\channel{\procA}{\procB},
\channel{\procB}{\procA}
\in \channels
$
allowing communication between the participants in the two directions.
The FSM 
$
  {A}_\procA =
    (Q_\procA, \AlphAsync_{\negthinspace\procA}, \delta_\procA, q_{0, \procA}, F_\procA)
$
describes the possible actions of participant~$\procA$.
A transition $(q_{\procA}, \snd{\procA}{\procB}{\val}, q'_{\procA}) \in \delta_\procA$ indicates that when $\procA$ takes a step
from $q_{\procA}$ to $q'_{\procA}$, it will send a message $\val$ to $\procB$
by enqueuing it in channel $\channel{\procA}{\procB}$.
Similarly,
$(q_{\procA}, \rcv{\procB}{\procA}{\val}, q'_{\procA}) \in \delta_\procA$
prescribes the reception by $\procA$ of message $\val$ from
the channel $\channel{\procB}{\procA}$.
A CSM's run always starts with empty channels and each participant running its respective initial state.
We denote the set of all reachable configurations (from the initial configuration) by $\reach(\CSMabb{A})$.
A~\emph{deadlock} of $\CSM{A}$ is a reachable configuration with no outgoing transition that has at least one non-empty channel or at least one participant not in a (local) final~state.
\end{definition}

The formal definition is given in \appendixref{app:csms}.
\noindent
As an example, \cref{fig:example1-csm} shows the three state machines constituting a CSM.

The goal of a protocol designer is to define
a protocol that can be realised as a CSM.
The \emph{projectable} languages are exactly those protocols which can.

\begin{definition}[Projections and Projectability]
A language $L \subseteq \AlphAsync^\infty$ is said to be \emph{projectable} if there exists a \emph{deadlock-free} CSM $\CSM{A}$ such that
 it generates the same language \emph{(protocol fidelity)}, i.e.,
$L = \lang(\CSM{A})$.
We say that $\CSM{A}$ is a projection of $L$.
\end{definition}

The asynchronous nature of CSMs makes them unable to enforce
the order between certain events without explicit synchronisation.
For instance, any CSM producing a word
$ 
 \snd{\procA}{\procB}{\val} \cat
 \snd{\procC}{\procD}{\val'} \cat w
$ 
will necessarily produce also
$ 
 \snd{\procC}{\procD}{\val'} \cat
 \snd{\procA}{\procB}{\val} \cat w
$. 
Which events can be reordered is context-dependent:
the events in the word $
 \snd{\procA}{\procB}{\val} \cat
 \rcv{\procA}{\procB}{\val}
$ cannot be swapped, as the receive is only possible after the send.
But in
$
 \snd{\procA}{\procB}{\val} \cat
 \snd{\procA}{\procB}{\val} \cat
 \rcv{\procA}{\procB}{\val}
$
the last two events can be reordered.
This has been formalised as equivalence relation by~\citet{DBLP:conf/concur/MajumdarMSZ21}, which can be seen as an instance of Lamport's happens-before relation \cite{DBLP:journals/cacm/Lamport78} 
for systems with point-to-point FIFO channels.

\begin{definition}\label{def:indistinguishability-relation}
The \emph{indistinguishability relation}
${\interswap} \subseteq \AlphAsync^* \times \AlphAsync^*$
is the smallest equivalence relation such that
\begin{enumerate}[label=\textnormal{(\arabic*)}]
\item
\label{rule:indistinguishability-relation-send-send}
If $\procA ≠ \procC$, then
$
 w \cat
 \snd{\procA}{\procB}{\val} \cat
 \snd{\procC}{\procD}{\val'} \cat
 u
 \; \interswap \;
 w \cat
 \snd{\procC}{\procD}{\val'} \cat
 \snd{\procA}{\procB}{\val} \cat
 u
$.

\item
\label{rule:indistinguishability-relation-receive-receive}
If $\procB ≠ \procD$, then
$
 w \cat
 \rcv{\procA}{\procB}{\val} \cat
 \rcv{\procC}{\procD}{\val'} \cat
 u
 \; \interswap \;
 w \cat
 \rcv{\procC}{\procD}{\val'} \cat
 \rcv{\procA}{\procB}{\val} \cat
 u
$.

\item
\label{rule:indistinguishability-relation-send-receive-indep}
If $\procA ≠ \procD \land (\procA ≠ \procC ∨ \procB ≠ \procD)
$, then
$
 w \cat
 \snd{\procA}{\procB}{\val} \cat
 \rcv{\procC}{\procD}{\val'} \cat
 u
 \; \interswap \;
 w \cat
 \rcv{\procC}{\procD}{\val'} \cat
 \snd{\procA}{\procB}{\val} \cat
 u
$.
\item
\label{rule:indistinguishability-relation-send-receive-history}
If $\card{w \wproj_{\snd{\procA}{\procB}{}}} >
    \card{w \wproj_{\rcv{\procA}{\procB}{}}}$,
then
$
 w \cat
 \snd{\procA}{\procB}{\val} \cat
 \rcv{\procA}{\procB}{\val'} \cat
 u
 \; \interswap \;
 w \cat
 \rcv{\procA}{\procB}{\val'} \cat
 \snd{\procA}{\procB}{\val} \cat
 u
$.
\end{enumerate}
We define $u \preceq_\interswap v$ if there is $w\in\AlphAsync^*$ such that $u \cat w \interswap v$.
Observe that $u \interswap v$ iff
$u \preceq_\interswap v$ and $v \preceq_\interswap u$.
For infinite words $u, v\in\AlphAsync^\omega$, we define $u \preceq_\interswap^\omega v$
if for each finite prefix $u'\negthinspace$ of $u$, there is a finite prefix~$v'$ of~$v$ such that
$u' \preceq_\interswap v'$.
Define $u \interswap v$ iff $u \preceq_\interswap^\omega v$ and $v\preceq_\interswap^\omega u$.
We lift the equivalence relation $\interswap$ on words to languages.
For a language $L$, we define
{ \small
\[
  \close(L) = \left\{ w' \;\middle|\; \bigvee
    \begin{array}{l}
    w' \in \AlphAsync^* \land ∃ w ∈ \AlphAsync^*. \; w \in L \text{ and } w' \interswap w \\
    w' ∈ \AlphAsync^ω \land \exists w \in \AlphAsync^\omega. \; w \in
    L \text{ and } w' \preceq_\interswap^\omega w
  \end{array} \right\}.
\]
}
\end{definition}

\begin{lemma}[\cite{DBLP:conf/concur/MajumdarMSZ21}]
For any CSM $\CSM{A}$, $\lang(\CSM{A}) = \close(\lang(\CSM{A}))$.
\end{lemma}

\begin{example}
\label{ex:interswap-and-fairness}
For finite words $\close(\hole)$ is standard.
For infinite words, though, the situation is a bit counterintuitive.
Let us consider
$w \is (\snd{\procA}{\procB}{\val} \cat \rcv{\procA}{\procB}{\val})^\omega$.
It is easy to construct a CSM $\CSM{A}$, with FSMs $A_\procA$ and $A_\procB$, that accepts $w$. CSMs do not promise any sort of fairness for infinite runs so there is an infinite run for
$(\snd{\procA}{\procB}{\val})^\omega$
where only $A_\procA$'s transitions are scheduled.
This is why $\close(\hole)$ is defined using
$\preceq^\omega_\interswap$, giving
$(\snd{\procA}{\procB}{\val})^\omega
\in
\close((\snd{\procA}{\procB}{\val} \cat \rcv{\procA}{\procB}{\val})^\omega)$.
\end{example}

\subsection{Protocol State Machines}
\label{sec:protocol-state-machines}
\label{sec:psm}

We now introduce PSMs as a mean to specify protocol languages
from a global, centralised perspective.
The idea, shared with both global types and HMSCs,
is to specify only a core subset of the admissible executions,
\eg the ones where there is a bounded delay between sends and matching receives,
and obtain the full set of admissible executions by closing the core language
using $\close(\hole)$.

We adapt the notion of $B$-bounded from~\cite{DBLP:journals/fuin/GenestKM07}
to formalise the idea of ``bounded delay'' between matching events.

\begin{definition}[$B$-bounded and \sumBounded{B}]Let $B \in \Nat$ be a natural number.
A \channelcompliant word $w$ is \emph{$B$-bounded}, resp.\ \emph{\sumBounded{B}},  if for every prefix $w'$ of $w$ and participants $\procA, \procB \in \Procs$, it holds that
\mbox{$\card{w' \wproj_{\snd{\procA}{\procB}{}}} -
\card{w' \wproj_{\rcv{\procA}{\procB}{}}} \leq B$},
resp.\,
{ \small
$ \sum_{\procA \neq \procB \in \Procs}
    \left(\card{w' \wproj_{\snd{\procA}{\procB}{}}} -
    \card{w' \wproj_{\rcv{\procA}{\procB}{}}}\right)
    \leq B
$. }
We define the (regular) set of \mbox{$B$-bounded} FIFO words:
$\FIFO_B \is \set{ w \in \FIFO \mid w \text{ is $B$-bounded}}$.
\end{definition}

\begin{definition}[Protocol State Machine]
A dense FSM $\PSM = (Q, \AlphAsync, \delta, q_{0}, F)$
is a $B$-PSM if $\lang(\PSM) \subs \FIFO_B$ and $\lang(\PSM)$ satisfies feasible eventual reception.
The semantics of $\PSM$ defined as
$\semantics(\PSM) \is \close(\lang(\PSM))$.
Moreover, $\PSM$ is a PSM if it is a $B$-PSM for some~$B$.
\end{definition}

By definition, PSMs specify FIFO languages;
importantly, although the core language $\lang(\PSM)$ is \pre B-bounded,
the semantics $\close(\lang(\PSM))$ 
includes non-$B$-bounded words
and will not even be regular in general.
Note that, it is decidable to check if an FSM is a \pre B-PSM.
 
In \appendixref{app:FER-reflected-and-preserved}, we show that $\close(\hole)$ preserves and reflects feasible eventual reception:
if $L \subseteq \AlphAsync^\infty$ satisfies feasible eventual reception, then $\close(L)$ does,
and if $\close(L)$ satisfies feasible eventual reception, then $L$ does.
More generally, every property that is preserved by $\close(\hole)$ can be soundly checked on the core language of a PSM.
If the property is also reflected by $\close(\hole)$, the property holds if and only if it holds for the core language.

\begin{definition}
\label{def:sumOnePSM}
A PSM $\PSM$ is a \sumOnePSM if its core language $\lang(\PSM)$ is \sumBounded{\textit{1}}.
We may abuse notation and use $\AlphSync_\Procs$ as alphabet for \sumOnePSMs.
\end{definition}

\begin{example}[Kindergarten Leader Election]
\label{ex:kindergarten-leader-election}
We consider a protocol between two participants \p{e} (evens) and \p{o} (odds).
It can be used to quickly settle a dispute between children (hence the name).
Both children pick~$0$ or~$1$ and tell each other their pick at the same time.
Child \p{e} wins if the sum is even while \p{o} wins if the sum is odd.
At the end, the loser concedes by sending the message $\lbl{win}$ to the winner.
The protocol is specified as a PSM in \cref{fig:kle-psm}
(and as an HMSC in \cref{fig:kle-hmsc}).
Note that specifying this protocol requires
the ability of issuing send and receive events independently.
If one insisted on issuing send and matching receives together,
as in global types and \sumOnePSMs,
one of the children would be forced to reveal their hand first,
undermining the purpose of the protocol.
\end{example}

\begin{figure}[t]
\begin{subfigure}[b]{0.62\textwidth}
\resizebox{0.95\textwidth}{!}{
 \begin{tikzpicture}[psm]
  \node[state,] (s0) {} [
    psm tree,
    level 1/.style={sibling distance=4em},
    level 2/.style={sibling distance=2em},
  ]
    child { node[state] {}
      child { node[state] {}
        child {node[state]{} child {node[state] (s00) {} [recv=e from o:0]} [recv=o from e:0]}
      [send=o to e:0]}
child { node[state] {}
        child {node[state]{} child {node[state] (s01) {} [recv=e from o:1]} [recv=o from e:0]}
      [down,send=o to e:1]}
[send=e to o:0]}
child { node[state] {}
      child { node[state] {}
        child {node[state]{} child {node[state] (s11) {} [recv=e from o:1]} [recv=o from e:1]}
      [send=o to e:1]}
child { node[state] {}
        child {node[state]{} child {node[state] (s10) {} [recv=e from o:0]} [recv=o from e:1]}
      [down,send=o to e:0]}
[down,send=e to o:1]
    }
  ;

  \node[state,between=s11 and s00,right=2em] (wE) {}
  [psm tree,level distance=4em]
    child { node[state] {}
      child { node[state,final] {} [recv=e from o:\win]}
    [send=o to e:\win]}
  ;
  \node[state,between=s01 and s10,right=2em] (wO) {}
  [psm tree,level distance=4em]
    child { node[state] {}
      child { node[state,final] {} [down,recv=o from e:\win]}
    [down,send=e to o:\win]}
  ;
  \draw
    (s00) edge (wE)
    (s11) edge (wE)
    (s01) edge (wO)
    (s10) edge (wO)
  ;
\draw[trans] (s0.west) ++(-2ex,0) -- (s0);
\end{tikzpicture}
 }
\vspace{-1ex}
\caption{KLE as a PSM.}
 \label{fig:kle-psm}
\end{subfigure}
\begin{subfigure}[b]{0.37\textwidth}
\resizebox{0.97\textwidth}{!}{
 \begin{tikzpicture}[hmsc,baseline,msg/.append style={tight,above=2pt,pos=.2}]
\begin{scope}[msc=s1,final]
  \begin{scope}[participant=e]
    \node[head];
    \node[event] {};
    \node[event] {};
    \node[event] {};
  \end{scope}

  \begin{scope}[participant=o]
    \node[head];
    \node[event] {};
    \node[event] {};
    \node[event] {};
  \end{scope}

  \draw[messages]
    (e-2) edge node[msg]{0} (o-3)
    (o-2) edge node[msg]{0} (e-3)
    (o-4) edge node[msg,midway]{\win} (e-4)
  ;
\end{scope}

\begin{scope}[msc=s2,final]
  \begin{scope}[participant=e]
    \node[head,right=of s1o];
    \node[event] {};
    \node[event] {};
    \node[event] {};
  \end{scope}

  \begin{scope}[participant=o]
    \node[head];
    \node[event] {};
    \node[event] {};
    \node[event] {};
  \end{scope}

  \draw[messages]
    (e-2) edge node[msg]{1} (o-3)
    (o-2) edge node[msg]{1} (e-3)
    (o-4) edge node[msg,midway]{\win} (e-4)
  ;
\end{scope}

\begin{scope}[msc=s3,final]
  \begin{scope}[participant=e]
    \node[head,right=of s2o];
    \node[event] {};
    \node[event] {};
    \node[event] {};
  \end{scope}

  \begin{scope}[participant=o]
    \node[head];
    \node[event] {};
    \node[event] {};
    \node[event] {};
  \end{scope}

  \draw[messages]
    (e-2) edge node[msg]{1} (o-3)
    (o-2) edge node[msg]{0} (e-3)
    (e-4) edge node[msg,midway]{\win} (o-4)
  ;
\end{scope}

\begin{scope}[msc=s4,final]
  \begin{scope}[participant=e]
    \node[head,right=of s3o];
    \node[event] {};
    \node[event] {};
    \node[event] {};
  \end{scope}

  \begin{scope}[participant=o]
    \node[head];
    \node[event] {};
    \node[event] {};
    \node[event] {};
  \end{scope}

  \draw[messages]
    (e-2) edge node[msg]{0} (o-3)
    (o-2) edge node[msg]{1} (e-3)
    (e-4) edge node[msg,midway]{\win} (o-4)
  ;
\end{scope}
\scoped[exclude from bounding box]
\path (s2.north) -- node[initial,empty block,name=s0,above=1em]{} (s3.north);
\draw[trans]
  (s0) edge[bend right=10pt] (s1.north east)
       edge (s2.north)
       edge (s3.north)
       edge[bend left=10pt] (s4.north west)
  ;

\end{tikzpicture} }
\vspace{-1ex}
 \caption{KLE as an HMSC.}
 \label{fig:kle-hmsc}
\end{subfigure}
\caption{Kindergarten Leader Election (KLE).}
\vspace{-1ex}
\label{fig:kindergarten-leader-election}
\end{figure}
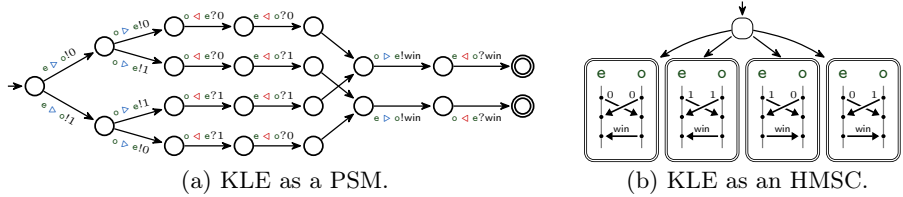
 \section{Projection: From PSMs to CSMs}
\label{sec:projection-dec-undec}

A CSM~$\CSMabb{A}$ is a projection of a PSM~$M$,
if $\CSMabb{A}$ is a projection of $ \semantics(M) $.
In this section, we explain two main results.
The first is positive:
we show that sound and complete projection
is decidable for Tame PSMs.
The second is negative:
we show that the sender-driven restriction of Tame PSMs
is necessary: if we drop the restriction,
projectability becomes undecidable even for sink-final \sumOnePSMs.
The full proofs can be found in
\appendixref{app:projection-dec-undec}.

\subsection{Sound and Complete Projection for Tame PSMs}
\label{sec:sound-and-complete-projection-for-tame-PSMs}

The idea of the decidability result is to reduce projectability of a Tame PSM to
projectability of a (different) sender-driven global type,
which can then be handled
using the sound and complete algorithm of \cite{DBLP:conf/cav/LiSWZ23}. Furthermore, the reduction is such that a projection of the original
PSM can be read off a projection of the global type.
Before sketching the idea behind the reduction, we define Tame PSMs formally.
Tame PSMs satisfy three conditions:
they are sink-final, sender-driven, and satisfy some more fine-grained
bounds on the message queues.

\begin{definition}[Choice restrictions for PSMs]Let $\PSM = (Q, \AlphAsync, \delta, q_{0}, F)$ be a PSM.
The PSM $\PSM$ satisfies \emph{sender-driven choice}
if
there is a function
$\choicefunction \from Q \to \Procs$ such that for all states $q, q'$ such that $q \xrightarrow{x} q'$ with $x \in \AlphAsync_!$, it holds that $\choicefunction(q)$ is the sender for $x$, i.e., $x = \snd{\choicefunction(q)}{\_}{\_}$.
In addition, we say $\PSM$ is \emph{directed} if for every state $q$, there is also a dedicated receiver $\procA$, i.e., all transition labels from $q$ are of the form $\snd{\choicefunction(q)}{\procA}{\_}$.
Last, if there is no dedicated sender but all transitions are still distinct, \ie $\PSM$ is deterministic, we say that it satisfies \emph{mixed choice}.
\end{definition}

\begin{definition}[Channel bounds for PSMs]
We define \emph{channel bounds} as a partial function $\chanBounds \from \channels \pto \Nat$ from channels to natural numbers, where $\domainOf(\chanBounds)$ denotes the domain of $\chanBounds$.
Given a PSM $\PSM$, we say that $\PSM$ respects $\chanBounds$ if the following holds for every
$\channel{\procA}{\procB} \in \channels$:
\begin{itemize}
\item If $\channel{\procA}{\procB} \notin \domainOf(\chanBounds)$,
        then every message from $\procA$ to $\procB$ is immediately followed by a receive:
        for every state $q$ and transition from $q$ to $q'$ labelled with $\snd{\procA}{\procB}{\val}$, it holds that there is a single transition from $q'$ and it is labelled with $\rcv{\procA}{\procB}{\val}$.
\item If $\channel{\procA}{\procB} \in \domainOf(\chanBounds)$,
        then
$w \wproj_{\snd{\procA}{\procB}{}, \rcv{\procA}{\procB}{}}$
        is
        $\chanBounds(\channel{\procA}{\procB})$-bounded
        for every $w \in \semantics(\PSM)$.
\end{itemize}
\end{definition}

A PSM that respects $\chanBounds$ with $\chanBounds = \emptyset$
is a PSM which only uses paired events, just like global types do.
Thus checking the condition with $\chanBounds = \emptyset$ is a
trivial syntactic check.
For general PSMs, it is possible to generate valid channel bounds
with a sound algorithm we propose
in \cref{sec:checking-tameness-for-psms}.
We conjecture the algorithm to be also complete,
\ie to always output some bounds if they exist.

\begin{definition}[Tame PSMs]
A \emph{Tame} PSM is a pair $(M,\chanBounds)$
where the PSM~$M$
is sender-driven, \sinkfinal, and
respects the channel bounds~$\chanBounds$.
\end{definition}

We can now sketch the idea behind the reduction.
Fundamentally, the gap in expressivity between Tame PSMs and
sender-driven global types is that in PSMs
sends and matching receives do not need
to appear one right after the other.
One can observe, however, that one could replicate the same asynchrony
of some trace $ \ev>{p->q:m} \cdots \ev<{p->q:m} $ by introducing an
intermediary participant $\p{(p, q)}$ that is always ready to forward
messages from $\p{p}$ to $\p{q}$,
leading to a trace
$
  \act{p->(p,q):m}
  \cdots
  \act{(p,q)->q:m}
$ where the sends and matching receives between participants and the intermediaries are now immediately adjacent.
The channel bounds~$\chanBounds$ tell us exactly for which channels we need to
introduce intermediaries; moreover the bound on the buffers induced by $\chanBounds$ makes sure that these intermediaries will not introduce any spurious dependency in the executions.
To consolidate the idea, we show how it applies to our KLE example.

\begin{figure}[t]
\centering
\scalebox{0.85}{
 \begin{tikzpicture}[psm]
  \node[state,] (s0) {} [
    psm tree,
    level distance=5em,
    level 1/.style={sibling distance=4em},
    level 2/.style={sibling distance=2em},
  ]
    child { node[state] {}
      child { node[state] {}
        child {node[state]{} child {node[state] (s00) {} [recvchanpaired=e from o:0]} [recvchanpaired=o from e:0]}
      [sendchanpaired=o to e:0]}
child { node[state] {}
        child {node[state]{} child {node[state] (s01) {} [recvchanpaired=e from o:1]} [recvchanpaired=o from e:0]}
      [down,sendchanpaired=o to e:1]}
[sendchanpaired=e to o:0]}
child { node[state] {}
      child { node[state] {}
        child {node[state]{} child {node[state] (s11) {} [recvchanpaired=e from o:1]} [recvchanpaired=o from e:1]}
      [sendchanpaired=o to e:1]}
child { node[state] {}
        child {node[state]{} child {node[state] (s10) {} [recvchanpaired=e from o:0]} [recvchanpaired=o from e:1]}
      [down,sendchanpaired=o to e:0]}
[down,sendchanpaired=e to o:1]
    }
  ;

  \node[state,between=s11 and s00,right=2em] (wE) {}
  [psm tree,level distance=6em]
    child { node[state] {}
      child { node[state,final] {} [recvchanpaired=e from o:\win]}
    [sendchanpaired=o to e:\win]}
  ;
  \node[state,between=s01 and s10,right=2em] (wO) {}
  [psm tree,level distance=6em]
    child { node[state] {}
      child { node[state,final] {} [down,recvchanpaired=o from e:\win]}
    [down,sendchanpaired=e to o:\win]}
  ;
  \draw
    (s00) edge (wE)
    (s11) edge (wE)
    (s01) edge (wO)
    (s10) edge (wO)
  ;
\draw[trans] (s0.west) ++(-2ex,0) -- (s0);
\end{tikzpicture}
 }
\vspace{-1ex}
 \caption{Kindergarten Leader Election after the Channel-participant Encoding.}
 \label{fig:KLE-chan-part-enc}
\vspace{-2ex}
\end{figure}
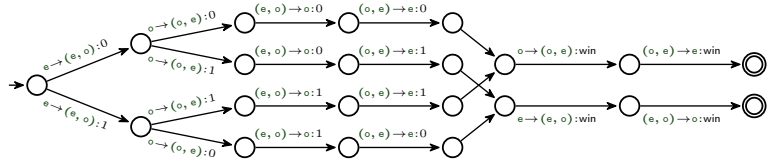

\begin{example}[Revisiting the KLE protocol]
\label{ex:revisiting-KLE}
In \cref{ex:kindergarten-leader-election}, we introduced the Kindergarten Leader Election protocol, whose communication pattern cannot be represented as a \sumOnePSM/global type: both children need to commit to the number they send before they receive the other's message.
Its PSM (\cref{fig:kindergarten-leader-election})
is however tame:
it is sink-final, sender-driven and respects $\chanBounds(\p{e},\p{o})=\chanBounds(\p{o},\p{e})=1$.
The ``intermediary forwarders'' idea applied to the protocol
results in a protocol where some teachers (the intermediaries)
will act as depositories for the initial choices of the two children.
After committing their choice, each child is allowed to learn from the teacher
the choice of the other child.
The resulting PSM is given in \cref{fig:KLE-chan-part-enc}.
The names of the additional participants indicate the direction of communication: $\procChanWo{\p{e}}{\p{o}}$ forwards messages from $\p{e}$ to $\p{o}$.
Obviously, this encoding does not specify the same protocol.
Still, our construction shows that one can obtain a projection of the original protocol from a projection of the modified one, by appropriately removing the forwarding actions of the teachers. \end{example}

The example illustrates the simple case where $\chanBounds(\hole) \leq 1$;
in the more general case, the reduction is more involved and requires
more intermediaries.

\noindent
\begin{minipage}[b]{.48\textwidth}
The workflow of our encoding is visualised in \cref{fig:workflow-encoding}.
Given a PSM~$\PSM$, one first computes its encoding $\encchanpsm(\PSM)$. Second, one synthesizes a projection $\CSM{A}$ of the encoded protocol using results from~\cite{DBLP:conf/cav/LiSWZ23}.
Third, one decodes
to obtain a projection $\CSMl{\decchanfsm(A_\procA)}$ of~$\PSM$.
\end{minipage}
\hfill
\begin{minipage}[b]{.52\textwidth}
\centering \resizebox{1.02\textwidth}{!}{
    \begin{tikzpicture}[
  thick,
  annot/.style={
    auto,
    sloped,font=\tiny,
  },
  number/.style={
    inner sep=.5pt,
    outer sep=0pt,
    circle,
    font=\tiny,
    draw,
    semithick,
  },
  num/.style={
    label={[number]left:#1},
  },
]
  \matrix[
    matrix of math nodes,
    column sep=4em,
    row sep=5em,
    ampersand replacement=\&,
  ] (M) {
    \PSM \& \encchanpsm(\PSM)\\
    \CSMl{\decchanfsm(A_\procA)} \& \CSM{A}\\
  };
  \draw[densely dashed,->]
    (M-2-1) edge node[annot]{implements} (M-1-1)
    (M-2-2) edge node[annot]{implements} (M-1-2)
    ;
  \draw[->]
    (M-1-2) edge[bend left] node[annot,num=2]{synthesis} (M-2-2)
    (M-1-1) edge node[annot,num=1]{$\encchanpsm(\hole)$} (M-1-2)
    (M-2-2) edge node[annot,num=3,xshift=1em]{$\decchanfsm(\hole)$} (M-2-1);
\end{tikzpicture}
 }
\vspace{-5ex}

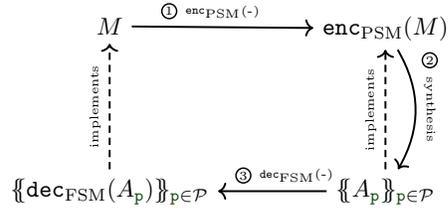
\captionof{figure}{Workflow of encoding.}
\label{fig:workflow-encoding}
\end{minipage}

\begin{restatable}{theorem}{implDecPSMsChoice}
\label{thm:implDecPSMsChoice}
Checking projectability of Tame PSMs is in \mbox{PSPACE}.
One can also synthesize a projection in \mbox{PSPACE}.
\end{restatable}

\subsection{Mixed Choice yields Undecidable Projectability}
\label{sec:mixed-choice-yields-undecidable-projection}

Now, we show that the sender-driven choice restriction for Tame PSMs is necessary for projectability to be decidable. 
General PSMs inherit undecidability of projectability from HMSCs,
which in turn was proven by~\citet[Thm.\,3.4]{DBLP:journals/tcs/Lohrey03}.
Given our positive result for Tame PSMs,
the proof for undecidability ought to break
in the presence of sender-driven choice.
The original proof goes through several (often implicitly given or omitted)
automata-based transformations
and does not give any insights about where and how the transformations
break under the assumption of sender-driven choice.

\begin{restatable}{theorem}{checkingImplSinkfinalMixedChoicePSMUndec}
\label{thm:checking-implementability-sinkfinal-mixed-choice-PSM-undecidable}
The projectability problem for \sinkfinal mixed-choice \sumOnePSMs is undecidable. 
\end{restatable}

We reduce the membership problem for Turing Machines to checking projectability of a \sinkfinal mixed-choice \sumOnePSM with five participants. Initially, there is a branching which only two participants are involved in and learn about. 
Subsequently, all participants communicate Turing machine computations in the form of configurations in both branches.
If the (projected) language of one of the other participants is not the same for both branches, 
the PSM cannot be implementable because they do not know which branch to comply with and easily deadlock.
We also show that the reverse is the case. 
Hence, we specify a language for each branch and make both coincide if and only if the Turing Machine has no accepting computation, which is the case if and only if the PSM is projectable. 

The full proof is in 
\appendixref{app:checking-implementability-mixed-choice-undec}.
We adopt the proof strategy of Lohrey to PSMs and 
make every transformation explicit and carefully check
which structural properties the transformations preserve,
yielding a stronger undecidability result concerning the most rudimentary of PSMs: \sumOnePSMs.

 \section{Typing Programs against CSMs}
\label{sec:typing-for-csms}

We now overview the key ideas behind AMP's type system.
The formal details and full proofs can be found in 
\appendixref{app:typing-for-csms}.
To define programs,
we take inspiration from the process calculus
with session interleaving and delegation of~\cite{DBLP:journals/pacmpl/ScalasY19}.
The syntax of AMP's programs is reproduced in
\cref{fig:prog-syntax}.
The processes~$P$ represent the static program text.
As is standard, $\zero$ is the terminated process, $\parallel$ denotes parallel composition,
$ \pn{Q}[\vec{c}] $ denotes a sequential process running the code
defined by a finite set of definitions~$\Defs$.
The prefixes~
  $\IntCh_{i \in I} c[\procB_i] ! \labelAndMsg{l_i}{c_i}$ and
  $\ExtCh_{i \in I} c[\procB_i] ? \labelAndVar{l_i}{y_i}$
denote internal and external choice respectively,
with a non-empty finite set of indices $I$.
The endpoint of participant $\procA$ of a channel between $\procA$ and $\procB$ in a session~$s$, is denoted by $s[\procA][\procB]$;
$\procA$ can send a label~$l$ and some payload~$p$ to $\procB$ in session~$s$
by $s[\procA][\procB] ! \labelAndMsg{l}{p}$,
the dual reception is denoted by $s[\procB][\procA] ? \labelAndMsg{l}{x}$
(which binds the payload to~$x$).
To model delegation, a process must be able to send to another the capability
to act as participant~$\procA$ in session~$s$, denoted $s[\procA]$; the receiving process will bind such capability to a variable~$x$ and use it to form endpoints~$x[\procB]$; we thus have in general send/receive actions
on $c[\procB]$ where $c$ can be a variable or some $s[\procA]$.

The process $(\restr s \hasType \CSMabb{A}) \, P$
denotes the creation of a new bound session~$s$ used in~$P$.
The session is annotated with a (computationally irrelevant) CSM $\CSMabb{A}$,
taking the place of what is often a global type.
So far, we treated messages in CSMs very abstractly
as elements of a finite alphabet.
In processes, messages are more structured:
they have a label (from a finite set) and a payload (of some type).
The messages used by the CSM will thus be pairs $\labelAndType{l}{t}$
of a label~$l$ and a payload type~$t$, with the convention that
if, from a state~$q$, there are two outgoing transitions with the same sender,
receiver and label, they will agree on the type.

In applications, the payload can be of any base type
(\eg integers, strings),
or be a channel capability $s[\procA]$ (for delegation).
Since supporting base types is a simple exercise, we follow
\cite{DBLP:journals/pacmpl/ScalasY19}
and focus on the harder case of channel capabilities as payloads.
When using a CSM $\CSMabb{A} = \CSM{A}$
as a protocol specification for a session~$s$,
it is natural to consider the (control) states~$Q_{\procA}$ of $A_{\procA}$
to be the local types that can be associated to~$s[\procA]$. Therefore, in our setting we will consider
the set~$L$ of the states of
any $\CSMabb{A}$ annotating the process, as the possible payload types.
For simplicity, we assume all CSMs use disjoint sets of states,
so that we can unambiguously refer to the transitions from any state~$q$
by $\delta(q)$.

In particular, if the protocol specified by~$\CSMabb{A}$ can delegate
channels of a session following some CSM~$\CSMabb{B}$,
then the message alphabet of $\CSMabb{A}$ will include states of $\CSMabb{B}$.
When the CSMs are obtained through projection, it is natural to first
obtain~$\CSMabb{B}$ so we can refer to its states in writing the PSM that projects to $\CSMabb{A}$.
We thus assume there is an acyclic ``delegation partial order''
between the CSMs of a process:
$\CSMabb{B} < \CSMabb{A}$ means that $\CSMabb{A}$ can use the states of $\CSMabb{B}$ in its messages. 

\begin{figure}[t]
    \adjustfigure[\small]
    \begin{grammar}
     c \is
            x
        |   s[\procA]
\\
     P \is
            \zero
        |   P_1 \parallel P_2
        |   (\restr s \hasType \CSMabb{A}) \, P
        |   \IntCh_{i \in I} c[\procB_i] ! \labelAndMsg{l_i}{c_i} \seq P_i
        |   \ExtCh_{i \in I} c[\procB_i] ? \labelAndVar{l_i}{y_i} \seq P_i
        |   \pn{Q}[\vec{c}]
        \\
     R \is
            \zero
        |   R_1 \parallel R_2
        |   (\restr s \hasType \CSMabb{A}) \, R
        |   \smash{\IntCh_{i \in I} c[\procB_i] ! \labelAndMsg{l_i}{c_i} \seq P_i}
        |   \smash{\ExtCh_{i \in I} c[\procB_i] ? \labelAndVar{l_i}{y_i} \seq P_i}
        |   \pn{Q}[\vec{c}]
        \\ & \hspace{5.77ex}
        |   \queueProc{s}{\queuecontent}
        |   \err
        \\
     \Defs \is
        \bigl(\pn{Q}[\vec{x}] =
        \smash{\IntCh_{i \in I} c[\procB_i]} ! \labelAndMsg{l_i}{c_i} \seq P_i\bigr); \; \Defs
        | \bigl(\pn{Q}[\vec{x}] =
        \smash{\ExtCh_{i \in I} c[\procB_i]} ? \labelAndVar{l_i}{y_i} \seq P_i\bigr); \; \Defs
        | \emptystring
    \end{grammar}
    \vspace{-1ex}
    \caption{Syntax of AMP's \picalc.}
    \label{fig:prog-syntax}
  \vspace{-2ex}
\end{figure}

The semantics of the calculus is defined on runtime configurations~$R$
(defined in \cref{fig:prog-syntax}),
which are processes which additionally contain message queues
$\queueProc{s}{\queuecontent}$ for each active session~$s$.
Here $\queuecontent$ is a map from pairs of participants to
sequences of messages.
The reduction semantics is standard (cf.~\appendixref{app:typing-for-csms}).
The only reduction rules we highlight here are the ones leading to an error
configuration:
{ 
\begin{proofrules}
\small
  \inferrule*{
\forall i \in I \st
          \queuecontent(\procB_i, \procA) = \labelAndMsg{l}{\_} \ldots
          \land
          l_i \neq l
  }{
          \ExtCh_{i \in I} s[\procA][\procB_i] ? \labelAndVar{l_i}{y_i} \seq P_i
          \parallel
          \queueProc{s}{\queuecontent}
              \redto
          \err
  }

  \inferrule*{
          \queuecontent(\procA, \procB) \neq \emptystring \text{ for some } \procA, \procB
  }{
          (\restr s \hasType \CSMabb{A}) \; \queueProc{s}{\queuecontent}
          \redto
          \err
  }
\end{proofrules}
}

\noindent 
The first rule models \emph{unsafe communication}:
a process is stuck
because all the queues it is waiting to receive from
are not empty, but the labels of the first messages do not match any of the cases the process is expecting.
The second rule models \emph{orphan messages}: a session where all participants terminated but that has still non-empty message queues.
The safety guarantees of our type system will rule out both cases.
Note that \cite{DBLP:journals/pacmpl/ScalasY19,NODBLP:ScalasY19techreport} focuses on communication safety.
In addition, they consider S-deadlock freedom, which implies no orphan messages, but is an undecidable property that needs to be checked and is not necessarily transferred to processes by the type system: the property only holds if there is only one session, in which case much stronger conditions transfer. 
In our setting, deadlock freedom is transferred throughout by projection and the type system, yielding no orphan messages.

\begin{figure}[t]
  \adjustfigure[\small]
  \begin{proofrules}
  \infer*[right=\procTypingZero]{
}{
      \typingContextOne
      \typingContextCat
      \emptyset
          \types
      \zero
  }

  \infer*[right=\procTypingEnd]{
      \typingContextOne
      \typingContextCat
      \typingContextTwo \types P \\
\EndState(q)
}{
      \typingContextOne
      \typingContextCat
      c \hasType q,
      \typingContextTwo
          \types
      P
  }

  \infer*[right=\procTypingExtCh]{
      \delta(q) =
      \set{(\rcv{\procB_i}{\procA}{\labelAndType{l_i}{p_i}}, q_i) \mid i \in I} \\
      \meta{\forall i \in I \st}
      \typingContextOne \typingContextCat
          \typingContextTwo,
          c \hasType q_i,
          y_i \hasType p_i
          \types P_i
  }{
      \typingContextOne
      \typingContextCat
      \typingContextTwo,
      c \hasType q
          \types
      \ExtCh_{i \in I} c[\procB_i] ? \labelAndVar{l_i}{y_i} \seq P_i
  }

  \infer*[right=\procTypingIntCh]{
\delta(q) \sups
      \set{(\snd{\procA}{\procB_i}{\labelAndType{l_i}{p_i}}, q_i) \mid i \in I}
      \\
\meta{\forall i \in I \st}
      \typingContextOne \typingContextCat
          \typingContextTwo , c \hasType q_i,
          \set{c_j \hasType p_j}_{j \in I\setminus \set{i}}
          \types P_i
  }{
      \typingContextOne
      \typingContextCat
      \typingContextTwo,
      c \hasType q,
      \set{c_i \hasType p_i}_{i \in I}
          \types
      \IntCh_{i \in I} c[\procB_i] ! \labelAndMsg{l_i}{c_i} \seq P_i
  }

  \infer*[right=\procTypingRestr]{
      \typingContextTwo_s =
          \set{s[\procA] \hasType \initialState(\CSMabb{A}_{\procA})}_{\procA \in \ProcsOf{\!\!\CSMabb{A}}} \\
\typingContextOne
          \typingContextCat
          \typingContextTwo,
          \typingContextTwo_s
          \types
          P
  }{
      \typingContextOne
          \typingContextCat
          \typingContextTwo
          \types
          (\restr s \hasType \CSMabb{A})\, P
  }
  \end{proofrules}
  \vspace{-1ex}
  \caption{
    Typing rules for processes; 
    $\initialState(\hole)$ denotes a CSM's initial state.
  }
  \label{fig:proc-typing-short}
  \vspace{-2ex}
\end{figure}

\Cref{fig:proc-typing-short} shows the crucial rules of AMP's type system.
The typing judgement
$
  \typingContextOne
  \typingContextCat
  \typingContextTwo
  \types P
$
uses
a process~$P$\!,
a typing context $\typingContextOne$
for the types of the parameters $\vec{c}$ of
sequential processes $\pn{Q}[\vec{c}]$
(the definitions of which are typed separately against $\typingContextOne$);
a typing context $\typingContextTwo$
associating the variables~$x$
and the channel capabilities~$s[\procA]$
occurring free in~$P$, with some CSM state $q\in L$.
Rule~\procTypingZero says that a terminated process
is typable in the environment with no capabilities.
Rule~\procTypingEnd permits to discard the capabilities that have terminated:
$\EndState(q)$ holds for final states with no outgoing receive transition.
Rules~\procTypingExtCh and~\mbox{\procTypingIntCh} deal with communication.
Assume $c=s[\p{p}]$.
According to \procTypingExtCh,
to receive a message as participant \p{p} in session~$c$,
we look for the type~$q$ of~$c$ in the typing context and check the CSM transitions $\delta(q)$ are all receives $\set{(\rcv{\procB_i}{\procA}{\labelAndType{l_i}{p_i}}, q_i) \mid i \in I}$.
Then the process needs to be able to receive any branch~$i$,
resulting in the continuation $P_i$ which is typed in the context
extended with the corresponding payload type $y_i:p_i$,
and with the type of~$c$ changed to~$q_i$.
According to \procTypingIntCh,
to send, as $c$ with type~$q$, a message non-deterministically picked from a number of branches~$i\in I$, we have to make sure $q$ allows each branch,
including matching the types of the payloads.
Then each branch~$i$ continues as $P_i$ which is typed in a context
where $c$ has type $q_i$ and we lost ownership of the payload~$c_i$.
Finally, \procTypingRestr types a new session~$s$ used by~$P$,
by adding to the context a new binding $s[\p{p}] : q_{\p{p}}$ for each
participant \p{p} of the CSM~$\CSMabb{A}$ annotating the session,
with $q_{\p{p}}$ being the initial state of~$\p{p}$ in~$\CSMabb{A}$.

The first correctness criterion for the type system is to prove
subject reduction: if a process is typable, then every configuration reachable from it will be typable.
Thus, to state subject reduction,
we need to define when a runtime configuration is typable.
For this purpose, we define a second judgement
$
  \typingContextOne
  \typingContextCat
  \typingContextTwo
  \typingContextCat
  \typingContextThree
  \types R
$
that includes a third typing context $\typingContextThree$
used to type session message queues:
associating to each channel $s[\p{p}][\p{q}]$ a sequence of message types
$ l(q) $ (label and payload type).
The key to make typing of runtime configurations an inductive invariant,
is the following rule:
\vspace{-1ex}
{ 
\[
\small
  \inferrule*[right=\runtimeTypingRestr]{
(\vec{q}, \xi) \in \reach(\CSMabb{A})
      \\
      \typingContextTwo^s_{\vec{q}} =
          \set{s[\procA] \hasType \vec{q}_\procA}_{\procA \in \ProcsOf{\CSMabb{A}}}
      \\
      \typingContextThree^s_{\xi} =
          \set{s[\procA][\procB] \hasQueueType \xi(\procA,\procB)}_{\procA,\procB \in \ProcsOf{\CSMabb{A}}}
      \\
\typingContextOne
          \typingContextCat
          \typingContextTwo,
          \typingContextTwo^s_{\vec{q}}
          \typingContextCat
          \typingContextThree,
          \typingContextThree^s_{\xi}
      \types
      R
}{
      \typingContextOne
          \typingContextCat
          \typingContextTwo
          \typingContextCat
          \typingContextThree
          \types
          (\restr s \hasType \CSMabb{A})\, R
  }
\]
}
\vspace{-1ex}

\noindent
The main difference between \runtimeTypingRestr and \procTypingRestr is
that the typing context is not populated with capabilities
associated to initial states; instead the prover can pick any CSM
configuration $(\vec{q}, \xi)$
---where $\vec{q}$ collects the local state of each participant, and $\xi$ the contents of the message queues---
that is reachable from the initial configuration of $\CSMabb{A}$.
The states and the queues are used to initialise the typing context
to type the process~$R$ using the restricted session.

In what follows we assume
the definitions~$\Defs$ can be typed according to~$\typingContextOne$.
We say a process/runtime configuration is \emph{well-annotated} if
every CSM appearing in it is
\begin{enumerate*}[(1)]
  \item deadlock-free, and
  \item satisfies feasible eventual reception.
\end{enumerate*}
Here, \emph{annotated} indicates that the CSMs have no computational meaning but \emph{well} shows the need for certain guarantees, which our type system can preserve. 
Note that a process is automatically well-annotated
if the CSMs are obtained via projection.

\begin{theorem}[Subject Reduction]
\label{th:subj-red-closed}
  Given a well-annotated~$R$,
  if $
    \typingContextOne
    \typingContextCat
    \emptyset
    \typingContextCat
    \emptyset
    \types R
  $ and $R \redto R'$,
  then
  $
    \typingContextOne
    \typingContextCat
    \emptyset
    \typingContextCat
    \emptyset
    \types R'.
  $
\end{theorem}

\begin{corollary}[Type Safety] 
\label{cor:safety}
  For a well-annotated~$R$,
  if $
    \typingContextOne
    \typingContextCat
    \emptyset
    \typingContextCat
    \emptyset
    \types R
  $ and \mbox{$R \redto^* R'$},
  then $R'$ does not contain~$\err$.
\end{corollary}

For progress, the situation is more delicate:
just like in \cite{DBLP:journals/pacmpl/ScalasY19} and most other MST systems, allowing session interleaving may introduce inter-session dependencies that are not modelled in the global protocol (which only pertains intra-session dependencies).
We thus prove progress under these assumptions:
\begin{enumerate*}[label=(\roman*)]
  \item there is only one session running, and
  \item that each of its participants is implemented by a distinct process, and
  \item the CSM annotating it is sink-final.
\end{enumerate*}
To encode these extra restrictions, we define a ``Session Fidelity'' variant of the typing judgement
$
  \typingContextOne
  \typingContextCat
  \typingContextTwo
  \typingContextCat
  \typingContextThree
  \typesSFs R
$
which uses a subset of the rules of $\types$ to enforce the restrictions above.
Let $\typingContextTwo^s_{\vec{q}}$ and $\typingContextThree^s_{\xi}$ be
defined as in the premises of~\runtimeTypingRestr.

\begin{theorem}[Progress]
\label{thm:progress}
  Let $(\vec{q}, \xi)$ be a configuration of a
  sink-final, deadlock-free
  CSM $\CSMabb{A}$ satisfying feasible eventual reception.
  If $
    \typingContextOne
    \typingContextCat
    \typingContextTwo^s_{\vec{q}}
    \typingContextCat
    \typingContextThree^s_{\xi}
    \typesSFs R
  $,
  and $ (\vec{q}, \xi) $ can take a step,
  then there exist some $R'$ and $ (\pvec{q}', \xi') $,
  such that
    $R \redto R'$,
    $(\vec{q}, \xi)\redto(\pvec{q}', \xi')$,
    and
    $
      \typingContextOne
      \typingContextCat
      \typingContextTwo^s_{\pvec{q}'}
      \typingContextCat
      \typingContextThree^s_{\xi'}
      \typesSFs R'.
    $
\end{theorem}

Progress hinges on deadlock freedom of the CSM. 
In general, any (language) property of a PSM that is preserved and reflected by $\interswaplang(\hole)$ holds for its projection. 
However, as for progress, it is not necessarily easy to make the type system enforce the preservation of these properties at the global process level and requires careful treatment. 
\cite{CoppoDYP16} demonstrated how Kobayashi-style techniques~\cite{Kobayashi06}
that can be used to show progress in the presence of session interleaving.
We conjecture a similar system can be
added on top of AMP's type~system.

 \section{Applications of AMP to MST Frameworks} 
\label{sec:reconstructing-global-types}

Standard (expression-based) global types from MST frameworks
can be seen as restricted special cases of PSMs.
What is gained from using AMP for global types seen as PSMs?
Is anything lost in doing so?
In this section, we evaluate AMP as a backend for projection/typing of
standard global types.
The key consequences of our results are:
\begin{enumerate}[labelindent=0pt,labelwidth=\widthof{(b)},label=\textnormal{(\alph*)},itemindent=0em,leftmargin=!] \item \label{item:key-consequence-a}
    Every sender-driven global type
    is a tame \sinkfinal \sumOnePSM.
  \item \label{item:key-consequence-b}
    Tame \sinkfinal \sumOnePSMs can be represented as a
    sender-driven global type.
  \item \label{item:key-consequence-c}
    Every collection of (expression-based) local types
    \CSM{L}
    can be expressed as a CSM
    \CSM{A}
    and vice versa.
  \item \label{item:key-consequence-d}
   AMP's projection is deadlock-free by construction,
    but MSTs typically insist on freedom of a stricter notion of deadlock
    which we call \emph{soft deadlock}. 
    We show AMP's projection can also be set
    to ensure soft deadlock freedom, without losing completeness.
\end{enumerate}

\noindent
These results help us settle two open questions:
\begin{itemize}
  \item
    Expression-based global/local types are equi-expressive with respect to
    state-machine-based global/local specifications.
  \item
    Allowing mixed-choice in global types makes
    projectability undecidable.
\end{itemize}

Here, we give an overview while 
\appendixref{app:reconstructing-global-types} presents the results in detail.

\paragraph{Global and local types.}
In most MST frameworks,
protocols are specified using expression-based global types~($G$),
which get projected to expression-based local types~($L$).
Their syntax is specified as follows:

\vspace{-2ex}
 \begingroup\small
  \begin{grammar}
    G \is
       \zero
     | \sum_{i ∈ I} \msgFromTo{\procA_i}{\procB_i}{\val_i} \seq G_i
     | μ X \seq G
     | X
    \quad
    \text{(global types)}
     \\
     L \is \zero
         | \IntCh_{i ∈ I} \snd{}{\procB_i}{\val_i} \seq L_i
         | \ExtCh_{i ∈ I} \rcv{\procB_{i}}{}{\val_i} \seq L_i
         | μ X \seq L
         | X
    \quad
    \text{(local types)}
  \end{grammar}
\endgroup
\noindent 
where 
$\zero$ explicitly denotes termination and 
$μ X$ binds a recursion variable $X$. The remaining operators specify how messages are exchanged: 
for local types, sending and receiving are separate actions,
while for global types they are specified in a single paired event.
Typically only \emph{deterministic} global types are considered,
\ie where
  every $\sum_{i ∈ I} \msgFromTo{\procA_i}{\procB_i}{\val_i} \seq G_i$
  has no $i\ne j$ with $ \msgFromTo{\procA_i}{\procB_i}{\val_i} = \msgFromTo{\procA_j}{\procB_j}{\val_j} $.
The choice restrictions we discussed, can be imposed on global types,
e.g.\ sender-driven choice requires that, for 
  $\sum_{i ∈ I} \msgFromTo{\procA_i}{\procB_i}{\val_i} \seq G_i$,
for all $i,j \in I$, $\procA_i=\procA_j$.

The standard semantics of global types has been given as a transition system,
or as sets of traces.
In both cases,
the semantics allows reordering of events that are not causally related,
\eg $\act{p->q:1}\seq \act{r->s:2} \seq \zero$
allows $\p{r}$ and $\p{s}$ to communicate before $\p{p}$ and $\p{q}$.
This is formalised, in the presentation of~\cite{DBLP:conf/concur/MajumdarMSZ21,DBLP:conf/ecoop/Stutz23,DBLP:conf/cav/LiSWZ23}
(which we adopt here)
as defining the semantics of a global type to be a set of traces
closed under the indistinguishability relation~$\interswap$.
With this view, it is immediate to represent any global type as a PSM.
Given the restricted format of global types,
the PSMs corresponding to translations of global types
(like the one in \cref{fig:example1-psm})
are \sumOnePSMs with a specific shape:
they are tree-like, \sinkfinal and recursion only happens at leaves and to ancestors~\cite{DBLP:conf/ecoop/Stutz23}.
On the face of it,
it is unclear whether every \sumOnePSM\ can be modelled as global type.

\begin{restatable}{theorem}{sumOnePSMasGlobalType}
\label{thm:sumOnePSMasGlobalType}
For every \sinkfinal \sumOnePSM \PSM, there is a global type $\semglobalsync(M)$ with the same core language (and hence the same semantics).
If \PSM is non-deterministic (mixed-choice, sender-driven, or directed, resp.), then $\semglobalsync(M)$ is non-deterministic (mixed-choice, sender-driven, or directed, resp.).
\end{restatable}

The main idea of \cref{thm:sumOnePSMasGlobalType} is that one can
see a global type as a special regular expression, and thus
we can adapt techniques like Arden's lemma and Brzozowski derivatives
to the case of PSMs.
The key difficulty in the proof
lays in showing the branching conditions are preserved:
the standard automata transformations change the branching structure,
and we need to produce new variants that do.

Similarly, local types can be directly read as the FSMs of a CSM.
We can also provide an inverse transformation (preserving branching).

\begin{restatable}{theorem}{localTypesEquiExpressive}
\label{lm:local-types-equi-expressive}
Let $A_\procA$ be a \sinkfinal FSM over~$\AlphAsync_\procA$ without mixed-choice states
for a participant $\procA$. 
One can construct a local type $L_\procA$ for $\procA$ with $\lang(L_\procA) = \lang(A_\procA)$. 
\end{restatable}

\paragraph{Deadlocks and protocol termination.}
In MSTs, local types can only terminate with a $\zero$,
which signals at the same time that it is valid to stop the protocol,
and that there is no further potential action.
This implies that for a global type to be projectable into local types,
all the participants need to know unambiguously when the protocol terminated globally.
In contrast, using CSMs,
it is possible to mark as final a state with outgoing transitions.
Consider for instance the (directed) global type
{
$
 \small
    \GG \is
(\msgFromTo{\procA}{\procB}{\val_1} \seq
        \msgFromTo{\procA}{\procC}{\val_1} \seq \zero)
+
        (\msgFromTo{\procA}{\procB}{\val_2} \seq \zero)
$}.
AMP's projection would produce the FSM
\scalebox{0.65}{\begin{tikzpicture}[psm, node distance=5em and 4em,baseline=-3]
    \node[state, accepting, initial left, initial text = ] (q0) {};
    \node[state, accepting, right= 5em of q0] (q1) {};
\path (q0) edge node [above] {$\rcv{\procA}{\procC}{\val_1}$} (q1);
\end{tikzpicture}
 }
as the projection for~$\procC$.
It contains a non-sink final state:
$\procC$ is not informed of which branch was taken and needs to be prepared
to terminate, or receive one more message.

AMP's projection produces deadlock-free CSMs, where deadlocks are defined
as configurations which cannot take a step, but 
their queues are not empty or some participant is in a non-final state.
Projections to local types ensure the absence of another type of deadlock:
a \emph{soft deadlock}, \ie a configuration that is a deadlock, or that
cannot take a step but where some participant is in a \emph{non-sink} state.
Is the possibility of soft deadlocks desirable?
We argue that this depends on the domain of application:
in distributed computing, it would be fine if a server kept listening for incoming requests while, in embedded computing, it can be key that all participants eventually~stop.
We can show that it is possible to use AMP in both scenarios, without giving up on completeness.

\begin{definition}[Strong Projectability]
A language $L \subseteq \AlphAsync^\omega$ is said to be \emph{strongly projectable} if there exists a CSM $\CSM{B}$ such that
$\CSM{B}$ is free from soft deadlocks
 \emph{(soft deadlock freedom)},
 and
$L$ is the language of $\CSM{B}$
 \emph{(protocol fidelity)}.
We say that $\CSM{B}$ is a strong projection of $L$.
\end{definition}

\begin{restatable}{theorem}{softImplSinkFinalRed}
\label{thm:soft-impl-sink-final-red}
Let $\GG$ be a projectable global type.
Then, the subset construction $\subsetcons{\GG}{\procA}$ 
\textnormal{\cite[Def.\,5.4]{DBLP:conf/cav/LiSWZ23}} is \sinkfinal for every participant $\procA$ if and only if there is a CSM that is a strong projection of $\GG$ and this CSM satisfies feasible eventual reception or every of its state machines is deterministic. \end{restatable}

If we aim for a strong projection of a projectable global type, we construct the global type's subset construction and check if it is \sinkfinal.
If it is not, there is no strong projection of it.
If this is undesirable, the protocol needs redesigning.
\cref{lm:local-types-equi-expressive} can yield local types and 
$\semlocal(L)$ is the FSM for a local type~$L$. 

\paragraph{Undecidability for mixed-choice.}
Finally, these results together with our results from
\cref{sec:mixed-choice-yields-undecidable-projection},
can settle the open question of whether we can project
mixed-choice global types algorithmically.

\begin{restatable}{corollary}{checkingImplMixedChoiceGlobalTypesUndec}
Both the projectability problem and the strong projectability problem for mixed-choice global types are undecidable.
\end{restatable}
 \section{Related Work}
\label{sec:related-work}

\paragraph{Multiparty session types.}

Inspired by linear logic \cite{DBLP:journals/tcs/Girard87},
\citet{DBLP:conf/concur/Honda93} proposed binary sessions types
for sessions with two participants.
Multiparty session types \cite{DBLP:conf/popl/HondaYC08} extended the idea to multiple participants.
\citet{DBLP:conf/esop/DenielouY12}
were the first to extensively explore the relation between CSMs and local types, but their projection is not complete and only supports directed choice;
moreover the approach was found to be somewhat flawed~\cite{DBLP:journals/pacmpl/ScalasY19}.
MSTs have been incorporated into a number of programming languages
\cite{DBLP:journals/ftpl/AnconaBB0CDGGGH16,DBLP:conf/icfp/JespersenML15,DBLP:conf/haskell/LindleyM16,DBLP:conf/icse/LangeNTY18,DBLP:conf/ecoop/ScalasY16,DBLP:conf/cc/NeykovaHYA18,DBLP:conf/ecoop/ChenBT22}.
They have also been applied to various other domains like
    operating systems~\cite{DBLP:conf/eurosys/FahndrichAHHHLL06},
web services~\cite{DBLP:conf/tgc/YoshidaHNN13},
    distributed algorithms~\cite{DBLP:journals/corr/abs-1902-01353},
    timed systems~\cite{DBLP:conf/esop/BocchiMVY19},
    cyber-physical systems~\cite{DBLP:conf/ecoop/MajumdarPYZ19}, and
    smart contracts~\cite{DBLP:conf/csfw/DasB0PS21}.
A number of works are devoted to mechanising MST meta-theory~\cite{DBLP:conf/ppdp/Thiemann19,DBLP:conf/pldi/Castro-Perez0GY21,dblp:journals/pacmpl/jacobsbk22a,dblp:journals/pacmpl/jacobsbk22}.
Our results could potentially extend the expressivity of the types involved in these applications.

\paragraph{MST projectability/projection.}
Via a reduction to globally-cooperative HMSCs, \citet{DBLP:conf/ecoop/Stutz23} proved MST projectability to be decidable for the class of global types that can ---but do not have to--- terminate (called $\zero$-reachable).
\citet{DBLP:conf/cav/LiSWZ23} provided a direct MST projection algorithm that is complete for sender-driven global types, providing a PSPACE upper bound. Our results use a reduction to these later developments.
The global specifications in \cite{DBLP:conf/cav/LiSWZ23}
can be shown to be special cases of Tame \sumOnePSMs
so our results strictly expand the reach of their results. 
For example, the protocols of 
\cref{fig:psm-no-hmsc-possible,fig:kindergarten-leader-election}
are all tame, yet out of the reach of both works. 
We also clarify 
the discrepancy between the notion of deadlock in global types and in 
\cite{DBLP:conf/cav/LiSWZ23} 
(cf.~\cref{sec:reconstructing-global-types}).
Finally, \cite{DBLP:conf/ecoop/Stutz23,DBLP:conf/cav/LiSWZ23}
do not have a type system, providing no way to link properties of projections with the ones of processes.
Preliminary versions of some of our results
appeared in the PhD thesis of the first author~\cite{DBLP:phd/dnb/Stutz24}.

The almost totality of asynchronous MST works can only handle directed choice.
An exception is~\cite{DBLP:journals/corr/abs-1203-0780},
where unrestricted global types are considered (without a type system).
They propose an incomplete projection algorithm that is correct with respect to
a different notion of correct projection than the standard one we adopt and generalise.
We refer to \cite{DBLP:conf/concur/MajumdarMSZ21} for a survey on choice restrictions.

\citet{DBLP:conf/fase/HuY17} propose a scheme with global types
and an incomplete projection,
where the global types are not safe by construction
and the restrictions on choice only appear at the local types.
The safety of global types is ensured by a combination of model-checking
with message buffers of size 1,
and syntactic restrictions that ensure that any unsafety
that might arise, will be visible in the 1-bounded executions.
For PSMs satisfying the syntactic restrictions,
the same approach could be applied.
The types of~\cite{DBLP:conf/fase/HuY17} also include connect/disconnect actions,
which can be emulated in AMP by excluding deadlocks (but not soft deadlocks)
and using non-sink final states.

\paragraph{Choreography automata and languages.}
Choreography Automata \cite{DBLP:conf/coordination/BarbaneraLT20} are syntactically similar to \sumOnePSMs,
but do not employ any closure operation, requiring the user to specify all the allowed interleavings, and preventing finite state representations for many common communication patterns.
In addition, \citet{DBLP:conf/concur/MajumdarMSZ21} showed that their conditions for projectability are flawed for the asynchronous case
(fixes for the synchronous case appeared in \cite[p.\,8]{DBLP:conf/ecoop/GheriLSTY22}).
\citet{DBLP:conf/coordination/BarbaneraLT22} applies a language-theoretic
approach to a limited class of \emph{synchronous} choreographies
(with no claim of completeness of projection).

\paragraph{Bottom-up MSTs.}
A number of MST-based works deviate from the top-down approach.
For instance, \cite{DBLP:journals/pacmpl/ScalasY19}
proposes a type system that only requires local types and not a global type.
The typing ensures some operational correspondence between local types and processes, making it possible to model check local types to determine properties of the program.
Their local types in the asynchronous setting are Turing-powerful, and therefore model checking is of limited use.
By virtue of the decoupling achieved by our type system,
AMP can be used in a bottom-up way too:
safety of the CSM used to type a process, implies safety of the process,
regardless of whether the CSM is obtained by projection or just given.
\citet{DBLP:conf/coordination/DagninoGD21} and \citet{DBLP:journals/corr/abs-2203-12876} also use a bottom-up strategy
by reconstructing
a so-called \emph{deconfined} global type
from the parallel composition of local programs
of a single session.
Deconfined global types are not automatically safe,
and checking their safety is shown to be undecidable.

\paragraph{Extensions for MSTs.}
A number of extensions for MSTs have been considered
(see~\cite{DBLP:journals/programming/BejleriDVEM19} for a survey),
including:
parametrisation~\cite{dblp:journals/scp/charalambidesda16,DBLP:journals/corr/abs-1208-6483},
dependent types~\cite{DBLP:conf/ppdp/ToninhoCP11,DBLP:journals/corr/abs-1208-6483,DBLP:conf/fossacs/ToninhoY18},
graduality~\cite{DBLP:journals/jfp/IgarashiTTVW19},
fault-tolerance~\cite{DBLP:journals/pacmpl/VieringHEZ21,DBLP:conf/concur/BarwellSY022}.
Context-free session types~\cite{dblp:conf/icfp/thiemannv16,dblp:journals/toplas/keizerbp22} specify binary sessions that are not representable with finite-control.
It would be interesting to consider projection for PSMs generated by pushdown~automata.

\paragraph{Distributed synthesis.}
In automata theory,
distributed synthesis seeks a way to transform a sequential specification into an equivalent distributed implementation, which is close in spirit
with the idea of extracting local types from global types.
One of the few positive results in this area is Zielonka's theorem~\cite{Zielonka87},
which shows that every regular trace language can be recognised by a so-called ``asynchronous automaton''.
Despite their name, asynchronous automata can be seen as a parallel composition
of participants interacting through \emph{synchronous} actions.
In contrast, PSMs and CSMs represent non-regular FIFO languages,
giving rise to a harder challenge.

\paragraph{High-level message sequence charts.}
HMSCs were defined in an industry standard~\cite{z120-standard},
inspiring extensive academic research~\cite{DBLP:conf/sdl/MauwR97,
DBLP:conf/ac/GenestMP03,DBLP:conf/acsd/GenestM05,DBLP:conf/concur/GazagnaireGHTY07,DBLP:journals/tosem/RoychoudhuryGS12}.
Projectability has been studied for HMSCs under the name ``safe realisability''
\cite{DBLP:journals/jcss/GenestMSZ06,DBLP:journals/tcs/Lohrey03,DBLP:journals/tcs/AlurEY05},
and was shown to be undecidable in general \cite{DBLP:journals/tcs/Lohrey03}.
Several restrictions of HMSCs have been proposed to make projectability decidable.
For a detailed survey, we refer to \cite{DBLP:conf/ecoop/Stutz23}.
Compared to PSMs, HMSCs only model finite runs; their PSM encoding equips them with an infinite run semantics.
With our developments, it is fairly straightforward to obtain a projection operation for  sender-driven, \sinkfinal \mbox{HMSCs} that respect some channel bounds.
This class is incomparable to any of the decidable HMSC classes proposed in the literature.
Since our type system only depends on CSMs, regardless of how they are obtained,
AMP can type check a program against a projectable HMSC.

\paragraph{Choreographic programming.}
Choreographic programming \cite{DBLP:journals/tcs/Cruz-FilipeM20,DBLP:conf/ecoop/GiallorenzoMPRS21,DBLP:journals/pacmpl/HirschG22} adopts the top-down approach even more radically than MSTs.
In choreographic programming, the endpoint projection (EPP) aims at synthesizing a fully-featured program implementation directly from the global specification.
As a result, the global specification describes the local computation alongside the communication structure (requiring infinite-state formalisms).
In choreographies, one typically works with non-finite-control-state global specifications, 
so the hopes for a complete and decidable EPP are slim, justifying giving up on completeness. By only considering local computation in processes, the MST/AMP approach avoids this issue.
Nevertheless, our results 
could still be useful for EPP
when applied to the pure communication structure of choreographies.
Notably, our method can project examples that cannot be projected using EPPs from the literature.
Consider the choreography
$
    \textbf{if } \procA.{\star}
        \textbf{ then }
            (\msgFromTo{\procA}{\procB}{\_} \seq \msgFromTo{\procB}{\procD}{\_})
        \textbf{ else }
            (\msgFromTo{\procA}{\procD}{\_} \seq
            \msgFromTo{\procD}{\procB}{\_})
$,
where message payloads are irrelevant and hence omitted 
and $\procA.{\star}$ denotes non-determi\-nistic choice by $\procA$.
The example is syntactically valid in \cite{DBLP:conf/sac/Cruz-FilipeMP18}
and can be easily encoded as a global type with sender-driven choice.
However, their EPP
would be undefined for $\procB$ and $\procD$:
it 
uses the merge from
\cite{DBLP:journals/toplas/CarboneHY12},
which can only merge same sender receives. Our results would instead produce the desired projection.

\paragraph{Communicating state machines.}
CSMs are the canonical automata model for distributed systems.
They have been studied in the context of model checking projections and do not apply a top-down methodology.
The verification problem is undecidable in general since CSMs are Turing-powerful \cite{DBLP:journals/jacm/BrandZ83}.
Several strategies to yield decidable classes have been proposed:
assuming channels are lossy \cite{DBLP:conf/cav/AbdullaBJ98}, restricting the communication topology~\cite{DBLP:journals/acta/PengP92, DBLP:conf/tacas/TorreMP08}, or
only allowing half-duplex communication for two participants~\cite{DBLP:journals/iandc/CeceF05}.
The concept of existential boundedness~\cite{DBLP:journals/fuin/GenestKM07} was initially defined for CSMs and yields decidability of control state reachability. The same holds for synchronisability~\cite{DBLP:conf/cav/BouajjaniEJQ18,DBLP:conf/fossacs/GiustoLL20}, which, intuitively, requires that every execution can be re-ordered (up to $\interswap$) into phases of sends and receives such that messages can only be received in the same phase.
Global types can only express $1$-synchronisable and half-duplex communication~\cite{DBLP:journals/corr/abs-2209-10328}. 

\phantomsection\label{paper-last-page}

\begin{credits}
\subsubsection{\ackname} The authors thank Anca Muscholl and Jorge A.\ Pérez for their discussions. 
This work was partially supported by the Luxembourg National Research Fund (FNR) under the grant agreement C22/IS/17238244/AVVA
  and
by the ERC Consolidator Grant for the project ``PERSIST'' under the EU's Horizon 2020 research and innovation programme (grant No.~101003349). \end{credits}

\clearpage
\bibliographystyle{splncs04nat}

\appendix

\iftoggle{arxiv}
{
\section{Additional Material for \cref{sec:automata-based-protocols}}
\label{app:automata-based-protocols}

\subsection{Communicating State Machines}
\label{app:csms}

\begin{definition}[Semantics of Communicating State Machines]
\label{app:semantics-csm}
We denote the set of channels with $\channels = \set{\channel{\procA}{\procB} \mid \procA,\procB\in \Procs, \procA\neq \procB}$.
The set of global states of a CSM is given by $\prod_{\procA \in \Procs} Q_\procA$.
Given a global state $q$, $q_\procA$ denotes the state of $\procA$ in $q$.
A~\emph{configuration} of a CSM $\CSM{A}$ is a pair $(q, \xi)$, where $q$ is a global state and
$\xi : \channels \rightarrow \MsgVals^*$ is a mapping of each channel to its current content.
The initial configuration $(q_0, \xi_\emptystring)$ consists of  a global state $q_0$ where the state of each role is the initial state  $q_{0,\procA}$ of $A_\procA$ and a mapping~$\xi_\emptystring$, which maps each channel to the empty word~$\emptystring$.
A~configuration $(q, \xi)$ is said to be \emph{final} iff each individual local state $q_\procA$ is final for every $\procA$ and $\xi$ is $\xi_{\emptystring}$.

The global transition relation $\rightarrow$ is defined as follows:
\vspace{-1ex}
\begin{itemize}
\item
$(q,\xi) \xrightarrow{\snd{\procA}{\procB}{\val}} (q',\xi')$ if
$(q_\procA, \snd{\procA}{\procB}{\val}, q'_\procA)\in\delta_\procA$,
$q_\procC = q'_\procC$ for every role $\procC \neq \procA$,
$\xi'(\channel{\procA}{\procB}) =  \xi(\channel{\procA}{\procB})\cdot\val$ and $\xi'(c) = \xi(c)$ for every other channel $c\in \channels$.

\item
$(q,\xi) \xrightarrow{\rcv{\procA}{\procB}{\val}} (q',\xi')$ if
$(q_\procB, \rcv{\procA}{\procB}{\val}, q'_\procB)\in\delta_\procB$,
$q_\procC = q'_\procC$ for every role $\procC \neq \procB$,
$\xi(\channel{\procA}{\procB}) = \val\cdot \xi'(\channel{\procA}{\procB})$
and $\xi'(c) = \xi(c)$ for every other channel $c\in \channels$.

\item
$(q,\xi) \xrightarrow{\emptystring} (q',\xi)$ if
$(q_\procA, \emptystring, q'_\procA)\in\delta_\procA$ for some role
$\procA$, and
$q_\procB = q'_\procB$ for every role $\procB \neq \procA$.
\end{itemize}
\noindent
A run of the CSM always starts with an initial configuration $(q_0, \xi_0)$, and is a finite or infinite sequence
$(q_0, \xi_0) \xrightarrow{w_0} (q_1, \xi_1) \xrightarrow{w_1} \ldots$
for which  $(q_i,\xi_i) \xrightarrow{w_i} (q_{i+1},\xi_{i+1})$.
The word $w_0w_1 \ldots \in\Sigma^\infty$ is said to be the \emph{trace} of the run.
A run is called maximal if it is infinite or finite and ends in a final configuration. As before, the trace of a maximal run is maximal.
The language $\lang(\CSM{A})$ of the CSM $\CSM{A}$ consists of its set of maximal traces.
A configuration is a deadlock if it is not final and has no outgoing transitions.
\end{definition}

\citet[Lm.\,22]{DBLP:conf/concur/MajumdarMSZ21} showed that the semantics of CSMs are closed under the indistinguishability relation $\interswap$.

\begin{lemma}
\label{lm:csm-closed-interswap}
Let $\CSM{A}$ be a CSM.
Then, $\lang(\CSM{A}) = \interswaplang(\lang(\CSM{A}))$.
\end{lemma}

\subsection{Feasible Eventual Reception is Preserved and Reflected by the Indistinguishability Closure}
\label{app:FER-reflected-and-preserved}

\begin{lemma}[Feasible eventual reception is preserved and reflected by $\close(\hole)$]
\label{lm:FER-reflected-and-preserved}
Let $L \subseteq \AlphAsync^\infty$ be a language.
It holds that $L$ satisfies feasible eventual reception
iff
$\close(L)$ satisfies feasible eventual reception.
\end{lemma}
\begin{proof}
Let us recall the definition of feasible eventual reception:
for every finite word $w \is w_1 \ldots w_n \in \pref(L)$ such that $w_i$ is an unmatched send event, there is an extension $w \preforder w' \in L$ of $w$ such that $w_i$ is matched in $w'$\negthinspace.

For the direction from left to right, it suffices to observe the following.
For every (finite) prefix $u \in \pref(\close(L))$ with an unmatched send event, there is $w \in L$ with $u \interswap w$ and, by assumption, there is $w \preforder w'$ where this send event is matched.
For finite words, $\interswap$ does not affect matching (as it only considers the send and receive events for one channel).
It is obvious that we can also use $w'$ as witness for $u$.

Let us consider the direction from right to left.
We need to show that for every finite word $w \is w_1 \ldots w_n \in \pref(L)$ such that $w_i$ is an unmatched send event, there is an extension $w \preforder w' \in L$ of $w$ such that $w_i$ is matched in $w'$\negthinspace.
By assumption, we know there is $u' = w_1 \ldots w_n \ldots \in \close(L)$ with $w \preforder u'$ where $w_i$ is matched.
We  do a case distinction whether $u'$ is finite.
If $u'$ is finite, there is $u \in L$ with $u \interswap u'$ by definition.
It is straightforward that $w_i$ is also matched in $u$, proving the claim.
Let us suppose that $u'$ is infinite.
Despite, by definition, there is $j$ such that $w_j$ is the matching receive event for $w_i$.
(It is obvious that $j > n$.)
By definition of $\close(\hole)$, there is $u \in L$ with $u' \preceq_\interswap^\omega u$.
This means for each finite prefix $v' \negthinspace$ of $u'$, there is a finite prefix~$v$ of~$u$ such that
$v' \cat v'' \interswap v$ for some $v''$.
We choose $v' = w_1 \ldots w_n \ldots w_j$.
Hence, we can obtain $v \in \pref(L)$ where $w_j$ occurs and matches $w_i$ because, as argued before, $\interswap$ does not affect matching for finite words.
For the definition of feasible eventual reception, we need a word in $L$.
We choose $u$, which is a continuation of $v$, where $w_j$ is still matched.
\proofEndSymbol
\end{proof}

 \section{Additional Material for \cref{sec:projection-dec-undec}}
\label{app:projection-dec-undec}

Without loss of generality,
in what follows we assume there are no $\emptystring$-transitions in a PSM $\PSM$ if $\lang(\PSM) = \set{\emptystring}$.

\subsection{Finding Channel Bounds for PSMs}
\label{sec:checking-tameness-for-psms}

Our algorithm to generate channel bounds consists of three phases. 

First, we detect all channels for which we need a channel bound. 
For this, we consider every transition with label 
$(q_1, \snd{\procA}{\procB}{\val}, q_2)$ 
and check if there is only a single transition from $q_2$ such that
$(q_2, \snd{\procA}{\procB}{\val}, q_3)$. 
If not, the channel $\channel{\procA}{\procB}$ needs to have a channel bound. 
We call these \emph{detected} channels.

Second, we check if the channel bound can be upperbounded.
For every detected channel $\channel{\procA}{\procB}$, we consider every loop with a message from $\procA$ to $\procB$ and check the following. 
We check if there is a sequence of messages from $\procB$ back to $\procA$ (possibly involving other intermediary participants), 
i.e., does the loop have a subword of the form \[
\snd{\procB}{\procC_1}{m_1} \cat
\rcv{\procB}{\procC_1}{m_1} \cat 
\snd{\procC_1}{\procC_2}{m_2} \cat 
\rcv{\procC_1}{\procC_2}{m_2} \cat 
\ldots 
\snd{\procC_n}{\procA}{m_n} \cat 
\rcv{\procC_n}{\procA}{m_n} 
\text{ for } n \geq 0.
\]
where subword is defined in the expected way. 

Intuitively, the condition enforces that any message sent in one iteration from $\procA$ to $\procB$ needs to be received before a message from $\procA$ to $\procB$ will be sent in the next iteration of the loop.
Hence, if satisfied, $\procA$ will not be able to go ahead with loop executions without $\procB$ receiving its messages. 
If this condition is violated for any detected channel, the PSM will not respect any channel bounds.
This condition is inspired from bounded HMSCs \cite{DBLP:conf/concur/AlurY99,DBLP:conf/mfcs/MuschollP99}, where they require the communication topology of every loop to form a strongly connected component, yielding that all channels will always be bounded.

Third, we compute the channel bounds.
In the second step, we ensured that loops will not add unmatched messages to the channels. 
Hence, we do not need to account for loops in the computation of channel bounds. 
We consider every loop-free path from the initial state. 
For every detected channel, we compute the number of messages in the channel and the maximum over all loop-free paths yields the channel bound.
We shall not restrict ourselves to paths ending in final states as we also consider infinite executions and there might be loop-free paths for which no extension to a final state exists.

\subsection{Projection Synthesis and Projectability for a Class of PSMs}
\label{sec:channel-participant-encoding}

\citet{DBLP:conf/cav/LiSWZ23} propose a sound and complete algorithm to generate projections of global types from MSTs.
Transferred to our setting, this amounts to \sinkfinal sender-driven \sumOnePSMs.
We propose an encoding that encodes PSMs as \mbox{\sumOnePSMs.}
This encoding preserves choice restrictions and is \sinkfinal if and only if the original PSM is.

With \cref{ex:revisiting-KLE},
we gave intuition about the channel participant encoding for
the KLE example.
\cref{fig:KLE-implementations} provides the respective projection and decoded FSM for $\p{e}$.

\begin{figure}[t]
\begin{subfigure}[b]{0.49\textwidth}
\centering
\resizebox{0.97\textwidth}{!}{
 \begin{tikzpicture}[csm, node distance=.5em and 4em]
  \node[state,init,label=above left:\p{e}](s){};
  \node[state,above right=of s] (s0) {};
  \node[state,right=5em of s0]      (sE) {};
  \node[state,below right=of s] (s1) {};
  \node[state,right=5em of s1]      (sO) {};
  \node[state,below right=of sE,final](end){};
  \draw
    (s)  edge node[sendchan=e to o:0]{} (s0)
    (s0) edge[out=-45,in=120] node[recvchan=o from e:1]{} (sO)
    (s0) edge[bend left=10pt] node[recvchan=o from e:0]{} (sE)
    (s)  edge node[down,sendchan=e to o:1]{} (s1)
    (s1) edge[out=45,in=-120] node[down,recvchan=o from e:1]{} (sE)
    (s1) edge[bend right=10pt] node[down,recvchan=o from e:0]{} (sO)
(sE)  edge node[recvchan=o from e:\win]{} (end)
    (sO)  edge node[down,sendchan=e to o:\win]{} (end)
  ;
\end{tikzpicture}
 }
 \caption{Local projection of \p{e} for KLE after Channel-participant Encoding.}
 \label{fig:KLE-impl-chan-part-enc}
\end{subfigure}
\begin{subfigure}[b]{0.49\textwidth}
\centering
\resizebox{0.97\textwidth}{!}{
 \begin{tikzpicture}[csm, node distance=.5em and 3em]
  \node[state,init,label=above left:\p{e}](s){};
  \node[state,above right=of s] (s0) {};
  \node[state,right=5em of s0]      (sE) {};
  \node[state,below right=of s] (s1) {};
  \node[state,right=5em of s1]      (sO) {};
  \node[state,below right=of sE,final](end){};
  \draw
    (s)  edge node[send=e to o:0]{} (s0)
    (s0) edge[out=-45,in=120] node[recv=e from o:1]{} (sO)
    (s0) edge[bend left=10pt] node[recv=e from o:0]{} (sE)
    (s)  edge node[down,send=e to o:1]{} (s1)
    (s1) edge[out=45,in=-120] node[down,recv=e from o:1]{} (sE)
    (s1) edge[bend right=10pt] node[down,recv=e from o:0]{} (sO)
(sE)  edge node[recv=e from o:\win]{} (end)
    (sO)  edge node[down,send=e to o:\win]{} (end)
  ;
\end{tikzpicture} }
 \caption{Local projection of \p{e} for KLE. \\ \phantom{empty line}}
 \label{fig:KLE-impl}
\end{subfigure}
\caption{Local projections for participant \p{e} for KLE variants.}
\label{fig:KLE-implementations}
\end{figure}
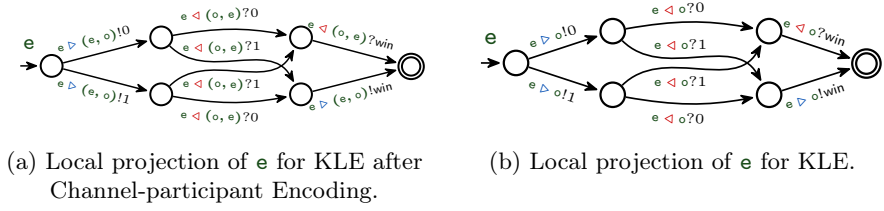

We also hinted at the fact that only having one channel participant per channel can be unsound.
The following example shows this.

\begin{example}
\label{ex:example-issue-with-only-one-channel-participant}
Consider the word
$
\snd{\procA}{\procB}{\val} \cat
\snd{\procA}{\procB}{\val} \cat
\rcv{\procA}{\procB}{\val} \cat
\rcv{\procA}{\procB}{\val}
$.
For now, let us say that $\procC$ acts as channel participant from $\procA$ to $\procB$.
Then, the encoded word would be
$
\msgFromTo{\procA}{\procC}{\val} \cat
\msgFromTo{\procA}{\procC}{\val} \cat
\msgFromTo{\procC}{\procB}{\val} \cat
\msgFromTo{\procC}{\procB}{\val}
$.
While the original version allows us to swap the 2nd and 3rd event:
$
\snd{\procA}{\procB}{\val} \cat
\rcv{\procA}{\procB}{\val} \cat
\snd{\procA}{\procB}{\val} \cat
\rcv{\procA}{\procB}{\val}
$,
this is not possible for the encoded word as $\procC$ induces an ordering between them.
Intuitively, this means that our encoding would not consider certain scenarios, possibly leading to unsoundness.
\end{example}

To avoid this, intuitively, we would want to introduce a new channel participant for every message.
However, we only know how to check projectability for a finite number of participants.
Thus, we introduce the notion of \emph{channel bounds} that allows us to restrict the number of channel participants.
Intuitively, they restrict the number of messages in flight in a channel.
This allows re-use of the channel participants for later messages.
Obviously, if we imposed this on every channel, we would end up with finite state systems.
Thus, we distinguish between channels for which we need channel participants and for which we do not need them: 
channel $\channel{\procA}{\procB}$
does not need a channel bound if all its
actions of shape $\snd{\procA}{\procB}{\val}$
are immediately followed by actions of shape
$\rcv{\procA}{\procB}{\val}$.

We now formalise our encoding and prove it correct.

Formally, we introduce channel alphabets to define channel bounds.

\begin{definition}[Channel alphabets]
For readability, we define the alphabet of a channel $\channel{\procA}{\procB} \in \channels$:
$
  \AlphAsync_{\channel{\procA}{\procB}} \is
    \AlphAsync_{{\channel{\procA}{\procB}}, !}
    \dunion
    \AlphAsync_{{\channel{\procA}{\procB}}, ?}
$
where
$
    \AlphAsync_{{\channel{\procA}{\procB}}, !} \is
        \set{ \snd{\procA}{\procB}{\val} \mid \val \in \MsgVals}
$ and $
    \AlphAsync_{{\channel{\procA}{\procB}}, ?} \is
        \set{ \rcv{\procA}{\procB}{\val} \mid \val \in \MsgVals}
$.
\end{definition}

Note that we already used these alphabets to define channel bounds in the main text but did not use a special symbol for them.
Here, we also need to know when a word satisfies channel bounds.

\begin{definition}[More on channel bounds] We say a word $w \subseteq \AlphAsync^\infty$ respects the channel bounds $\chanBounds$ if
for every channel $\channel{\procA}{\procB} \in \channels$, the following holds:
if $\chanBounds(\channel{\procA}{\procB})$ is defined, then
 $w \wproj_{\AlphAsync_{\channel{\procA}{\procB}}}$ is
$\chanBounds(\channel{\procA}{\procB})$-bounded.
\end{definition}

\begin{proposition}
\label{prop:respecting-channel-bounds-from-psm-to-semantics}
 If a PSM $\PSM$ respects channel bounds $\chanBounds$, then every word in its semantics does.
\end{proposition}

\begin{remark}[Notation]
Intuitively, if a PSM respects channel bounds $\chanBounds$, for channels that are not in the domain of $\chanBounds$, the reception of a message is specified right after it was sent. 
We may call such events \emph{paired} and will abuse notation: we assume transitions of PSMs are labelled with $\AlphSync$ if and only if the respective channel has an undefined channel bound in $\chanBounds$.
\end{remark}

Given a channel bound for a channel, we know the number of messages in this channel will never exceed this bound.
Intuitively, we can therefore treat this channel as a ring buffer with a producer, i.e.\ the sender, and a consumer, i.e.\ the receiver.
For both, we will keep track which is the next position to send to or to receive from.
To this end, we use the additive group of integers modulo $\chanBounds(\channel{\procA}{\procB})$ for every channel $\channel{\procA}{\procB}$:
$(\Int_{\chanBounds(\channel{\procA}{\procB})}, +)$.
It is common to write
$5 \modOp 2 \modSuff \chanBounds(\channel{\procA}{\procB})$ for
$\chanBounds(\channel{\procA}{\procB}) = 3$.
In our setting,
$\chanBounds(\channel{\procA}{\procB})$
will always be clear from context.
Therefore, we will omit it.

\begin{definition}[Channel participants]
Given channel bounds $\chanBounds$, we define a set of \emph{channel participants} for the channels in the domain of $\chanBounds$:
\[
    \ProcsCh \is \Union_{0 \leq \roleFmt{i} < \chanBounds(\channel{\procA}{\procB})}
                    \set{\procChanI{\procA}{\procB}{i}}
\quad \text{ and } \quad
    \ProcsWCh \is \Procs \dunion \ProcsCh
\]
to obtain the new set with both the original participants and the channel participants.
The set of events $\AlphAsyncWCh$ for $\ProcsWCh$ is defined as follows:
{ \small
\begin{align*}
 \AlphAsyncWCh \is \set{
    & \snd{\procA}{\procChan{\procA}{\procB}}{\val} \;,\;
      \rcv{\procA}{\procChan{\procA}{\procB}}{\val} \;,\;
      \snd{\procChan{\procA}{\procB}}{\procB}{\val} \;,\;
      \rcv{\procChan{\procA}{\procB}}{\procB}{\val} \\
    & \mid \procA, \procB ∈ \Procs, \procChan{\procA}{\procB} \in \ProcsCh \text{ and } \val ∈ \MsgVals},
\end{align*}
}
Again, we define syntactic sugar for paired send and receive events:
\begin{align*}
    \msgFromTo{\procA}{\procChan{\procA}{\procB}}{\val}
    & \is
    \snd{\procA}{\procChan{\procA}{\procB}}{\val} \cat
    \rcv{\procA}{\procChan{\procA}{\procB}}{\val}
    \quad \text{ and } \quad
    \\
    \msgFromTo{\procChan{\procA}{\procB}}{\procB}{\val}
    & \is
    \snd{\procChan{\procA}{\procB}}{\procB}{\val} \cat
    \rcv{\procChan{\procA}{\procB}}{\procB}{\val}
\end{align*}
for
$\procA, \procB ∈ \Procs, \procChan{\procA}{\procB} \in \ProcsCh \text{ and } \val ∈ \MsgVals$.
The set of these is denoted by $\AlphSyncWCh\negthinspace$.
Note that, in $\AlphAsyncWCh\negthinspace$, participants do only send to and receive from the respective channel participants.
Consequently, the alphabet for channel participants is
$
    \AlphAsyncWCh_{\procChan{\procA}{\procB}}
    \is
    \set{
        \snd{\procChan{\procA}{\procB}}{\procB}{\val},
        \rcv{\procA}{\procChan{\procA}{\procB}}{\val}
        \mid \val \in \MsgVals
    }
$
while the alphabet for $\procA \in \Procs$ is
{ \small
\begin{align*}
    \AlphAsyncWCh_{\procA}
    \is
    \set{
        \snd{\procA}{\procChan{\procA}{\procB}}{\val},
        \rcv{\procChanI{\procB}{\procA}{j}}{\procA}{\val}
        \mid \;
        & \val \in \MsgVals, \procB \in \Procs,
        0 \leq i < \chanBounds(\channel{\procA}{\procB}),
        \\ &
        \text{ and }
        0 \leq j < \chanBounds(\channel{\procB}{\procA})
    }
    \enspace .
\end{align*}
}
\end{definition}

\begin{definition}[Channel-ordered]
We say a word $w$ over $\AlphSyncWCh$ is \emph{channel-ordered} if,
for every channel $\channel{\procA}{\procB} \in \channels$,
the following holds:
\begin{itemize}
 \item Let $w'_0 w'_1 \ldots \is w \wproj_{\msgFromTo{\procA}{\procChanI{\procA}{\procB}{\_}}{\_}}$.
 It holds that $w'_i = \msgFromTo{\procA}{\procChanI{\procA}{\procB}{j}}{\_}$ for
 $\roleFmt{j} \modOp i$ for all $i$.
 \item Let $w'_0 w'_1 \ldots \is w \wproj_{\msgFromTo{\procChanI{\procA}{\procB}{\_}}{\procB}{\_}}$.
 It holds that $w'_i = \msgFromTo{\procChanI{\procA}{\procB}{j}}{\procB}{\_}$ for
 $\roleFmt{j} \modOp i$ for all $i$.
\end{itemize}
We say that a word $w$ over $\AlphAsyncWCh$ is channel-ordered if there is channel-ordered word $w' \in (\AlphSyncWCh)^\infty$ such that $w' \interswap w$.

For participant $\procA$, we say a word $w$ over $\AlphSyncWCh_\procA$ is \emph{channel-ordered} if,
for every channel $\channel{\procA}{\procB} \in \channels$,
the following holds:
\begin{itemize}
 \item Let $w'_0 w'_1 \ldots \is w \wproj_{\snd{\procA}{\procChanI{\procA}{\procB}{\_}}{\_}}$.
 It holds that $w'_i = \snd{\procA}{\procChanI{\procA}{\procB}{j}}{\_}$ for
 $\roleFmt{j} \modOp i$ for all $i$.
 \item Let $w'_0 w'_1 \ldots \is w \wproj_{\rcv{\procChanI{\procA}{\procB}{\_}}{\procB}{\_}}$.
 It holds that $w'_i = \rcv{\procChanI{\procA}{\procB}{j}}{\procB}{\_}$ for
 $\roleFmt{j} \modOp i$ for all $i$.
\end{itemize}
\end{definition}

Equipped with these definitions, we can define the channel participant encoding for a PSM that respects channel bounds, for which we briefly give an intuition.

We introduce channel participants
$\procChanI{\procA}{\procB}{i}$
for every channel $\channel{\procA}{\procB} \in \domainOf(\chanBounds)$,
concretely $\chanBounds(\channel{\procA}{\procB})$ of them.
From a PSM~$\PSM$, we obtain an encoded PSM $\encchanpsm(\PSM)$ by substituting every send and receive event by the respective (paired) event involving the channel participants,
\eg
we change 
$\snd{\procA}{\procB}{\val}$
into 
$\msgFromTo{\procA}{\procChanI{\procA}{\procB}{i}}{\val}$
if it is the $i^{\text{th}}$ send event.
Provided that we can obtain a projection $\CSMwCh{A}$ of $\encchanpsm(\PSM)$,
we can compute a CSM by rebending the messages to and from channel participants, yielding the decoding function $\decchanfsm(A_\procA)$ for every $\procA$.
Note that, in such an encoded PSM, a participant $\procA$ only communicates with channel participants,
\ie all its communications are of form
$\snd{\procA}{\procChanI{\procA}{\_}{i}}{\_}$
or
$\rcv{\procChanI{\_}{\procA}{i}}{\procA}{\_}$.
For channel participants, it is even more restricted: a channel participant $\procChan{\procA}{\procB}$ solely receives from $\procA$ and sends to $\procB$.

\begin{definition}[Channel-participant Encoding]
\label{def:channel-participant-encoding}
The channel-participant encoding for a PSM $\PSM$ that respects channel bounds $\chanBounds$ is a $6$-tuple
\[
\left( \;
    \encchanpsm,
    \encchanwords,
    \decchanwords,
    \encchanfsm,
    \decchanfsm,
    A_{(\procA, \procB)} \text{ for } (\procA, \procB) \in \ProcsCh
\; \right)
\]
with the following definitions:
\begin{itemize}
 \item $\encchanpsm$ turns a PSM into a \sumOnePSM:
    \[
    \encchanpsm((Q, \AlphAsync, \delta, q_{0}, F)) \is
    (Q', \AlphSyncWCh, \delta', q'_0, F')
    \qquad \text{where}
    \]
    \begin{itemize}
     \item $Q' \is
                Q \times
                (\set{\channel{\procA}{\procB} \to \Int_{\chanBounds(\channel{\procA}{\procB})}}_{\channel{\procA}{\procB} \in \domainOf(\chanBounds)})^2 $,
     \item $q'_0 \is
            \left( q_{0}, (\set{\channel{\procA}{\procB} \mapsto 0}_{\channel{\procA}{\procB} \in \domainOf(\chanBounds)})^2 \right)$,
     \item $F' \is
            F \times (\set{\channel{\procA}{\procB} \mapsto 0}_{\channel{\procA}{\procB} \in \domainOf(\chanBounds)})^2$,
    \end{itemize}
    and the transition relation $\delta'$ is defined as follows:
    \begin{itemize}
     \item $(
                (q_1, \chanBoundsMapSnd, \chanBoundsMapRcv),
                \; x, \;
                (q_2, \chanBoundsMapSnd, \chanBoundsMapRcv)
            ) \in \delta'$ if
            $x \in \AlphSync$ and
            $(q_1, \; x, \; q_2) \in \delta$,
     \item $(
                (q_1, \chanBoundsMapSnd, \chanBoundsMapRcv),
                \; \msgFromTo{\procA}{\procChanI{\procA}{\procB}{j}}{\val}, \;
                (q_2,
                \chanBoundsMapSnd[\channel{\procA}{\procB} \mapsto \chanBoundsMapSnd(\channel{\procA}{\procB}) + 1],
                \channel{\procA}{\procB}), \chanBoundsMapRcv)
            ) \in \delta'$ if \\
            $(q_1, \; \snd{\procA}{\procB}{\val}, \; q_2) \in \delta$ and
            $\roleFmt{j} = \chanBoundsMapSnd(\channel{\procA}{\procB})$,
     \item $(
                (q_1, \chanBoundsMapSnd, \chanBoundsMapRcv),
                \; \msgFromTo{\procChanI{\procB}{\procA}{j}}{\procA}{\val}, \;
                (q_2, \chanBoundsMapSnd,
                \chanBoundsMapRcv[\channel{\procB}{\procA} \mapsto \chanBoundsMapRcv(\channel{\procB}{\procA}) + 1]
            ) \in \delta'$ if \\
            $(q_1, \; \rcv{\procB}{\procA}{\val}, \; q_2) \in \delta$ and
            $\roleFmt{j} = \chanBoundsMapRcv(\channel{\procB}{\procA})$.
    \end{itemize}

 \item $\encchanwords \from \AlphAsync^\infty \to (\AlphSyncWCh)^\infty$: \[
        \encchanwords(w_1 w_2 \ldots) \is
            \encchanwords(w_1) \cat
            \encchanwords(w_2) \ldots
    \quad \text{where }
    \]
    \vspace{-2.5ex}
    {
    \small
    \[
    \encchanwords(w_i) \is
    \begin{cases}
        \msgFromTo{\procA}{\procChanI{\procA}{\procB}{j}}{\val}
            & \text{if } w_i = \snd{\procA}{\procB}{\val}
            \text{ and } \roleFmt{j} \modOp \card{(w_1 \ldots w_{i-1}) \wproj_{\snd{\procA}{\procB}{\_}}}
            \\
        \msgFromTo{\procChanI{\procA}{\procB}{j}}{\procB}{\val}
            & \text{if } w_i = \rcv{\procA}{\procB}{\val}
            \text{ and } \roleFmt{j} \modOp \card{(w_1 \ldots w_{i-1}) \wproj_{\rcv{\procA}{\procB}{\_}}}
            \\
        w_i
            & \text{if }
            (w_i = \snd{\procA}{\procB}{\val}
            \text{ or }
            w_i = \rcv{\procA}{\procB}{\val}) \\
            & \quad
            \text{ and } \chanBounds(\channel{\procA}{\procB})
            \text{ is undefined}
    \end{cases}
\enspace .
    \]
    }
 \item $\decchanwords \from (\AlphSyncWCh)^\infty \to \AlphAsync^\infty$:
    \[
        \decchanwords(w_1 w_2 \ldots) \is
        \decchanwords(w_1) \cat \decchanwords(w_2) \ldots
    \quad \text{where }
    \]
    \vspace{-2.5ex}
    {
    \small
    \[
    \decchanwords(x) \is
    \begin{cases}
    \snd{\procA}{\procB}{\val}
        & \text{if } x = \msgFromTo{\procA}{\procChan{\procA}{\procB}}{\val} \\
\rcv{\procA}{\procB}{\val}
        & \text{if } x = \msgFromTo{\procChan{\procA}{\procB}}{\procB}{\val} \\
    \end{cases}
\enspace .
    \]
    }
 \item $\encchanfsm$ that turns an FSM over $\AlphAsync_\procA$ to an FSM over $\AlphAsyncWCh_\procA$:
    \[
        \encchanfsm((Q, \AlphAsync_\procA, \delta, q_{0}, F)) \is
        (Q', \AlphAsyncWCh_\procA, \delta', q'_0, F')
        \qquad \text{where}
    \]
    \begin{itemize}
     \item $Q' \is
             Q \times
             \set{\channel{\procA}{\procB} \to \Int_{\chanBounds(\channel{\procA}{\procB})}}_{\channel{\procA}{\procB} \in \domainOf(\chanBounds)}  \times
             \set{\channel{\procB}{\procA} \to \Int_{\chanBounds(\channel{\procB}{\procA})}}_{\channel{\procB}{\procA} \in \domainOf(\chanBounds)} $,
     \item $q'_0 \is
            \left(
                q_{0},
                \set{\channel{\procA}{\procB} \mapsto 0}_{\channel{\procA}{\procB} \in \domainOf(\chanBounds)},
                \set{\channel{\procB}{\procA} \mapsto 0}_{\channel{\procB}{\procA} \in \domainOf(\chanBounds)}
            \right)$,
     \item $F' \is
            F \times
                (\set{\channel{\procA}{\procB} \mapsto 0}_{\channel{\procA}{\procB} \in \domainOf(\chanBounds)}) \times
                (\set{\channel{\procB}{\procA} \mapsto 0}_{\channel{\procB}{\procA} \in \domainOf(\chanBounds)}) $,
    \end{itemize}
    and the transition relation $\delta'$ is defined as follows:
    \begin{itemize}
     \item $(
                (q_1, \chanBoundsMapSnd, \chanBoundsMapRcv),
                \; \snd{\procA}{\procB}{\val}, \;
                (q_2, \chanBoundsMapSnd, \chanBoundsMapRcv)
            ) \in \delta'$ if
            $(q_1, \snd{\procA}{\procB}{\val}, q_2) \in \delta$ and
           $\channel{\procA}{\procB} \notin \domainOf(\chanBounds)$,
     \item $(
                (q_1, \chanBoundsMapSnd, \chanBoundsMapRcv),
                \; \rcv{\procB}{\procA}{\val}, \;
                (q_2, \chanBoundsMapSnd, \chanBoundsMapRcv)
            ) \in \delta'$ if
            $(q_1, \rcv{\procB}{\procA}{\val}, q_2) \in \delta$ and
            $\channel{\procA}{\procB} \notin \domainOf(\chanBounds)$,
     \item $(
                (q_1, \chanBoundsMapSnd, \chanBoundsMapRcv),
                \; \snd{\procA}{\procChanI{\procA}{\procB}{j}}{\val}, \;
                (q_2,
                \chanBoundsMapSnd[\channel{\procA}{\procB} \mapsto \chanBoundsMapSnd(\channel{\procA}{\procB}) + 1],
                \chanBoundsMapRcv)
            ) \in \delta'$ if \\
            $(q_1, \; \snd{\procA}{\procB}{\val}, \; q_2) \in \delta$ and
            $\roleFmt{j} = \chanBoundsMapSnd(\channel{\procA}{\procB})$, and
     \item $(
                (q_1, \chanBoundsMapSnd, \chanBoundsMapRcv),
                \; \rcv{\procChanI{\procB}{\procA}{j}}{\procA}{\val}, \;
                (q_2, \chanBoundsMapSnd,
                \chanBoundsMapRcv[\channel{\procB}{\procA} \mapsto \chanBoundsMapRcv(\channel{\procB}{\procA}) + 1],
                )
            ) \in \delta'$ if \\
            $(q_1, \; \rcv{\procB}{\procA}{\val}, \; q_2) \in \delta$ and
            $\roleFmt{j} = \chanBoundsMapRcv(\channel{\procB}{\procA})$.
    \end{itemize}
\item $\decchanfsm$ that turns an FSM over $\AlphAsyncWCh_\procA$ to an FSM over $\AlphAsync_\procA$:
  \[
    \decchanfsm((Q, \AlphAsyncWCh_\procA, \delta, q_{0}, F)) \is
    (Q, \AlphAsync_\procA, \delta', q_{0}, F)
    \quad \text{where }
  \]
        $(q_1, h(x), q_2) \in \delta'$ if $(q_1, x, q_2) \in \delta$
        with
        {
        \small
        $
        h(x) \is
        \begin{cases}
                \snd{\procA}{\procB}{\val}
                & \text{if } x = \snd{\procA}{\procChan{\procA}{\procB}}{\val} \\
                \rcv{\procB}{\procA}{\val}
                & \text{if } x = \rcv{\procChan{\procB}{\procA}}{\procA}{\val} \\
        \end{cases}
\enspace .
        $
        }

 \item
    $A_{\procChan{\procA}{\procB}} \is
    \left((q_{0,f} \dunion \Union_{\val \in \MsgVals} q_{\val}), \AlphAsync_{(\procA, \procB)}, \delta, q_{0,f}, \set{q_{0,f}}\right)$ where \\
    $(q_{0,f}, \rcv{\procA}{\procChan{\procA}{\procB}}{\val}, q_{\val}) \in \delta$
    and
    $(q_{\val}, \snd{\procChan{\procA}{\procB}}{\procB}{\val}, q_{0,f}) \in \delta$
    for every $\val \in \MsgVals$.
\end{itemize}
\end{definition}

Eventually, we want to prove the following:
first, if the encoded PSM is projectable, then the original PSM is;
second, if the original PSM is projectable, then the encoded PSM is projectable.
In particular, we will show that the respective projections can be constructed.
The second condition is important to obtain completeness.

What follows is a sequence of lemmas that establishes properties of the channel-participant encoding, ultimately leading to the desired result.

From now on, we fix a PSM $\PSM$, channel bounds $\chanBounds$
and a corresponding channel-participant encoding
\[
\left( \;
    \encchanpsm,
    \encchanwords,
    \decchanwords,
    \encchanfsm,
    \decchanfsm,
    A_{(\procA, \procB)} \text{ for } (\procA, \procB) \in \ProcsCh
\; \right)
\enspace .
\]
For readability, we assume that
$\encchanfsm(A_{\procChan{\procA}{\procB}}) = A_{\procChan{\procA}{\procB}}$.

In addition, without loss of generality, we assume that for every CSM, it holds that the FSM for every participant is deterministic. 
This is possible as, if there is a CSM (which is what we will need to show for projectability), each FSM that can be determinised individually. 
Observe that $\encchanfsm(\hole)$ and $\decchanfsm(\hole)$ preserve determinism.

The following lemma characterises the property that would not hold if we only introduced one channel participant per channel (cf.\ \cref{ex:example-issue-with-only-one-channel-participant}).

\begin{lemma}
\label{lm:circular-matching}
Let $w, u \in \AlphAsync^\infty$ be two \channelcompliant words that respect $\chanBounds$ for which $w \interswap u$ holds.
For every $\procChan{\procA}{\procB} \in \ProcsCh$, it holds that
\[
\encchanwords(w) \wproj_{\AlphAsyncWCh_{\procChan{\procA}{\procB}}}
=
\encchanwords(u) \wproj_{\AlphAsyncWCh_{\procChan{\procA}{\procB}}}
\enspace .
\]
\end{lemma}
\begin{proof}
We define a function that, for an offset $j$, keeps every $k^{\text{th}}$ send and every $k^{\text{th}}$ receive event.
Let $v \is v_1 \ldots \in \AlphAsync^\infty$.
\[
h(v, j, k) \is
    h'(\emptystring, v_1, j, k) \cat
    h'(v_1, v_2, j, k) \cat
\ldots \cat
    h'(v_1 \ldots v_{i-1}, v_i, j, k) \cat
    \ldots
\]
\[
\text{ where }
h'(v', v_i, j, k) \is
    \begin{cases}
        v_i
            & \text{if } v_i = \snd{\procA}{\procB}{\val}
            \text{ and } j \modOp \card{v' \wproj_{\snd{\procA}{\procB}{\_}}} \modSuff k
            \\
v_i
            & \text{if } v_i = \rcv{\procA}{\procB}{\val}
            \text{ and } j \modOp \card{v' \wproj_{\rcv{\procA}{\procB}{\_}}} \modSuff k
            \\
        \emptystring
        & \text{otherwise}
    \end{cases}
    \enspace .
\]

\textit{Claim 1:}
Let $v$ be a \channelcompliant word that respects $\chanBounds$.
For all channels $\channel{\procA}{\procB}$ and every $0 \leq j < \chanBounds(\channel{\procA}{\procB})$, it holds that
\[
    h(v, j, \chanBounds(\channel{\procA}{\procB})) \in
    \set{x \cat y \mid
    x \in \AlphAsync_{{\channel{\procA}{\procB}}, !} \text{ and }
    y \in \AlphAsync_{{\channel{\procA}{\procB}}, ?}
    }^\infty
    \enspace .
\]

\textit{Proof of Claim 1:}
Towards a contradiction, assume that
$
    h(v, j, \chanBounds(\channel{\procA}{\procB})) \notin
    \set{x \cat y \mid x = \snd{\procA}{\procB}{\_} \land y = \rcv{\procA}{\procB}{\_} }^\infty
$
for some $j$.
We know that $v$ is \channelcompliant.
Thus, two send events occur next to each other:
$
h(v, j, \chanBounds(\channel{\procA}{\procB}))
=
\ldots \cat \snd{\procA}{\procB}{\_} \cat \snd{\procA}{\procB}{\_} \cat \ldots
$.
We situate the above pattern in $v$:
\[
 v' \cat \snd{\procA}{\procB}{\_} \cat v'' \cat \snd{\procA}{\procB}{\_} \cat v'''
 \enspace .
\]
By construction,
$h(v, j, \chanBounds(\channel{\procA}{\procB}))$
only collects every $k^{\text{th}}$ send and every $k^{\text{th}}$ receive event.
Thus, $v''$ contains
$\chanBounds(\channel{\procA}{\procB}) - 1$
send events
$\snd{\procA}{\procB}{\_}$.
Notably, it can also contain receive events of shape
$\rcv{\procA}{\procB}{\_}$.
However, these will all match send events in $v'$ because
the matching receive event for the first
$\snd{\procA}{\procB}{\_}$ from above is also in
$h(v, j, \chanBounds(\channel{\procA}{\procB}))$ and thus can only appear in $v'''$.
In other words, $v''$ is empty for $h(v, j, \chanBounds(\channel{\procA}{\procB}))$. 
Thus,
$v' \cat \snd{\procA}{\procB}{\_} \cat v'' \cat \snd{\procA}{\procB}{\_}$
contains at least
$\chanBounds(\channel{\procA}{\procB}) + 1$
unmatched send events so $v$ does not respect $\chanBounds$, yielding a contradiction.

\claimProofEnd{1}

\textit{Claim 2:}
For every channel $\channel{\procA}{\procB}$ and
$0 \leq j < \chanBounds(\channel{\procA}{\procB})$, it holds that
$h(w, j, k) = h(u, j, k)$.

\textit{Proof of Claim 2:}
Towards a contradiction, suppose that
$h(w, j, k) \neq h(u, j, k)$.
By Claim 1, we know that both
$h(u, j, k)$ and $h(w, j, k)$ alternate between send and receive events.
We consider the first difference and do a case analysis if it is a send or receive event.
First, suppose that they differ on a receive event.
Then, the received message is different from the message that was sent, contradicting \channelcompliancy for at least one of $u$ or $w$.
Second, suppose that they differ on a send event.
Let us establish an intermediate fact that we contradict later.
From $w \interswap u$ and \cite[Lm.\,23]{DBLP:conf/concur/MajumdarMSZ21}, we know that
$w \wproj_{\AlphAsync_\procA} = u \wproj_{\AlphAsync_\procA}$.
It is straightforward that this implies that
$w \wproj_{\snd{\procA}{\procB}{\_}} = u \wproj_{\snd{\procA}{\procB}{\_}}$ and, thus,
also that every
$\chanBounds(\channel{\procA}{\procB})$-th occurrence on both sides are the same (for every offset).
However, the fact that
$h(w, j, k)$ and $h(u, j, k)$
differ on some send contradicts this observation, yielding a contradiction.

\claimProofEnd{2}

We now show that Claim 2 implies
$
\encchanwords(w) \wproj_{\AlphAsyncWCh_{\procChan{\procA}{\procB}}}
=
\encchanwords(u) \wproj_{\AlphAsyncWCh_{\procChan{\procA}{\procB}}}
$
for every $\procChan{\procA}{\procB} \in \ProcsCh$, which is the goal of this lemma.
By construction of $\encchanwords(\hole)$, we know that
$\encchanwords(v) \wproj_{\AlphAsync_{\procChanI{\procA}{\procB}{i}}}$
is built from $v$ in the following way:
starting with the $i$-th occurrence,
every $\chanBounds(\channel{\procA}{\procB})$-th instance of
a send event
$\snd{\procA}{\procB}{\_}$
turns into
$\rcv{\procA}{\procChanI{\procA}{\procB}{i}}{\_}$
and
every $\chanBounds(\channel{\procA}{\procB})$-th instance of
a receive event
$\rcv{\procA}{\procB}{\_}$
turns into
$\snd{\procChanI{\procA}{\procB}{i}}{\procB}{\_}$.
The function
$h(\hole, \hole, \hole)$
considers precisely this pattern to keep letters,
\ie keeping every
$\chanBounds(\channel{\procA}{\procB})$-th letter
starting from any offset.
Thus, Claim~2 implies our claim.
\proofEndSymbol
\end{proof}

We continue by proving various properties about the channel participant encoding. 

\begin{restatable}{lemma}{propertiesChanPartEnc}The following properties hold:
\begin{enumerate}[label=\textnormal{(\arabic*)}]
 \item
    \label{def:channel-encoding-prop-words-almost-bij-1}
    For every word $w \in \AlphAsync^\infty$, it holds that
    $\decchanwords(\encchanwords(w)) = w$.
 \item
    \label{def:channel-encoding-prop-words-almost-bij-2}
    For every channel-ordered word $w \in (\AlphSyncWCh)^\infty$, we have
    \mbox{$\encchanwords(\decchanwords(w)) = w$}.
 \item \label{def:channel-encoding-prop-psm-func-channel-ordered}
 Every word in $\semantics(\encchanpsm(\PSM))$ is channel-ordered.
 \item \label{def:channel-encoding-prop-fsm-func-channel-ordered}
 Let $\CSM{A}$ be a CSM. Then, for every $\procA$, every word in $\lang(\encchanfsm(A_\procA))$ is channel-ordered. 
 \item \label{def:channel-encoding-prop-words-interplay-interswap}
 Let $w, u \in \AlphAsync^\infty$ be two \channelcompliant words that respect $\chanBounds$ such that $w \interswap u$. \\
 Then, it holds that $\encchanwords(w) \interswap \encchanwords(u)$.
 \item \label{def:channel-encoding-prop-words-interplay-interswap-inv}
 Let $w, u \in (\AlphSyncWCh)^\infty$ be two words such that $w \interswap u$.
 Then, it holds that $\decchanwords(w) \interswap \decchanwords(u)$.
 \item \label{def:channel-encoding-prop-words-and-psm-func-align}
 For every $w \in \AlphAsync^\infty$, it holds that
 $w \in \semantics(\PSM)$ iff
 $\encchanwords(w) \in \semantics(\encchanpsm(\PSM))$.
 \item \label{def:channel-encoding-prop-words-and-psm-func-rev}
 For every channel-ordered word $w \in (\AlphSyncWCh)^\infty$, it holds that \\
    $w \in \semantics(\encchanpsm(\PSM))$
    iff
    $\decchanwords(w) \in \semantics(\PSM)$.
\end{enumerate}
\end{restatable}
\begin{proof}
We consider every claim and use previously proven ones later:
\begin{itemize}\item
    \ref{def:channel-encoding-prop-words-almost-bij-1}:
    Let $w \is w_1 \cat w_2 \ldots$.
    Then, we have that
    \begin{align*}
    \; & \decchanwords(\encchanwords(w)) \\
    = \; & \decchanwords(\encchanwords(w_1 \cat w_2 \ldots)) \\
    = \; & \decchanwords(\encchanwords(w_1) \cat \encchanwords(w_2) \ldots) \\
    = \; & \decchanwords(\encchanwords(w_1)) \cat \decchanwords(\encchanwords(w_2)) \ldots \\
    = \; & w_1 \cat w_2 \ldots
    \end{align*}
 \item
    \ref{def:channel-encoding-prop-words-almost-bij-2}:
    Let $w \is w_1 \cat w_2 \ldots$ with $w_i \in \AlphSyncWCh$.
    Then, we have that
    \begin{align*}
    \; & \encchanwords(\decchanwords(w)) \\
    = \; & \encchanwords(\decchanwords(w_1 \cat w_2 \ldots)) \\
    = \; & \encchanwords(\decchanwords(w_1) \cat \decchanwords(w_2) \ldots) \\
    = \; & \encchanwords(\decchanwords(w_1)) \cat \encchanwords(\decchanwords(w_2)) \ldots \\
    = \; & w_1 \cat w_2 \ldots
    \end{align*}
    where the last equality hinges on the fact that $w$ is channel-ordered as it makes sure that the indices for channel participants are introduced in the right way.

 \item \ref{def:channel-encoding-prop-psm-func-channel-ordered}:
    The construction of $\encchanpsm(\hole)$ keeps track of which channel participants for each channel shall be used next, both for sending and receiving.
    The semantics allows reordering under $\interswap$.
    Still, events from the same participant cannot be reordered beyond each other.
    Thus, all words are channel-ordered.
    Note that the events of different channel participants might be reordered, \eg,
    $
    \rcv{\procA}{\procChanI{\procA}{\procB}{{i+1}}}{\val_1}
    \cat
    \rcv{\procA}{\procChanI{\procA}{\procB}{i}}{\val_2}
        \interswap
    \rcv{\procA}{\procChanI{\procA}{\procB}{i}}{\val_2}
    \cat
    \rcv{\procA}{\procChanI{\procA}{\procB}{{i+1}}}{\val_1}
    $
    but the matching send events by $\procA$ will happen in the latter order.

 \item \ref{def:channel-encoding-prop-fsm-func-channel-ordered}
    Again, the property follows directly by construction, which keeps track of which channel participants to use next, both for sending and receiving.

 \item \ref{def:channel-encoding-prop-words-interplay-interswap}:
    With Lemma 23 from \cite{DBLP:conf/concur/MajumdarMSZ21},
    we know that $u \interswap w$ if and only if
    $u \wproj_{\AlphAsync_\procA} = w \wproj_{\AlphAsync_\procA}$ for every $\procA \in \Procs$.
    Again, because of \cite[Lm.\,23]{DBLP:conf/concur/MajumdarMSZ21},
    it suffices to show that
    $\encchanwords(u) \wproj_{\AlphAsync_\procC} = \encchanwords(w) \wproj_{\AlphAsync_\procC}$ for every $\procC \in \ProcsWCh$.
    Let $\procC \in \ProcsWCh$.
    We do a case analysis if $\procC \in \Procs$ or $\procC \in \ProcsCh$.

    First, let $\procC \in \Procs$.
    By assumption, we know that
    $u \wproj_{\AlphAsync_\procC} = w \wproj_{\AlphAsync_\procC}$.
    Applying $\encchanwords(\hole)$ on both sides yields
    $\encchanwords(u \wproj_{\AlphAsync_\procC}) = \encchanwords(w \wproj_{\AlphAsync_\procC})$.
    By induction, it is straightforward to show that applying $\encchanwords(\hole)$ first and then
    $(\hole) \wproj_{\AlphAsync_\procC}$ yields the same result, proving the claim.
    Intuitively, $\encchanwords(\hole)$ requires the prefix of a word but it actually does only consider the part that is left after applying
    $(\hole) \wproj_{\AlphAsync_\procC}$ anyway.

    Second, let $\procC = \procChanI{\procA}{\procB}{i} \in \ProcsCh$.
    We need to show that
    $\encchanwords(u) \wproj_{\AlphAsync_{\procChanI{\procA}{\procB}{i}}} =
     \encchanwords(w) \wproj_{\AlphAsync_{\procChanI{\procA}{\procB}{i}}}$, which follows from
    \cref{lm:circular-matching}.

 \item \ref{def:channel-encoding-prop-words-interplay-interswap-inv}:
    With Lemma 23 from \cite{DBLP:conf/concur/MajumdarMSZ21},
    it suffices to show that
    $\decchanwords(u) \wproj_{\AlphAsync_\procC} = \decchanwords(w) \wproj_{\AlphAsync_\procC}$ for every $\procC \in \Procs$.
    It suffices to consider $\Procs$ since those are the only participants that are left in $\decchanwords(\hole)$.
    Let $\procC \in \Procs$ be some participant.
    We show that
    $\decchanwords(u) \wproj_{\AlphAsync_\procC} = \decchanwords(w) \wproj_{\AlphAsync_\procC}$.
    Again, from Lemma 23 from \cite{DBLP:conf/concur/MajumdarMSZ21},
    we know that
    $u \wproj_{\AlphSyncWCh_\procA} = w \wproj_{\AlphSyncWCh_\procA}$ for every $\procA \in \ProcsWCh$.
    We instantiate this with $\procC$ to obtain:
    $u \wproj_{\AlphSyncWCh_\procC} = w \wproj_{\AlphSyncWCh_\procC}$.
    We apply $\decchanwords(\hole)$ on both sides to obtain:
    $\decchanwords(u \wproj_{\AlphSyncWCh_\procC}) = \decchanwords(w \wproj_{\AlphSyncWCh_\procC})$.
    By induction, it is straightforward to show that applying $\decchanwords(\hole)$ first and then
    $(\hole) \wproj_{\AlphAsync_\procC}$ yields the same result, proving the claim.
Intuitively, we only keep events of $\procC$ because of the projection and the decoding $\decchanwords(\hole)$ can also be applied to the whole word first (and vice versa).

 \item \ref{def:channel-encoding-prop-words-and-psm-func-align}:
    For both directions, we use the following notation for the PSMs: \\
    \[ \PSM = (Q, \AlphAsync, \delta, q_{0}, F) \text{ and } \] 
    { \scriptsize
    \[\encchanpsm(\PSM) =
        (Q \times
            (\set{\channel{\procA}{\procB} \mapsto [0, \chanBounds(\channel{\procA}{\procB})-1]}_{\channel{\procA}{\procB} \in \domainOf(\chanBounds)})^2,
        \AlphSyncWCh, \delta',
        (q_{0}, \chanBoundsMapSnd_0, \chanBoundsMapRcv_0), F')
        \enspace . \]
    }

    The transition relation will be clear from context, so we use $\rightarrow$ for both.

    For the direction from left to right, we assume that
    $w' \in \semantics(\PSM)$ and show that
    $\encchanwords(w') \in \semantics(\encchanpsm(\PSM))$.
    By definition, there is $w \interswap w'$ such that
    $w \in \lang(\PSM)$.
    Thus, there is a run $\run$ in $\PSM$ with $\trace(\run) = w$.

    We show that the run for $w$ in $\PSM$
    can be simulated in
    $\encchanpsm(\PSM)$ for $\encchanwords(w)$.
    Let us assume there exists run
    $ \run \is
     q_0 \xrightarrow{w_0}
     q_1 \xrightarrow{w_1}
     \ldots
    $
    for $\PSM$ for
    $w \is w_0 \cat w_1 \cat \ldots$.
    We claim the following:
    for every $i$
    with
    $
     q_i \xrightarrow{w_i}
     q_{i+1}
    $ and
     $(s_i, \chanBoundsMapSnd_i, \chanBoundsMapRcv_i)$
    such that
    \begin{enumerate}[label=(\alph*)]
     \item \label{cond:wCh-sim-woCh-PSM-1}
            $q_i = s_i$,
     \item \label{cond:wCh-sim-woCh-PSM-2}
        $\chanBoundsMapSnd_i(\channel{\procA}{\procB}) \modOp \card{(w_1 \ldots w_{i-1}) \wproj_{\snd{\procA}{\procB}{\_}}}$
        for every $\channel{\procA}{\procB} \in \domainOf(\chanBounds)$, and
     \item \label{cond:wCh-sim-woCh-PSM-3}
        $\chanBoundsMapRcv_i(\channel{\procA}{\procB}) \modOp \card{(w_1 \ldots w_{i-1}) \wproj_{\rcv{\procA}{\procB}{\_}}}$
        for every $\channel{\procA}{\procB} \in \domainOf(\chanBounds)$.
    \end{enumerate}
    there is
     $(s_{i+1}, \chanBoundsMapSnd_{i+1}, \chanBoundsMapRcv_{i+1})$
    such that
    \begin{enumerate}[label=(\alph*')]
     \item \label{prop:wCh-sim-woCh-PSM-1}
            $q_{i+1} = s_{i+1}$,
     \item \label{prop:wCh-sim-woCh-PSM-2}
        $\chanBoundsMapSnd_{i+1}(\channel{\procA}{\procB}) \modOp \card{(w_1 \ldots w_{i}) \wproj_{\snd{\procA}{\procB}{\_}}}$
        for every $\channel{\procA}{\procB} \in \domainOf(\chanBounds)$,
     \item \label{prop:wCh-sim-woCh-PSM-3}
        $\chanBoundsMapRcv_{i+1}(\channel{\procA}{\procB}) \modOp \card{(w_1 \ldots w_{i}) \wproj_{\rcv{\procA}{\procB}{\_}}}$
        for every $\channel{\procA}{\procB} \in \domainOf(\chanBounds)$, and
     \item \label{prop:wCh-sim-woCh-PSM-4}
        $(s_i, \chanBoundsMapSnd_i, \chanBoundsMapRcv_i)
         \xrightarrow{\encchanwords(w_i)}
         (s_{i+1}, \chanBoundsMapSnd_{i+1}, \chanBoundsMapRcv_{i+1})$.
    \end{enumerate}

    For our simulation argument, we need one more ingredient: the two initial states, \ie for $i=0$, satisfy the above conditions.
    It is easy to check that this is the case.
    Then, the claim can be used to mimic the run in $\encchanpsm(\PSM)$ for $\encchanwords(w)$.

    Let us prove the claim.
    Let $i \in \Nat$.
    We assume that
    \ref{cond:wCh-sim-woCh-PSM-1}
    to
    \ref{cond:wCh-sim-woCh-PSM-3}
    hold and prove
    \ref{prop:wCh-sim-woCh-PSM-1}
    to
    \ref{prop:wCh-sim-woCh-PSM-4}.
    We choose
    $(s_{i+1}, \chanBoundsMapSnd_{i+1}, \chanBoundsMapRcv_{i+1})$
    such that
    \ref{prop:wCh-sim-woCh-PSM-1}
    to
    \ref{prop:wCh-sim-woCh-PSM-3}
    are satisfied, so it remains to prove that
    \ref{prop:wCh-sim-woCh-PSM-4}
    holds.
    We do a simultaneous case analysis on the shape of $w_i$ and whether the corresponding channel is in the domain~$\chanBounds$.
    \begin{itemize}
     \item $\channel{\procA}{\procB} \notin \domainOf(\chanBounds)$ and either $w_i = \snd{\procA}{\procB}{\val}$ or $w_i = \rcv{\procA}{\procB}{\val}$: \\
        Since $\channel{\procA}{\procB} \notin \domainOf(\chanBounds)$,
        we have
        $\encchanwords(w_i) = w_i$ by definition
        as well as
        \[
        \card{(w_1 \ldots w_{i-1}) \wproj_{\snd{\procA}{\procB}{\_}}}
        \modOp
        \card{(w_1 \ldots w_{i}) \wproj_{\snd{\procA}{\procB}{\_}}}
        \text{ and }
        \phantom{.}
        \]
        \[
        \card{(w_1 \ldots w_{i-1}) \wproj_{\rcv{\procA}{\procB}{\_}}}
        \modOp
        \card{(w_1 \ldots w_{i}) \wproj_{\rcv{\procA}{\procB}{\_}}}
        \phantom{ and }
        \enspace.
        \]
        Again, since $\channel{\procA}{\procB} \notin \domainOf(\chanBounds)$,
        we also have that
        $\chanBoundsMapSnd_i = \chanBoundsMapSnd_{i+1}$
        and
        $\chanBoundsMapRcv_i = \chanBoundsMapRcv_{i+1}$
        by construction of $\encchanpsm(\PSM)$.
        Thus,
        \ref{prop:wCh-sim-woCh-PSM-4}
        is satisfied.

     \item $\channel{\procA}{\procB} \in \domainOf(\chanBounds)$ and $w_i = \snd{\procA}{\procB}{\val}$: \\
        By definition,
        $\encchanwords(w_i) = \msgFromTo{\procA}{\procChanI{\procA}{\procB}{j}}{\val}$
        for
        $\roleFmt{j} \modOp
        \card{(w_1 \ldots w_{i-1}) \wproj_{\snd{\procA}{\procB}{\_}}}$.
        By construction, there is a transition labelled with $\encchanwords(w_i)$:
        \[
         (s_i, \chanBoundsMapSnd_i, \chanBoundsMapRcv_i)
         \xrightarrow{\encchanwords(w_i)}
         (s'_{i+1}, \widehat{\chanBoundsMapSnd_{i+1}}, \widehat{\chanBoundsMapSnd_{i+1}})
         \enspace .
         \]
        By construction, it is obvious that $s_{i+1} = s'_{i+1}$.
        Let us consider the changes for the channel bounds due to
        \ref{prop:wCh-sim-woCh-PSM-2} and
        \ref{prop:wCh-sim-woCh-PSM-3}.
        Since $w_i$ is not a receive event, we have
        $\chanBoundsMapRcv_{i+1} = \chanBoundsMapRcv_{i}$.
        For every channel $\channel{\procC}{\procD}$ different from $\channel{\procA}{\procB}$, we also have
        $\chanBoundsMapSnd_{i+1}(\channel{\procC}{\procD}) = \chanBoundsMapSnd_{i}(\channel{\procC}{\procD})$.
        For $\channel{\procA}{\procB}$, it holds that
        $\chanBoundsMapSnd_{i+1}(\channel{\procA}{\procB}) \equiv \chanBoundsMapSnd_{i}(\channel{\procA}{\procB}) + 1$.
        This matches with the changes due to the semantics for $\encchanpsm(\PSM)$.
        Thus,
        $
        \widehat{\chanBoundsMapSnd_{i+1}}
        =
        \chanBoundsMapSnd_{i+1}
        $ and $
        \widehat{\chanBoundsMapRcv_{i+1}}
        =
        \chanBoundsMapRcv_{i+1}
        $, which shows that \ref{prop:wCh-sim-woCh-PSM-4} is satisfied.

     \item $\channel{\procA}{\procB} \in \domainOf(\chanBounds)$ and $w_i = \rcv{\procA}{\procB}{\val}$: \\
        This case is analogous to the previous one and, thus, omitted.

    \end{itemize}

    This shows that there is a run in
    $\encchanpsm(\PSM)$ for $\encchanwords(w)$.
    With $w' \interswap w$ and
    \ref{def:channel-encoding-prop-words-interplay-interswap},
    we have that
    $\encchanwords(w') \interswap \encchanwords(w)$.
    Together, it follows that $\encchanwords(w') \in \semantics(\encchanpsm(\PSM))$, which concludes this direction.

    For the direction from right to left, we claim it suffices to show that,
    for all $u' \in (\AlphSyncWCh)^\infty$, it holds that,
    if $u' \in \semantics(\encchanpsm(\PSM))$,
    then $\decchanwords(u') \in \semantics(\PSM)$. \\
    Let us first explain why this claim is sufficient:
    we instantiate $u' = \encchanwords(w)$, which gives that
    if $\encchanwords(w) \in \semantics(\encchanpsm(\PSM))$,
    then $\decchanwords(\encchanwords(w)) \in \semantics(\PSM)$.
    By
    \ref{def:channel-encoding-prop-words-almost-bij-1}, we have that
    $\decchanwords(\encchanwords(w)) = w$, proving the claim. \\
    Let us now prove this claim.
    We assume that
    $u' \in \semantics(\encchanpsm(\PSM))$
    and show that
    $\decchanwords(u') \in \semantics(\PSM)$.

    By definition, there is $u \interswap u'$ such that
    $u \in \lang(\encchanpsm(\PSM))$.
    Thus, there is a run $\run$ in $\encchanpsm(\PSM)$ with $\trace(\run) = u$.
    We claim that there is a run~$\run'$ in $\PSM$ with $\trace(\run') = \decchanwords(u)$.
    As for the other direction, we prove a simulation-like argument for which both initial states are trivially related.
    Let us assume there exists run
    $ \run \is
     (s_0, \chanBoundsMapSnd_0, \chanBoundsMapRcv_0) \xrightarrow{u_0}
     (s_1, \chanBoundsMapSnd_1, \chanBoundsMapRcv_1) \xrightarrow{u_1}
     \ldots
    $
    in $\PSM$ for
    $u \is u_0 \cat u_1 \cat \ldots$.
    We claim the following:
    for every $i$
    with
    $
     (s_i, \chanBoundsMapSnd_i, \chanBoundsMapRcv_i) \xrightarrow{u_i}
     (s_{i+1}, \chanBoundsMapSnd_{i+1}, \chanBoundsMapRcv_{i+1})
    $,
    it holds that
    $s_i \xrightarrow{\decchanwords(u_i)} s_{i+1}$.
    Let $i \in \Nat$.
    By construction of $\encchanpsm(\PSM)$, there is
    $s_i \xrightarrow{x} s_{i+1}$ such that
    $\encchanwords(x) = u_i$.
    We apply $\decchanwords(\hole)$ on both sides and obtain
    $\decchanwords(\encchanwords(x)) = \decchanwords(u_i)$.
    By
    \ref{def:channel-encoding-prop-words-almost-bij-1},
    it get $x = \decchanwords(u_i)$, which proves the claim.
    Together, this shows that there is a run with trace $\decchanwords(u)$ in $\PSM$.

    With $u \interswap u'$ and
    \ref{def:channel-encoding-prop-words-interplay-interswap-inv}
    we have that
    $\decchanwords(u) \interswap \decchanwords(u')$.
    Thus, we have $\decchanwords(u) \in \semantics(\PSM)$, which concludes this proof.

 \item \ref{def:channel-encoding-prop-words-and-psm-func-rev}:
 The direction from left to right was proven as part of the proof for~\ref{def:channel-encoding-prop-words-and-psm-func-align}.
 The direction from right to left follows from
 \ref{def:channel-encoding-prop-words-and-psm-func-align}
 if one instantiates $w = \decchanwords(w)$:
 if $\decchanwords(w) \in \semantics(\PSM)$ then
 $\encchanwords(\decchanwords(w)) \in \semantics(\encchanpsm(\PSM))$;
 with
 \ref{def:channel-encoding-prop-words-almost-bij-2}, we obtain that
 $\encchanwords(\decchanwords(w)) = w$, proving the claim.

\end{itemize}
This concludes all proofs. 
\proofEndSymbol
\end{proof}

\begin{lemma}
Let $\CSM{A}$ be a CSM.
 \label{def:channel-encoding-prop-words-and-csm-func-align} Then, for every $w \in \AlphAsync^\infty$ that respects $\chanBounds$, it holds that 
 $w \in \lang(\CSM{A})$ iff
 $\encchanwords(w) \in \lang(\CSMlwCh{\encchanfsm(A_\procA)})$.
\end{lemma}
\begin{proof}
We show that the run for $w$ in $\CSM{A}$
    can be simulated
    for $\encchanwords(w)$ in
    $\CSMlwCh{\encchanfsm(A_\procA)}$,
    and vice versa.

    Prior, we establish some notation that will be used for both cases. \\
    Let $w \is w_0 \cat w_1 \ldots$ and 
    $\run$ be the run in $\CSM{A}$ for $w$:
    $
     (\vec{q}_0, \xi_0) \xrightarrow{w_0}
     (\vec{q}_1, \xi_1) \xrightarrow{w_1}
     \ldots
    $.
    For
    $\CSMlwCh{\encchanfsm(A_\procA)}$,
    we split the states of participants as follows:
    for participants from $\Procs$, we have $\vec{s}_i$ as well as $\vec{\chanBoundsMapSnd_i}$ and $\vec{\chanBoundsMapRcv_i}$ for the respective channel bounds;
    for participants from $\ProcsCh$, we have $\vec{t}_i$.
    Then, let $\run'$ be the run in
    $\CSMlwCh{\encchanfsm(A_\procA)}$
    for $\encchanwords(w)$:
    \[
     (\vec{s}_0, \vec{\chanBoundsMapSnd_0}, \vec{\chanBoundsMapRcv_0}, \vec{t}_0, \xi'_0) \xrightarrow{\encchanwords(w_0)}
     (\vec{s}_1, \vec{\chanBoundsMapSnd_1}, \vec{\chanBoundsMapRcv_1}, \vec{t}_1, \xi'_1) \xrightarrow{\encchanwords(w_1)}
     \ldots
     \enspace .
    \]
    Intuitively, the channel participants store the channel content while the actual channels are empty.
    More precisely,
    We introduce notation for this.
    Intuitively, $\xi(\channel{\procA}{\procB})$ contains what is stored in states
    $\procChanI{\procA}{\procB}{i}$ for $\roleFmt{i}$ between
                from
                $\chanBoundsMapRcv(\channel{\procA}{\procB})$ to
                $\chanBoundsMapSnd(\channel{\procA}{\procB})$ (excl.)
    and all the other channel participants, \ie
                $\chanBoundsMapSnd(\channel{\procA}{\procB})$ to
                $\chanBoundsMapRcv(\channel{\procA}{\procB})$ (excl.)
    are in their initial state;
    where we consider the indices for channel participants to form a ring buffer, \eg,
    $5$ to $2$ (excl.) gives the sequence $5, 6, 0, 1$ if the channel bound is $7$.

    Formally, we define
    $\contentOf(\vec{t}, \vec{\chanBoundsMapSnd}, \vec{\chanBoundsMapRcv}, \channel{\procA}{\procB}) \is w_i \ldots w_j$
    for the states $\vec{t}$ of channel participants if the following holds:
    \begin{itemize}
        \item $j = \vec{\chanBoundsMapSnd_\procA}(\channel{\procA}{\procB})$,
        \item $i = \vec{\chanBoundsMapRcv_\procB}(\channel{\procA}{\procB})$,
        \item for every $i \leq k < j$, it holds that $\vec{t}_{\procChanI{\procA}{\procB}{k}} = q_\val$ and $w_k = \val$, and
        \item for every $j \leq k < i$, it holds that $\vec{t}_{\procChanI{\procA}{\procB}{k}} = q_{0,f}$.
    \end{itemize}

    First, let us assume there exists run $\run$ for $\CSM{A}$ for $w  = w_0 \cat w_1 \cat \ldots$.
    We claim the following:
    for every $i$
    with
    $
     (\vec{q}_i, \xi_i) \xrightarrow{w_i}
     (\vec{q}_{i+1}, \xi_{i+1})
    $ and
     $(\vec{s}_i, \vec{\chanBoundsMapSnd_i}, \vec{\chanBoundsMapRcv_i}, \vec{t}_i, \xi'_i)$
    such that
    \begin{enumerate}[label=(\alph*)]
     \item \label{cond:wCh-sim-woCh-CSM-1}
            $\vec{q}_i = \vec{s}_i$,
     \item \label{cond:wCh-sim-woCh-CSM-2}
$\xi'_i(\channel{\procA}{\procB}) = \xi_i(\channel{\procA}{\procB})$
            for every $\channel{\procA}{\procB} \notin \domainOf(\chanBounds)$,
     \item \label{cond:wCh-sim-woCh-CSM-3}
$\xi'_i(\channel{\procA}{\procB}) = \emptystring$
            for every $\channel{\procA}{\procB} \in \domainOf(\chanBounds)$,
     \item \label{cond:wCh-sim-woCh-CSM-4}
$\xi_i(\channel{\procA}{\procB}) =
        \contentOf(\vec{t}_i, \vec{\chanBoundsMapSnd_i}, \vec{\chanBoundsMapRcv_i}, \channel{\procA}{\procB})$
            for every $\channel{\procA}{\procB} \in \domainOf(\chanBounds)$,
    \end{enumerate}
    there is
     $(\vec{s}_{i+1}, \vec{\chanBoundsMapSnd_{i+1}}, \vec{\chanBoundsMapSnd_{i+1}}, \vec{t}_{i+1}, \xi'_{i+1})$
    such that
    \begin{enumerate}[label=(\alph*')]
     \item \label{prop:wCh-sim-woCh-CSM-1}
           $\vec{q}_{i+1} = \vec{s}_{i+1}$,
     \item \label{prop:wCh-sim-woCh-CSM-2}
$\xi'_{i+1}(\channel{\procA}{\procB}) = \xi_{i+1}(\channel{\procA}{\procB})$
            for every $\channel{\procA}{\procB} \notin \domainOf(\chanBounds)$,
     \item \label{prop:wCh-sim-woCh-CSM-3}
$\xi'_{i+1}(\channel{\procA}{\procB}) = \emptystring$
            for every $\channel{\procA}{\procB} \in \domainOf(\chanBounds)$,
     \item \label{prop:wCh-sim-woCh-CSM-4}
$\xi_{i+1}(\channel{\procA}{\procB}) =
        \contentOf(\vec{t}_{i+1}, \vec{\chanBoundsMapSnd}_{\kern-.3em i+1}, \vec{\chanBoundsMapSnd}_{\kern-.3em i+1}, \channel{\procA}{\procB})$
            for every $\channel{\procA}{\procB} \in \domainOf(\chanBounds)$, and
     \item \label{prop:wCh-sim-woCh-CSM-5}
    $
     (\vec{s}_i, \vec{\chanBoundsMapSnd_i}, \vec{\chanBoundsMapRcv_i}, \vec{t}_i, \xi'_i) \xrightarrow{\encchanwords(w_{i})}\!\!^2 \,
     (\vec{s}_{i+1}, \vec{\chanBoundsMapSnd}_{\kern-.3em i+1}, \vec{\chanBoundsMapRcv}_{\kern-.3em i+1}, \vec{t}_{i+1}, \xi'_{i+1})
    $.
    \end{enumerate}

    With the above claim, we need one more ingredient to show that the run for $w$ can be mimicked with one for $\encchanwords(w)$:
    we need to relate the initial states.
    Concretely, we show that
    \ref{cond:wCh-sim-woCh-CSM-1}
    to
    \ref{cond:wCh-sim-woCh-CSM-4}
    hold for $i = 0$.
    In fact, all of these conditions are trivial: initially, all participants are in their initial state and all channels are empty.

    Having related the initial states, it is straightforward that the above claim can be used to mimic the run for $w$ in $\CSM{A}$ in
    $\CSMlwCh{\encchanfsm(A_\procA)}$, having trace $\encchanwords(w)$.
    Intuitively, we start with the initial states and, for every $w_i$, we apply the above claim to obtain the mimicked run.

    It remains to prove the above claim.
    Let $i$ be some number.
    We assume that
    \ref{cond:wCh-sim-woCh-CSM-1}
    to
    \ref{cond:wCh-sim-woCh-CSM-4}
    hold.
    We show that
    \ref{prop:wCh-sim-woCh-CSM-1}
    to
    \ref{prop:wCh-sim-woCh-CSM-5}
    hold.
    We use \ref{prop:wCh-sim-woCh-CSM-5} to obtain the witness
    $
     (\vec{s}_{i+1}, \vec{\chanBoundsMapSnd}_{\kern-.3em i+1}, \vec{\chanBoundsMapRcv}_{\kern-.3em i+1}, \vec{t}_{i+1}, \xi'_{i+1})
    $.
    We will still need to show that this transition is possible.
    We do a simultaneous case analysis whether $w_i$ is a send or receive transition and whether the corresponding channel is in the domain of $\chanBounds$.
    \begin{itemize}
     \item $w_i = \snd{\procA}{\procB}{\val}$ and
            $\channel{\procA}{\procB} \notin \domainOf(\chanBounds)$: \\
            Because of $\channel{\procA}{\procB} \notin \domainOf(\chanBounds)$,
            we have that
            $\encchanwords(w_i) = w_i$.
            From \ref{cond:wCh-sim-woCh-CSM-1}, we know that $\vec{q}_i$ and $\vec{s}_i$ agree on the state for $\procA$, which is the active participant for $w_i$.
            Therefore, the send transition for
            \ref{prop:wCh-sim-woCh-CSM-5}
            is possible.
            All claims
            \ref{prop:wCh-sim-woCh-CSM-1}
            to
            \ref{prop:wCh-sim-woCh-CSM-4}
            trivially follow by the semantics of CSMs and the fact that channel participants are not involved for transitions for which the corresponding channel is not in the domain of $\chanBounds$.

     \item $w_i = \rcv{\procA}{\procB}{\val}$ and
            $\channel{\procA}{\procB} \notin \domainOf(\chanBounds)$: \\
            From \ref{cond:wCh-sim-woCh-CSM-1}, we know that $\vec{q}_i$ and $\vec{s}_i$ agree on the state for $\procB$, which is the active participant for $w_i$.
            In addition, we know that the channel content for $\channel{\procA}{\procB}$ are the same in
            $\xi_i$ and $\xi'_i$ from
            \ref{cond:wCh-sim-woCh-CSM-2}.
            Therefore, the receive transition for
            \ref{prop:wCh-sim-woCh-CSM-5}
            is possible.
            Again, the remaining claims
            \ref{prop:wCh-sim-woCh-CSM-1}
            to
            \ref{prop:wCh-sim-woCh-CSM-4}
            trivially follow by the semantics of CSMs and the fact that channel participants are not involved for transitions for which the corresponding channel is not in the domain of~$\chanBounds$.

     \item $w_i = \snd{\procA}{\procB}{\val}$ and
            $\channel{\procA}{\procB} \in \domainOf(\chanBounds)$: \\
            We have that
            $
            \encchanwords(w_i)
            =
            \msgFromTo{\procA}{\procChanI{\procA}{\procB}{j}}{\val}
            $
            where
            $
            \roleFmt{j} \modOp
            \card{(w_1 \ldots w_{i-1}) \wproj_{\snd{\procA}{\procB}{\_}}}
            $ 
            by construction.
            We first argue that the transition in
            \ref{prop:wCh-sim-woCh-CSM-5}
            is possible.
            In fact, it contains two transitions: one send for $\procA$ and one receive for $\procChanI{\procA}{\procB}{j}$.
            Because of
            \ref{cond:wCh-sim-woCh-CSM-1},
            the states for $\procA$ in $\vec{q}_i$ and $\vec{s}_i$ agree and, thus, the transition is possible.
            After this transition, message $\val$ is in the channel from $\procA$ to $\procChanI{\procA}{\procB}{i}$ and, by
            \ref{cond:wCh-sim-woCh-CSM-3}, this is the only message.
            From
            \ref{cond:wCh-sim-woCh-CSM-4},
            we also know that this channel participant is in its initial state in $\vec{t}_{i+1}$.
            Thus, it can receive $\val$ by construction, with which it updates its state to $q_{\val}$.
            Therefore, this channel is empty again and, thus, \ref{prop:wCh-sim-woCh-CSM-3} holds.
            After the first send transition, $\procA$ is in the same state and thus
            \ref{prop:wCh-sim-woCh-CSM-1} is satisfied by determinism. None of the other participants change their state and thus \ref{prop:wCh-sim-woCh-CSM-2} is (still) satisfied.
            There are only changes related $\channel{\procA}{\procB}$.
            Thus, we only need to consider these changes for \ref{prop:wCh-sim-woCh-CSM-4}.

            By the semantics of CSMs, we have
            $\xi_{i+1}(\channel{\procA}{\procB}) =
             \xi_{i}(\channel{\procA}{\procB}) \cat \val$.
            We also have
            $
            (\vec{\chanBoundsMapSnd}_{\kern-.3em i+1})_{\procC} = (\vec{\chanBoundsMapRcv_i})_{\procC}
            $
            for all $\procC$.

            Thus, it remains to show
            $
            \xi_{i}(\channel{\procA}{\procB}) \cat \val
            =
            \contentOf(\vec{t}_{i+1}, \vec{\chanBoundsMapSnd}_{\kern-.3em i+1}, \vec{\chanBoundsMapRcv_i}, \channel{\procA}{\procB})
            $.

            From
            $
            \contentOf(\vec{t}_{i}, \vec{\chanBoundsMapSnd_i}, \vec{\chanBoundsMapRcv_i}, \channel{\procA}{\procB})
            $
            to
            $
            \contentOf(\vec{t}_{i+1}, \vec{\chanBoundsMapSnd}_{\kern-.3em i+1}, \vec{\chanBoundsMapRcv_i}, \channel{\procA}{\procB})
            $,
            there are only the following changes: 
            the state change of
            $\procChanI{\procA}{\procB}{j}$
            in $\vec{t}_{i+1}$
            and the increase of
            $
            (\vec{\chanBoundsMapSnd}_{\kern-.3em i+1})_\procA(\channel{\procA}{\procB}) =
            (\vec{\chanBoundsMapSnd_i})_\procA(\channel{\procA}{\procB}) + 1
            $,
            while
            $
            (\vec{\chanBoundsMapSnd}_{\kern-.3em_i+1})_{\procC} = (\vec{\chanBoundsMapSnd_i})_{\procC}
            $
            for all $\procC \neq \procA$.
            It is easy to show that these add precisely $\val$ to
            $
            \contentOf(\vec{t}_{i+1}, \vec{\chanBoundsMapSnd}_{\kern-.3em i+1}, \vec{\chanBoundsMapRcv_i}, \channel{\procA}{\procB})
            $
            to
            $
            \contentOf(\vec{t}_{i}, \vec{\chanBoundsMapSnd_{i}}, \vec{\chanBoundsMapRcv_{i}}, \channel{\procA}{\procB})
            $
            which proves the claim.

     \item $w_i = \rcv{\procA}{\procB}{\val}$ and
            $\channel{\procA}{\procB} \in \domainOf(\chanBounds)$: \\
            We have that
            $
            \encchanwords(w_i)
            =
            \msgFromTo{\procChanI{\procA}{\procB}{j}}{\procB}{\val}
            $
            where
            $
            \roleFmt{j} \modOp
            \card{(w_1 \ldots w_{i-1}) \wproj_{\rcv{\procA}{\procB}{\_}}}
            $
            by construction.
            We first argue that the transition in
            \ref{prop:wCh-sim-woCh-CSM-5}
            is possible.
            In fact, it contains two transitions:
            one send for $\procChanI{\procA}{\procB}{j}$ and
            one receive for $\procB$.
            Regarding the second transition:
            because of
            \ref{cond:wCh-sim-woCh-CSM-1},
            the states for $\procA$ in $\vec{q}_i$ and $\vec{s}_i$ agree and, thus, the transition is possible.
            Regarding the first transition:
            since the previous transition is possible in $\CSM{A}$, we know that $\val$ is at the head of the channel
            $\xi(\channel{\procA}{\procB})$;
            from \ref{cond:wCh-sim-woCh-CSM-4}, we know that the state of
            $\procChanI{\procA}{\procB}{j}$ is $q_{\val}$ in $\vec{t}_i$, making the first transition possible and putting it back to its initial state.
            As for the other cases, it is straightforward that
            \ref{prop:wCh-sim-woCh-CSM-1}
            to
            \ref{prop:wCh-sim-woCh-CSM-3}
            are satisfied after these transitions.

            It remains to show
            \ref{prop:wCh-sim-woCh-CSM-4}.
            There are only changes related $\channel{\procA}{\procB}$.
            Thus, we only need to consider these changes.

            By the semantics of CSMs, we have
            $
            \xi_{i}(\channel{\procA}{\procB})
             =
            \val \cat \xi_{i+1}(\channel{\procA}{\procB})
             $.
            We also have
            $
            (\vec{\chanBoundsMapSnd}_{\kern-.3em i+1})_{\procC} = (\vec{\chanBoundsMapSnd_i})_{\procC}
            $
            for all $\procC$.

            The only changes
            from
            $
            \contentOf(\vec{t}_{i}, \vec{\chanBoundsMapSnd_i}, \vec{\chanBoundsMapRcv_i}, \channel{\procA}{\procB})
            $
            to
            $
            \contentOf(\vec{t}_{i+1}, \vec{\chanBoundsMapSnd_i}, \vec{\chanBoundsMapSnd}_{\kern-.3em i+1}, \channel{\procA}{\procB})
            $
            are: the state change of
            $\procChanI{\procA}{\procB}{j}$
            in $\vec{t}_{i+1}$
            and the increase of
            $
            (\vec{\chanBoundsMapSnd}_{\kern-.3em i+1})_\procA(\channel{\procA}{\procB}) =
            (\vec{\chanBoundsMapRcv_i})_\procA(\channel{\procA}{\procB}) + 1
            $,
            while
            $
            (\vec{\chanBoundsMapSnd}_{\kern-.3em i+1})_{\procC} = (\vec{\chanBoundsMapRcv_i})_{\procC}
            $
            for all $\procC \neq \procA$.
            It is easy to show that these remove precisely $\val$,
            which proves the claim.
    \end{itemize}
    This concludes the proof for this direction.

    Second, let us assume there exists run $\run$ for
    $\CSMlwCh{\encchanfsm(A_\procA)}$
    for $\encchanwords(w)$.
    We claim the following:
    for every $i$
    with
    \[
     (\vec{s}_i, \vec{\chanBoundsMapSnd_i}, \vec{\chanBoundsMapRcv_i}, \vec{t}_i, \xi'_i) \xrightarrow{\encchanwords(w_{i})}\!\!^2 \,
     (\vec{s}_{i+1}, \vec{\chanBoundsMapSnd}_{\kern-.3em i+1}, \vec{\chanBoundsMapSnd}_{\kern-.3em i+1}, \vec{t}_{i+1}, \xi'_{i+1})
    \]
    and
    $(\vec{q}_i, \xi_i)$
    such that
    \begin{enumerate}[label=(\alph*)]
     \item \label{cond:woCh-sim-wCh-CSM-1}
            $\vec{q}_i = \vec{s}_i$,
     \item \label{cond:woCh-sim-wCh-CSM-2}
$\xi'_i(\channel{\procA}{\procB}) = \xi_i(\channel{\procA}{\procB})$
            for every $\channel{\procA}{\procB} \notin \domainOf(\chanBounds)$,
     \item \label{cond:woCh-sim-wCh-CSM-3}
$\xi'_i(\channel{\procA}{\procB}) = \emptystring$
            for every $\channel{\procA}{\procB} \in \domainOf(\chanBounds)$,
     \item \label{cond:woCh-sim-wCh-CSM-4}
$\xi_i(\channel{\procA}{\procB}) =
        \contentOf(\vec{t}_i, \vec{\chanBoundsMapSnd_i}, \vec{\chanBoundsMapRcv_i}, \channel{\procA}{\procB})$
            for every $\channel{\procA}{\procB} \in \domainOf(\chanBounds)$,
    \end{enumerate}
    there is
    $(\vec{q}_{i+1}, \xi_{i+1})$
    such that
    \begin{enumerate}[label=(\alph*')]
     \item \label{prop:woCh-sim-wCh-CSM-1}
           $\vec{q}_{i+1} = \vec{s}_{i+1}$,
     \item \label{prop:woCh-sim-wCh-CSM-2}
$\xi'_{i+1}(\channel{\procA}{\procB}) = \xi_{i+1}(\channel{\procA}{\procB})$
            for every $\channel{\procA}{\procB} \notin \domainOf(\chanBounds)$,
     \item \label{prop:woCh-sim-wCh-CSM-3}
$\xi'_{i+1}(\channel{\procA}{\procB}) = \emptystring$
            for every $\channel{\procA}{\procB} \in \domainOf(\chanBounds)$,
     \item \label{prop:woCh-sim-wCh-CSM-4}
$\xi_{i+1}(\channel{\procA}{\procB}) =
        \contentOf(\vec{t}_{i+1}, \vec{\chanBoundsMapSnd_{i+1}}, \vec{\chanBoundsMapSnd_{i+1}}, \channel{\procA}{\procB})$
            for every $\channel{\procA}{\procB} \in \domainOf(\chanBounds)$, and
     \item \label{prop:woCh-sim-wCh-CSM-5}
            $
            (\vec{q}_i, \xi_i) \xrightarrow{w_i}
            (\vec{q}_{i+1}, \xi_{i+1})
            $.
    \end{enumerate}

    As before, we need one more ingredient to show that one run simulates the other: the conditions hold for the initial states, \ie $i=0$.
    Again, this is trivially true.
    In fact, one can also prove the above claim analogously so we omit the full proof.

 \end{proof}

We want the channel participants only to faithfully forward messages.
To formalise this, we define the notions of \emph{forwarding} and \emph{almost forwarding}.

\begin{definition}[Forwarding and almost forwarding]
A word $w = w_1 \ldots  \in (\AlphAsync_{\procChan{\procA}{\procB}})^\infty$
if \emph{forwarding} is for every odd $j$, it holds that
$w_j = \rcv{\procA}{\procChan{\procA}{\procB}}{\val}$
and
$w_{j+1} = \snd{\procChan{\procA}{\procB}}{\procB}{\val}$.
If $w = w_1 \ldots w_n$ is finite, it is \emph{almost forwarding} if $w_1 \ldots w_{n-1}$ is forwarding and
$w_n = \rcv{\procA}{\procChan{\procA}{\procB}}{\val}$.
A language $L \subseteq (\AlphAsync_{\procChan{\procA}{\procB}})^\infty$ is forwarding if for every word is.
We say that an FSM~$A$ over alphabet $\AlphAsync_{\procChan{\procA}{\procB}}$ is \emph{forwarding} if its language $\lang(A)$ is forwarding. 
\end{definition}

Intuitively, for finite words, we only want forwarding words. 
For infinite words, though, it can happen that some send events will never be matched. 
Because of the channel bounds, this, however, cannot happen indefinitely, which is why we use the notion of almost forwarding for infinite words whose projection onto channel participants will be finite.

\begin{lemma}
\label{lm:chan-part-wproj-forwarding}
Let $\PSM$ be a PSM that respects channel bounds $\chanBounds$ and $\procChanI{\procA}{\procB}{i}$ be a channel participant.
Then, the following holds
for every $u \in \semantics(\encchanpsm(\PSM))$:
\begin{itemize}
 \item  If $u$ is finite, then $w = u \wproj_{\AlphAsyncWCh_{\procChanI{\procA}{\procB}{i}}}$ is forwarding.
 \item  If $u$ is infinite, then $w = u \wproj_{\AlphAsyncWCh_{\procChanI{\procA}{\procB}{i}}}$ is almost forwarding or forwarding.
\end{itemize}
\end{lemma}
\begin{proof}
Let $u \in \semantics(\encchanpsm(\PSM))$ be a word.
We do a case distinction if $u$ is finite of infinite.

Suppose $u$ is finite.
Then, we have to show that
$w = u \wproj_{\AlphAsyncWCh_{\procChanI{\procA}{\procB}{i}}}$ is forwarding.
By assumption, we know that $w$ respects $\chanBounds$.
We show that $w$ is forwarding with reasoning from \cref{lm:circular-matching}.
From Claim 1 in
\cref{lm:circular-matching},
we know the following:
for every $0 \leq j < \chanBounds(\channel{\procA}{\procB})$, it holds that
\[
    h(w, j, \chanBounds(\channel{\procA}{\procB})) \in
    \set{x \cat y \mid
    x \in \AlphAsync_{{\channel{\procA}{\procB}}, !} \text{ and }
    y \in \AlphAsync_{{\channel{\procA}{\procB}}, ?}
    }^* \enspace .
\]
Later in that lemma, we have drawn the connection between
$h(w, j, \chanBounds(\channel{\procA}{\procB}))$
and
$\encchanwords(w) \wproj_{\AlphAsyncWCh_{\procChan{\procA}{\procB}}}$.
In particular, they consider the same pattern to keep letters:
starting from the $j$-th occurrence,
it turns
every $\chanBounds(\channel{\procA}{\procB})$-th instance of
a send event
$\snd{\procA}{\procB}{\_}$
turns into
$\snd{\procA}{\procChanI{\procA}{\procB}{i}}{\_}$
and
every $\chanBounds(\channel{\procA}{\procB})$-th instance of
a receive event
$\rcv{\procA}{\procB}{\_}$
turns into
$\rcv{\procChanI{\procA}{\procB}{i}}{\procB}{\_}$.
By construction and the fact that $w$ is \channelcompliant, we also know that these occurrences match, \ie they carry the same message.
Together, it follows that
\[
    w \in
    \set{x \cat y \mid
    x = \rcv{\procA}{\procChanI{\procA}{\procB}{i}}{\val}
    \text{ and }
    y = \snd{\procChanI{\procA}{\procB}{i}}{\procB}{\val}
    \text{ for some } \val \in \MsgVals
    }^* \enspace .
\]
With this, it is obvious that $w$ is forwarding.

Suppose that $u$ is infinite.
In this case, the same reasoning almost applies.
The only difference is that it is possible that some send events are unmatched.
However, because $u$ respects $\chanBounds$, we know that there can be at most $\chanBounds(\channel{\procA}{\procB})$ such unmatched send events.
Thus, for $\procChanI{\procA}{\procB}{i}$, there can be at most one unmatched send event.
This is precisely the difference between almost forwarding and forwarding.
\proofEndSymbol
\end{proof}

It will not be sufficient to only check that an implementation is (almost) forwarding.
What is important is that the channel participant can react to all possible messages properly (and forwards them).
We call this \emph{amicable}.

\begin{definition}[Amicable]
Let $A_\procA$ be an FSM over $\AlphSyncWCh_{\procA}$,
$A_{\procChanI{\procA}{\procB}{i}}$ be an FSM over $\AlphSyncWCh_{\procChanI{\procA}{\procB}{i}}$,
and $B$ be the channel bound for $\channel{\procA}{\procB}$.
We say that
$A_{\procChanI{\procA}{\procB}{i}}$
is \emph{amicable} with
$A_\procA$ if the following holds:
\begin{itemize}
 \item $\lang(A_\procA)$ is channel-ordered,
 \item $A_{\procChanI{\procA}{\procB}{i}}$ is forwarding, and
 \item for every run of $A_\procA$ with trace $u$ and
    $w = u \wproj_{\AlphSync_{\procChanI{\procA}{\procB}{i}}} = w_1 \ldots$,
    there is a run of $A_{\procChanI{\procA}{\procB}{i}}$ with trace
    $
    w_i \cat w'_0 \cat w_{i+1\cdot B} \cat w'_{1} \cat w_{i+2\cdot B} \cat w'_2 \cat \ldots
    $
    where,
    for every $j$,
    $w'_j = \snd{\procChanI{\procA}{\procB}{i}}{\procB}{\val}$
    if
    $w_{i+j\cdot B} = \rcv{\procA}{\procChanI{\procA}{\procB}{i}}{\val}$,
    and for every trace with prefix
    $w_i \cat w'_0 \cat \ldots w_{i+j\cdot B}$,
    the only continuation is $w'_{j}$.
\end{itemize}
We say a CSM $\CSMwCh{A}$ is amicable if,
for every $\procA \in \Procs$, 
$A_{\procChanI{\procA}{\procB}{i}}$
is \emph{amicable} with
$A_\procA$
for every $\procChanI{\procA}{\procB}{i} \in \ProcsCh$.
\end{definition}

Given an amicable projection of an encoded PSM, there will be FSMs for channel participants.
We observe that we can match the states of such FSMs to different letters which ought to be forwarded.

\begin{proposition}[Matching states and messages for amicable CSMs]
\label{prop:matching-states-and-messages}Let $\CSMwCh{A}$ be an amicable CSM that satisfies feasible eventual reception.
For every $\procChanI{\procA}{\procB}{i} \in \ProcsCh$ with
$A_{\procChanI{\procA}{\procB}{i}} = (Q, \AlphSync_{\procChanI{\procA}{\procB}{i}}, \delta, q_{0}, F)$,
there is a function $f \from Q \to \MsgVals \dunion \set{\emptystring}$
such that:
\begin{itemize}
 \item if $f(q) = \emptystring$, for every $(q, x, q') \in \delta$, it holds that
 $x = \rcv{\procA}{\procChanI{\procA}{\procB}{i}}{\val}$ for some $\val \in \MsgVals$,
 \item if $f(q) = \val$, then $q$ is not final and for every $(q, x, q') \in \delta$, it holds that
 $x = \snd{\procChanI{\procA}{\procB}{i}}{\procB}{\val}$ for some $\val \in \MsgVals$.
\end{itemize}
\end{proposition}

We now prove that amicable CSMs with channel participants can be used to mimic CSMs without channel participants and vice versa, when the right encoding and decoding functions for words is applied. 

\begin{lemma}
    \label{lm:channel-encoding-prop-words-and-csm-func-rev}
Let $\CSMwCh{A}$ be an amicable CSM. For every word $w \in (\AlphAsyncWCh)^\infty$
 that respects~$\chanBounds$,
 it holds that \\
    $\encchanwords(w) \in \lang(\CSMwCh{A})$
    iff $\decchanwords(\encchanwords(w)) \in \lang(\CSMl{\decchanfsm(A_\procA)})$.
\end{lemma}
\begin{proof}
We show that the run for $\encchanwords(w)$ in $\CSMwCh{A}$
    can be simulated in
    $\CSMl{\decchanfsm(A_\procA)}$
    for $\decchanwords(\encchanwords(w))$,
    and vice versa.

    Prior, we establish some notation that will be used for both cases.
    We denote $\encchanwords(w) \is \encchanwords(w_1) \cat \encchanwords(w_2) \cat \ldots$,
    assuming that $w \is w_1 \cat w_2 \cat \ldots$
    where $w_i \in \AlphAsync$ for every $i$.
    Let $\run$ be the run in $\CSMwCh{A}$ for $\encchanwords(w)$:
    \[
     (\vec{s}_0, \vec{t}_0, \xi'_0) \xrightarrow{\encchanwords(w_1)}\!\!^2 \,
     (\vec{s}_1, \vec{t}_1, \xi'_1) \xrightarrow{\encchanwords(w_2)}\!\!^2 \,
     \ldots
    \]
    where we split the states for participants from $\Procs$, given by $\vec{s}_i$,
    and the ones from $\ProcsCh$, given by~$\vec{t}_i$.
    Let $\run'$ be the run in
    $\CSMl{\decchanfsm(A_\procA)}$
    for $\decchanwords(\encchanwords(w))$:
    $
     (\vec{q}_0, \xi_0) \xrightarrow{\decchanwords(\encchanwords(w_1))}
     (\vec{q}_1, \xi_1) \xrightarrow{\decchanwords(\encchanwords(w_2))}
     \ldots
    $.

    Compared to the proof of
    \cref{def:channel-encoding-prop-words-and-csm-func-align},
    we are not given $\chanBoundsMapSnd$ and $\chanBoundsMapRcv$ explicitly.
    However, we can construct them from the part of $w$ which has been consumed already.
    Formally, we define
    $\contentOf(\vec{t}, w_1 \ldots w_l, \channel{\procA}{\procB}) \is u_i \ldots u_j$
    for the states $\vec{t}$ of channel participants and a prefix $w_1 \ldots w_l$ of $w$ if the following holds:
    \begin{itemize}
        \item $j \modOp \card{(w_1 \ldots w_{i}) \wproj_{\snd{\procA}{\procB}{\_}}}$,
        \item $i \modOp \card{(w_1 \ldots w_{l}) \wproj_{\rcv{\procA}{\procB}{\_}}}$,
        \item for every $i \leq k < j$, it holds that $\vec{t}_{\procChanI{\procA}{\procB}{k}} = q'$ with $f(q') = \val$ and $w_k = \val$, and
        \item for every $j \leq k < i$, it holds that $\vec{t}_{\procChanI{\procA}{\procB}{k}} = q'$ with $f(q') = \emptystring$.
    \end{itemize}
    where we use the function $f(\hole)$ from
    \cref{prop:matching-states-and-messages}
    to distinguish between send and receive states of the channel participants.
    Note that we use the same name $\contentOf(\hole, \hole, \hole)$ but different parameters.
    Also, we have $\contentOf(\vec{t}, \emptystring, \channel{\procA}{\procB}) = \emptystring$ by definition.
    In fact, this only happens for initial states and this coincides with the fact that all channel participants are in their initial states initially, not storing any messages in transit.

    First, let us assume there exists run $\run$ 
    for $\encchanwords(w)$ 
    in $\CSMlwCh{\encchanfsm(A_\procA)}$. 
    We claim the following:
    for every $i$
    with
    $
     (\vec{s}_i, \vec{t}_i, \xi'_i) \xrightarrow{\encchanwords(w_{i})}\!\!^k \,
     (\vec{s}_{i+1}, \vec{t}_{i+1}, \xi'_{i+1})
    $
    for $k \in \set{1,2}$
    and
    $(\vec{q}_i, \xi_i)$
    such that
    \begin{enumerate}[label=(\alph*)]
     \item \label{cond2:woCh-sim-wCh-CSM-1}
            $\vec{q}_i = \vec{s}_i$,
     \item \label{cond2:woCh-sim-wCh-CSM-2}
$\xi'_i(\channel{\procA}{\procB}) = \xi_i(\channel{\procA}{\procB})$
            for every $\channel{\procA}{\procB} \notin \domainOf(\chanBounds)$,
     \item \label{cond2:woCh-sim-wCh-CSM-3}
$\xi'_i(\channel{\procA}{\procB}) = \emptystring$
            for every $\channel{\procA}{\procB} \in \domainOf(\chanBounds)$,
     \item \label{cond2:woCh-sim-wCh-CSM-4}
$\xi_i(\channel{\procA}{\procB}) =
        \contentOf(\vec{t}_i, w_1 \ldots w_{i-1}, \channel{\procA}{\procB})$
            for every $\channel{\procA}{\procB} \in \domainOf(\chanBounds)$,
    \end{enumerate}
    there is
    $(\vec{q}_{i+1}, \xi_{i+1})$
    such that
    \begin{enumerate}[label=(\alph*')]
     \item \label{prop2:woCh-sim-wCh-CSM-1}
           $\vec{q}_{i+1} = \vec{s}_{i+1}$,
     \item \label{prop2:woCh-sim-wCh-CSM-2}
$\xi'_{i+1}(\channel{\procA}{\procB}) = \xi_{i+1}(\channel{\procA}{\procB})$
            for every $\channel{\procA}{\procB} \notin \domainOf(\chanBounds)$,
     \item \label{prop2:woCh-sim-wCh-CSM-3}
$\xi'_{i+1}(\channel{\procA}{\procB}) = \emptystring$
            for every $\channel{\procA}{\procB} \in \domainOf(\chanBounds)$,
     \item \label{prop2:woCh-sim-wCh-CSM-4}
$\xi_{i+1}(\channel{\procA}{\procB}) =
        \contentOf(\vec{t}_{i+1}, w_1 \ldots w_{i}, \channel{\procA}{\procB})$
            for every $\channel{\procA}{\procB} \in \domainOf(\chanBounds)$, and
     \item \label{prop2:woCh-sim-wCh-CSM-5}
            $
            (\vec{q}_i, \xi_i) \xrightarrow{w_i}
            (\vec{q}_{i+1}, \xi_{i+1})
            $.
    \end{enumerate}

    To use this claim, we need to argue that the initial states are related, \ie the conditions hold for $i=0$.
    Conditions
    \ref{cond2:woCh-sim-wCh-CSM-1}
    to
    \ref{cond2:woCh-sim-wCh-CSM-3}
    are trivially satisfied.
    For \ref{cond2:woCh-sim-wCh-CSM-4},
    it suffices to recall that $\contentOf(\vec{t}, \emptystring, \channel{\procA}{\procB}) = \emptystring$.

    Now, we prove the claim.
    Let $i \in \Nat$.
    We assume that
    \ref{cond2:woCh-sim-wCh-CSM-1}
    to
    \ref{cond2:woCh-sim-wCh-CSM-4}
    hold and prove
    \ref{prop2:woCh-sim-wCh-CSM-1}
    to
    \ref{prop2:woCh-sim-wCh-CSM-5}.
    We use
    \ref{prop2:woCh-sim-wCh-CSM-5}
    to obtain
    $
       (\vec{q}_{i+1}, \xi_{i+1})
    $.
    We will still need to show, though, that this transition is possible.
    Let us do a case analysis on the shape of~$w_i$.
    We do a simultaneous case analysis whether $w_i$ is a send or receive transition and whether the corresponding channel is in the domain of $\chanBounds$.
    \begin{itemize}
     \item $w_i = \snd{\procA}{\procB}{\val}$ and
            $\channel{\procA}{\procB} \notin \domainOf(\chanBounds)$: \\
        Because of $\channel{\procA}{\procB} \notin \domainOf(\chanBounds)$, we have that
        $\encchanwords(w_i) = w_i$.
        Participant $\procA$ is in the same state in $\vec{s}$ and $\vec{q}$ by
        \ref{cond2:woCh-sim-wCh-CSM-1}.
        Since $w_i$ is a send transition, it is always enabled and, thus, possible.
        The channels change in the same way in both runs: $\val$ is appended to $\channel{\procA}{\procB}$.
        The states also change in the same way: only the one of $\procA$ changes but in the same way due to determinism.
        Thus, it is easy to check that all properties
        \ref{prop2:woCh-sim-wCh-CSM-1}
        to
        \ref{prop2:woCh-sim-wCh-CSM-4}
        are satisfied.

     \item $w_i = \snd{\procA}{\procB}{\val}$ and
            $\channel{\procA}{\procB} \in \domainOf(\chanBounds)$: \\
        By definition,
        $\encchanwords(w_i) =
        \msgFromTo{\procA}{\procChanI{\procA}{\procB}{j}}{\val}$
        for
        $\roleFmt{j} \modOp \card{(w_1 \ldots w_{i-1}) \wproj_{\snd{\procA}{\procB}{\_}}}$.
        Thus, there are two transitions in $\CSMwCh{A}$:
        one send transition labelled with
        $\snd{\procA}{\procChanI{\procA}{\procB}{j}}{\val}$
        and one receive transition labelled with
        $\rcv{\procA}{\procChanI{\procA}{\procB}{j}}{\val}$.
        Participant $\procA$ is in the same state in $\vec{s}$ and $\vec{q}$ by
        \ref{cond2:woCh-sim-wCh-CSM-1}.
        Since $w_i$ is a send transition, it is always enabled and, thus, possible.
        \ref{prop2:woCh-sim-wCh-CSM-1} holds by determinism of the state machines for~$\procA$.
        \ref{prop2:woCh-sim-wCh-CSM-2} is satisfied as we deal with a transition for which
        $\channel{\procA}{\procB} \in \domainOf(\chanBounds)$.
        \ref{prop2:woCh-sim-wCh-CSM-3} is satisfied as the only message enqueued into a channel, with the first transition, is immediately received, with the second one.
        For the remaining property, we show that the changes related to the second transition match the ones of the channel in $\CSMl{\decchanfsm(A_\procA)}$.
        It is obvious that we only need to consider the changes of the channel $\channel{\procA}{\procB}$.
        For this, we have
        $
            \xi_{i+1}(\channel{\procA}{\procB}) =
            \xi_{i}(\channel{\procA}{\procB}) \cat \val
        $.
        With
        $\xi_i(\channel{\procA}{\procB}) =
        \contentOf(\vec{t}_i, w_1 \ldots w_{i-1}, \channel{\procA}{\procB})$,
        it suffices to show that
        $
        \contentOf(\vec{t}_i, w_1 \ldots w_{i-1}, \channel{\procA}{\procB}) \cat \val
        =
        \contentOf(\vec{t}_{i+1}, w_1 \ldots w_{i}, \channel{\procA}{\procB})
        $.
        And this easily follows from the definition of $\contentOf(\hole, \hole, \hole)$ because the only state to change is the one of $\procChanI{\procA}{\procB}{j}$ and adding $w_i$ changes the indices in the way that the state of this channel participant is considered.

     \item $w_i = \rcv{\procA}{\procB}{\val}$ and
            $\channel{\procA}{\procB} \notin \domainOf(\chanBounds)$: \\
        Because of $\channel{\procA}{\procB} \notin \domainOf(\chanBounds)$, we have that
        $\encchanwords(w_i) = w_i$.
        Participant $\procB$ is in the same state in $\vec{s}$ and $\vec{q}$ by
        \ref{cond2:woCh-sim-wCh-CSM-1}.
        From \ref{cond2:woCh-sim-wCh-CSM-2} and the fact that the transition is possible in $\CSMwCh{A}$, we know that the transition is also possible in $\CSMl{\decchanfsm(A_\procA)}$.
        The channels change in the same way in both runs: $\val$ is removed from $\channel{\procA}{\procB}$.
        The states also change in the same way: only the one of $\procA$ changes but in the same way due to determinism.
        Thus, it is easy to check that all properties
        \ref{prop2:woCh-sim-wCh-CSM-1}
        to
        \ref{prop2:woCh-sim-wCh-CSM-4}
        are satisfied.

     \item $w_i = \rcv{\procA}{\procB}{\val}$ and
            $\channel{\procA}{\procB} \in \domainOf(\chanBounds)$: \\
        By definition,
        $\encchanwords(w_i) =
        \msgFromTo{\procChanI{\procA}{\procB}{j}}{\procB}{\val}$
        for
        $\roleFmt{j} \modOp \card{(w_1 \ldots w_{i-1}) \wproj_{\rcv{\procA}{\procB}{\_}}}$.
        Thus, there are two transitions in $\CSMwCh{A}$:
        one send transition labelled with
        $\snd{\procChanI{\procA}{\procB}{j}}{\procB}{\val}$
        and one receive transition labelled with
        $\rcv{\procChanI{\procA}{\procB}{j}}{\procB}{\val}$.
        In $\CSMl{\decchanfsm(A_\procA)}$, we only have one transition.

        We first argue that this transition is possible:
        participant $\procB$ is in the same state in $\vec{s}$ and $\vec{q}$ by
        \ref{cond2:woCh-sim-wCh-CSM-1} so it is able to receive $\val$ and,
        from \ref{cond2:woCh-sim-wCh-CSM-4}, we know that $\val$ is enqueued at the head of $\channel{\procA}{\procB}$ in $\xi_{i}$.

        It remains to show that the other properties are satisfied, for which the arguments are analogous.
        We still repeat them for understandability.
        \ref{prop2:woCh-sim-wCh-CSM-1} holds by determinism of the state machines for~$\procA$.
        \ref{prop2:woCh-sim-wCh-CSM-2} is satisfied as we deal with a transition for which
        $\channel{\procA}{\procB} \in \domainOf(\chanBounds)$.
        \ref{prop2:woCh-sim-wCh-CSM-3} is satisfied as the only message enqueued into a channel, with the first transition, is immediately received, with the second one.
        For the remaining property, we show that the changes related to the first transition match the ones of the channel in $\CSMl{\decchanfsm(A_\procA)}$.
        It is obvious that we only need to consider the changes of the channel $\channel{\procA}{\procB}$.
        For this, we have
        $
            \xi_{i}(\channel{\procA}{\procB}) =
            \val \cat \xi_{i+1}(\channel{\procA}{\procB})
        $.
        With
        $\xi_i(\channel{\procA}{\procB}) =
        \contentOf(\vec{t}_i, w_1 \ldots w_{i-1}, \channel{\procA}{\procB})$,
        it suffices to show that
        $
        \contentOf(\vec{t}_i, w_1 \ldots w_{i-1}, \channel{\procA}{\procB})
        =
        \val \cat \contentOf(\vec{t}_{i+1}, w_1 \ldots w_{i}, \channel{\procA}{\procB})
        $.
        And this easily follows from the definition of $\contentOf(\hole, \hole, \hole)$ because the only state to change is the one of $\procChanI{\procA}{\procB}{j}$ and adding $w_i$ changes the indices in the way that the state of this channel participant is considered.
    \end{itemize}
    This concludes the proof for this direction.

    Second, let us assume there exists run $\run'$ for
    $\CSMl{\decchanfsm(A_\procA)}$ for $w$.
    We claim the following:
    for every $i$
    with
    $
     (\vec{q}_i, \xi_i)
        \xrightarrow{w_{i}}
     (\vec{q}_{i+1}, \xi_{i+1})
    $
    and
    $(\vec{s}_i, \vec{t}_i, \xi'_i)$
    such that
    \begin{enumerate}[label=(\alph*)]
     \item \label{cond3:woCh-sim-wCh-CSM-1}
            $\vec{q}_i = \vec{s}_i$,
     \item \label{cond3:woCh-sim-wCh-CSM-2}
$\xi'_i(\channel{\procA}{\procB}) = \xi_i(\channel{\procA}{\procB})$
            for every $\channel{\procA}{\procB} \notin \domainOf(\chanBounds)$,
     \item \label{cond3:woCh-sim-wCh-CSM-3}
$\xi'_i(\channel{\procA}{\procB}) = \emptystring$
            for every $\channel{\procA}{\procB} \in \domainOf(\chanBounds)$,
     \item \label{cond3:woCh-sim-wCh-CSM-4}
$\xi_i(\channel{\procA}{\procB}) =
        \contentOf(\vec{t}_i, w_1 \ldots w_{i-1}, \channel{\procA}{\procB})$
            for every $\channel{\procA}{\procB} \in \domainOf(\chanBounds)$,
    \end{enumerate}
    there is
    $(\vec{s}_{i+1}, \vec{t}_{i+1}, \xi'_{i+1})$
    such that
    \begin{enumerate}[label=(\alph*')]
     \item \label{prop3:wCh-sim-woCh-CSM-1}
           $\vec{q}_{i+1} = \vec{s}_{i+1}$,
     \item \label{prop3:wCh-sim-woCh-CSM-2}
$\xi'_{i+1}(\channel{\procA}{\procB}) = \xi_{i+1}(\channel{\procA}{\procB})$
            for every $\channel{\procA}{\procB} \notin \domainOf(\chanBounds)$,
     \item \label{prop3:wCh-sim-woCh-CSM-3}
$\xi'_{i+1}(\channel{\procA}{\procB}) = \emptystring$
            for every $\channel{\procA}{\procB} \in \domainOf(\chanBounds)$,
     \item \label{prop3:wCh-sim-woCh-CSM-4}
$\xi_{i+1}(\channel{\procA}{\procB}) =
        \contentOf(\vec{t}_{i+1}, w_1 \ldots w_{i}, \channel{\procA}{\procB})$
            for every $\channel{\procA}{\procB} \in \domainOf(\chanBounds)$, and
     \item \label{prop3:wCh-sim-woCh-CSM-5}
            $
            (\vec{s}_{i}, \vec{t}_{i}, \xi'_{i})
            \xrightarrow{\encchanwords(w_i)}
            (\vec{s}_{i+1}, \vec{t}_{i+1}, \xi'_{i+1})
            $.
    \end{enumerate}

    The conditions are basically the same.
    Thus, the initial states are still related, making the above claim sufficient to show our simulation argument.
    The proof is analogous to the one before and, thus, omitted.
    There is one minor difference: one exploits the fact that a CSM is amicable to show that the channel participants faithfully forward the respective messages.
\end{proof}

Equipped with these, we can show that both ways of the encoding are fine to take.

\begin{restatable}{lemma}{PSMimplementableIffSumOnePSMimplementable}
\label{lm:PSMimplementableIffSumOnePSMimplementable}
Let $\PSM$ be a PSM and $\CSMwCh{A}$ and $\CSM{B}$ be CSMs.
\begin{enumerate}[label=\textnormal{(\alph*)}]
 \item \label{lm:PSMimplementableIffSumOnePSMimplementableCoupledToDecoupled}
 If $\CSMwCh{A}$ is an amicable projection of $\encchanpsm(\PSM)$, \\
    then $\CSMl{\decchanfsm(A_\procA)}$ is a projection of~$\PSM$.
 \item \label{lm:PSMimplementableIffSumOnePSMimplementableDecoupledToCoupled}
 If $\CSM{B}$ is a projection of $\PSM$, then $\CSMlwCh{\encchanfsm(B_\procA)}$ is a projection of $\encchanpsm(\PSM)$.\end{enumerate}
\end{restatable}
\begin{proof}
First, we prove \ref{lm:PSMimplementableIffSumOnePSMimplementableCoupledToDecoupled}.
Thus, we assume that
$\CSMwCh{A}$ is deadlock-free and forwarding as well as
$\lang(\CSMwCh{A}) = \semantics(\encchanpsm(\PSM))$.
We prove deadlock freedom and protocol fidelity.

For deadlock freedom,
we assume, towards a contradiction,
$\CSMl{\decchanfsm(A_\procA)}$
has a deadlock.
Let $w$ be the trace of this run ending in a deadlock.
We used a simulation argument to show that there is a run with trace $\encchanwords(w)$ in
$\CSMwCh{A}$
in the second part of the proof for
\cref{lm:channel-encoding-prop-words-and-csm-func-rev}.
We can apply this result and reach also reach a configuration in $\CSM{A}$ from which no further transitions can be taken.
By determinism of each $A_\procA$, this configuration is unique.
Hence, this is a deadlock and, yielding a contradiction.

For protocol fidelity, the following equivalences prove the claim:
\smallskip

{ \small
\begin{tabular}{cl}
    & $w \in \semantics(\PSM)$ \\
$\overset{\ref{def:channel-encoding-prop-words-and-psm-func-align}}{\Leftrightarrow}$
    & $\encchanwords(w) \in \semantics(\encchanpsm(\PSM))$ \\
$\overset{Assumption}{\Leftrightarrow}$
    & $\encchanwords(w) \in \lang(\CSMwCh{A})$ \\
$\overset{\cref{lm:channel-encoding-prop-words-and-csm-func-rev}}{\Leftrightarrow}$
    & $\decchanwords(\encchanwords(w)) \in \lang(\CSMl{\decchanfsm(A_\procA)})$ \\
$\overset{\ref{def:channel-encoding-prop-words-almost-bij-1}} {\Leftrightarrow}$
    & $w \in \lang(\CSMl{\decchanfsm(A_\procA)})$ \\
\end{tabular}
}

\medskip

Second, we prove \ref{lm:PSMimplementableIffSumOnePSMimplementableDecoupledToCoupled}.
Thus, we assume that
$\CSM{B}$ is deadlock-free and
$\lang(\CSM{B}) = \semantics(\PSM)$.
We prove deadlock freedom and protocol fidelity.

For deadlock freedom,
we assume, towards a contradiction,
$\CSMlwCh{\encchanfsm(B_\procA)}$
has a deadlock.
Let $w$ be the trace of this run ending in a deadlock.
By construction of channel participants, there is $w' \in (\AlphSyncWCh)^*$ such that $w' \interswap w$.
Thus, there is $u \in \AlphAsync^*$ such that $w' = \encchanwords(u)$.
We used a simulation argument to show that there is a run with trace $u$ in
$\CSMwCh{B}$
in the second part of the proof for
\cref{def:channel-encoding-prop-words-and-csm-func-align}.
We can apply this result and also reach a configuration in $\CSM{A}$ from which no further transitions can be taken.
By determinism of each $B_\procA$, this configuration is unique.
Hence, this is a deadlock and, yielding a contradiction.

For protocol fidelity, the following equivalences prove the claim: \\

{ \small
\begin{tabular}{cl}
    & $w \in \semantics(\encchanpsm(\PSM))$ \\
$\overset{\text{def.\ \& constr.}}{\Leftrightarrow}$
    & $\exists w' \in (\AlphSyncWCh)^\infty \text{ such that }
    w' \interswap w \text{ and }
    w' \in \semantics(\encchanpsm(\PSM))$ \\
$\overset{\ref{def:channel-encoding-prop-words-and-psm-func-rev}}{\Leftrightarrow}$
    & $\exists w' \in (\AlphSyncWCh)^\infty \text{ such that }
    w' \interswap w \text{ and }
    \decchanwords(w') \in \semantics(\PSM)$ \\
$\overset{Assumption}{\Leftrightarrow}$
    & $\exists w' \in (\AlphSyncWCh)^\infty \text{ such that }
    w' \interswap w \text{ and }
    \decchanwords(w') \in \lang(\CSM{B})$ \\
$\overset{\cref{def:channel-encoding-prop-words-and-csm-func-align}}{\Leftrightarrow}$
    & $\exists w' \in (\AlphSyncWCh)^\infty \text{ such that }
    w' \interswap w \text{ and }$ \\
    & \hspace{17mm}
    $\encchanwords(\decchanwords(w')) \in \lang(\CSMlwCh{\encchanfsm(B_\procA)})$ \\
$\overset{\ref{def:channel-encoding-prop-words-almost-bij-2}}{\Leftrightarrow}$
    & $\exists w' \in (\AlphSyncWCh)^\infty \text{ such that }
    w' \interswap w \text{ and }
    w' \in \lang(\CSMlwCh{\encchanfsm(B_\procA)})$ \\
$\overset{(*)}{\Leftrightarrow}$
    & $w \in \lang(\CSMlwCh{\encchanfsm(B_\procA)})$
\end{tabular}
}
\medskip

where the first equivalence follows by definition of the semantics but we keep the semantics $\semantics(\hole)$ instead of the language $\lang(\hole)$, and $(*)$ follows from the fact that CSMs are closed under $\interswap$ \cite[Lm.\,22]{DBLP:conf/concur/MajumdarMSZ21}.
Furthermore,
$\decchanwords(w')$ respects~$\chanBounds$ because every word in $\semantics(\PSM)$ does
(\cref{prop:respecting-channel-bounds-from-psm-to-semantics})
and $w'$ is channel-ordered by
\ref{def:channel-encoding-prop-psm-func-channel-ordered}.~\proofEndSymbol
\end{proof}

It remains to show that we can obtain an amicable projection.
For this, we use the fact that we can obtain a local language preserving projection.

\begin{definition}[Local language preserving projection \cite{DBLP:conf/cav/LiSWZ23}]
Let $L \subseteq \AlphAsync^\infty$ be a language and $\CSM{A}$ be a projection of $L$.
We say that $\CSM{A}$ is \emph{local language preserving} if it holds that
$L \wproj_{\AlphAsync_\procA} = \lang(A_\procA)$
for every $\procA \in \Procs$.
For a CSM $\CSMwCh{B}$ to be local language preserving, we require
$L \wproj_{\AlphAsync_\procA} = \lang(B_\procA)$
for every $\procA \in \ProcsWCh$.
\end{definition}

\begin{lemma}[\cite{DBLP:conf/cav/LiSWZ23}] \label{thm:cav23-and-arxiv24-result}
\label{thm:cav-result}
Let $\PSM$ be a projectable \sinkfinal sender-driven \sumOnePSM and let $\CSM{A}$ be its subset projection from \textnormal{\cite[Def.\,5.4]{DBLP:conf/cav/LiSWZ23}}. 
The subset projection can be computed in PSPACE and $\CSM{A}$ is a local language preserving projection.
\end{lemma}
\begin{proof}
Their subset projection is defined for global types but the algorithm works on a global state machine which is a \sinkfinal \sumOnePSM.
The rest follows from~\cite[Cor.\,8.2]{DBLP:conf/cav/LiSWZ23}.
\proofEndSymbol
\end{proof}

\begin{lemma}
\label{lm:local-language-preserving-implies-forwarding}
Let $\PSM$ be a PSM and let $\CSMwCh{A}$ be a CSM.
If $\CSMwCh{A}$ is a local language preserving projection for $\encchanpsm(\PSM)$, then $\CSMwCh{A}$ is amicable.
\end{lemma}
\begin{proof}
For every channel participant $\procChan{\procA}{\procB}$, we have to show that $A_{\procChan{\procA}{\procB}}$ is amicable with $A_\procA$.
Let $\procChanI{\procA}{\procB}{i}$ be a channel participant.
By construction of $\encchanpsm$, we know that
$\procChanI{\procA}{\procB}{i}$ solely receives from $\procA$ in $\encchanpsm(\PSM)$.
From
\cref{lm:chan-part-wproj-forwarding},
we know that
$\semantics(\encchanpsm(\PSM)) \wproj_{\AlphAsyncWCh_{\procChan{\procA}{\procB}}}$
is forwarding or almost forwarding (only possible if the projected word is infinite).
Since $\CSMwCh{A}$ is a local language preserving projection of $\encchanpsm(\PSM)$,
we know that all words in
$\semantics(\encchanpsm(\PSM)) \wproj_{\AlphAsyncWCh_{\procChan{\procA}{\procB}}}$
have a trace in $A_{\procChan{\procA}{\procB}}$.
Last, by construction, $\semantics(\encchanpsm(\PSM)) \wproj_{\AlphAsync_{\procA}}$ is channel-ordered.
Thus, again, by local language preservation, $A_\procA$ is channel-ordered.
Thus, it follows that $A_{\procChan{\procA}{\procB}}$ is amicable with $A_\procA$.
\proofEndSymbol
\end{proof}

Our main theorem combines the previous observations.

\implDecPSMsChoice*
\begin{proof}
  Let $(\PSM, \chanBounds)$ be a Tame PSM, \ie be sender-driven, \sinkfinal and respect channel bounds $\chanBounds$.
  We apply the channel-participant encoding to obtain
  the \sinkfinal \sumOnePSM $\encchanpsm(\PSM)$.

  We claim that $\PSM$ is projectable if and only if $\encchanpsm(\PSM)$ is projectable.

  For the direction from left to right,
  let $\CSM{B}$ be the witness for projectability of $\PSM$.
  With \cref{lm:PSMimplementableIffSumOnePSMimplementable}\ref{lm:PSMimplementableIffSumOnePSMimplementableDecoupledToCoupled},
  $\CSMlwCh{\encchanfsm(B_\procA)}$ is a projection of $\encchanpsm(\PSM)$, proving its projectability.

  For the direction from right to left,
  we apply
  \cref{thm:cav23-and-arxiv24-result} 
  to obtain a local language preserving projection $\CSM{A}$ for $\encchanpsm(\PSM)$.
  With
  \cref{lm:local-language-preserving-implies-forwarding},
  it follows that $\CSM{A}$ is amicable.
  By
  \cref{lm:PSMimplementableIffSumOnePSMimplementable}\ref{lm:PSMimplementableIffSumOnePSMimplementableCoupledToDecoupled},
  $\CSMlwCh{\decchanfsm(A_\procA)}$ is a projection of~$\PSM$, proving its projectability.

  This proves that $\PSM$ is projectable if and only if $\encchanpsm(\PSM)$ is projectable.

  The encodings (and decodings) can be constructed in polynomial time.
  Thus, we can check projectability and also obtain an projection of $\PSM$ in \mbox{PSPACE} (\cref{thm:cav-result}).

\proofEndSymbol
\end{proof}
 
\subsection{Additional Material for \cref{sec:mixed-choice-yields-undecidable-projection}}
\label{app:checking-implementability-mixed-choice-undec}

\checkingImplSinkfinalMixedChoicePSMUndec*

\begin{proof}
We consider the problem of checking if a word is accepted by a Turing Machine.
This is known to be undecidable and we reduce it to checking projectability of a mixed-choice \sinkfinal \sumOnePSM.
We basically construct a PSM with two branches at the top. 
For each, we construct a language and they coincide -- which will be necessary for projectability -- if and only if the Turing Machine does not halt in a final configuration.
We assume familiarity with the concept of Turing Machines and refer to \cite{DBLP:books/aw/HopcroftU79} for further details.

Let $\TM$ be a Turing Machine with tape alphabet $\TapeAlph$ and states $\TMStates$ such that $\TapeAlph \inters \TMStates = \emptyset$.
We have that $q_0 \in \TMStates$ is the initial state and $q_f \in \TMStates$ is, without loss of generality, the only final state.
A configuration of $\TM$ is given by a word $a_1, \ldots, a_i, q, b_1, \ldots, b_j \in \TapeAlph^{\!*} \TMStates \TapeAlph^{\!*}$. The initial configuration for input word $w$ is
$q_0 w$ while any configuration from $\TapeAlph^{\!*} q_f \TapeAlph^{\!*}$ is final.
A computation is a sequence of configurations
$(u_1, \ldots, u_m)$ such that
$u_{i+1}$ is the next configuration of $\TM$,
also denoted by $u_i \TMNext u_{i+1}$.
A computation $(u_1, \ldots, u_m)$ accepts $w$ if $u_1 = q_0 w$ and $u_m \in \TapeAlph^{\!*} q_f \TapeAlph^{\!*}$.

For our encoding, we use five participants $\procA_1, \ldots, \procA_5$ who send configurations to each other.
Thus, messages are from the set
$\TapeAlph \dunion
\set{
    \markNewRound,
    \markBeginConf,
    \markEndConf,
    \markEndLoop
}
\dunion Q$
where $\markNewRound$ is sent by $\procA_3$ to indicate the start of a new pair of configurations
and $\markBeginConf$ and $\markEndConf$ delimit a configuration.

We introduce the notation 
$\msgBackForth{\procA}{\procB}{\val}$
to abbreviate 
$\msgFromTo{\procA}{\procB}{\val} \cat
 \msgFromTo{\procB}{\procA}{\val}$.
We only specify interactions using
$\msgBackForth{\_}{\_}{\_}$.
Using these, we will also define regular expressions and complements thereof and consider
$\msgBackForth{\procA}{\procB}{\val}$
as their single letters.
By construction, every PSM will be \sumBounded{1} and, in fact, every message is immediately acknowledged.

For a word $w = w_1 \ldots w_i$, we write
$\msgBackForth{\procA}{\procB}{w}$
for
$\msgBackForth{\procA}{\procB}{w_1}
 \cdots
 \msgBackForth{\procA}{\procB}{w_i}$.

For words $C_1, D_1, C_2, D_2 \ldots, C_m, D_m \in (\TapeAlph \dunion \TMStates)^*$, we define the word
{ \scriptsize
\begin{align*}
   w(C_1, D_1, C_2, D_2, \ldots, C_m, D_m)
   \is \quad
   &
    \msgBackForth{\procA_3}{\procA_2}{\markNewRound} \cat
    \msgBackForth{\procA_2}{\procA_1}{\markBeginConf} \cat
    \msgBackForth{\procA_2}{\procA_1}{C_1} \cat
    \msgBackForth{\procA_2}{\procA_1}{\markEndConf} \cat \invisibleEndLine
    \\ &
    \msgBackForth{\procA_3}{\procA_4}{\markNewRound} \cat
    \msgBackForth{\procA_4}{\procA_5}{\markBeginConf} \cat
    \msgBackForth{\procA_4}{\procA_5}{D_1} \cat
    \msgBackForth{\procA_4}{\procA_5}{\markEndConf} \cat \invisibleEndLine
    \\ &
    \msgBackForth{\procA_3}{\procA_2}{\markNewRound} \cat
    \msgBackForth{\procA_2}{\procA_1}{\markBeginConf} \cat
    \msgBackForth{\procA_2}{\procA_1}{C_2} \cat
    \msgBackForth{\procA_2}{\procA_1}{\markEndConf} \cat \invisibleEndLine
    \\ &
    \msgBackForth{\procA_3}{\procA_4}{\markNewRound} \cat
    \msgBackForth{\procA_4}{\procA_5}{\markBeginConf} \cat
    \msgBackForth{\procA_4}{\procA_5}{D_2} \cat
    \msgBackForth{\procA_4}{\procA_5}{\markEndConf} \cat
    \\ & \cdots \\ &
    \msgBackForth{\procA_3}{\procA_2}{\markNewRound} \cat
    \msgBackForth{\procA_2}{\procA_1}{\markBeginConf} \cat
    \msgBackForth{\procA_2}{\procA_1}{C_m} \cat
    \msgBackForth{\procA_2}{\procA_1}{\markEndConf} \cat \invisibleEndLine
    \\ &
    \msgBackForth{\procA_3}{\procA_4}{\markNewRound} \cat
    \msgBackForth{\procA_4}{\procA_5}{\markBeginConf} \cat
    \msgBackForth{\procA_4}{\procA_5}{D_m} \cat
    \msgBackForth{\procA_4}{\procA_5}{\markEndConf} \phantom{\cat}
    \enspace .
\end{align*}
}

Intuitively, $\procA_2$ sends the sequence $C_i$ to $\procA_1$ while
$\procA_4$ sends the sequence $D_i$ to $\procA_5$.
Each sequence is started by a $\markBeginConf$-message and finished by a $\markEndConf$-message between the respective pair.
The participant $\procA_3$ starts each round by sending $\markNewRound$.

Note that the communication of $C_i$ between $\procA_1$ and $\procA_2$ can happen concurrently to both $D_{i-1}$ and $D_i$ between $\procA_4$ and $\procA_5$.
(This will later allow us to both detect if $C_i$ and $D_i$ do not coincide or $C_i$ is no successor configuration of $D_i$.)

We define two languages $L_l$ and $L_r$, which we later use for two branches of the PSM encoding. 
\begin{align*}
 L_l & \is
    \set{
        \interswaplang(w(C_1, D_1, \ldots, C_m, D_m)) \mid
        m \geq 1, C_1, D_1, \ldots, C_m, D_m \in (\TapeAlph \dunion \TMStates)^*
    }
 \\
 L_r & \is
    L_l \setminus
    \set{
        \interswaplang(w(u_1, u_1, \ldots, u_m, u_m)) \mid
        (u_1, \ldots, u_m) \text{ is an accepting computation}
    }
\end{align*}

To make the resulting PSM \sinkfinal, we define a sequence of messages that indicates the end of an execution:
\[
 w_{\mathit{end}} \is
    \msgBackForth{\procA_3}{\procA_2}{\markEndLoop} \cat
    \msgBackForth{\procA_2}{\procA_1}{\markEndLoop} \cat
    \msgBackForth{\procA_3}{\procA_4}{\markEndLoop} \cat
    \msgBackForth{\procA_4}{\procA_5}{\markEndLoop}
\]
and append it to obtain
$L'_l  \is \set{w \cat w_{\mathit{end}} \mid w \in L_l}$
and
    $L'_r  \is \set{w \cat w_{\mathit{end}} \mid w \in L_r}$.

We will show that both $L_l$ and $L_r$, and thus, $L'_l$ and $L'_r$, can be specified as \sumOnePSMs.
Provided with PSMs for $L_l$ and $L_r$, it is straightforward to construct a $\PSM_{\TM}$ such that
\[
\semantics(\PSM_{\TM}) =
    \set{\msgFromTo{\procA_2}{\procA_3}{\lbl{l}} \cat w \mid w \in L'_l}
    \dunion
    \set{\msgFromTo{\procA_2}{\procA_3}{\lbl{r}} \cat w \mid w \in L'_r}
    \enspace .
\]
By definition of $L'_l$ and $L'_r$, every word ends with $w_{\mathit{end}}$ so $\PSM_{\TM}$ is \sinkfinal.
(In fact, if there is an implementation, the FSM for each participant will also be \sinkfinal.)
We will show that $\PSM_{\TM}$ is projectable if and only if $\TM$ does not accept the input $w$.

For this, it suffices to establish the following four facts:
\begin{itemize}
 \item Claim 1: $L_l$ and $L_r$ can be specified as \sumOnePSMs.
 \item Claim 2: $L_l$ is projectable.
 \item Claim 3: If $\TM$ has no accepting computation for $w$, then $\lang(\PSM_{\TM})$ is projectable.
 \item Claim 4: If $\TM$ has an accepting computation for $w$, then $\lang(\PSM_{\TM})$ is not projectable.
\end{itemize}

\claimStm{1}
Both $L_l$ and $L_r$ can be specified as \sumOnePSMs.

\claimProofStart{1}
It is easy to construct a PSM from a regular expression.
Thus, for conciseness, we give regular expressions for the languages we consider or for their complements.
For this, we introduce some more notation for concise specifications when using sets of messages:
\[
    \msgBackForth{\procA_2}{\procA_1}{\set{x_1, \ldots, x_n}}
    \is
    (
    \msgBackForth{\procA_2}{\procA_1}{x_1}
    +
    \ldots
    +
    \msgBackForth{\procA_2}{\procA_1}{x_n}
    )
\enspace .
\]

First, let us consider $L_l$.
Inspired by the definition of $w(C_1, D_1, \ldots, C_m, D_m)$, we construct this regular expression $\regex[l]$ for $L_l$:
\begin{align*}
    &
    (
    \msgBackForth{\procA_3}{\procA_2}{\markNewRound} \cat
    \msgBackForth{\procA_2}{\procA_1}{\markBeginConf} \cat
    (\msgBackForth{\procA_2}{\procA_1}{(\TapeAlph \dunion Q)})^* \cat
    \msgBackForth{\procA_2}{\procA_1}{\markEndConf} \cat \invisibleEndLine
    \\
    &
    \phantom{(}
    \msgBackForth{\procA_3}{\procA_4}{\markNewRound} \cat
    \msgBackForth{\procA_4}{\procA_5}{\markBeginConf} \cat
    (\msgBackForth{\procA_4}{\procA_5}{(\TapeAlph \dunion Q)})^* \cat
    \msgBackForth{\procA_4}{\procA_5}{\markEndConf}
    )^*
    \enspace.
\end{align*}

Second, let us consider $L_r$.
Recall that $L_r$ should admit the encoding of all sequences of configurations except for accepting ones.
We provide an exhaustive list of how such a sequence can fail to be an accepting computation.
We provide a language $L_{r,i}$ for each and $L_r$ is their union.\footnote{We renumbered the languages because some of Lohrey's construction does not apply to this undecidability proof.}

\begin{itemize}
 \item $L_{r,1}$ contains all sequences of configurations for which some $C_k$ or $D_k$ is actually not a configuration, \ie not from $\TapeAlph^{\!*} \TMStates \TapeAlph^{\!*}$.
 \item $L_{r,2}$ contains all sequences for which $C_1$ is not the correct initial configuration, \ie it does not have the shape $q_0, a_1, a_2, \ldots a_n$ where $w = a_1 \cdots a_n$.
 \item $L_{r,3}$ contains all sequences for which $q_f$ does not occur in $C_m$.
 \item $L_{r,4}$ contains all sequences where $C_k$ and $D_k$ differ in some position. \item $L_{r,5}$ contains all sequences for which $C_{k+1}$ is no successor configuration for~$D_k$.
\end{itemize}

For each $L_{r,i}$, we show that it can be specified as PSM (or regular expression).
It is straightforward to obtain a PSM for $L_r$ by adding one initial state and adding a transition from this one to the initial state for the PSM of $L_{r,i}$ for each~$i$.

\textit{Language $L_{r,1}$}:

We construct a regular expression for
$w(C_1, D_1, \ldots, C_m, D_m)$ for any $m$ such that there is some $C_i$ or $D_i$ with either no message from $Q$ or at least two messages from $Q$:
{ \tiny
\begin{align*}
\regex[1] \is \quad &
    (\msgBackForth{\procA_3}{\procA_2}{\markNewRound} \cat
    \msgBackForth{\procA_2}{\procA_1}{\markBeginConf} \cat
    (\msgBackForth{\procA_2}{\procA_1}{\TapeAlph})^* \cat
    \msgBackForth{\procA_2}{\procA_1}{\TMStates} \cat
    (\msgBackForth{\procA_2}{\procA_1}{\TapeAlph})^* \cat
    \msgBackForth{\procA_2}{\procA_1}{\markEndConf} \cat \invisibleEndLine
    \\
    &
    \phantom{(}
    \msgBackForth{\procA_3}{\procA_4}{\markNewRound} \cat
    \msgBackForth{\procA_4}{\procA_5}{\markBeginConf} \cat
    (\msgBackForth{\procA_4}{\procA_5}{\TapeAlph})^* \cat
    \msgBackForth{\procA_4}{\procA_5}{\TMStates} \cat
    (\msgBackForth{\procA_4}{\procA_5}{\TapeAlph})^* \cat
    \msgBackForth{\procA_4}{\procA_5}{\markEndConf} )^* \cat
    \\
    &
    (\regex[l0] + \regex[l2] + \regex[r0] + \regex[r2]) \cat \regex[l]^*
    \\
\text{where} \\
\regex[l0] \is \quad &
    \msgBackForth{\procA_3}{\procA_2}{\markNewRound} \cat
    \msgBackForth{\procA_2}{\procA_1}{\markBeginConf} \cat
    (\msgBackForth{\procA_2}{\procA_1}{\TapeAlph})^* \cat
    \msgBackForth{\procA_2}{\procA_1}{\markEndConf} \cat \invisibleEndLine
    \\
    &
    \msgBackForth{\procA_3}{\procA_4}{\markNewRound} \cat
    \msgBackForth{\procA_4}{\procA_5}{\markBeginConf} \cat
    (\msgBackForth{\procA_4}{\procA_5}{\TapeAlph \dunion \TMStates})^* \cat
    \msgBackForth{\procA_4}{\procA_5}{\markEndConf} \cat
    \\
\regex[l2] \is \quad &
    \msgBackForth{\procA_3}{\procA_2}{\markNewRound} \cat
    \msgBackForth{\procA_2}{\procA_1}{\markBeginConf} \cat
    \\ & \quad
    (\msgBackForth{\procA_2}{\procA_1}{\TapeAlph})^* \cat
    \msgBackForth{\procA_2}{\procA_1}{\TMStates} \cat
    (\msgBackForth{\procA_2}{\procA_1}{\TapeAlph})^* \cat
    \msgBackForth{\procA_2}{\procA_1}{\TMStates} \cat
    (\msgBackForth{\procA_2}{\procA_1}{\TapeAlph \dunion \TMStates})^* \cat
    \msgBackForth{\procA_2}{\procA_1}{\markEndConf} \cat \invisibleEndLine
    \\
    &
    \msgBackForth{\procA_3}{\procA_4}{\markNewRound} \cat
    \msgBackForth{\procA_4}{\procA_5}{\markBeginConf} \cat
    (\msgBackForth{\procA_4}{\procA_5}{\TapeAlph \dunion \TMStates})^* \cat
    \msgBackForth{\procA_4}{\procA_5}{\markEndConf} \cat
    \\
\regex[r0] \is \quad &
    \msgBackForth{\procA_3}{\procA_2}{\markNewRound} \cat
    \msgBackForth{\procA_2}{\procA_1}{\markBeginConf} \cat
    (\msgBackForth{\procA_2}{\procA_1}{\TapeAlph \dunion \TMStates})^* \cat
    \msgBackForth{\procA_2}{\procA_1}{\markEndConf} \cat \invisibleEndLine
    \\
    &
    \msgBackForth{\procA_3}{\procA_4}{\markNewRound} \cat
    \msgBackForth{\procA_4}{\procA_5}{\markBeginConf} \cat
    (\msgBackForth{\procA_4}{\procA_5}{\TapeAlph})^* \cat
    \msgBackForth{\procA_4}{\procA_5}{\markEndConf} \cat
    \\
\regex[r2] \is \quad &
    \msgBackForth{\procA_3}{\procA_2}{\markNewRound} \cat
    \msgBackForth{\procA_2}{\procA_1}{\markBeginConf} \cat
    (\msgBackForth{\procA_2}{\procA_1}{\TapeAlph \dunion \TMStates})^* \cat
    \msgBackForth{\procA_2}{\procA_1}{\markEndConf} \cat \invisibleEndLine
    \\
    &
    \msgBackForth{\procA_3}{\procA_4}{\markNewRound} \cat
    \msgBackForth{\procA_4}{\procA_5}{\markBeginConf} \cat
    \\ & \quad
    (\msgBackForth{\procA_4}{\procA_5}{\TapeAlph})^* \cat
    \msgBackForth{\procA_4}{\procA_5}{\TMStates} \cat
    (\msgBackForth{\procA_4}{\procA_5}{\TapeAlph})^* \cat
    \msgBackForth{\procA_4}{\procA_5}{\TMStates} \cat
    (\msgBackForth{\procA_4}{\procA_5}{\TapeAlph \dunion \TMStates})^* \cat
    \msgBackForth{\procA_4}{\procA_5}{\markEndConf} \cat
\end{align*}
}

Let us explain how $\regex[1]$ works.
In the beginning, there are only $C_i$ and $D_i$ with one message from $\TMStates$.
At some point, one of the regular expressions
$\regex[l0], \regex[l2], \regex[r0]$ or $\regex[r2]$
has to match.
These specify some way how the number of messages for $\TMStates$ can be wrong:
$\regex[l0]$ has no messages from $\TMStates$ between $\procA_2$ and $\procA_1$ while
$\regex[r0]$ has no messages from $\TMStates$ between $\procA_4$ and $\procA_5$;
$\regex[l2]$ has more than one message from $\TMStates$ between $\procA_2$ and $\procA_1$ while
$\regex[r2]$ has more than one message from $\TMStates$ between $\procA_4$ and $\procA_5$.
For each, the other pair can communicate any number of messages from $\TMStates$ to account for sequences where both $C_i$ and $D_i$ do not match.
Subsequently, we use $\regex[l]$ to simply allow a sequence of any configuration.

\textit{Language $L_{r,2}$}:

Let $w = a_1 \ldots a_n$ be the input word for $\TM$.
Then, we can specify $L_{r,2}$ as follows:
\[
 \set{
    \interswaplang(w(C_1, D_1, C_2, D_2 \ldots, C_m, D_m))
        \mid
        m \geq 1 \land C_1 \neq q_0, a_1, \ldots, a_n
    }
\]

It is easy to see that we can change $\regex[l]$ to obtain a regular expression for
$\set{w(C_2, D_2, \ldots, C_m, D_m) \mid m \geq 1}$.
Thus, it suffices to show that
$w(C_1, D_1)$
with $C_1 \neq q_0, a_1, \ldots, a_n$
can be specified as PSM.
Again, we give a regular expression for the complement.
We observe that the communication between $\procA_2$ and $\procA_1$ about the configuration can easily be specified as regular expression:
\[
 \msgBackForth{\procA_2}{\procA_1}{q_0} \cat
 \msgBackForth{\procA_2}{\procA_1}{a_1}
 \ldots
 \msgBackForth{\procA_2}{\procA_1}{a_n}
\enspace .
\]
It is straightforward to construct a PSM for the complement of this regular expression.
(Before we did not use the complement for $L_{r,1}$ because we could not guarantee that the same number of configurations would be communicated between both pairs.)
When combined, this gives us the following regular expression (with the complement operator as syntactic sugar) for $L_{r,2}$:
\begin{align*}
    &
    \msgBackForth{\procA_3}{\procA_2}{\markNewRound} \cat
    \msgBackForth{\procA_2}{\procA_1}{\markBeginConf} \cat \invisibleEndLine
    \\
    &
    \overline{
    \msgBackForth{\procA_2}{\procA_1}{q_0} \cat
    \msgBackForth{\procA_2}{\procA_1}{a_1}
    \cat \ldots \cat
    \msgBackForth{\procA_2}{\procA_1}{a_n}
} \cat
    \msgBackForth{\procA_2}{\procA_1}{\markEndConf} \cat \invisibleEndLine
    \\
    &
    \msgBackForth{\procA_3}{\procA_4}{\markNewRound} \cat
    \msgBackForth{\procA_4}{\procA_5}{\markBeginConf} \cat
    (\msgBackForth{\procA_4}{\procA_5}{(\TapeAlph \dunion Q)})^* \cat
    \msgBackForth{\procA_4}{\procA_5}{\markEndConf} \cat
    \regex[l]
    \enspace.
\end{align*}

\textit{Language $L_{r,3}$}:

The following regular expression specifies all sequences for which the last $C_m$ does not contain the final state $q_f \in \TMStates$:
\begin{align*}
    &
    \regex[l] \cat \invisibleEndLine
    \\ &
    \msgBackForth{\procA_3}{\procA_2}{\markNewRound} \cat
    \msgBackForth{\procA_2}{\procA_1}{\markBeginConf} \cat
    (\msgBackForth{\procA_2}{\procA_1}{(\TapeAlph \dunion Q \setminus \set{q_f})})^* \cat
    \msgBackForth{\procA_2}{\procA_1}{\markEndConf} \cat \invisibleEndLine
    \\
    &
    \msgBackForth{\procA_3}{\procA_4}{\markNewRound} \cat
    \msgBackForth{\procA_4}{\procA_5}{\markBeginConf} \cat
    (\msgBackForth{\procA_4}{\procA_5}{(\TapeAlph \dunion Q)})^* \cat
    \msgBackForth{\procA_4}{\procA_5}{\markEndConf}
    \enspace.
\end{align*}

\textit{Language $L_{r,4}$}:

Intuitively, we can merge the loops for both $C_i$ and $D_i$ to check that some message at the same position is different.
This is possible because all languages are closed under $\interswap$ by definition.
We introduce this notation
\[
    \{
    \msgBackForth{\procA_2}{\procA_1}{x} \cat
    \msgBackForth{\procA_4}{\procA_5}{x}
    \mid x \in \set{y_1, \ldots, y_n}
    \}
\]
which is an abbreviation for
\[
    (
    \msgBackForth{\procA_2}{\procA_1}{y_1} \cat
    \msgBackForth{\procA_4}{\procA_5}{y_1}
    )
    +
    \ldots
    +
    (
    \msgBackForth{\procA_2}{\procA_1}{y_n} \cat
    \msgBackForth{\procA_4}{\procA_5}{y_n}
    )
    \enspace .
\]

With this, the following is a regular expression for $L_{r,4}$:
{ \scriptsize
\begin{align*}
\regex[4] \is \quad &
    \regex[l] \cat
    \\ &
    \msgBackForth{\procA_3}{\procA_2}{\markNewRound} \cat
    \msgBackForth{\procA_3}{\procA_4}{\markNewRound} \cat
    \msgBackForth{\procA_2}{\procA_1}{\markBeginConf} \cat
    \msgBackForth{\procA_4}{\procA_5}{\markBeginConf} \cat \invisibleEndLine
    \\ &
    (\{
    \msgBackForth{\procA_2}{\procA_1}{x_1} \cat
    \msgBackForth{\procA_4}{\procA_5}{x_1}
    \mid x_1 \in (\TapeAlph \union \TMStates)
    \})^* \cat \invisibleEndLine
    \\ &
    (\regex[a] + \regex[b] + \regex[c])
    \cat \regex[l]
    \\
\text{where} \\
\regex[a] \is \quad &
    (
    \{
    \msgBackForth{\procA_2}{\procA_1}{x_2} \cat
    \msgBackForth{\procA_4}{\procA_5}{x_3}
    \mid x_2, x_3 \in (\TapeAlph \union \TMStates)
    \text{ and } x_2 \neq x_3
    \}
    )
    \cat \invisibleEndLine
    \\ & \quad
    (\{
    \msgBackForth{\procA_2}{\procA_1}{x_4} \cat
    \msgBackForth{\procA_4}{\procA_5}{x_5}
    \mid x_4, x_5 \in (\TapeAlph \union \TMStates)
    \})^* \cat \regex[l]
\\
\regex[b] \is \quad &
    \{
    \msgBackForth{\procA_2}{\procA_1}{x_7} \cat
    \msgBackForth{\procA_4}{\procA_5}{\markEndConf}
    \mid x_7 \in (\TapeAlph \union \TMStates)
    \}
    \cat
    (\msgBackForth{\procA_2}{\procA_1}{\TapeAlph \union \TMStates})^* \cat
    \msgBackForth{\procA_2}{\procA_1}{\markEndConf}
    \cat \regex[l]
\\
\regex[c] \is \quad &
    \{
    \msgBackForth{\procA_2}{\procA_1}{\markEndConf} \cat
    \msgBackForth{\procA_4}{\procA_5}{x_6}
    \mid x_6 \in (\TapeAlph \union \TMStates)
    \}
    \cat
    (\msgBackForth{\procA_4}{\procA_5}{\TapeAlph \union \TMStates})^* \cat
    \msgBackForth{\procA_4}{\procA_5}{\markEndConf}
    \cat \regex[l]
    \enspace .
\end{align*}
}

The regular expression checks that at some point two configurations $C_i$ and $D_i$ do not agree on some position (using $\regex[a]$), or $C_i$ is longer than $D_i$ (using $\regex[b]$), or $D_i$ is longer than $C_i$ (using $\regex[c]$).

\textit{Language $L_{r,5}$}:

We use the same idea of merging the loops to compare as for the previous case and also use the same notation.
We can give a regular expression that consists of different phases.
First, we let $\procA_2$ and $\procA_1$ communicate about $C_1$ in order to then compare $D_i$ with $C_{i+1}$ for any $i$ in a loop.
We want that $C_{i+1}$ is no successor of $D_i$ for some $i$.
Thus, we check if the changes from $D_i$ to $C_{i+1}$ are a valid transition for $\TM$.
The regular expression $\regex[5]$ is defined as follows:

{ \tiny
\begin{align*}
\regex[5] \is \quad
    &
    \msgBackForth{\procA_3}{\procA_2}{\markNewRound} \cat
    \msgBackForth{\procA_2}{\procA_1}{\markBeginConf} \cat
    (\msgBackForth{\procA_2}{\procA_1}{(\TapeAlph \dunion \TMStates)})^* \cat
    \msgBackForth{\procA_2}{\procA_1}{\markEndConf} \cat \invisibleEndLine
    \\ &
    (\regex[d] + \regex[e])^*
        \cat
    (\regex[f] + \regex[g])
        \cat \invisibleEndLine
    \\ &
    \msgBackForth{\procA_3}{\procA_4}{\markNewRound} \cat
    \msgBackForth{\procA_4}{\procA_5}{\markBeginConf} \cat
    (\msgBackForth{\procA_4}{\procA_5}{(\TapeAlph \dunion \TMStates)})^* \cat
    \msgBackForth{\procA_4}{\procA_5}{\markEndConf} \cat \invisibleEndLine
    \\ &
    \regex[l]
    \\
\text{where}
    \\
\regex[d] \is \quad
    &
    \msgBackForth{\procA_3}{\procA_4}{\markNewRound} \cat
    \msgBackForth{\procA_4}{\procA_5}{\markBeginConf} \cat
    \msgBackForth{\procA_3}{\procA_2}{\markNewRound} \cat
    \msgBackForth{\procA_2}{\procA_1}{\markBeginConf} \cat \invisibleEndLine
    \\ &
    (\{
    \msgBackForth{\procA_4}{\procA_5}{x_1} \cat
    \msgBackForth{\procA_2}{\procA_1}{x_1}
    \mid x_1 \in (\TapeAlph \union \TMStates)
    \})^* \cat \invisibleEndLine
    \\ &
    (\{
    \msgBackForth{\procA_4}{\procA_5}{a_1} \cat
    \msgBackForth{\procA_2}{\procA_1}{a_2} \cat
    \msgBackForth{\procA_4}{\procA_5}{b_1} \cat
    \msgBackForth{\procA_2}{\procA_1}{b_2} \cat
    \msgBackForth{\procA_4}{\procA_5}{c_1} \cat
    \msgBackForth{\procA_2}{\procA_1}{c_2}
    \\ &
    \hspace{3ex} \mid
        a_1, a_2, b_1, b_2, c_1, c_2
        \in (\TapeAlph \union \TMStates)^*,
        a_1 \neq a_2
    \text{ and }
        \exists w_1, w_2 \in \TapeAlph^{\!*} \st
        w_1 a_1 b_1 c_1 w_2 \TMNext w_1 a_2 b_2 c_2 w_2
    \}) \cat \invisibleEndLine
    \\ &
    (\{
    \msgBackForth{\procA_4}{\procA_5}{x_2} \cat
    \msgBackForth{\procA_2}{\procA_1}{x_2}
    \mid x_2 \in (\TapeAlph \union \TMStates)
    \})^* \cat \invisibleEndLine
    \\ &
    \msgBackForth{\procA_4}{\procA_5}{\markEndConf} \cat
    \msgBackForth{\procA_2}{\procA_1}{\markEndConf}
    \\
\regex[e] \is \quad
    &
    \msgBackForth{\procA_3}{\procA_4}{\markNewRound} \cat
    \msgBackForth{\procA_4}{\procA_5}{\markBeginConf} \cat
    \msgBackForth{\procA_3}{\procA_2}{\markNewRound} \cat
    \msgBackForth{\procA_2}{\procA_1}{\markBeginConf} \cat \invisibleEndLine
    \\ &
    (\{
    \msgBackForth{\procA_4}{\procA_5}{x_1} \cat
    \msgBackForth{\procA_2}{\procA_1}{x_1}
    \mid x_1 \in (\TapeAlph \union \TMStates)
    \})^* \cat \invisibleEndLine
    \\ &
    (\{
    \msgBackForth{\procA_4}{\procA_5}{a_1} \cat
    \msgBackForth{\procA_2}{\procA_1}{a_2} \cat
    \msgBackForth{\procA_4}{\procA_5}{b_1} \cat
    \msgBackForth{\procA_2}{\procA_1}{b_2} \cat
    \msgBackForth{\procA_4}{\procA_5}{\markEndConf} \cat
    \msgBackForth{\procA_2}{\procA_1}{c_2} \cat
    \msgBackForth{\procA_2}{\procA_1}{\markEndConf}
    \\ &
    \hspace{3ex} \mid
        a_1, a_2, b_1, b_2, c_2
        \in (\TapeAlph \union \TMStates)^*,
        a_1 \neq a_2
    \text{ and }
        \exists w_1 \in \TapeAlph^{\!*} \st
        w_1 a_1 b_1 \TMNext w_1 a_2 b_2 c_2
    \})
    \\
\regex[f] \is \quad
    &
    \msgBackForth{\procA_3}{\procA_4}{\markNewRound} \cat
    \msgBackForth{\procA_4}{\procA_5}{\markBeginConf} \cat
    \msgBackForth{\procA_3}{\procA_2}{\markNewRound} \cat
    \msgBackForth{\procA_2}{\procA_1}{\markBeginConf} \cat \invisibleEndLine
    \\ &
    (\{
    \msgBackForth{\procA_4}{\procA_5}{x_1} \cat
    \msgBackForth{\procA_2}{\procA_1}{x_1}
    \mid x_1 \in (\TapeAlph \union \TMStates)
    \})^* \cat \invisibleEndLine
    \\ &
    (\{
    \msgBackForth{\procA_4}{\procA_5}{a_1} \cat
    \msgBackForth{\procA_2}{\procA_1}{a_2} \cat
    \msgBackForth{\procA_4}{\procA_5}{b_1} \cat
    \msgBackForth{\procA_2}{\procA_1}{b_2} \cat
    \msgBackForth{\procA_4}{\procA_5}{c_1} \cat
    \msgBackForth{\procA_2}{\procA_1}{c_2}
    \\ &
    \hspace{3ex} \mid
        a_1, a_2, b_1, b_2, c_1, c_2
        \in (\TapeAlph \union \TMStates)^*,
        a_1 \neq a_2
    \text{ and }
        \nexists w_1, w_2 \in \TapeAlph^{\!*} \st
        w_1 a_1 b_1 c_1 w_2 \TMNext w_1 a_2 b_2 c_2 w_2
    \}) \cat \invisibleEndLine
    \\ &
    ( \msgBackForth{\procA_4}{\procA_5}{\TapeAlph \union \TMStates} )^* \cat
    ( \msgBackForth{\procA_2}{\procA_1}{\TapeAlph \union \TMStates} )^* \cat \invisibleEndLine
    \\ &
    \msgBackForth{\procA_4}{\procA_5}{\markEndConf} \cat
    \msgBackForth{\procA_2}{\procA_1}{\markEndConf}
\\
\regex[g] \is \quad
    &
    \msgBackForth{\procA_3}{\procA_4}{\markNewRound} \cat
    \msgBackForth{\procA_4}{\procA_5}{\markBeginConf} \cat
    \msgBackForth{\procA_3}{\procA_2}{\markNewRound} \cat
    \msgBackForth{\procA_2}{\procA_1}{\markBeginConf} \cat \invisibleEndLine
    \\ &
    (\{
    \msgBackForth{\procA_4}{\procA_5}{x_1} \cat
    \msgBackForth{\procA_2}{\procA_1}{x_1}
    \mid x_1 \in (\TapeAlph \union \TMStates)
    )^* \cat \invisibleEndLine
    \\ &
    (\{
\msgBackForth{\procA_4}{\procA_5}{c_1} \cat
    \msgBackForth{\procA_2}{\procA_1}{\markEndConf}
\mid
c_1
        \in (\TapeAlph \union \TMStates)^*
\}) \cat \invisibleEndLine
    \\ &
    ( \msgBackForth{\procA_4}{\procA_5}{\TapeAlph \union \TMStates} )^* \cat
\msgBackForth{\procA_4}{\procA_5}{\markEndConf}
\end{align*}
}

We distinguish two types of transitions:
the ones that simply change letters in the middle of the configurations ($\regex[d]$) and the ones that extend the tape ($\regex[e]$).
If one is matched against, we recurse using
$(\regex[d] + \regex[e])^*$.
If not, $(\regex[f] + \regex[g])$ is matched against and we, subsequently, allow any possible subsequent pair configurations using $\regex[l]$.
The regular expression $\regex[f]$ checks that the transition is not possible while $\regex[g]$ checks if $C_{i+1}$ is shorter than $D_i$.
Without loss of generality, we can assume that the tape never shrinks (as it could be encoded using an extra tape alphabet letter).
Note that the transition check is a local condition and it suffices to check at most two more messages after the first different message.
In fact, the words $w_1$ and $w_2$ in the conditions do not matter:
either it is a transition for all such pairs or none.
After the mismatch, we let $D_i$ catch up and continue with~$\regex[l]$.

\textit{Mixed choice}:

We explained how to construct a PSM for $L_{r,i}$ for every $i$.
They are \mbox{\sumOnePSMs} by construction.
In fact, each of them individually also satisfies sender-driven choice.
However, when we combine both PSMs for $L_{r,4}$ and $L_{r,5}$ in order to obtain a PSM for $L_r$, the resulting PSM exposes mixed choice.
Intuitively, this happens because $L_{r,4}$ checks $C_i$ against $D_i$ and $L_{r,5}$ checks $D_i$ against $C_{i+1}$.
Technically, when merging both PSMs, we reach a state after the sequence
\[
    \msgBackForth{\procA_3}{\procA_2}{\markNewRound} \cat
    \msgBackForth{\procA_2}{\procA_1}{\markBeginConf}
\]
for which $L_{r,4}$ requires to have
$ \msgBackForth{\procA_3}{\procA_4}{\markNewRound} $
next while
$L_{r,5}$ requires to have a loop with
$
    \msgBackForth{\procA_2}{\procA_1}{\TapeAlph \dunion \TMStates}. 
$
It is not possible to let $\procA_2$ send a message to distinguish both branches as the indistinguishability of both branches is necessary so that $C_{i+1}$ can be compared to both $D_i$ and $D_{i+1}$.

\claimProofEnd{1}

\claimStm{2}
    $L_l$ is projectable.

\claimProofStart{2}
 By \cref{thm:sumOnePSMasGlobalType}, the PSM for $L_l$ can be represented as a global type with mixed choice.
 It is also $0$-reachable \cite{DBLP:conf/ecoop/Stutz23}, \ie one can reach a final state from every state.
 For a $0$-reachable global type $\GG$,
 \cite[Lm.\,4.10]{DBLP:conf/ecoop/Stutz23} showed that projections for $\lang_{\fin}(\GG)$ generalise to $\lang_{\inf}(\GG)$.
 They do not consider mixed choice but the proof generalises to the mixed choice setting as it does not use any restrictions on choice.
 \cite{DBLP:journals/tse/AlurEY03} showed that a language~$L$ of finite words is projectable if and only if two closure conditions $CC_2$ and $CC_3$ hold.
 Hence, it suffices to show that both $CC_2$ and $CC_3$. 

In what follows, a word $u$ is said to be complete if 
 $u \wproj_{\snd{\procA}{\procB}{\_}} = 
  u \wproj_{\rcv{\procA}{\procB}{\_}}$. 
\\
 \textbf{$CC_2$}:
If $w \in \AlphAsync^*$ is
 \channelcompliant, complete,
and for every participant $\procA \in \Procs$, there is $v \in L$ with
 $w \wproj_{\AlphAsync_\procA} = v \wproj_{\AlphAsync_\procA}$,
 then $w \in L$. 
\\
 \textbf{$CC_3$}:
 If $w$ is \channelcompliant and for every participant $\procA \in \Procs$, there is $v \in \pref(L)$ with
 $w \wproj_{\AlphAsync_\procA} = v \wproj_{\AlphAsync_\procA}$,
 then $v \in \pref(L)$.

 We show that $L_l$ satisfies $CC_2$.
 The proof for $CC_3$ is analogous.
 Let $w \in \AlphAsync^*$ be a \channelcompliant and complete word such that for every participant $\procA \in \Procs$,
 there is $v \in L_l$ with
 $w \wproj_{\AlphAsync_\procA} = v \wproj_{\AlphAsync_\procA}$.
 Let us give the structure of $w \wproj_{\AlphAsync_{\procA_i}}$ for $i \in \set{1,2,3}$.
{ \footnotesize
\begin{align*}
   w \wproj_{\AlphAsync_{\procA_3}}
   = \quad &
    (\snd{\procA_3}{\procA_2}{\markNewRound} \cat
    \rcv{\procA_2}{\procA_3}{\markNewRound} \cat
    \snd{\procA_3}{\procA_4}{\markNewRound} \cat
    \rcv{\procA_4}{\procA_3}{\markNewRound})^{k_3}
    \text{ for some } k_3
\\
   w \wproj_{\AlphAsync_{\procA_2}}
   = \quad &
    \rcv{\procA_3}{\procA_2}{\markNewRound} \cat
    \snd{\procA_2}{\procA_1}{\markBeginConf} \cat
    \rcv{\procA_1}{\procA_2}{\markBeginConf} \, \cat
    \\ &
    (\snd{\procA_2}{\procA_1}{a_{1,1}} \cat
    \rcv{\procA_1}{\procA_2}{a_{1,1}} \cat
    \ldots
    \snd{\procA_2}{\procA_1}{a_{1,i_1}} \cat
    \rcv{\procA_1}{\procA_2}{a_{1,i_1}}) \cat \invisibleEndLine
    \\ &
    \ldots
    \\ &
    (\snd{\procA_2}{\procA_1}{a_{k_2,1}} \cat
    \rcv{\procA_1}{\procA_2}{a_{k_2,1}} \cat
    \ldots
    \snd{\procA_2}{\procA_1}{a_{k_2,i_{k_2}}} \cat
    \rcv{\procA_1}{\procA_2}{a_{k_2,i_{k_2}}})
    \\ &
    \text{for some } k_2, i_1, \ldots, i_{k_2}
\\
   w \wproj_{\AlphAsync_{\procA_2}}
   = \quad &
    \rcv{\procA_2}{\procA_1}{\markBeginConf} \cat
    \snd{\procA_1}{\procA_2}{\markBeginConf} \, \cat
    \\ &
    (\rcv{\procA_2}{\procA_1}{b_{1,1}} \cat
    \snd{\procA_1}{\procA_2}{b_{1,1}} \cat
    \ldots
    \rcv{\procA_2}{\procA_1}{b_{1,j_1}} \cat
    \snd{\procA_1}{\procA_2}{b_{1,j_1}}) \cat \invisibleEndLine
    \\ &
    \ldots
    \\ &
    (\rcv{\procA_2}{\procA_1}{b_{k_1,1}} \cat
    \snd{\procA_1}{\procA_2}{b_{k_1,1}} \cat
    \ldots
    \rcv{\procA_2}{\procA_1}{b_{k_1,j_{k_1}}} \cat
    \snd{\procA_1}{\procA_2}{b_{k_1,j_{k_1}}})
    \\ &
    \text{for some } k_1, j_1, \ldots, j_{k_1}
\end{align*}
}

(The projections for $\procA_4$ and $\procA_5$ are analogous to $\procA_2$ and~$\procA_1$ and analogous reasoning applies.)
By the fact that $w$ is finite and complete, we know that $k_1 = k_2 = k_3$ and
$i_l = j_l$ for every $1 \leq l \leq k$.
By the fact that $w$ is \channelcompliant, the letters coincide, \ie
$a_{i, l} = b_{i, l}$
for every $i$ and~$l$.
Therefore, it is straightforward that $w$ can be obtained by reordering
$w(C_1, D_1, \ldots, C_k, D_k)$
using $\interswap$
and, thus, $w \in L_l$.

\claimProofEnd{2}

\claimStm{3}
 If $\TM$ has no accepting computation for $w$, then $\lang(\PSM_{\TM})$ is projectable.

\claimProofStart{3}
Recall that
\[
\lang(\PSM_{\TM}) =
\set{\msgFromTo{\procA_2}{\procA_3}{\lbl{l}} \cat w \mid w \in L'_l}
\dunion
\set{\msgFromTo{\procA_2}{\procA_3}{\lbl{r}} \cat w \mid w \in L'_r}
\enspace .
\]
If $w$ is not accepted, then $L_l$ and $L_r$ and, hence, $L'_l$ and $L'_r$ coincide by construction.
Thus, it is irrelevant for $\procA_1, \procA_4$, and $\procA_5$ which branch was taken.
By Claim 2, $L_l$ is projectable and so is~$\lang(\PSM_\TM)$.

\claimProofEnd{3}

\claimStm{4}
 If $\TM$ has an accepting computation for $w$, then $\lang(\PSM_{\TM})$ is not projectable.

\claimProofStart{4}
Let $(u_1, \ldots, u_m)$ be an accepting computation for $w$.
Then, there is
$w_u = w(u_1, u_1, \ldots, u_m, u_m)$
and, by construction, it holds that
$w_u \notin L_r$.
By definition of $\lang(\PSM_{\TM})$, it holds that
$\msgFromTo{\procA_2}{\procA_3}{\lbl{r}} \cat w_u \cat w_{\mathit{end}} \notin \lang(\PSM_{\TM})$.
However, for every participant $\procA \in \Procs$,
there is
$v \in \lang(\PSM_{\TM})$
such that
$w \wproj_{\AlphAsync_\procA} = v \wproj_{\AlphAsync_\procA}$.
Together, this contradicts closure condition $CC_2$ which is a necessary condition for projectability.
Thus, $\lang(\PSM_{\TM})$ is not projectable.

\claimProofEnd{4}
This concludes the proofs for all claims and hence the overall proof. 
\proofEndSymbol
\end{proof}
  \section{Additional Material for \cref{sec:typing-for-csms}}
\label{app:typing-for-csms}

In this section, we show CSMs are a good fit for type checking.
They can be seamlessly integrated into a type system.
To show this, we present a session type system that uses CSMs as interface for type checking, instead of local types. 
This clean separation of concerns makes our type system applicable to any type of protocol specification that can be projected to CSMs.
CSMs are strictly more expressive than local types.
Thus, the use of CSMs as intermediate interface improves generality without loosing efficiency.
For instance, with our projection result, we can take (Tame) PSMs as global protocol specifications and type check processes against them, with CSMs as intermediate interface for local specifications.

For the design of our type system, we follow
\cite{DBLP:journals/pacmpl/ScalasY19}
as a particularly streamlined instance of a session type system.
There are two main differences.
\citet{DBLP:journals/pacmpl/ScalasY19}
do not consider global specifications, so they do not check against a global protocol specification, but they still use local types.
In our type system, we use CSMs and show which of their properties will entail which properties of the typed program.

\subsection{Payload Types and Delegation}

In contrast to the explanation in the main text, we use global types (cf. \cref{def:global-types}) as representations for protocols in this section to elaborate more on the payload types and delegation.
We will show that every global type can be represented as PSM but, here, their syntactic nature makes examples more concise. 
While these are global specifications, all our examples transfer to CSMs but using them would be a bit harder to grasp.

So far, we treated message payloads as uninterpreted names from a fixed finite set.
In practice, each message would be a label and a payload type.
The label can indicate the branch that was taken while the payload type is interpreted as type of the data transmitted.
For instance, consider the following global type:

\[ \GG \is
    \msgFromTo{\procA}{\procB}{\labelAndType{l}{\stringtype}}
    +
    \msgFromTo{\procA}{\procB}{\labelAndType{r}{\inttype}}
\]

Here, $l$ and $r$ are labels so $\procB$ knows which branch was taken.
For the right branch, $\inttype$ is interpreted as type of the payload.
Thus, the payload should be of type integer, \eg the number $2$.
If there is no payload, we simply omit it and only write the label.

A global type specifies the intended behaviour for one session.
With a type system, it is interesting to consider systems with multiple sessions that possibly follow protocols given by different global types.
For session $s$, the endpoint of participant $\procA$ is denoted by $s[\procA]$.

Let us give an informal example of a process that follows the protocol specified by $\GG$; more precisely let $s$ be that session.
Then, the process
$P_\procA \parallel P_\procB$, where $\parallel$ is parallel composition,
would comply with the above global type:

\vspace{-2ex}
{ \small
\begin{align*}
    P_\procA \is & \;
        s[\procA][\procB] ! \labelAndMsg{l}{\text{``foo''}} \seq 0
        \, \IntCh \,
        s[\procA][\procB] ! \labelAndMsg{r}{2} \seq 0
    \\
    P_\procB \is & \;
        s[\procB][\procA] ? \labelAndVar{l}{x} \seq 0
        \, \ExtCh \,
        s[\procB][\procA] ? \labelAndVar{r}{y} \seq 0
\end{align*}
}

We use internal choice $\IntCh$ for
$P_\procA$ because it sends first:
$s[\procA][\procB] ! \labelAndMsg{r}{2}$
indicates that endpoint $s[\procA]$ sends message $\labelAndMsg{r}{2}$ to $s[\procB]$.
For $P_\procB$, we use external choice $\ExtCh$ as it cannot actually choose but receives the sent message.
It stores the payload in a variable which could be used subsequently, \eg
$s[\procB][\procA] ? \labelAndVar{l}{x}$
stores the message in $x$.
Both processes terminate after their only action.

If all payload types are base types like $\stringtype$, $\inttype$, $\booltype$, etc., type checking is rather simple.
It gets more complicated and interesting if one adds the possibility to send channel endpoints.
This amounts to sending an endpoint of a session, \eg $s[\procA]$.
Upon receiving, the receiver shall subsequently comply with what is specified.
This is why sending a channel endpoint is called \emph{delegation}.
Usually, local types are used to specify channel endpoint types.
However, we do not use local types so how do we specify such behaviour?
We use the states from the subset projections of a global type because this specifies the behaviour of a participant in our setting.
We assume that they are distinct across all considered state machines.
Our model uses FSMs, giving us a finite set of labels.
Note that one cannot only send the initial state.
Each state corresponds to a position in the local behaviour.
Sending non-initial states corresponds to sending subexpressions of local~types.

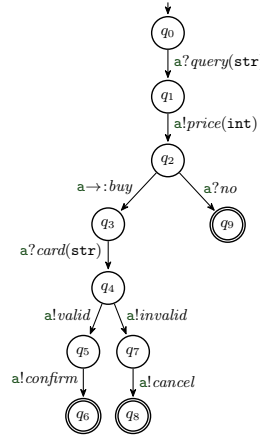
\begin{figure}
\centering
 \resizebox{0.3\textwidth}{!}{
    \begin{tikzpicture}[psm, node distance=0.7cm and 0.7cm]
    \node[state, initial above, initial text = ] (q0) {$q_0$};
    \node[state, below = of q0, yshift=1mm] (q1) {$q_1$};
    \node[state, below = of q1, yshift=1mm] (q2) {$q_2$};
    \node[state, below = of q2, xshift=-12mm, yshift=1mm] (q3) {$q_3$};
    \node[state, below = of q3, yshift=1mm] (q4) {$q_4$};
    \node[state, below = of q4, xshift=-5mm, yshift=1mm] (q5) {$q_5$};
    \node[finalstate, below = of q5, yshift=1mm] (q6) {$q_6$};
    \node[state, below = of q4, xshift=5mm, yshift=1mm] (q7) {$q_7$};
    \node[finalstate, below = of q7, yshift=1mm] (q8) {$q_{8}$};
    \node[finalstate, below = of q2, xshift=12mm, yshift=1mm] (q9) {$q_9$};
\path (q0) edge node[right, yshift=.5mm] {$\rcv{\buyerA}{}{\labelAndType{\query}{\texttt{str}}}$} (q1);
    \path (q1) edge node[right, yshift=1mm] {$\snd{}{\buyerA}{\labelAndType{\price}{\inttype}}$} (q2);
    \path (q2) edge node[left, yshift=1mm] {$\msgFromTo{\buyerA}{}{\buy}$} (q3);
    \path (q2) edge node[right, yshift=1mm] {$\rcv{\buyerA}{}{\no}$} (q9);
    \path (q3) edge node[left, yshift=1mm] {$\rcv{\buyerA}{}{\labelAndType{\ccard}{\stringtype}}$} (q4);
    \path (q4) edge node[left, yshift=1mm] {$\snd{}{\buyerA}{\valid}$} (q5);
    \path (q5) edge node[left, yshift=1mm] {$\snd{}{\buyerA}{\confirm}$} (q6);
    \path (q4) edge node[right, yshift=1mm] {$\snd{}{\buyerA}{\invalid}$} (q7);
    \path (q7) edge node[right, yshift=.5mm] {$\snd{}{\buyerA}{\cancel}$} (q8);

\end{tikzpicture}
  }
 \caption{Projection of the one buyer protocol onto seller $\seller$.}
 \label{fig:one-buyer-protocol-projection-seller}
\end{figure}

\begin{example}[Delegation]
\label{ex:global-types-delegation}
We consider a protocol where (one) buyer $\buyerA$ buys a book from a seller $\seller$ who delegates the payment to a payment service~$\payment$.
The one buyer protocol is specified as follows:

\vspace{-2ex}
{ \small
\begin{align*}
 \GG_1 & \is
 \msgFromTo{\buyerA}{\seller}{\labelAndType{\query}{\stringtype}} \seq
 \msgFromTo{\seller}{\buyerA}{\labelAndType{\price}{\inttype}} \seq
 +
 \begin{cases}
    \msgFromTo{\buyerA}{\seller}{\buy} \seq
    \msgFromTo{\buyerA}{\seller}{\labelAndType{\ccard}{\stringtype}} \seq
    \GG'
    \\
    \msgFromTo{\buyerA}{\seller}{\no} \seq \zero
 \end{cases}
 \text{ where}
 \\
 \GG' & \is
    +
    \begin{cases}
    \msgFromTo{\seller}{\buyerA}{\valid} \seq
    \msgFromTo{\seller}{\buyerA}{\confirm} \seq \zero
    \\
    \msgFromTo{\seller}{\buyerA}{\invalid} \seq
    \msgFromTo{\seller}{\buyerA}{\cancel} \seq \zero
    \end{cases}
\end{align*}
}

We project $\GG_1$ onto $\seller$ to obtain the state machine in \cref{fig:one-buyer-protocol-projection-seller} with its set of states $\set{q_1, \ldots, q_9}$.
We define the second global type, which specifies the interaction between the seller $\seller$ and the payment service $\payment$, including delegation.

\begin{align*}
 \GG_2 & \is
 +
 \begin{cases}
 \msgFromTo{\seller}{\payment}{\labelAndType{\price}{\inttype}} \seq
 \msgFromTo{\seller}{\payment}{\labelAndType{\deleg}{q_3}} \seq
    +
    \begin{cases}
    \msgFromTo{\payment}{\seller}{\labelAndType{\valid}{q_5}} \seq \zero
    \\
    \msgFromTo{\payment}{\seller}{\labelAndType{\invalid}{q_7}} \seq \zero
    \end{cases}
 \\
 \msgFromTo{\seller}{\payment}{\no} \seq \zero
 \end{cases}
\end{align*}

First, the seller delegates checking the card details to the payment service by sending $q_3$.
The payment service then takes care of the payment but we do not specify this here.
Afterwards, the payment service delegates control back to the seller:
depending on the outcome of the credit card check, they send $q_5$ or $q_7$.
Starting from there, the seller will either confirm or cancel.
This choice is not up to the seller but determined by the label, $\valid$ or $\invalid$, sent by the payment service earlier.

We use the states of a projection for seller $\seller$ as syntactic marker for the behaviour that is expected from the receiver of that channel endpoint.
Usually, this is achieved using local types.
In general, local types are less expressive than state machines.
However, with
\cref{lm:local-types-equi-expressive},
we will show that their expressivity coincides for \sinkfinal state machines, \ie the ones where final states have no outgoing transitions.
Hence, any \sinkfinal state machine can be turned into a local type.
Such a local type, however, would correspond to the initial state and, with delegation, we can also send non-initial states, \eg $q_3$.
Despite, it appears feasible to first construct a state machine that represents the same behaviour as starting from a non-initial state and, then, the same techniques for constructing a local type apply.
For instance, the behaviour specified by $q_3$ can also be represented as local~type:

\begin{align*}
    \rcv{\buyerA}{\seller}{\labelAndType{\ccard}{\stringtype}} \seq
    \IntCh
    \begin{cases}
    \snd{\seller}{\buyerA}{\valid} \seq
    \snd{\seller}{\buyerA}{\confirm} \seq \zero
    \\
    \snd{\seller}{\buyerA}{\invalid} \seq
    \snd{\seller}{\buyerA}{\cancel} \seq \zero
    \end{cases}
\end{align*}
\end{example}

We have seen an example for delegation.
There, in order to talk about states from a projection of $G_1$ in $G_2$, we need access to the projection of $G_1$ already before defining $G_2$.
If one considers CSMs, this is not necessarily the case.
However, we still need a similar condition to prove that well-typed processes do not leave messages in channels behind, so-called orphan messages.
Solely for this property, we assume a strict partial order $<$,
\ie it is irreflexive, antisymmetric, and transitive, for the CSMs under consideration.
With its acyclicity, this relation provides means to decide which CSMs can use the control states of which in a system:
for every $\CSMabb{A}$ and $\CSMabb{B}$ in a system,
if $\CSMabb{A} < \CSMabb{B}$, then $\CSMabb{B}$ can use states from $\CSMabb{A}$.
We expect that this condition could be worked around with more sophisticated techniques but leave this for future work.
Interestingly,
\citet{DBLP:journals/pacmpl/ScalasY19},
allow delegation using local types and do not impose such restrictions.
However, they also do not prove the absence of orphan messages.
It is unclear whether their type system can be extended to prove the absence of orphan messages without such restrictions.

\subsection{Process Calculus}

We first define processes and runtime configurations.

\begin{definition}
Processes, runtime configurations and process definitions are defined by the following grammar:
\begin{grammar}
 c \is
        x
    |   s[\procA]
\\
 P \is
        \zero
    |   P_1 \parallel P_2
    |   (\restr s \hasType \CSMabb{A}) \, P
    |   \IntCh_{i \in I} c[\procB_i] ! \labelAndMsg{l_i}{c_i} \seq P_i
    |   \ExtCh_{i \in I} c[\procB_i] ? \labelAndVar{l_i}{y_i} \seq P_i
    |   \pn{Q}[\vec{c}]
    \\
 R \is
        \zero
    |   R_1 \parallel R_2
    |   (\restr s \hasType \CSMabb{A}) \, R
    |   \smash{\IntCh_{i \in I} c[\procB_i] ! \labelAndMsg{l_i}{c_i} \seq P_i}
    |   \smash{\ExtCh_{i \in I} c[\procB_i] ? \labelAndVar{l_i}{y_i} \seq P_i}
    |   \pn{Q}[\vec{c}]
    \\ & \hspace{5.77ex}
    |   \queueProc{s}{\queuecontent}
    |   \err
    \\
 \Defs \is
    \bigl(\pn{Q}[\vec{x}] =
    \IntCh_{i \in I} c[\procB_i] ! \labelAndMsg{l_i}{c_i} \seq P_i\bigr); \; \Defs
    | \bigl(\pn{Q}[\vec{x}] =
    \ExtCh_{i \in I} c[\procB_i] ? \labelAndVar{l_i}{y_i} \seq P_i\bigr); \; \Defs
    | \emptystring
\end{grammar}

The term $c$ can either be a variable $x$ or a session endpoint of shape $s[\procA]$ (which will have type $q$ for some state $q$).
Let us explain the constructors for processes and runtime configurations in more detail.
The term $\zero$ denotes termination while
$\parallel$ is the parallel operator.
With $(\restr s \hasType \CSMabb{A})$, we restrict a new session~$s$ for which the CSM~$\CSMabb{A}$ specifies the intended session behaviour.
$\CSMabb{A}$ is ignored by our reduction semantics and solely used for type checking.
We have internal ($\IntCh$) and external ($\ExtCh$) choice and assume that $\card{I} > 0$.
For runtime configurations, the use of processes $P_i$ for continuations ensures that we only specify queue contents for active session restrictions.
A session restriction is \emph{active} if it is not \emph{guarded} by internal or external actions.
$\pn{Q}[\vec{c}]$ specifies the use of a process definition with identifier $\pn{Q} \in \ProcIds$ and parameters $\vec{c}$ for which
$\Defs$ provides the definitions.
We only consider guarded process definitions (and $c$ and $c_i$ can be in $\vec{x}$).
This allows us to properly distinguish between processes and runtime configurations later:
in the reduction rules, we will only add queues for active session restrictions.
We assume that $\Defs$ has one single definition for every process identifier in $\ProcIds$.
Thus, we define $\Defs(\pn{Q}, \vec{c})$ as unfolding of the process definition when its variables $\vec{x}$ are substituted by $\vec{c}$.
For runtime configurations,
$\queueProc{s}{\queuecontent}$
denotes that the queues of session $s$ are currently $\queuecontent$.
The function
$\queuecontent \from \Procs \times \Procs \to \Msg^*$
where
$\gramm{\Msg \ni m \is \labelAndMsg{l}{v}}$
specifies the queue content where $l$ is from a finite set of labels.
If all queues are empty, we exploit notation and use $\emptystring$, \ie
$\emptyqueuecontent(\procA,\procB) \is \emptystring $.
We will use the term $\err$ to specify if something went wrong.
\end{definition}

The correctness of our typing relies on the CSMs used for annotations
to be well-behaved, which we encode in the notion of well-annotated processes.
\begin{definition}[Well-annotated]
  We say a process/runtime configuration is \emph{well-annotated} if
  every CSM appearing in it is:
  \begin{enumerate}[(i)]
    \item deadlock-free, and
    \item satisfies feasible eventual reception.
  \end{enumerate}
\end{definition}
Note that any CSM obtained through projection automatically satisfies
the conditions of well-annotation.

In our process calculus, recursion can be achieved using process definitions $\pn{Q}[\vec{x}]$.
Not that, for global and local types, recursion is usually defined using $\mu t$, which binds a recursion variable $t$ that can be used subsequently 
(cf. \cref{def:global-types,def:local-types}).

\begin{figure}
  \adjustfigure[\small]\begin{align*}
   P \is & \;
      (\restr s_1 \hasType \CSMabb{A})
         (\restr s_2 \hasType \CSMabb{A}) \;
         P_{\buyerA} \parallel P_{\seller} \parallel P_{\payment} \quad \text{ where}
      \\
P_{\buyerA} \is & \;
      s_1[\buyerA][\seller] ! \labelAndMsg{\query}{\text{``Alice in Wonderland''}} \seq
s_1[\buyerA][\seller] ? \labelAndVar{\price}{p} \seq
      \\ & \;
      \iif p > 10
      \\ & \quad
      \ithen
      s_1[\buyerA][\seller] ! \no \seq \zero
      \\ & \quad
      \ielse
      s_1[\buyerA][\seller] ! \buy \seq
      s_1[\buyerA][\seller] ! \labelAndMsg{\ccard}{\text{``1234..., 08/2024, 532''}} \seq
      \\ & \quad \hspace{17ex}
      \ExtCh
          \begin{cases}
          s_1[\buyerA][\seller] ? \confirm \seq \zero
          \\
          s_1[\buyerA][\seller] ? \cancel \seq \zero
          \end{cases}
      \\
P_{\seller} \is & \;
      s_1[\seller][\buyerA] ? \labelAndVar{\query}{b} \seq
s_1[\seller][\buyerA] ! \labelAndMsg{\price}{\operatorname{prices}[b]} \seq
      \\ & \;
      \hspace{17ex}
      \ExtCh
      \begin{cases}
          s_1[\seller][\buyerA] ? \no \seq
          s_2[\seller][\payment] ! \no \seq \zero
          \\
          s_1[\seller][\buyerA] ? \buy \seq
          s_2[\seller][\payment] ! \labelAndMsg{\price}{\operatorname{prices}[b]} \seq
          s_2[\seller][\payment] ! \labelAndMsg{\deleg}{s_1[\seller]} \seq
          P'_{\seller}
      \end{cases}
      \\
P'_{\seller} \is & \;
      \ExtCh
      \begin{cases}
      s_2[\seller][\payment] ? \labelAndVar{\valid}{y_1} \seq
      y_1[\buyerA] ! \confirm \seq \zero
      \\
      s_2[\seller][\payment] ? \labelAndVar{\invalid}{y_2} \seq
      y_2[\buyerA] ! \cancel \seq \zero
      \end{cases}
      \\
P_{\payment} \is & \;
      \ExtCh
      \begin{cases}
          s_2[\payment][\seller] ? \no \seq \zero
          \\
          s_2[\payment][\seller] ! \labelAndVar{\price}{p} \seq
          s_2[\payment][\seller] ? \labelAndMsg{\deleg}{y} \seq
          y[\buyerA] ? \labelAndVar{\ccard}{z} \seq
          \\ \hspace{17ex}
          \iif \operatorname{is-valid}(z)
\ithen
          s_2[\payment][\seller] ! \labelAndMsg{\valid}{y} \seq \zero
\ielse
          s_2[\payment][\seller] ! \labelAndMsg{\invalid}{y} \seq \zero
      \end{cases}
  \end{align*}
  \caption{An example process following the protocol of \cref{ex:global-types-delegation}.}
  \label{fig:proc-ex-deleg}
\end{figure}

\begin{example}
\label{ex:processes-delegation}
In \cref{fig:proc-ex-deleg}, we give a process~$P$ that uses
projections of the global types from
\cref{ex:global-types-delegation}:
$\CSMabb{A}$ is a projection of $\GG_1$ and
$\CSMabb{B}$ is a projection of $\GG_2$.
We assume base types $\booltype$, $\stringtype$ and $\inttype$ as well as a construct for if-then-else for illustrative~purposes.
In $P$, $\operatorname{prices}[b]$ denotes a lookup for the price and $\operatorname{is-valid} \from \stringtype \to \booltype$ is a function that checks if credit card details are valid.
Note the use of variable $y$ in $P_{\payment}$ for the delegation.
In fact, it does not know the endpoint, or local type, it receives but needs to trust that it can perform the respective actions on it.
A type system can ensure~this.
\end{example}

\begin{definition}

We define a function $\procToRuntime{\hole}$ to convert a process into a runtime configuration by adding channel types for active sessions:
\begin{align*}
\procToRuntime{P_1 \parallel P_2}
    & \is \procToRuntime{P_1} \parallel \procToRuntime{P_2} \\
\procToRuntime{(\restr s \hasType \CSMabb{A}) \, P}
    & \is (\restr s \hasType \CSMabb{A}) \, (\procToRuntime{P} \parallel \queueProc{s}{\emptyqueuecontent}) \\
\procToRuntime{P}
    & \is P \text{ otherwise}
\end{align*}

\end{definition}

We define structural (pre)congruence.
Intuitively, this shows which kind of transformations do not change the meaning of a process or runtime configuration.
For instance, parallel composition of $P$ with $\zero$ is basically the same as $P$ itself.

\begin{definition}

For processes, the rules for structural congruence $\congr$ are the following:

\begin{itemize}
 \item $P_1 \parallel P_2
        \congr
        P_2 \parallel P_1$
 \item $(P_1 \parallel P_2) \parallel P_3
        \congr
        P_1 \parallel (P_2 \parallel P_3)$
 \item $P \parallel \zero
        \congr
        P$
 \item $(\restr s \hasType \CSMabb{A}) \, (\restr s' \hasType \CSMabb{B}) \, P
        \congr
        (\restr s' \hasType \CSMabb{B}) \, (\restr s \hasType \CSMabb{A}) \, P$
 \item $(\restr s \hasType \CSMabb{A}) \, (P_1 \parallel P_2)
        \congr
        P_1 \parallel (\restr s \hasType \CSMabb{A}) \, P_2$,
        if $s$ is not free in $P_1$
\end{itemize}

We define structural precongruence $\precongr$ for processes as the smallest precongruence relation that
includes $\congr$ and
    $(\restr s \hasType \CSMabb{A}) \, \zero
    \precongr
    \zero$.
For runtime configurations, the rules for structural congruence $\congr$ are the ones above.
We define structural precongruence~$\precongr$ for runtime configurations as the smallest precongruence relation that
includes $\congr$
    and
    $(\restr s \hasType \CSMabb{A}) \, \queueProc{s}{\emptystring} \precongr \zero$.
\end{definition}

We only define one direction for the rules
$(\restr s \hasType \CSMabb{A}) \, \zero \precongr \zero$
and
\mbox{$(\restr s \hasType \CSMabb{A}) \, \queueProc{s}{\emptystring} \precongr \zero$}.
This is solely required to prove that structural congruence preserves typability for both processes and runtime configurations (cf.\,\cref{lm:str-congr-preserves-typability-runtime-confs,lm:str-congr-preserves-typability-processes}).
Intuitively, the other direction would require to impose conditions on the CSM~$\CSMabb{A}$.
This treatment is not restrictive in terms of reductions: applying these rules from right to left will not change the possibility for reductions.

Last, we define the reduction rules for our process calculus.

\begin{definition}
For our reduction rules, we first define a reduction context:

\begin{grammar}
    \redContext \is
        \redContext \parallel R
    |   R \parallel \redContext
    |   (\restr s \hasType \CSMabb{A}) \, \redContext
    |   [\,]
\end{grammar}

In \cref{fig:reduction-rules} we define the reduction rules.
The rule \runtimeReductionProcName unfolds a process definition while \runtimeReductionContext allows us to descend for reductions using contexts.
Both rules \runtimeReductionOut and \runtimeReductionIn specify how a message is output to a queue or received as input from a queue.
\runtimeReductionCongr allows us to consider structurally precongruent runtime configurations for reductions.
\runtimeReductionErrOne yields an error if the next action is to receive but all possible incoming messages do not match any specified label.
Last, \runtimeReductionErrTwo yields an error if a session is over but there are non-empty queues for this session.
\end{definition}

\begin{figure}[t]
  \adjustfigure[\scriptsize]
  \begin{mathpar}

  \inferrule*[right=\runtimeReductionProcName]{
          \Defs(\pn{Q}, \vec{c})
          \parallel
          R
          \redto
          R'
  }{
          \pn{Q}[\vec{c}]
          \parallel
          R
          \redto
          R'
  }

  \inferrule*[right=\runtimeReductionContext]{
          R \redto R'
  }{
          \redContext[R] \redto \redContext[R']
  }

  \inferrule*[right=\runtimeReductionOut]{
          k \in I
  }{
          \IntCh_{i \in I} s[\procA][\procB_i] ! \labelAndMsg{l_i}{v_i} \seq P_i
          \parallel
          \queueProc{s}
          {\queuecontent[(\procA, \procB_k) \mapsto \vec{m}]}
              \redto
          \procToRuntime{P_k}
          \parallel
          \queueProc{s}
          {\queuecontent[(\procA, \procB_k) \mapsto \vec{m} \cat \labelAndMsg{l_k}{v_k}]}
  }

  \inferrule*[right=\runtimeReductionIn]{
          k \in I
  }{
          \ExtCh_{i \in I} s[\procA][\procB_i] ? \labelAndVar{l_i}{y_i} \seq P_i
          \parallel
          \queueProc{s}
          {\queuecontent[(\procB_k, \procA) \mapsto \labelAndMsg{l_k}{v_k} \cat \vec{m}]}
          \redto
          \procToRuntime{P_k[v_k / y_k]}
          \parallel
          \queueProc{s}
          {\queuecontent[(\procA, \procB_k) \mapsto \vec{m}]}
  }

  \inferrule*[right=\runtimeReductionCongr]{
          R_1 \precongr R'_1
          \\
          R'_1 \redto R'_2
          \\
          R'_2 \precongr R_2
  }{
          R_1 \redto R_2
  }

  \inferrule*[right=\runtimeReductionErrOne]{
\forall i \in I \st
          \queuecontent(\procB_i, \procA) = \labelAndMsg{l}{\_} \cat \vec{m}
          \text{ and }
          l_i \neq l
  }{
          \ExtCh_{i \in I} s[\procA][\procB_i] ? \labelAndVar{l_i}{y_i} \seq P_i
          \parallel
          \queueProc{s}{\queuecontent}
              \redto
          \err
  }

  \inferrule*[right=\runtimeReductionErrTwo]{
          \queuecontent(\procA, \procB) \neq \emptystring \text{ for some } \procA, \procB
  }{
          (\restr s \hasType \CSMabb{A}) \; \queueProc{s}{\queuecontent}
          \redto
          \err
  }
  \end{mathpar}
  \caption{Reduction rules for runtime configurations.}
  \label{fig:reduction-rules}
\end{figure}

\begin{remark}
To prove the absence of runtime error \runtimeReductionErrOne, we require that feasible eventual reception holds.
If the CSM does not have this property, there might be messages that could be left behind.
\end{remark}

\begin{example}
\label{ex:reduction-processes-delegation}
We give a reduction for the process specified in
\cref{ex:processes-delegation}.
First, we apply the function $\procToRuntime{\hole}$ to turn the process into a runtime configuration, yielding
\[
 R \is
    (\restr s_1 \hasType \CSMabb{A})
       (\restr s_2 \hasType \CSMabb{B}) \,
       P_{\buyerA} \parallel P_{\seller} \parallel P_{\payment}
       \parallel \queueProc{s_1}{\emptystring}
       \parallel \queueProc{s_2}{\emptystring}
\]
There is only one possible reduction step: the message $\labelAndMsg{\query}{\text{``Alice in Wonderland''}}$ is sent by $s_1[\buyerA]$ to $s_1[\seller]$.
Then, we obtain the following runtime configuration:

\vspace{-2ex}
{ \small
\begin{align*}
 R' \is & \;
    (\restr s_1 \hasType \CSMabb{A})
       (\restr s_2 \hasType \CSMabb{B}) \;
       P'_{\buyerA} \parallel P_{\seller} \parallel P_{\payment}
       \\
       & \; \hspace{20ex}
       \parallel \queueProc{s_1}{\emptystring[(\buyerA, \seller) \mapsto \labelAndMsg{\query}{\text{``Alice in Wonderland''}}]}
       \parallel \queueProc{s_2}{\emptystring}
    \\
\text{ where}
    \\
 P'_{\buyerA} \is & \;
s_1[\buyerA][\seller] ? \labelAndVar{\price}{p} \seq
    \\ & \;
    \iif p > 10
    \\ & \quad
    \ithen
    s_1[\buyerA][\seller] ! \no \seq \zero
    \\ & \quad
    \ielse
    s_1[\buyerA][\seller] ! \buy \seq
    s_1[\buyerA][\seller] ! \labelAndMsg{\ccard}{\text{``1234..., 08/2024, 111''}} \seq
\ExtCh
        \begin{cases}
        s_1[\buyerA][\seller] ? \confirm \seq \zero
        \\
        s_1[\buyerA][\seller] ? \cancel \seq \zero
        \end{cases}
\end{align*}
}

Despite dealing with runtime configurations, we can specify processes because the queues are specified at top level.
\end{example}

\subsection{Type System for Processes and Runtime Configurations}

Typing for base types is well-understood.
Thus, we focus on the more difficult case of delegation, following  work by
\citet{DBLP:journals/pacmpl/ScalasY19}.
Integration of base types is mostly orthogonal and would distract from the main concerns here so we briefly remark differences in the treatment of base types after presenting our type system.

For processes, we have two typing contexts:
$\typingContextOne$ and $\typingContextTwo$.
We consider states as syntactic markers for local specifications so we use $L$ as type for such payloads. So far, we only considered a fixed set of participants $\Procs$.
In a system with multiple CSMs, we write $\ProcsOf{\CSMabb{A}}$ to denote the subset of participants of $\CSMabb{A}$ and $\channelsOf{\CSMabb{A}}$ for the respective channels.
We might also use the session $s$ instead of the respective~CSM.

Prior to giving definitions for our type system, let us remark that we make use of the \emph{Barendregt Variable Convention} \cite{DBLP:books/daglib/0067558}, which assumes that the names of bound variables is always distinct from the ones of free variables.
This allows us not to explicitly rename variables, simplifying the formalisation both for writing and reading.

\begin{definition}
The process definition typing context $\typingContextOne$ is a function from process identifiers to types for its parameters:
$\typingContextOne \from \ProcIds \to \vec{L}$.
A \emph{syntactic typing context} is defined by the following grammar:
\begin{grammar}
    \typingContextTwo \is
        \typingContextTwo, s[\procA] \hasType L
    |   \typingContextTwo, x \hasType L
    |   \emptyset
\end{grammar}
A syntactic typing context is a \emph{typing context} if every element has at most one type.
Here, we do only consider typing contexts.
We consider typing contexts to be equivalent up to reordering and, thus, we may also treat them as mappings.
We use notation $\set{\typingContextTwo_i}_{i \in I}$ to denote that we split $\typingContextTwo$ into $\card{I}$ typing contexts.
\end{definition}

Equipped with these typing contexts, we can give the typing rules for processes.
The first two rules solely deal with the process definition typing context, which provides the types for process definitions.
The other rules deal with the different constructs of our process calculus and how to type them.
Most importantly, our type system ensures that all information in the typing context is used exactly once.

\begin{figure}[t]
  \adjustfigure[\small]
  \begin{mathpar}
\inferrule*[right=\procTypingProcDefEmpty]{
}{
      \types
      \emptystring \hasType \typingContextOne
  }

  \inferrule*[right=\procTypingProcDef]{
      \typingContextOne
      \typingContextCat
      \vec{x} \hasType \vec{L}
      \types
      P \\
\typingContextOne(\pn{Q}) = \vec{L}
  }{
      \types
      (\pn{Q}[\vec{x}] = P); \Defs \hasType \typingContextOne
  }

  \inferrule*[right=\procTypingProcName]{
      \typingContextOne(\pn{Q}) = \vec{L} }{
      \typingContextOne
          \typingContextCat
          \vec{c}  \hasType \vec{L}
\types
          \pn{Q}[\vec{c}]
  }

  \inferrule*[right=\procTypingZero]{
}{
      \typingContextOne
      \typingContextCat
      \emptyset
          \types
      \zero
  }

  \inferrule*[right=\procTypingEnd]{
      \typingContextOne
      \typingContextCat
      \typingContextTwo \types P \\
\EndState(q)
}{
      \typingContextOne
      \typingContextCat
      c \hasType q,
      \typingContextTwo
          \types
      P
  }

  \inferrule*[right=\procTypingParallel]{
      \typingContextOne
          \typingContextCat
          \typingContextTwo_1
          \types
          P_1 \\
\typingContextOne
          \typingContextCat
          \typingContextTwo_2
          \types
          P_2 \\
  }{
      \typingContextOne
          \typingContextCat
          \typingContextTwo_1,
          \typingContextTwo_2
          \types
          P_1 \parallel P_2
  }

  \inferrule*[right=\procTypingIntCh]{
\delta(q) \sups
      \set{(\snd{\procA}{\procB_i}{\labelAndType{l_i}{L_i}}, q_i) \mid i \in I}
      \\
\meta{\forall i \in I \st}
      \typingContextOne \typingContextCat
          \typingContextTwo , c \hasType q_i,
          \set{c_j \hasType L_j}_{j \in I\setminus \set{i}}
          \types P_i \\
  }{
      \typingContextOne
      \typingContextCat
      \typingContextTwo,
      c \hasType q,
      \set{c_i \hasType L_i}_{i \in I}
          \types
      \IntCh_{i \in I} c[\procB_i] ! \labelAndMsg{l_i}{c_i} \seq P_i
  }

  \inferrule*[right=\procTypingExtCh]{
      \delta(q) =
      \set{(\rcv{\procB_i}{\procA}{\labelAndType{l_i}{L_i}}, q_i) \mid i \in I} \\
      \meta{\forall i \in I \st}
      \typingContextOne \typingContextCat
          \typingContextTwo,
          c \hasType q_i,
          y_i \hasType L_i
          \types P_i \\
  }{
      \typingContextOne
      \typingContextCat
      \typingContextTwo,
      c \hasType q
          \types
      \ExtCh_{i \in I} c[\procB_i] ? \labelAndVar{l_i}{y_i} \seq P_i
  }

  \inferrule*[right=\procTypingRestr]{
      \typingContextTwo_s =
          \set{s[\procA] \hasType \initialState(\CSMabb{A}_{\procA})}_{\procA \in \ProcsOf{\CSMabb{A}}} \\
\typingContextOne
          \typingContextCat
          \typingContextTwo,
          \typingContextTwo_s
          \types
          P
  }{
      \typingContextOne
          \typingContextCat
          \typingContextTwo
          \types
          (\restr s \hasType \CSMabb{A})\, P
  }
  \end{mathpar}
  \caption{
    The typing rules for processes.
    $\initialState(\hole)$ denotes the initial state of a CSM.
  }
  \label{fig:proc-typing}
\end{figure}

\begin{definition}
We define $\EndState(q)$ to hold when $q$ is final and does not have outgoing receive transitions.
The typing rules for processes are shown in \cref{fig:proc-typing}.
We assume that all CSMs in typing derivations are deadlock-free and satisfy feasible eventual reception.

The rules \procTypingProcDefEmpty and \procTypingProcDef ensure that the process definition typing context provides the right types for parameters.
This is then used to type process definitions in a process, using \procTypingProcName.
\procTypingZero types $\zero$ with an empty second typing context while \procTypingEnd can be used to remove type bindings $c \hasType q$ where $q$ is a final state without outgoing receive transitions.
The rules \procTypingIntCh and \procTypingExtCh can be used to type internal and external choice.
The rule \procTypingParallel allows us to split the typing contexts and type the respective processes independently.
Last, \procTypingRestr adds type bindings for a session~$s$ and requires that the remaining process is typed using~this.
\end{definition}

Our type system is \emph{linear}, \ie it requires that every type binding is used once and they can only be dropped if they correspond to final states without outgoing receive transitions.
This ensures that all the actions specified by the CSM are actually taken and the participants of a session cannot stop earlier.

It might seem that $\procA$ in \procTypingIntCh and \procTypingExtCh is unbound but, by assumption, $q$ is distinct across all considered FSMs so it is clear from context.
Let us explain \procTypingIntCh in more detail.
To type
$
    \typingContextOne
    \typingContextCat
    \typingContextTwo,
    c \hasType q,
    \set{c_i \hasType L_i}_{i \in I}
        \types
    \IntCh_{i \in I} c[\procB_i] ! \labelAndMsg{l_i}{c_i} \seq P_i
$,
we require all send actions to be possible from $q$, \ie
\mbox{
$
    \delta(q) \sups
    \set{(\snd{\procA}{\procB_i}{\labelAndType{l_i}{L_i}}, q_i) \mid i \in I}
$}.
We do not require all of them to be possible though, in contrast to the receive actions in \procTypingExtCh.
In addition, for every $i \in I$, we require the following: 
$
    \typingContextOne \typingContextCat
        \typingContextTwo , c \hasType q_i,
        \set{c_j \hasType L_j}_{j \in I\setminus \set{i}}
        \types P_i
$,
which gives $c$ the new binding $q_i$ and removes the type binding for the payload $c_i \hasType L_i$.
Intuitively, it is transferred when sending a message.
Thus, in its counterpart \procTypingExtCh, the payload's type will be used to type the continuation after receiving it, \ie $y_i \hasType L_i$.

The assumption that we only consider typing contexts, and not syntactic typing contexts, ensures that $\typingContextTwo$ in the conclusion of \procTypingRestr does not contain any $s$ for instance.
The same will hold for the typing rules for runtime configurations.

\begin{remark}[$\EndState(\hole)$ and final non-sink states]
We use $\EndState(q)$ to check if $q$ is a final state without outgoing receive transitions.
Following standard MST frameworks, we would simply require $q$ to be final.
Thus, our type system is slightly more general. 
However, when using our type system for projections from our approach, this will not be exploited: 
we base our projection on results by 
\cite{DBLP:conf/cav/LiSWZ23} and their (complete) conditions do not allow final state with outgoing send transitions. 
\end{remark}

\newcommand{\markerEndStateType}{\texttt{end}}
\newcommand{\markerEndStateState}{\textit{end}\,}

\begin{figure}[tb]
\hfill
\begin{subfigure}[b]{0.18\textwidth}
\centering
 \resizebox{0.7\textwidth}{!}{
    \begin{tikzpicture}[psm, node distance=5em and 4em]
    \node[state, initial above, initial text = ] (q0) {$q_0$};
    \node[finalstate, below =2em of q0] (q1) {$q_1$};
\path (q0) edge node[right] {$\snd{\procA}{\procB}{\labelAndType{l}{\hole}}$} (q1);
\end{tikzpicture}
  }
\caption{$\GG_1 \tproj_\procA$.}
\label{fig:project-G1-onto-p}
\end{subfigure}
\hfill
\begin{subfigure}[b]{0.18\textwidth}
\centering
 \resizebox{0.7\textwidth}{!}{
    \begin{tikzpicture}[psm, node distance=5em and 4em]
    \node[state, initial above, initial text = ] (q0) {$q_2$};
    \node[finalstate, below =2em of q0] (q1) {$q_3$};
\path (q0) edge node[right] {$\rcv{\procA}{\procB}{\labelAndType{l}{\hole}}$} (q1);
\end{tikzpicture}
  }
\caption{$\GG_1 \tproj_\procB$.}
\label{fig:project-G1-onto-q}
\end{subfigure}
\hfill
\begin{subfigure}[b]{0.28\textwidth}
\centering
 \resizebox{0.85\textwidth}{!}{
    \begin{tikzpicture}[psm, node distance=5em and 4em]
    \node[state, initial above, initial text = ] (q0) {$q_4$};
    \node[finalstate, below=2em of q0, xshift=-1.5em] (q1) {$q_5$};
    \node[finalstate, below=2em of q0, xshift=1.5em] (q3) {$q_6$};
\path (q0) edge node[left,yshift=1ex] {$\snd{\procA}{\procC}{\labelAndType{l_1}{q_0}}$} (q1);
    \path (q0) edge node[right,yshift=1ex] {$\snd{\procA}{\procC}{\labelAndType{l_2}{\hole}}$} (q3);
\end{tikzpicture}
  }
\caption{$\GG_2 \tproj_\procA$.}
\label{fig:project-G2-onto-p}
\end{subfigure}
\hfill
\begin{subfigure}[b]{0.28\textwidth}
\centering
 \resizebox{0.85\textwidth}{!}{
    \begin{tikzpicture}[psm, node distance=5em and 4em]
    \node[state, initial above, initial text = ] (q0) {$q_7$};
    \node[finalstate, below=2em of q0, xshift=-1.5em] (q1) {$q_8$};
    \node[finalstate, below=2em of q0, xshift=1.5em] (q3) {$q_9$};
\path (q0) edge node[left,yshift=1ex] {$\rcv{\procA}{\procC}{\labelAndType{l_1}{q_0}}$} (q1);
    \path (q0) edge node[right,yshift=1ex] {$\rcv{\procA}{\procC}{\labelAndType{l_2}{\hole}}$} (q3);
\end{tikzpicture}
  }
\caption{$\GG_2 \tproj_\procC$.}
\label{fig:project-G2-onto-r}
\end{subfigure}
\hfill
\caption{Projections of two global types.}
\end{figure}
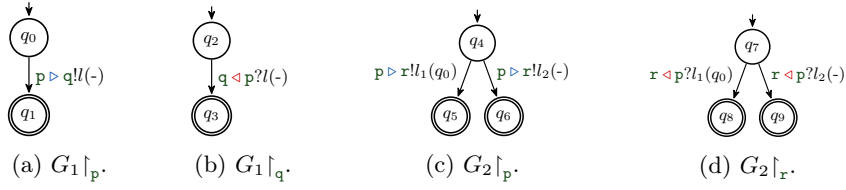

\begin{example}
\label{ex:proc-typing-delegation}
Let us illustrate delegation with an example.
We have the following global type:
\[
 \GG_1 \is
    \msgFromTo{\procA}{\procB}{\labelAndType{l}{\markerEndStateType}} \seq \zero
\]
for which we could model the payload of $\labelAndType{l}{\markerEndStateType}$ with an arbitrary state $\markerEndStateState$ such that $\EndState(\markerEndStateState)$.
For readability, we omit its treatment in the typing derivations.
We obtain a projection $\CSMabb{A}$ of $\GG_1$,
giving \cref{fig:project-G1-onto-p} and \cref{fig:project-G1-onto-q}.
Using the states of these, we define delegation in the second global type:

\vspace{-2ex}
{ \small
\[
 \GG_2 \is
    +
    \begin{cases}
    \msgFromTo{\procA}{\procC}{\labelAndType{l_1}{q_0}} \seq \zero
    \\
    \msgFromTo{\procA}{\procC}{\labelAndType{l_2}{\markerEndStateType}} \seq \zero
    \end{cases}
\]
}

A projection $\CSMabb{B}$ for $\GG_2$ is given in
\cref{fig:project-G2-onto-p} and
\cref{fig:project-G2-onto-r}.
Let us define a process that uses both global~types:

\vspace{-2ex}
{ \small
\begin{align*}
 P \is & \;
    (\restr s_1 \hasType \CSMabb{A})
    (\restr s_2 \hasType \CSMabb{B}) \;
    P_{\procA} \parallel
    P_{\procB} \parallel
    P_{\procC}
    \text{ where}
    \\
P_{\procA} \is & \;
    \IntCh
    \begin{cases}
    s_1[\procA][\procC] ! \labelAndMsg{l_1}{s_1[\procA]} \seq \zero
    \\
    s_1[\procA][\procC] ! \labelAndMsg{l_2}{\markerEndStateState} \seq
    s_1[\procA][\procB] ! \labelAndMsg{l}{\markerEndStateState} \seq \zero
    \end{cases}
    \\
P_{\procB} \is & \;
    s_1[\procB][\procA] ? \labelAndVar{l}{x} \seq \zero
    \\
P_{\procC} \is & \;
    \ExtCh
    \begin{cases}
    s_2[\procC][\procA] ? \labelAndVar{l_1}{x} \seq
    x[\procB] ! \labelAndMsg{l}{\markerEndStateState} \seq \zero
    \\
    s_2[\procC][\procA] ? \labelAndType{l_2}{x} \seq \zero
    \end{cases}
\end{align*}
}

In this example, the process definition typing context is always empty so we omit~it.
To fit it within the page limits, we give the typing derivation in pieces.
We use numbers $(0)-(5)$ to refer to the typing derivation for the respective branch.
It should be read from bottom to top, starting from $(0)$.
We start with the initial part until typing we arrive at typing the parallel composition.

\vspace{-2ex}
{ \small
\begin{mathpar}
    \inferrule*[right=\procTypingParallel (twice)]{
        \inferrule*{
            (3)
            }{
            s_1[\procA] \hasType q_0,
            s_2[\procA] \hasType q_4
            \types
            P_{\procA}
        }
        \\
\inferrule*{
            (4)
            }{
            s_1[\procB] \hasType q_2
            \types
            P_{\procB}
        }
        \\
\inferrule*{
            (5)
            }{
            s_2[\procC] \hasType q_7
            \types
            P_{\procC}
        }
    }{
        (2): \quad
        \typingContextTwo_{s_1},
        \typingContextTwo_{s_2}
        \types
        P_{\procA} \parallel
        P_{\procB} \parallel
        P_{\procC}
    }
\end{mathpar}
\begin{mathpar}
    \inferrule*[right=\procTypingRestr]{
        \typingContextTwo_{s_2} =
            s_2[\procA] \hasType q_4,
            s_2[\procC] \hasType q_7 \\
\inferrule*{
            (2)
        }{
            \typingContextTwo_{s_1},
            \typingContextTwo_{s_2}
            \types
            P_{\procA} \parallel
            P_{\procB} \parallel
            P_{\procC}
        }
    }{
        (1): \quad
        \typingContextTwo_{s_1}
        \types
        (\restr s_2 \hasType \CSMabb{B}) \;
        P_{\procA} \parallel
        P_{\procB} \parallel
        P_{\procC}
    }
\end{mathpar}
\begin{mathpar}
    \inferrule*[right=\procTypingRestr]{
        \typingContextTwo_{s_1} =
            s_1[\procA] \hasType q_0,
            s_1[\procB] \hasType q_2 \\
\inferrule*{
            (1)
        }{
            \typingContextTwo_{s_1}
            \types
            (\restr s_2 \hasType \CSMabb{B}) \;
            P_{\procA} \parallel
            P_{\procB} \parallel
            P_{\procC}
        }
    }{
        (0): \quad
        \emptyset
        \types
        (\restr s_1 \hasType \CSMabb{A})
        (\restr s_2 \hasType \CSMabb{B}) \;
        P_{\procA} \parallel
        P_{\procB} \parallel
        P_{\procC}
    }
\end{mathpar}
}

We apply the rule for restrictions \procTypingRestr and then the one for parallel composition~\procTypingParallel.
Let us give the typing derivations for the individual branches.

\vspace{-2ex}
{ \scriptsize
\begin{mathpar}
    \inferrule*[right=\procTypingIntCh]{
        \inferrule*[left=\procTypingEnd]{
            \inferrule*[left=\procTypingIntCh]{
                \delta(q_0) =
                \set{
                (\snd{\procA}{\procB}{\labelAndType{l}{\markerEndStateType}}, q_1)
                }
                \\
\inferrule*[left=\procTypingEnd]{
                    \inferrule*[left=\procTypingZero]{
                    }{
                        \emptyset \types \zero
                    }
                    \\
\EndState(q_1)
                }{
                    s_1[\procA] \hasType q_1
                    \types
                    \zero
                }
            }{
                s_1[\procA] \hasType q_0
                \types
                s_1[\procA][\procB] ! \labelAndMsg{l}{\markerEndStateState} \seq \zero
            }
        }{
            s_1[\procA] \hasType q_0,
            s_2[\procA] \hasType q_6
            \types
            s_1[\procA][\procB] ! \labelAndMsg{l}{\markerEndStateState} \seq \zero
        }
        \\
\inferrule*[left=\procTypingEnd]{
            \inferrule*[left=\procTypingZero]{
            }{
                \emptyset \types \zero
            }
            \\
            \EndState(q_5)
        }{
            s_2[\procA] \hasType q_5
            \types
            \zero
        }
        \\
\delta(q_4) =
        \set{
        (\snd{\procA}{\procC}{\labelAndType{l_1}{q_0}}, q_5),
        (\snd{\procA}{\procC}{\labelAndType{l_2}{\markerEndStateType}}, q_6)
        }
    }{
        (3): \quad
        s_1[\procA] \hasType q_0,
        s_2[\procA] \hasType q_4
        \types
        P_{\procA}
    }
\end{mathpar}

\begin{mathpar}
    \inferrule*[right=\procTypingExtCh]{
        \delta(q_2) =
        \set{
        (\rcv{\procA}{\procB}{\labelAndType{l}{\markerEndStateType}}, q_3)
        }
        \\
\inferrule*[right=\procTypingEnd]{
            \inferrule*[left=\procTypingZero]{
            }{
                \emptyset \types \zero
            }
            \\
            \EndState(q_3)
        }{
            s_1[\procB] \hasType q_3
            \types
            \zero
        }
    }{
        (4): \quad
        s_1[\procB] \hasType q_2
        \types
        P_{\procB}
    }

\end{mathpar}

\begin{mathpar}
    \inferrule*[right=\procTypingExtCh]{
        \inferrule*[left=\procTypingEnd]{
            \inferrule*[left=\procTypingIntCh]{
                \delta(q_0) =
                \set{
                (\snd{\procA}{\procB}{\labelAndType{l}{\markerEndStateType}}, q_1)
                }
                \\
\inferrule*[left=\procTypingEnd]{
                    \inferrule*[left=\procTypingZero]{
                    }{
                        \emptyset \types \zero
                    }
                    \\
\EndState(q_1)
                }{
                    x \hasType q_1 \types \zero
                }
            }{
                x \hasType q_0
                \types
                x[\procB] ! \labelAndMsg{l}{\markerEndStateState} \seq \zero
            }
            \\
\EndState(q_8)
        }{
            s_2[\procC] \hasType q_8,
            x \hasType q_0
            \types
            x[\procB] ! \labelAndMsg{l}{\markerEndStateState} \seq \zero
        }
        \\
\inferrule*[right=\procTypingEnd]{
            \inferrule*[right=\procTypingZero]{
            }{
                \emptyset \types \zero
            }
            \\
            \EndState(q_9)
        }{
            s_2[\procC] \hasType q_9
            \types
            \zero
        }
        \\
\delta(q_7) =
        \set{
        (\rcv{\procA}{\procC}{\labelAndType{l_1}{q_0}}, q_8),
        (\rcv{\procA}{\procC}{\labelAndType{l_2}{\markerEndStateType}}, q_9)
        }
    }{
        (5): \quad
        s_2[\procC] \hasType q_7
        \types
        P_{\procC}
    }
\end{mathpar}
}
\end{example}

During runtime, we have queues for each session.
For these, we define queue types and use them in queue typing contexts.
\begin{definition}
Queue types are defined by the following grammar:
\begin{grammar}
    \queueType \is
        \labelAndVar{l}{L} \cat \queueType
    |   \emptystring
\end{grammar}
A \emph{syntactic queue typing context} is defined by the following grammar:
\begin{grammar}
    \typingContextThree \is
        \typingContextThree, s[\procA][\procB] \hasQueueType \queueType
    |   \emptyset
\end{grammar}
A syntactic queue typing context is a \emph{queue typing context} if every element has at most one type.
Here, we do only consider queue typing contexts.
We consider queue typing contexts to be equivalent up to reordering and, thus, we may also treat them as mappings.
\end{definition}

While the (second) typing context specifies states for each participant, the queue typing context specifies the content of the queues.
This is all we need to define reductions following the respective communicating state machine.

\newcommand{\typingReductionIntCh}{\textsc{TR-$\IntCh$}\xspace}
\newcommand{\typingReductionExtCh}{\textsc{TR-$\ExtCh$}\xspace}

\begin{definition}
We define the reductions for typing contexts as follows:
\begin{mathpar}
\inferrule*[right=\typingReductionIntCh]{
    q
        \xrightarrow{\snd{\procA}{\procB}{\labelAndVar{l}{L}}}
    q'
}{
    s[\procA] \hasType q,
        \typingContextTwo
        \typingContextCat
        s[\procA][\procB] \hasQueueType \queueType,
        \typingContextThree
        \redto
    s[\procA] \hasType q',
        \typingContextTwo
        \typingContextCat
        s[\procA][\procB] \hasQueueType \queueType \cat \labelAndVar{l}{L},
        \typingContextThree
}

\inferrule*[right=\typingReductionExtCh]{
    q
        \xrightarrow{\rcv{\procB}{\procA}{\labelAndType{l}{L}}}
    q'
}{
    s[\procA] \hasType q,
        \typingContextTwo
        \typingContextCat
        s[\procB][\procA] \hasQueueType \labelAndVar{l}{L} \cat \queueType,
        \typingContextThree
        \redto
    s[\procA] \hasType q',
        \typingContextTwo
        \typingContextCat
        s[\procB][\procA] \hasQueueType \queueType,
        \typingContextThree
}
\end{mathpar}
These rules mimic exactly the semantics of communicating state machines.
\end{definition}

We show that reductions for typing contexts are preserved when adding type bindings to the typing contexts.

\begin{lemma}
\label{lm:typing-reduction-cong}
Let $\typingContextTwo_1, \typingContextTwo'_1$, and $\typingContextTwo_2$ be typing contexts and
$ \typingContextThree_1, \typingContextThree'_1$, and $\typingContextThree_2$ be queue typing contexts.
If
    $
        \typingContextTwo_1
            \typingContextCat
            \typingContextThree_1
        \redto
        \typingContextTwo'_1
            \typingContextCat
            \typingContextThree'_1
    $,
then
    $
        \typingContextTwo_1,
            \typingContextTwo_2
            \typingContextCat
            \typingContextThree_1,
            \typingContextThree_2
        \redto
        \typingContextTwo_1',
            \typingContextTwo_2
            \typingContextCat
            \typingContextThree_1',
            \typingContextThree_2
    $.
\end{lemma}
\begin{proof}
We do inversion on
    $
        \typingContextTwo_1
            \typingContextCat
            \typingContextThree_1
        \redto
        \typingContextTwo'_1
            \typingContextCat
            \typingContextThree'_1
    $, yielding two cases.
First, we have
\begin{mathpar}
    \small
    \inferrule*[right=\typingReductionIntCh]{
        q
            \xrightarrow{\snd{\procA}{\procB}{\labelAndVar{l}{L}}}
        q'
    }{
        s[\procA] \hasType q,
            \hat{\typingContextTwo}_1
            \typingContextCat
            s[\procA][\procB] \hasQueueType \queueType,
            \hat{\typingContextThree}_1
            \redto
        s[\procA] \hasType q',
            \hat{\typingContextTwo}_1
            \typingContextCat
            s[\procA][\procB] \hasQueueType \queueType \cat \labelAndVar{l}{L},
            \hat{\typingContextThree}_1
    }
\end{mathpar}
With this, it is obvious that the following holds:
\begin{mathpar}
    \small
    \inferrule*[right=\typingReductionIntCh]{
        q
            \xrightarrow{\snd{\procA}{\procB}{\labelAndVar{l}{L}}}
        q'
    }{
        s[\procA] \hasType q,
            \hat{\typingContextTwo}_1,
            \typingContextTwo_2
            \typingContextCat
            s[\procA][\procB] \hasQueueType \queueType,
            \hat{\typingContextThree}_1,
            \typingContextThree_2
            \redto
        s[\procA] \hasType q',
            \hat{\typingContextTwo}_1,
            \typingContextTwo_2
            \typingContextCat
            s[\procA][\procB] \hasQueueType \queueType \cat \labelAndVar{l}{L},
            \hat{\typingContextThree}_1,
            \typingContextThree_2
    }
\end{mathpar}
which is precisely what we have to show.
The case for \typingReductionExtCh is analogous and, therefore, omitted.
\proofEndSymbol
\end{proof}

After this small intermezzo on reductions for typing contexts, we now define the typing rules for runtime configurations.

\begin{figure}[t]
  \adjustfigure[\small]
  \begin{mathpar}
\inferrule*[right=\runtimeTypingProcDefEmpty]{
}{
      \types
      \emptystring \hasType \typingContextOne
  }

  \inferrule*[right=\runtimeTypingProcDef]{
      \typingContextOne
          \typingContextCat
          \vec{x} \hasType \vec{L}
      \types
      P \\
\typingContextOne(\pn{Q}) = \vec{L}
  }{
      \types
      (\pn{Q}[\vec{x}] = P); \Defs \hasType \typingContextOne
  }

  \inferrule*[right=\runtimeTypingProcName]{
      \typingContextOne(\pn{Q}) = \vec{L} }{
      \typingContextOne
          \typingContextCat
          \vec{c} \hasType \vec{L}
\typingContextCat
          \emptyset
      \types
      \pn{Q}[\vec{c}]
  }

  \inferrule*[right=\runtimeTypingZero]{
}{
      \typingContextOne
          \typingContextCat
          \emptyset
          \typingContextCat
          \emptyset
      \types
      \zero
  }

  \inferrule*[right=\runtimeTypingEnd]{
      \typingContextOne
          \typingContextCat
          \typingContextTwo
          \typingContextCat
          \typingContextThree
      \types
      R \\
\EndState(q)
}{
      \typingContextOne
      \typingContextCat
      \typingContextTwo,
      c \hasType q
      \typingContextCat
      \typingContextThree
          \types
      R
  }

  \inferrule*[right=\runtimeTypingParallel]{
      \typingContextOne
          \typingContextCat
          \typingContextTwo_1
          \typingContextCat
          \typingContextThree_1
          \types
          R_1 \\
\typingContextOne
          \typingContextCat
          \typingContextTwo_2
          \typingContextCat
          \typingContextThree_2
          \types
          R_2 \\
  }{
      \typingContextOne
          \typingContextCat
          \typingContextTwo_1,
          \typingContextTwo_2
          \typingContextCat
          \typingContextThree_1,
          \typingContextThree_2
          \types
          R_1 \parallel R_2
  }

  \inferrule*[right=\runtimeTypingIntCh]{
      \delta(q) \sups
      \set{(\snd{\procA}{\procB_i}{\labelAndType{l_i}{L_i}}, q_i) \mid i \in I}
      \\
\meta{\forall i \in I \st}
      \typingContextOne \typingContextCat
          \typingContextTwo , c \hasType q_i,
          \set{c_j \hasType L_j}_{j \in I\setminus \set{i}}
          \typingContextCat
          \typingContextThree
          \types P_i
  }{
      \typingContextOne
          \typingContextCat
          \typingContextTwo,
          c \hasType q,
          \set{c_i \hasType L_i}_{i \in I}
          \typingContextCat
          \typingContextThree
      \types
      \IntCh_{i \in I} c[\procB_i] ! \labelAndMsg{l_i}{c_i} \seq P_i
  }

  \inferrule*[right=\runtimeTypingExtCh]{
      \delta(q) =
      \set{(\rcv{\procB_i}{\procA}{\labelAndType{l_i}{L_i}}, q_i) \mid i \in I} \\
      \meta{\forall i \in I \st}
      \typingContextOne
          \typingContextCat
          \typingContextTwo,
          y_i \hasType L_i,
          c \hasType q_i
          \typingContextCat
          \typingContextThree
      \types
      P_i \\
  }{
      \typingContextOne
          \typingContextCat
          \typingContextTwo,
          c \hasType q
          \typingContextCat
          \typingContextThree
      \types
      \ExtCh_{i \in I} c[\procB_i] ? \labelAndVar{l_i}{y_i} \seq P_i
  }

  \inferrule*[right=\runtimeTypingRestr]{
(\vec{q}, \xi) \in \reach(\CSMabb{A})
      \\
      \typingContextTwo_s =
          \set{s[\procA] \hasType \vec{q}_\procA}_{\procA \in \ProcsOf{\CSMabb{A}}}
      \\
      \typingContextThree_s =
          \set{s[\procA][\procB] \hasQueueType \xi(\procA,\procB)}_{(\procA,\procB) \in \channelsOf{\CSMabb{A}}}
      \\
\typingContextOne
          \typingContextCat
          \typingContextTwo,
          \typingContextTwo_s
          \typingContextCat
          \typingContextThree,
          \typingContextThree_s
      \types
      R
}{
      \typingContextOne
          \typingContextCat
          \typingContextTwo
          \typingContextCat
          \typingContextThree
          \types
          (\restr s \hasType \CSMabb{A})\, R
  }

  \inferrule*[right=\runtimeTypingEmptyQueue]{
}{
      \typingContextOne
          \typingContextCat
          \emptyset
          \typingContextCat
          \set{s[\procA][\procB] \hasQueueType \emptystring}_{(\procA, \procB) \in \channelsOf{s}}
          \types
          \queueProc{s}{\emptyqueuecontent}
  }

  \inferrule*[right=\runtimeTypingQueue]{
      \typingContextOne
          \typingContextCat
          \typingContextTwo
          \typingContextCat
          \typingContextThree,
          s[\procA][\procB] \hasQueueType \queueType
          \types
          \queueProc{s}{\queuecontent[
              (\procA, \procB) \mapsto \vec{m}
          ]}
  }{
      \typingContextOne
          \typingContextCat
          \typingContextTwo,
          v \hasType L
          \typingContextCat
          \typingContextThree,
          s[\procA][\procB] \hasQueueType \labelAndVar{l}{L} \cat \queueType
          \types
          \queueProc{s}{\queuecontent[
              (\procA, \procB) \mapsto \labelAndMsg{l}{v} \cat \vec{m}
          ]}
  }
  \end{mathpar}
  \caption{Typing rules for runtime configurations; 
    $\reach(\hole)$ denotes the set of reachable configurations of the given CSM.}
  \label{fig:runtime-typing}
\end{figure}

\begin{definition}
The typing rules for runtime configurations are defined in \cref{fig:runtime-typing}.
We require all CSMs in typing derivations to be deadlock-free and to satisfy feasible eventual reception.

Most rules are analogous to the rules for processes.
For \runtimeTypingRestr, though, we do not require the CSM configuration to be initial but solely reachable, yielding typability of runtime configurations during execution.
The rules for queues are standard:
\runtimeTypingQueue types queues from the first to the last message in the queue while
\runtimeTypingEmptyQueue types empty queues.
\end{definition}

\begin{example}
\label{ex:runtime-conf-typing-delegation}
In
\cref{ex:proc-typing-delegation},
we gave a typing derivation for a process $P$ using delegation.
It is straightforward that this typing derivation can be mimicked for $\procToRuntime{P}$.
Here, we want to give a typing derivation after one reduction step, for the case where delegation happens.
We have
\begin{align*}
 R' \is  & \;
    (\restr s_1 \hasType \CSMabb{A})
    (\restr s_2 \hasType \CSMabb{B}) \;
    \zero \parallel
    P_{\procB} \parallel
    P_{\procC} \parallel
    \queueProc{s_1}{\emptystring} \parallel
    \queueProc{s_2}{\emptystring[(\procA, \procC) \mapsto \labelAndMsg{l_1}{s_1[\procA]}]}
    \enspace .
\end{align*}
We give a typing derivation for
$
        \emptyset
        \typingContextCat
        \emptyset
        \types
        R'
$, with label $(0)$.
Again, we omit the process definition typing context since it is empty throughout.

\vspace{-2ex}
{ \small
\begin{mathpar}
    \inferrule*[right=\procTypingParallel (4$\times$)]{
        \inferrule*[left=\runtimeTypingEmptyQueue]{
}{
            \emptyset
            \typingContextCat
            \typingContextThree_1
            \types
            \queueProc{s_1}{\emptystring}
        }
        \\
\inferrule*{
            (6)
        }{
            s_1[\procA] \hasType q_0,
            \typingContextCat
            \typingContextThree_2
            \types
            \queueProc{s_2}{\emptystring[(\procA, \procC) \mapsto \labelAndMsg{l_1}{s_1[\procA]}]}
        }
        \\
\inferrule*{
            (3)
            }{
            s_2[\procA] \hasType q_5
            \typingContextCat
            \emptyset
            \types
            \zero
        }
        \\
\inferrule*{
            (4)
            }{
            s_1[\procB] \hasType q_2
            \typingContextCat
            \emptyset
            \types
            P_{\procB}
        }
        \\
\inferrule*{
            (5)
            }{
            s_2[\procC] \hasType q_7
            \typingContextCat
            \emptyset
            \types
            P_{\procC}
        }
    }{
        (2): \quad
        \typingContextTwo_{s_1},
        \typingContextTwo'_{s_2}
        \typingContextCat
        \typingContextThree_{s_1},
        \typingContextThree'_{s_2}
        \types
        \zero \parallel
        P_{\procB} \parallel
        P_{\procC} \parallel
        \queueProc{s_1}{\emptystring} \parallel
        \queueProc{s_2}{\emptystring[(\procA, \procC) \mapsto \labelAndMsg{l_1}{s_1[\procA]}]}
    }
\end{mathpar}
\begin{mathpar}
    \inferrule*[right=\runtimeTypingRestr]{
        \inferrule*{
            (2)
        }{
            \typingContextTwo_{s_1},
            \typingContextTwo'_{s_2}
            \typingContextCat
            \typingContextThree_{s_1},
            \typingContextThree'_{s_2}
            \types
            \zero \parallel
            P_{\procB} \parallel
            P_{\procC} \parallel
            \queueProc{s_1}{\emptystring} \parallel
            \queueProc{s_2}{\emptystring[(\procA, \procC) \mapsto \labelAndMsg{l_1}{s_1[\procA]}]}
        }
        \\
\typingContextTwo'_{s_2} =
            s_2[\procA] \hasType q_5,
            s_2[\procC] \hasType q_7 \\
\typingContextThree'_{s_2} =
            s_2[\procA][\procC] \hasQueueType \labelAndType{l_1}{q_0},
            s_2[\procC][\procA] \hasQueueType \emptystring \\
}{
        (1): \quad
        \typingContextTwo_{s_1}
        \typingContextCat
        \typingContextThree_{s_1}
        \types
        (\restr s_2 \hasType \CSMabb{B}) \;
        \zero \parallel
        P_{\procB} \parallel
        P_{\procC} \parallel
        \queueProc{s_1}{\emptystring} \parallel
        \queueProc{s_2}{\emptystring[(\procA, \procC) \mapsto \labelAndMsg{l_1}{s_1[\procA]}]}
    }
\end{mathpar}
\begin{mathpar}
    \inferrule*[right=\runtimeTypingRestr]{
        \inferrule*{
            (1)
        }{
            \typingContextTwo_{s_1}
            \typingContextCat
            \typingContextThree_{s_1}
            \types
            (\restr s_2 \hasType \CSMabb{B}) \;
            \zero \parallel
            P_{\procB} \parallel
            P_{\procC} \parallel
            \queueProc{s_1}{\emptystring} \parallel
            \queueProc{s_2}{\emptystring[(\procA, \procC) \mapsto \labelAndMsg{l_1}{s_1[\procA]}]}
        }
        \\
\typingContextTwo_{s_1} =
            s_1[\procA] \hasType q_0,
            s_1[\procB] \hasType q_2 \\
\typingContextThree_{s_1} =
            s_1[\procA][\procB] \hasQueueType \emptystring,
s_1[\procB][\procA] \hasQueueType \emptystring \\
}{
        (0): \quad
        \emptyset
        \typingContextCat
        \emptyset
        \types
        (\restr s_1 \hasType \CSMabb{A})
        (\restr s_2 \hasType \CSMabb{B}) \;
        \zero \parallel
        P_{\procB} \parallel
        P_{\procC} \parallel
        \queueProc{s_1}{\emptystring} \parallel
        \queueProc{s_2}{\emptystring[(\procA, \procC) \mapsto \labelAndMsg{l_1}{s_1[\procA]}]}
    }
\end{mathpar}
}

The typing derivations for $(4)$ and $(5)$ are analogous to the ones in \cref{ex:proc-typing-delegation}.
The typing derivation for $(3)$ is straightforward with \runtimeTypingEnd and \runtimeTypingZero, similar to what we presented for $(3)$ in
\cref{ex:proc-typing-delegation}.
We give the typing derivation for~$(6)$:

\vspace{-2ex}
{ \small
\begin{mathpar}
    \inferrule*[right=\runtimeTypingQueue]{
        \emptyset
        \typingContextCat
        s_2[\procA][\procC] \hasQueueType \emptystring,
        s_2[\procC][\procA] \hasQueueType \emptystring
        \types
        \queueProc{s_2}{\emptystring}
    }{
        (6): \quad
        s_1[\procA] \hasType q_0,
        \typingContextCat
        s_2[\procA][\procC] \hasQueueType \labelAndType{l_1}{q_0},
        s_2[\procC][\procA] \hasQueueType \emptystring
        \types
        \queueProc{s_2}{\emptystring[(\procA, \procC) \mapsto \labelAndMsg{l_1}{s_1[\procA]}]}
    }
\end{mathpar}
}
\end{example}

The presentation of our type system is inspired by work from
\citet{DBLP:journals/pacmpl/ScalasY19}.
Thus, we want to highlight key differences.
First, we handle sender-driven choice, which allows a participant to send to different receivers and to receive from different senders, while they only consider directed choice.
In fact, our processes can choose between different options when sending messages while their work restricts to a single choice.
However, their treatment could be combined with control flow like if-then-else to cover more scenarios.
Also, \cite{DBLP:journals/pacmpl/ScalasY19} employs subtyping similar to what we do for send actions so there can still be multiple branches when sending.
In a sender-driven choice setting, there are subtleties for subtyping as one cannot simply add receives for instance.
For details, we refer to \cite{DBLP:conf/esop/LiSW24} where CSMs are considered for subtyping and protocol refinement.
With our framework, one can directly make use of their results: we can apply their subtyping algorithm to obtain another CSM against which we type check.
Second, \cite{DBLP:journals/pacmpl/ScalasY19} only considers one error scenario:
a participant would like to receive something but the first message in the respective queue does not match.
We generalise this scenario to our setting and require that the first message in all respective queues does not match.
In addition, we consider the error case where a session ended with non-empty queues.
\cite{DBLP:journals/pacmpl/ScalasY19} also considers S-deadlock freedom and we refer to a discussion in \cref{sec:typing-for-csms}.
In the next section, we will prove that our type system prevents both these scenarios.

\begin{remark}[Adding base types]
The treatment of base types and expressions in type systems is well-understood.
Hence, it should be straightforward to extend our type system to add expressions.
More specifically, one could then send the result of expressions and, hence, variables can be bound to values.
Provided with (Boolean) expressions, it is also standard to add features of control-flow like if-then-else.
For most flexible use of our results, it would make most sense if a user provided a type system for the expressions and base types they need.
Then, the type system would use what we defined for local types (and hence delegation) and the provided type system for expressions.
There is one important difference between both type systems.
While ours is linear, the one for expressions does not need to be.
With a non-linear type system, one can duplicate and drop type bindings from the typing context, using rules called \emph{contraction} and \emph{weakening}.
\end{remark}

\subsection{Soundness of Type System}

For conciseness, we assume a process definition typing context $\typingContextOne$, typing contexts
$\typingContextTwo$, $\typingContextTwo_1$, $\typingContextTwo_2$, $\ldots$,
and queue typing contexts
$\typingContextThree$, $\typingContextThree_1$, $\typingContextThree_2$, $\ldots$
in this section.

We presented a type system for processes and runtime configurations.
While we closed the reduction semantics under structural precongruence $\precongr$, we have not stated the respective rules for our type~system:

\vspace{-2ex}
{ \small
\begin{mathpar}
    \inferrule*[right=\procTypingCongr]{
        \typingContextOne
            \typingContextCat
            \typingContextTwo
            \types
            P \\
P \precongr P'
    }{
        \typingContextOne
            \typingContextCat
            \typingContextTwo
            \types
            P'
    }

    \inferrule*[right=\runtimeTypingCongr]{
        \typingContextOne
            \typingContextCat
            \typingContextTwo
            \typingContextCat
            \typingContextThree
            \types
            R \\
R \precongr R'
    }{
        \typingContextOne
            \typingContextCat
            \typingContextTwo
            \typingContextCat
            \typingContextThree
            \types
            R'
    }
\end{mathpar}
}

We show that these rules are admissible, \ie they can be added without changing the capabilities of the type system.
This allows us not to consider these rules in the following proofs but still use them if convenient.

\begin{lemma}[Admissibility of structural precongruence for runtime configuration typing]
\label{lm:str-congr-preserves-typability-runtime-confs}
Let $R_1$ and $R_2$ be well-annotated runtime configurations.
If
    $
        \typingContextOne
        \typingContextCat
        \typingContextTwo
        \typingContextCat
        \typingContextThree
            \types R_1
    $
    and
    $
        R_1 \precongr R_2
    $,
then
    $
        \typingContextOne
        \typingContextCat
        \typingContextTwo
        \typingContextCat
        \typingContextThree
            \types R_2
    $.
\end{lemma}
\begin{proof}
We first consider the cases for structural congruence $\congr$ and then the additional ones for structural precongruence.
We do a case analysis on $\congr$ and reason for both directions.
Subsequently, we consider the two rules for $\precongr$.
\begin{itemize}
 \item $R_1 \parallel R_2
        \congr
        R_2 \parallel R_1$: \\
        By inversion, we know that \runtimeTypingParallel is the first rule applied in the typing derivation.
        This rule is symmetric so basically the typing derivation works.
 \item $(R_1 \parallel R_2) \parallel R_3
        \congr
        R_1 \parallel (R_2 \parallel R_3)$: \\
        By inversion, we know that \runtimeTypingParallel is the first and second rule applied in the typing derivation.
        It is easy to see that the typing derivation can be rearranged to match the structure.
 \item $R \parallel \zero
        \congr
        R$: \\
        First, assume that there is a typing derivation

        \vspace{-2ex}
        { \small
        \begin{mathpar}
            \inferrule*[right=\runtimeTypingParallel]{
                \typingContextOne
                    \typingContextCat
                    \typingContextTwo_1
                    \typingContextCat
                    \typingContextThree_1
                    \types
                    R \\
\typingContextOne
                    \typingContextCat
                    \typingContextTwo_2
                    \typingContextCat
                    \typingContextThree_2
                    \types
                    \zero \\
            }{
                \typingContextOne
                    \typingContextCat
                    \typingContextTwo_1,
                    \typingContextTwo_2
                    \typingContextCat
                    \typingContextThree_1,
                    \typingContextThree_2
                    \types
                    R \parallel \zero
            }
        \end{mathpar}
        }

        We show there is a typing derivation
        $
            \typingContextOne
                \typingContextCat
                \typingContextTwo_1,
                \typingContextTwo_2
                \typingContextCat
                \typingContextThree_1,
                \typingContextThree_2
                \types
                R
        $.
        Inversion yields that two rules can be applied for the given typing derivation
        $
            \typingContextOne
                \typingContextCat
                \typingContextTwo_2
                \typingContextCat
                \typingContextThree_2
                \types
                \zero
        $:
        \runtimeTypingEnd
        and
        \runtimeTypingZero.
        Thus, it follows that
        $\typingContextThree_2 = \emptyset$.
Also,
        $
            \typingContextTwo_2 =
                \set{s[\procA] \hasType \vec{q}_\procA}_{\procA \in S}
        $
        for some set of participants $S$
        and
        $\EndState(\vec{q}_\procA)$
for every
        $\procA \in S$.
        By inversion, there is a typing derivation for
        $
                \typingContextOne
                    \typingContextCat
                    \typingContextTwo_1
                    \typingContextCat
                    \typingContextThree_1
                    \types
                    R
        $.
        With
        $\typingContextThree_2 = \emptyset$,
        it remains to show that there is a typing derivation
        $
            \typingContextOne
                \typingContextCat
                \typingContextTwo_1,
                \typingContextTwo_2
                \typingContextCat
                \typingContextThree_1
                \types
                R
        $.
        The only difference is the typing context
        $\typingContextTwo_2$.
        This, however, can be taken care of using \runtimeTypingEnd as in the other typing derivation, concluding this~case.
        \\
        Second, assume there is a typing derivation for $R$.
        We show there is a typing derivation for
        $
            \typingContextOne
                \typingContextCat
                \typingContextTwo
                \typingContextCat
                \typingContextThree
                \types
                R \parallel \zero
        $.
        We first apply \runtimeTypingParallel to obtain

        \vspace{-2ex}
        { \small
        \begin{mathpar}
            \inferrule*[right=\runtimeTypingParallel]{
                \typingContextOne
                    \typingContextCat
                    \typingContextTwo
                    \typingContextCat
                    \typingContextThree
                    \types
                    R \\
\inferrule*[right=\runtimeTypingZero]{
}{
                \typingContextOne
                    \typingContextCat
                    \emptyset
                    \typingContextCat
                    \emptyset
                    \types
                    \zero
                }
            }{
                \typingContextOne
                    \typingContextCat
                    \typingContextTwo
                    \typingContextCat
                    \typingContextThree
                    \types
                    R \parallel \zero
            }
        \end{mathpar}
        }

        for which the right premise is met with \runtimeTypingZero and the left premise is given by~assumption.

 \item $(\restr s \hasType \CSMabb{A}) \, (\restr s' \hasType \CSMabb{B}) \, R
        \congr
        (\restr s' \hasType \CSMabb{B}) \, (\restr s \hasType \CSMabb{A}) \, R$: \\
        By inversion, both typing derivations need to apply
        \runtimeTypingRestr
        twice in the beginning.
        It is straightforward that both rule applications do not interfere with each other, yielding the same premise to prove:
        \[
            \typingContextOne
                \typingContextCat
                \typingContextTwo,
                \typingContextTwo_s,
                \typingContextTwo_{s'}
                \typingContextCat
                \typingContextThree,
                \typingContextThree_s,
                \typingContextThree_{s'}
            \types
            R
        \]
        Thus, this is given by assumption.

 \item $(\restr s \hasType \CSMabb{A}) \, (R_1 \parallel R_2)
        \congr
        R_1 \parallel (\restr s \hasType \CSMabb{A}) \, R_2$
        and $s$ is not free in $R_1$: \\
        First, we assume there is a typing derivation for
        $(\restr s \hasType \CSMabb{A}) \, (R_1 \parallel R_2)$
        and show there is a typing derivation for
        $R_1 \parallel (\restr s \hasType \CSMabb{A}) \, R_2$.
        Applying inversion twice yields

        \vspace{-2ex}
        { \footnotesize
        \begin{mathpar}
            \inferrule*[right=\runtimeTypingRestr]{
                \inferrule*[left=\runtimeTypingParallel]{
                    \typingContextOne
                        \typingContextCat
                        \typingContextTwo_1
                        \typingContextCat
                        \typingContextThree_1
                    \types
                    R_1
                    \\
                    \typingContextOne
                        \typingContextCat
                        \typingContextTwo_2,
                        \typingContextTwo_s
                        \typingContextCat
                        \typingContextThree_2,
                        \typingContextThree_s
                    \types
                    R_2
                }{
                    \typingContextOne
                        \typingContextCat
                        \typingContextTwo_1,
                        \typingContextTwo_2,
                        \typingContextTwo_s
                        \typingContextCat
                        \typingContextThree_1,
                        \typingContextThree_2,
                        \typingContextThree_s
                    \types
                    R_1 \parallel R_2
                }
(\vec{q}, \xi) \in \reach(\CSMabb{A})
                \\
                \typingContextTwo_s =
                    \set{s[\procA] \hasType \vec{q}_\procA}_{\procA \in \ProcsOf{\CSMabb{A}}}
                \\
                \typingContextThree_s =
                    \set{s[\procA][\procB] \hasQueueType \xi(\procA,\procB)}_{(\procA,\procB) \in \channelsOf{\CSMabb{A}}}
                \\
}{
                \typingContextOne
                    \typingContextCat
                    \typingContextTwo_1,
                    \typingContextTwo_2
                    \typingContextCat
                    \typingContextThree_1,
                    \typingContextThree_2
                    \types
                    (\restr s \hasType \CSMabb{A}) \, (R_1 \parallel R_2)
            }
        \end{mathpar}
        }

        where
        $\typingContextTwo =
                    \typingContextTwo_1,
                    \typingContextTwo_2$
        and
        $\typingContextThree =
                    \typingContextThree_1,
                    \typingContextThree_2$.
        We claim we can assume that $\typingContextTwo_s$ and $\typingContextThree_s$ are used in the typing derivation for
        $
            \typingContextOne
                \typingContextCat
                \typingContextTwo_2,
                \typingContextTwo_s
                \typingContextCat
                \typingContextThree_2,
                \typingContextThree_s
            \types
            R_2
        $.
        By definition, these only contain type bindings related to $s$, which does not occur in $R_1$ by assumption.
        There might exist a typing derivation where parts of $\typingContextTwo_s$ or $\typingContextThree_s$ appear in the typing derivation for $R_1$ but these can only removed with the rules
        \runtimeTypingEnd
        (not even with
        \runtimeTypingEmptyQueue
        since this requires
        $\queueProc{s}{\emptyqueuecontent}$).
        Hence, such derivations can be mimicked in the typing derivation for $R_2$, justifying our treatment of $\typingContextTwo_s$ and $\typingContextThree_s$.
        We construct a typing derivation:

        \vspace{-2ex}
        { \small
        \begin{mathpar}
            \inferrule*[right=\runtimeTypingParallel]{
                \inferrule*[right=\runtimeTypingRestr]{
(\vec{q}, \xi) \in \reach(\CSMabb{A})
                    \\
                    \typingContextTwo_s =
                        \set{s[\procA] \hasType \vec{q}_\procA}_{\procA \in \ProcsOf{\CSMabb{A}}}
                    \\
                    \typingContextThree_s =
                        \set{s[\procA][\procB] \hasQueueType \xi(\procA,\procB)}_{(\procA,\procB) \in \channelsOf{\CSMabb{A}}}
                    \\
                    \typingContextOne
                        \typingContextCat
                        \typingContextTwo_2,
                        \typingContextTwo_s
                        \typingContextCat
                        \typingContextThree_2,
                        \typingContextThree_s
                    \types
                    R_2
                }{
                    \typingContextOne
                        \typingContextCat
                        \typingContextTwo_2
                        \typingContextCat
                        \typingContextThree_2
                    \types
                    (\restr s \hasType \CSMabb{A}) \, R_2
                }
                \\
\typingContextOne
                    \typingContextCat
                    \typingContextTwo_1
                    \typingContextCat
                    \typingContextThree_1
                \types
                R_1
            }{
                \typingContextOne
                    \typingContextCat
                    \typingContextTwo_1,
                    \typingContextTwo_2
                    \typingContextCat
                    \typingContextThree_1,
                    \typingContextThree_2
                    \types
                    R_1 \parallel (\restr s \hasType \CSMabb{A}) \, R_2
            }
        \end{mathpar}
        }

        All premises coincide with the ones of the original typing derivation, concluding this case.
        \\
        Second, we assume there is a typing derivation for
        $R_1 \parallel (\restr s \hasType \CSMabb{A}) \, R_2$
        and show there is a typing derivation for
        $(\restr s \hasType \CSMabb{A}) \, (R_1 \parallel R_2)$.
        The proof is analogous to the previous case but we do not need to reason about the treatment of $\typingContextTwo_s$ and $\typingContextThree_s$ but it suffices to show there is one typing and we can choose the respective treatment.

 \item $(\restr s \hasType \CSMabb{A}) \, \queueProc{s}{\emptystring} \precongr \zero$: \\
We assume there is a typing derivation for
        $
            \typingContextOne
                \typingContextCat
                \typingContextTwo
                \typingContextCat
                \typingContextThree
            \types
            (\restr s \hasType \CSMabb{A}) \, \queueProc{s}{\emptystring}
        $.
        We show there is a typing derivation for
        $
            \typingContextOne
                \typingContextCat
                \typingContextTwo
                \typingContextCat
                \typingContextThree
            \types
            \zero
        $.
        By inversion, we know that \runtimeTypingRestr is the last rule to be applied and we get one of the premises:
        \[
            \typingContextOne
                \typingContextCat
                \typingContextTwo,
                \typingContextTwo_s
                \typingContextCat
                \typingContextThree,
                \typingContextThree_s
            \types
            \queueProc{s}{\emptystring}
        \]
        with
        $
            \typingContextTwo_s =
                \set{s[\procA] \hasType \vec{q}_\procA}_{\procA \in \ProcsOf{\CSMabb{A}}}
        $ and $
            \typingContextThree_s =
                \set{s[\procA][\procB] \hasQueueType \xi(\procA,\procB)}_{(\procA,\procB) \in \channelsOf{\CSMabb{A}}}
        $.
        By inversion,
        \runtimeTypingEmptyQueue
        and
        \runtimeTypingEnd
        are the only rules that can be applied in the typing derivation for
        $
            \typingContextOne
                \typingContextCat
                \typingContextTwo,
                \typingContextTwo_s
                \typingContextCat
                \typingContextThree,
                \typingContextThree_s
            \types
            \zero
        $.
        Thus, we have
        $\typingContextThree = \emptyset$,
        changing our proof obligation to
        $
            \typingContextOne
                \typingContextCat
                \typingContextTwo
                \typingContextCat
                \typingContextThree_s
            \types
            \zero
        $.
        Since
        \runtimeTypingEmptyQueue
        does only change the queue typing context,
        we have that
        $
            \typingContextTwo =
                \Union_{s' \in \SessionName}
                \set{s'[\procA] \hasType \vec{q}_\procA}_{\procA \in \Procs_{s'}}
        $
        for a set of sessions $\SessionName$ that does not contain $s$ and
        $\EndState(\vec{q}_\procA)$
for every $\procA \in \Procs_{s'}$ and $s' \in \SessionName$.
        Therefore, we can also first apply \runtimeTypingEnd $\card{\typingContextTwo}$ times and
        last \runtimeTypingEmptyQueue to obtain a typing derivation for
        $
            \typingContextOne
                \typingContextCat
                \typingContextTwo
                \typingContextCat
                \typingContextThree_s
            \types
            \zero
        $.

\end{itemize}
This concludes the proof.
\proofEndSymbol
\end{proof}

We proved the admissibility lemma for the type system for runtime configurations.
The proof for the type system for processes is analogous for most cases.

\begin{lemma}[Admissibility of structural precongruence for process typing]
\label{lm:str-congr-preserves-typability-processes}
Let $P_1$ and $P_2$ be well-annotated processes.
If it holds that
    $
        \typingContextOne
        \typingContextCat
        \typingContextTwo
            \types P_1
    $
    and
    $
        P_1 \precongr P_2
    $,
then
    $
        \typingContextOne
        \typingContextCat
        \typingContextTwo
            \types P_2
    $.
\end{lemma}
\begin{proof}
The respective cases are analogous to the ones in the proof of
\cref{lm:str-congr-preserves-typability-runtime-confs}.
We only need to consider the case where
$(\restr s \hasType \CSMabb{A}) \, \zero
    \precongr
    \zero$: \\
We assume there is a typing derivation for
    $
        \typingContextOne
            \typingContextCat
            \typingContextTwo
        \types
        (\restr s \hasType \CSMabb{A}) \, \zero
    $.
    We show there is a typing derivation for
    $
        \typingContextOne
            \typingContextCat
            \typingContextTwo
        \types
        \zero
    $.
    By inversion, we know that \procTypingRestr is the last rule to be applied and we get one of the premises:
    $
        \typingContextOne
            \typingContextCat
            \typingContextTwo,
            \typingContextTwo_s
        \types
        \zero
    $
    with
    $
        \typingContextTwo_s =
            \set{s[\procA] \hasType \initialState(\CSMabb{A}_{\procA})}_{\procA \in \ProcsOf{\CSMabb{A}}}
    $.
    By inversion,
    \procTypingZero
    and
    \procTypingEnd
    are the only rules that can be applied in the typing derivation for
    $
        \typingContextOne
            \typingContextCat
            \typingContextTwo,
            \typingContextTwo_s
        \types
        \zero
    $.
    Since
    \procTypingZero
    needs the second typing context to be empty, it is applied last and all other derivations are applications of
    \procTypingEnd.
    Therefore,
    $
        \typingContextTwo =
            \Union_{s' \in \SessionName}
            \set{s'[\procA] \hasType \vec{q}_\procA}_{\procA \in \Procs_{s'}}
    $
    for some set of sessions $\SessionName$ that does not contain $s$
    and
    $\EndState(\vec{q}_\procA)$
for every $\procA \in \Procs_{s'}$ and $s' \in \SessionName$.
    Therefore, we can also first apply \runtimeTypingEnd $\card{\typingContextTwo}$ times and
    last \runtimeTypingZero to obtain a typing derivation for
    $
        \typingContextOne
            \typingContextCat
            \typingContextTwo
            \typingContextCat
            \emptyset
        \types
        \zero
    $.
\proofEndSymbol
\end{proof}

We proceed with a few observations about our type system that we will use later in the proof for our main result.

To start, we show that a term $x$ cannot appear in a runtime configuration if there is no type binding for it in the typing context.

\begin{lemma}
\label{lm:every-variable-in-runtime-configuration-needs-typing}
Let $R$ be a well-annotated runtime configuration.
If
$\typingContextOne
\typingContextCat
\typingContextTwo
\typingContextCat
\typingContextThree
\types R$
and $x$ is not in $\typingContextTwo$,
then $x$ cannot occur in $R$.
\end{lemma}
\begin{proof}
Towards a contradiction, assume that $x$ occurs in $R$.
Then, at some point in the typing derivation
\[
    \typingContextOne
    \typingContextCat
    \typingContextTwo
    \typingContextCat
    \typingContextThree
    \types R
\]
one of the following rules applies to handle $x$:
\runtimeTypingProcName,
\runtimeTypingIntCh, or
\runtimeTypingExtCh.
Each of them requires all variables to occur in their respective typing contexts, yielding a contradiction.
\proofEndSymbol
\end{proof}

We defined a type system for processes and one for runtime configurations.
Both are very similar and we defined runtime configurations to only have queues for active sessions.
Once a process becomes active, we turn it into a runtime configuration using~$\procToRuntime{\hole}$.
We show that this preserves typability with an empty queue typing context.

\begin{lemma}
\label{lm:typing-proc-to-runtime-configuration}
Let $P$ be a well-annotated process.
If
$\typingContextOne \typingContextCat \typingContextTwo \types P$,
then 
$\typingContextOne \typingContextCat \typingContextTwo \typingContextCat \emptyset \types \procToRuntime{P}$.
\end{lemma}
\begin{proof}
We prove this by induction on the structure of $P$.

For all except $P = P_1 \parallel P_2$ and $P = (\restr s \hasType \CSMabb{A}) \, P'$, it holds that
$\procToRuntime{P} = P$.
For all typing rules that processes and runtime configurations share, the queue typing context is not changed in the respective runtime configuration typing rule.
Thus,
$\typingContextOne \typingContextCat \typingContextTwo \typingContextCat \emptyset \types \procToRuntime{P}$.

For $P = P_1 \parallel P_2$, the claim follows directly by induction hypothesis.

Last, we consider $P = (\restr s \hasType \CSMabb{A}) \, P'$.
We have the following typing derivation

\vspace{-2ex}
{ \small
\begin{mathpar}
    \inferrule*[right=\procTypingRestr]{
        \typingContextTwo' =
            \set{s[\procA] \hasType \initialState(\CSMabb{A}_{\procA})}_{\procA \in \ProcsOf{\CSMabb{A}}} \\
\typingContextOne
            \typingContextCat
            \typingContextTwo,
            \typingContextTwo'
            \types
            P'
    }{
        \typingContextOne
            \typingContextCat
            \typingContextTwo
            \types
            (\restr s \hasType \CSMabb{A})\, P'
    }
\end{mathpar}
}

We show there is a typing derivation where the last rule is $\runtimeTypingRestr$.
This requires a reachable configuration in the respective CSM:
        $(\vec{q}, \xi) \in \reach(\CSMabb{A})$.
In this case, we choose the initial states and empty channels, which allows us to use \runtimeTypingEmptyQueue for the queues.

\vspace{-2ex}
{ \small
\begin{mathpar}
    \inferrule*[right=\runtimeTypingRestr]{
        \inferrule*[right=\runtimeTypingParallel]{
            \typingContextOne
                \typingContextCat
                \typingContextTwo,
                \typingContextTwo_s
                \typingContextCat
                \emptyset
            \types
            P'
            \\
            \inferrule*[right=\runtimeTypingEmptyQueue]{
}{
                \typingContextOne
                    \typingContextCat
                    \emptyset
                    \typingContextCat
                    \typingContextThree_s
                \types
                \queueProc{s}{\emptystring}
            }
        }{
            \typingContextOne
                \typingContextCat
                \typingContextTwo,
                \typingContextTwo_s
                \typingContextCat
                \typingContextThree_s
            \types
            (P' \parallel \queueProc{s}{\emptystring})
        }
        \\
(\vec{q}, \xi) \in \reach(\CSMabb{A})
        \\
        \typingContextTwo_s =
            \set{s[\procA] \hasType \vec{q}_\procA}_{\procA \in \ProcsOf{\CSMabb{A}}}
        \\
        \typingContextThree_s =
            \set{s[\procA][\procB] \hasQueueType \xi(\procA,\procB)}_{(\procA,\procB) \in \channelsOf{\CSMabb{A}}}
        \\
}{
        \typingContextOne
            \typingContextCat
            \typingContextTwo
            \typingContextCat
            \emptyset
            \types
            (\restr s \hasType \CSMabb{A})\, (P' \parallel \queueProc{s}{\emptystring})
    }
\end{mathpar}
}

where $\vec{q}_\procA = \initialState(\CSMabb{A}_{\procA})$ for every $\procA$ and $\xi(\procA, \procB) = \emptystring$ for every $\procA, \procB$.
Thus, $(\vec{q}, \xi)$ is clearly reachable.
Also, both second typing contexts then coincide: $\typingContextTwo' = \typingContextTwo_s$.
Thus, $ \typingContextOne
            \typingContextCat
            \typingContextTwo,
            \typingContextTwo_s
            \typingContextCat
            \emptyset
        \types
        P'
    $, the last premise to satisfy,
follows from the induction hypothesis.
\proofEndSymbol
\end{proof}

With the following lemma, we show that, if there is a typing derivation for a process, then the queue typing context is empty.

\begin{lemma}
\label{lm:proc-typing-queue-context-empty}
If $P$ is a well-annotated process such that
$
    \typingContextOne
    \typingContextCat
    \typingContextTwo
    \typingContextCat
    \typingContextThree
        \types
    P
$,
then $\typingContextThree = \emptyset$.
\end{lemma}
\begin{proof}
We do induction on the depth of the typing derivation.

For the base case, we consider \runtimeTypingZero, for which the claim trivially holds, and \runtimeTypingEmptyQueue, for which we reach a contradiction because $\queueProc{s}{\queuecontent}$ is no process.

Let us turn to the induction step:
\begin{itemize}
 \item \runtimeTypingProcName: trivially holds
 \item \runtimeTypingEnd: by inversion and induction hypothesis
 \item \runtimeTypingIntCh: by inversion and induction hypothesis for every $i \in I$
 \item \runtimeTypingExtCh: by inversion and induction hypothesis for every $i \in I$
 \item \runtimeTypingParallel: by inversion and induction hypothesis twice
 \item \runtimeTypingRestr: by inversion and induction hypothesis
 \item \runtimeTypingQueue: contradiction because $\queueProc{s}{\queuecontent[(\procA, \procB) \mapsto \labelAndMsg{l}{v} \cat \vec{m}]}$ is no process
\end{itemize}
This concludes the proof.
\proofEndSymbol
\end{proof}

In our type system, we type a queue from the first to the last element, when using \runtimeTypingQueue.
Thus, when applying inversion for this rule, we only get the type for the first element of a non-empty queue.
The following lemma allows us to also obtain the type for its last element.

\begin{lemma}[Message list reversal]
\label{lm:typing-message-list-reversal}
Let
$\queueType$
be a queue type.
If
\[
\typingContextOne \typingContextCat \typingContextTwo \typingContextCat
\typingContextThree,
s[\procA][\procB] \hasQueueType \queueType
\types
\queueProc{s}{\queuecontent[
            (\procA, \procB) \mapsto
                \vec{m}
        ]}, \text{ then}
\]
\[
\typingContextOne
\typingContextCat
\typingContextTwo,
v \hasType L
\typingContextCat
\typingContextThree,
s[\procA][\procB] \hasQueueType \queueType \cat \labelAndType{l}{L}
\types
\queueProc{s}{\queuecontent[
            (\procA, \procB) \mapsto
                \vec{m} \cat \labelAndMsg{l}{v}
        ]}
        \enspace .
\]
\end{lemma}
\begin{proof}
We prove this claim by induction on the length $n$ of $\vec{m} = \labelAndMsg{l_1}{v_1} \cat \ldots \cat \labelAndMsg{l_n}{v_n}$.

If $n = 0$, the claim is exactly the assumption.

For the induction step, we assume that
$\vec{m} = \labelAndMsg{l_1}{v_1} \cat \labelAndMsg{l_2}{v_2} \cat \ldots \cat \labelAndMsg{l_n}{v_n}$
and the induction hypothesis holds for $\labelAndMsg{l_2}{v_2} \cat \ldots \cat \labelAndMsg{l_n}{v_n}$.

We want to show that
\[
    \typingContextOne
    \typingContextCat
    \typingContextTwo
    \typingContextCat
    \typingContextThree,
    s[\procA][\procB] \hasQueueType \queueType
    \types
    \queueProc{s}{\queuecontent[
                (\procA, \procB) \mapsto
                \labelAndMsg{l_1}{v_1} \cat \labelAndMsg{l_2}{v_2} \cat \ldots \cat \labelAndMsg{l_n}{v_n}
            ]} \enspace \text{ implies}
\]
\[
    \typingContextOne
    \typingContextCat
    \typingContextTwo,
    v \hasType L
    \typingContextCat
    \typingContextThree,
    s[\procA][\procB] \hasQueueType \queueType \cat \labelAndType{l}{L}
    \types
    \queueProc{s}{\queuecontent[
                (\procA, \procB) \mapsto
                \labelAndMsg{l_1}{v_1} \cat \labelAndMsg{l_2}{v_2} \cat \ldots \cat \labelAndMsg{l_n}{v_n} \cat \labelAndMsg{l}{v}
            ]}
  \enspace .
\]
By inversion of the premise, we know that $\typingContextTwo = \typingContextTwo', v_1 \hasType L_1$ and $\queueType = L_1 \cat \queueType'$ in order to type $v_1$ with some type $L_1$.
Thus, we can apply \runtimeTypingQueue to our goal and then apply the induction hypothesis with $\typingContextTwo = \typingContextTwo'$ and $\queueType = \queueType'$ to conclude the proof.
\proofEndSymbol
\end{proof}

Next, we show that the queue types reflect what is in the queues of runtime~configurations.

\begin{lemma}
\label{lm:queue-types-model-queues}
If $
    \typingContextOne
        \typingContextCat
        \typingContextTwo
        \typingContextCat
        \typingContextThree,
        s[\procA][\procB] \hasQueueType
            \labelAndVar{l_1}{L_1} \cat \ldots
            \cat \labelAndVar{l_k}{L_k}
        \types
        \queueProc{s}{\queuecontent[
            (\procA, \procB) \mapsto \labelAndMsg{l'_1}{v'_1}, \ldots, \labelAndMsg{l'_n}{v'_n}
        ]}
   $,
then $k = n$ and,
for all $1 \leq i \leq k$, $l'_i = l_i$ and $v'_i \hasType L_i$.
\end{lemma}
\begin{proof}
We do an induction on the depth of the typing derivation.

For the induction base, we have the following typing derivation:

\vspace{-2ex}
    { \small
    \begin{mathpar}
    \inferrule*[right=\runtimeTypingEmptyQueue]{
}{
        \typingContextOne
            \typingContextCat
            \emptyset
            \typingContextCat
            \set{s[\procA][\procB] \hasQueueType \emptystring}_{(\procA, \procB) \in \channelsOf{s}}
            \types
            \queueProc{s}{\emptyqueuecontent}
    }
    \end{mathpar}
    }

It is obvious that $k = 0 = n$ and there are no messages to consider.

For the induction step, we have the following typing derivation:

    \vspace{-2ex}
    { \tiny
    \begin{mathpar}
    \inferrule*[right=\runtimeTypingQueue]{
        \typingContextOne
            \typingContextCat
            \typingContextTwo
            \typingContextCat
            \typingContextThree,
            s[\procA][\procB] \hasQueueType
            \labelAndVar{l_1}{L_1} \cat
            \labelAndVar{l_2}{L_2} \cat
            \ldots
            \labelAndVar{l_k}{L_k}
            \types
            \queueProc{s}{\queuecontent[
                (\procA, \procB) \mapsto
                \labelAndMsg{l'_2}{v'_2} \cat
                \ldots
                \labelAndMsg{l'_n}{v'_n}
            ]}
    }{
        \typingContextOne
            \typingContextCat
            \typingContextTwo,
            v'_1 \hasType L_1
            \typingContextCat
            \typingContextThree,
            s[\procA][\procB] \hasQueueType
            \labelAndVar{l_1}{L_1} \cat
            \labelAndVar{l_2}{L_2} \cat
            \ldots
            \labelAndVar{l_k}{L_k}
            \types
            \queueProc{s}{\queuecontent[
                (\procA, \procB) \mapsto
                \labelAndMsg{l_1}{v'_1} \cat
                \labelAndMsg{l'_2}{v'_2} \cat
                \ldots
                \labelAndMsg{l'_n}{v'_n}
            ]}
    }
    \end{mathpar}
    }

Let us first consider the length of the queue type and the queue:
by induction hypothesis, we know that $n-1 = k-1$ and, thus, $k = n$.
Second, let us consider the labels and payload types.
For $i = 1$, the typing rule requires the labels to match and $v'_1 \hasType L_1$ is required in the typing context.
For $i > 1$, the induction hypothesis~applies.
\proofEndSymbol
\end{proof}

We also provided typing context reductions.
Here, we show that these actually preserve reachability for the CSM associated with a session.

\begin{lemma}[Typing reductions preserve reachability]
\label{lm:typing-reductions-preserve-reachability}
Let $\typingContextTwo =
            \hat{\typingContextTwo},
            \set{\typingContextTwo_s}_{s \in \SessionName}
$
be a typing context
and $\typingContextThree =
            \hat{\typingContextThree},
            \set{\typingContextThree_s}_{s \in \SessionName}
$
be a queue typing context with a set of sessions~$\SessionName$.
Assume that
\begin{itemize}
    \item $
            \hat{\typingContextTwo},
            \set{\typingContextTwo_s}_{s \in \SessionName}
            \typingContextCat
            \hat{\typingContextThree},
            \set{\typingContextThree_s}_{s \in \SessionName}
                \redto
            \hat{\typingContextTwo}',
            \set{\typingContextTwo'_s}_{s \in \SessionName}
            \typingContextCat
            \hat{\typingContextThree}',
            \set{\typingContextThree'_s}_{s \in \SessionName}
          $, and
    \item for all $s \in \SessionName$, it holds that
       there is $
    (\vec{q}, \xi) \in \reach(\CSMabb{A})
   $ such that \\
    $\typingContextTwo_s =
        \set{s[\procA] \hasType q}_{\procA \in \ProcsOf{\CSMabb{A}}}
        $
    and
    $\typingContextThree_s =
        \set{s[\procA][\procB] \hasQueueType \xi(\procA,\procB)}_{(\procA,\procB) \in \channelsOf{\CSMabb{A}}}
   $
\end{itemize}
Then, for all $s \in \SessionName$, it holds that
       there is $
    (\pvec{q}', \xi') \in \reach(\CSMabb{A})
   $ such that \\
    $\typingContextTwo'_s =
        \set{s[\procA] \hasType q'}_{\procA \in \ProcsOf{\CSMabb{A}}}
        $
    and
    $\typingContextThree'_s =
        \set{s[\procA][\procB] \hasQueueType \xi'(\procA,\procB)}_{(\procA,\procB) \in \channelsOf{\CSMabb{A}}}
   $.
\end{lemma}
\begin{proof}
We do inversion on
    $
        \hat{\typingContextTwo},
        \set{\typingContextTwo_s}_{s \in \SessionName}
        \typingContextCat
        \hat{\typingContextThree},
        \set{\typingContextThree_s}_{s \in \SessionName}
            \redto
        \hat{\typingContextTwo}',
        \set{\typingContextTwo'_s}_{s \in \SessionName}
        \typingContextCat
        \hat{\typingContextThree}',
        \set{\typingContextThree'_s}_{s \in \SessionName}
    $,
yielding two cases.

First, we have
\vspace{-2ex}
{ \small
\begin{mathpar}

\inferrule*[right=\typingReductionIntCh]{
    q
        \xrightarrow{\snd{\procA}{\procB}{\labelAndVar{l}{L}}}
    q'
}{
    s[\procA] \hasType q,
        \bar{\typingContextTwo}
        \typingContextCat
        s[\procA][\procB] \hasQueueType \queueType,
        \bar{\typingContextThree}
        \redto
    s[\procA] \hasType q',
        \bar{\typingContextTwo}
        \typingContextCat
        s[\procA][\procB] \hasQueueType \queueType \cat \labelAndVar{l}{L},
        \bar{\typingContextThree}
}

\end{mathpar}
}

for some $s$, $\procA$, and $\procB$.
For every $s' \neq s$, the claim trivially holds.
For $s$,
the changes to $s[\procA]$ and $s[\procA][\procB]$ mimic the semantics of the CSM while the premise
$
    q
        \xrightarrow{\snd{\procA}{\procB}{\labelAndVar{l}{L}}}
    q'
$
ensures that such a transition is possible.

For the second case where we have

\vspace{-2ex}
{ \small
\begin{mathpar}

    \inferrule*[right=\typingReductionExtCh]{
        q
            \xrightarrow{\rcv{\procB}{\procA}{\labelAndMsg{l}{L}}}
        q'
    }{
        s[\procA] \hasType q,
            \bar{\typingContextTwo}
            \typingContextCat
            s[\procB][\procA] \hasQueueType \labelAndVar{l}{L} \cat \queueType,
            \bar{\typingContextThree}
            \redto
        s[\procA] \hasType q',
            \bar{\typingContextTwo}
            \typingContextCat
            s[\procB][\procA] \hasQueueType \queueType,
            \bar{\typingContextThree}
    }

\end{mathpar}
}

\noindent
the reasoning is analogous but mimics the receive case of the CSM semantics.
\proofEndSymbol
\end{proof}

With the substitution lemma, we prove that substituting a variable by a value with the same type preserves typability in our type system.

\begin{lemma}[Substitution Lemma]
\label{lm:substitution-lemma}
Let $R$ be a well-annotated runtime configuration.
For all $L$,
if it holds that
      $\typingContextOne
           \typingContextCat
           \typingContextTwo,
           x \hasType L
           \typingContextCat
           \typingContextThree
       \types
       R$,
then
      \[
        \typingContextOne
           \typingContextCat
           \typingContextTwo,
           v \hasType L
           \typingContextCat
           \typingContextThree
       \types
       R[v / x]
       \enspace .
      \]
\end{lemma}
\begin{proof}
We do an induction on the depth of the typing derivation and do a case analysis on the last applied rule of the derivation.

For the induction base, we consider both rules with depth $0$ and show that there is no $R'$ such that $R \redto R'$.
For both \runtimeTypingZero and \runtimeTypingEmptyQueue, the second typing context is empty, which contradicts our assumption that $x \hasType L$ or $v \hasType L$.

For the induction step, the induction hypothesis yields that the claim holds for typing derivations of smaller depth.

\begin{itemize}
 \item \runtimeTypingProcName: \\
    We have that

    \vspace{-2ex}
    { \small
    \begin{mathpar}
        \inferrule*[right=\runtimeTypingProcName]{
            \typingContextOne(\pn{Q}) = L_1,\dots,L_n \\
        }{
            \typingContextOne
                \typingContextCat
                c_1 \hasType L_1,
                \ldots,
                c_{i-1} \hasType L_{i-1},
                x \hasType L_{i},
                c_{i+1} \hasType L_{i+1},
                \ldots
                c_n \hasType L_n
                \typingContextCat
                \emptyset
                \types
                \pn{Q}[\vec{c}]
        }
    \end{mathpar}
    }

    It is straightforward that we need to show precisely the same premise for the desired typing derivation:
    \[
            \typingContextOne
                \typingContextCat
                c_1 \hasType L_1,
                \ldots,
                c_{i-1} \hasType L_{i-1},
                v \hasType L_{i},
                c_{i+1} \hasType L_{i+1},
                \ldots,
                c_n \hasType L_n
                \typingContextCat
                \emptyset
                \types
                (\pn{Q}[\vec{c}])[v / x]
            \enspace .
    \]

 \item \runtimeTypingEnd: We have two cases. \\
    First, we have

    \vspace{-2ex}
    { \small
    \begin{mathpar}
        \inferrule*[right=\runtimeTypingEnd]{
            \typingContextOne
            \typingContextCat
            \typingContextTwo
            \typingContextCat
            \typingContextThree
            \types R \\
\EndState(q)
}{
            \typingContextOne
            \typingContextCat
            x \hasType q,
            \typingContextTwo
            \typingContextCat
            \typingContextThree
                \types
            R
        }
    \end{mathpar}
    }

    We show that

    \vspace{-2ex}
    { \small
    \begin{mathpar}
        \inferrule*[right=\runtimeTypingEnd]{
            \typingContextOne
            \typingContextCat
            \typingContextTwo
            \typingContextCat
            \typingContextThree
            \types R[v/x] \\
\EndState(q)
}{
            \typingContextOne
            \typingContextCat
            v \hasType q,
            \typingContextTwo
            \typingContextCat
            \typingContextThree
                \types
            R[v/x]
        }
    \end{mathpar}
    }

    For the first typing derivation, we have $x \hasType q, \typingContextTwo$ as typing context.
    By the fact that $\typingContextTwo$ is a typing context (and no syntactic typing context),
    $x$ does not occur in $\typingContextTwo$.
    By inversion, we have
    $
            \typingContextOne
            \typingContextCat
            \typingContextTwo
            \typingContextCat
            \typingContextThree
            \types R
    $.
    Thus, $x$ cannot occur in $R$.
    If it did, $R$ could only be typed with a typing context with $x$, which does not occur in~$\typingContextTwo$, given by contraposition of
    \cref{lm:every-variable-in-runtime-configuration-needs-typing}.
    Hence $R = R[v/x]$ and, thus, both premises coincide.

    Second, we have

    \vspace{-2ex}
    { \small
    \begin{mathpar}
        \inferrule*[right=\runtimeTypingEnd]{
            \typingContextOne
            \typingContextCat
            \typingContextTwo,
            x \hasType L
            \typingContextCat
            \typingContextThree
            \types R \\
\EndState(q)
}{
            \typingContextOne
            \typingContextCat
            c \hasType q,
            \typingContextTwo,
            x \hasType L
            \typingContextCat
            \typingContextThree
                \types
            R
        }
    \end{mathpar}
    }

    We show that

    \vspace{-2ex}
    { \small
    \begin{mathpar}
        \inferrule*[right=\runtimeTypingEnd]{
            \typingContextOne
            \typingContextCat
            \typingContextTwo,
            v \hasType L
            \typingContextCat
            \typingContextThree
            \types R[v/x] \\
\EndState(q)
}{
            \typingContextOne
            \typingContextCat
            c \hasType q,
            \typingContextTwo,
            v \hasType L
            \typingContextCat
            \typingContextThree
                \types
            R[v/x]
        }
    \end{mathpar}
    }

    By inversion of the first typing derivation, we know that both premises hold.
    The second premise is the same for both derivations.
    For the first premise, the induction hypothesis applies.

 \item \runtimeTypingIntCh: \\
    Here, we do a case analysis if $x = c$, $x = c_k$ for some $k \in I$ or neither of both.

    For the last case, we can apply inversion and the induction hypothesis applies to all cases of the right premise.

    For $x = c$, the second premise follows from inversion and the induction hypothesis, instantiated with $L = q_i$.

    We consider the case for $x = c_k$ in more detail.

    We have

    \vspace{-2ex}
    { \small
    \begin{mathpar}
    \inferrule*[right=\runtimeTypingIntCh]{
        \delta(q) \sups
        \set{(\snd{\procA}{\procB_i}{\labelAndType{l_i}{L_i}}, q_i) \mid i \in I}
        \\
\meta{\forall i \in I \setminus \set{k} \st}
        \typingContextOne \typingContextCat
            \typingContextTwo, c \hasType q_i,
            x \hasType L,
            \set{c_j \hasType L_j}_{j \in I\setminus \set{i,k}}
            \typingContextCat
            \typingContextThree
            \types P_i
        \\
        \typingContextOne \typingContextCat
            \typingContextTwo, c \hasType q_i,
            \set{c_j \hasType L_j}_{j \in I\setminus \set{k}}
            \typingContextCat
            \typingContextThree
            \types P_k
    }{
        \typingContextOne
            \typingContextCat
            \typingContextTwo,
            c \hasType q,
            x \hasType L,
            \set{c_i \hasType L_i}_{i \in I \setminus \set{k}}
            \typingContextCat
            \typingContextThree
        \types
        \IntCh_{i \in I} c[\procB_i] ! \labelAndMsg{l_i}{c_i} \seq P_i
    }
    \end{mathpar}
    }

    By inversion, we obtain all premises.

    We show that

    \vspace{-2ex}
    { \small
    \begin{mathpar}
    \inferrule*[right=\runtimeTypingIntCh]{
        \delta(q) \sups
        \set{(\snd{\procA}{\procB_i}{\labelAndType{l_i}{L_i}}, q_i) \mid i \in I}
        \\
\meta{\forall i \in I \setminus \set{k} \st}
        \typingContextOne \typingContextCat
            \typingContextTwo, c \hasType q_i,
            v \hasType L,
            \set{c_j \hasType L_j}_{j \in I\setminus \set{i,k}}
            \typingContextCat
            \typingContextThree
            \types P_i[v/x]
        \\
        \typingContextOne \typingContextCat
            \typingContextTwo, c \hasType q_i,
            \set{c_j \hasType L_j}_{j \in I\setminus \set{k}}
            \typingContextCat
            \typingContextThree
            \types P_k[v/x]
    }{
        \typingContextOne
            \typingContextCat
            \typingContextTwo,
            c \hasType q,
            v \hasType L,
            \set{c_i \hasType L_i}_{i \in I \setminus \set{k}}
            \typingContextCat
            \typingContextThree
        \types
        (\IntCh_{i \in I} c[\procB_i] ! \labelAndMsg{l_i}{c_i} \seq P_i)[v/x]
    }
    \end{mathpar}
    }

    The first premise is the same.
    The second premise, for every $i \in I \setminus \set{k}$, follows by the induction hypothesis.
    For the third premise, we claim that $x$ cannot occur in $P_k$.
    In the conclusion of the first typing derivation, we have the typing context
    $
            \typingContextTwo,
            c \hasType q,
            x \hasType L,
            \set{c_i \hasType L_i}_{i \in I \setminus \set{k}}
    $.
    Thus, by assumption that each element has at most one type in a typing context, we know that $x$ cannot occur in
    $
            \typingContextTwo,
            c \hasType q,
            \set{c_i \hasType L_i}_{i \in I \setminus \set{k}}
    $.
    With \cref{lm:every-variable-in-runtime-configuration-needs-typing}, $x$ cannot occur in $P_k$.
    Thus, $P_k[v/x] = P_k$ and the third premise in both typing derivations coincide, concluding this~case.

 \item \runtimeTypingExtCh: \\
    We have

    \vspace{-2ex}
    { \small
    \begin{mathpar}
    \inferrule*[right=\runtimeTypingExtCh]{
        \delta(q) =
        \set{(\rcv{\procB_i}{\procA}{\labelAndType{l_i}{L_i}}, q_i) \mid i \in I} \\
        \meta{\forall i \in I \st}
        \typingContextOne
            \typingContextCat
            \typingContextTwo,
            y_i \hasType L_i,
            c \hasType q_i
            \typingContextCat
            \typingContextThree
        \types
        P_i
    }{
        \typingContextOne
            \typingContextCat
            \typingContextTwo,
            c \hasType q
            \typingContextCat
            \typingContextThree
        \types
        \ExtCh_{i \in I} c[\procB_i] ? \labelAndVar{l_i}{y_i} \seq P_i
    }
    \end{mathpar}
    }

    We do a case analysis if $x = c$ or not.

    If not, the claim follows by inversion and induction hypothesis for the second premise.

    If $x = c$, there is $x \hasType q_i$ in the second premise to type $P_i[v/x]$.
    The existence of such a typing derivation follows from inversion and induction hypothesis when instantiated with $L = q_i$.

 \item \runtimeTypingParallel: \\
    There are two symmetric cases.
    We only consider one of both.
    For this, we~have

    \vspace{-2ex}
    { \small
    \begin{mathpar}
        \inferrule*[right=\runtimeTypingParallel]{
            \typingContextOne
                \typingContextCat
                \typingContextTwo_1,
                x \hasType L
                \typingContextCat
                \typingContextThree_1
                \types
                R_1 \\
\typingContextOne
                \typingContextCat
                \typingContextTwo_2
                \typingContextCat
                \typingContextThree_2
                \types
                R_2 \\
        }{
            \typingContextOne
                \typingContextCat
                \typingContextTwo_1,
                x \hasType L,
                \typingContextTwo_2
                \typingContextCat
                \typingContextThree_1,
                \typingContextThree_2
                \types
                R_1 \parallel R_2
        }
    \end{mathpar}
    }

    We show that there is a typing derivation for
    \[
            \typingContextOne
                \typingContextCat
                \typingContextTwo_1,
                v \hasType L,
                \typingContextTwo_2
                \typingContextCat
                \typingContextThree_1,
                \typingContextThree_2
                \types
                (R_1 \parallel R_2)[v / x]
    \]
    By our assumption that typing contexts have at most one type per element, $\typingContextTwo_1$ and $\typingContextTwo_2$ cannot share any names.
    By inversion,
    we have
           $ \typingContextOne
                \typingContextCat
                \typingContextTwo_2
                \typingContextCat
                \typingContextThree_2
                \types
                R_2$.
    Thus, by \cref{lm:every-variable-in-runtime-configuration-needs-typing}, $x$ cannot occur in $R_2$.
    Hence, $(R_1 \parallel R_2)[v / x]
                =
                R_1[v/x] \parallel R_2$.
    We claim the following typing derivation exists:

    \vspace{-2ex}
    { \small
    \begin{mathpar}
        \inferrule*[right=\runtimeTypingParallel]{
            \typingContextOne
                \typingContextCat
                \typingContextTwo_1,
                v \hasType L,
                \typingContextCat
                \typingContextThree_1
                \types
                R_1 \\
\typingContextOne
                \typingContextCat
                \typingContextTwo_2
                \typingContextCat
                \typingContextThree_2
                \types
                R_2 \\
        }{
            \typingContextOne
                \typingContextCat
                \typingContextTwo_1,
                v \hasType L,
                \typingContextTwo_2
                \typingContextCat
                \typingContextThree_1,
                \typingContextThree_2
                \types
                R_1[v/x] \parallel R_2
        }
    \end{mathpar}
    }

    The second premise is the same as in the original typing derivation.
    The first premise can be obtained by inversion on the original typing derivation and applying the induction hypothesis.

 \item \runtimeTypingRestr: \\
    We have a typing derivation for
    \[
        \typingContextOne
            \typingContextCat
            \typingContextTwo,
            x \hasType L,
            \typingContextTwo_s
            \typingContextCat
            \typingContextThree,
            \typingContextThree_s
            \types R
        \enspace .
    \]
    This case follows easily from inversion and applying the induction hypothesis to obtain a typing derivation for
    \[
        \typingContextOne
            \typingContextCat
            \typingContextTwo,
            v \hasType L,
            \typingContextTwo_s
            \typingContextCat
            \typingContextThree,
            \typingContextThree_s
            \types R[v/x]
        \enspace .
    \]

 \item \runtimeTypingQueue: \\
    We have a typing derivation for

    \vspace{-2ex}
    { \small
    \begin{mathpar}
        \inferrule*[right=\runtimeTypingQueue]{
            \typingContextOne
                \typingContextCat
                \typingContextTwo
                \typingContextCat
                \typingContextThree,
                s[\procA][\procB] \hasQueueType \queueType
                \types
                \queueProc{s}{\queuecontent[
                    (\procA, \procB) \mapsto \vec{m}
                ]}
        }{
            \typingContextOne
                \typingContextCat
                \typingContextTwo,
                v' \hasType L
                \typingContextCat
                \typingContextThree,
                s[\procA][\procB] \hasQueueType \labelAndVar{l}{L} \cat \queueType
                \types
                \queueProc{s}{\queuecontent[
                    (\procA, \procB) \mapsto \labelAndMsg{l}{v'} \cat \vec{m}
                ]}
        }
    \end{mathpar}
    }

    We do a case analysis if $x = v'$ or not.
    If so, the same typing derivation can simply be re-used as $v' \hasType L$ also disappears in the original typing derivation.
    If not, the claim follows by inversion and application of the induction hypothesis.
\end{itemize}
This concludes the proof of the substitution lemma.
\proofEndSymbol
\end{proof}

Now, we turn to the main result about our type system: \emph{subject reduction}.
In short, if there is a runtime configuration with a typing derivation that can take a step, then the typing contexts can also take a step and can be used to type the new runtime configuration.

Note that \cref{th:subj-red-closed} in the main text follows from \cref{thm:subject-reduction}.

\begin{theorem}[Subject Reduction]
\label{thm:subject-reduction}
Let $R$ be a well-annotated runtime configuration with a set of active sessions~$\SessionName$.
If
\begin{enumerate}[label=\textnormal{(\arabic*)}]
 \item \label{lm:subject-reduction-assumption-1} $\types \Defs \hasType \typingContextOne$,
 \item $\typingContextOne
           \typingContextCat
           \typingContextTwo
           \typingContextCat
           \typingContextThree
       \types
       R$ with
        $\typingContextTwo = \hat{\typingContextTwo}, \set{\typingContextTwo_s}_{s \in \SessionName}$
        and
        $\typingContextThree = \hat{\typingContextThree}, \set{\typingContextThree_s}_{s \in \SessionName}$,
 \item for all $s \in \SessionName$, it holds that
       there is $
    (\vec{q}, \xi) \in \reach(\CSMabb{A})
   $ such that \\
    $\typingContextTwo_s =
        \set{s[\procA] \hasType q}_{\procA \in \ProcsOf{\CSMabb{A}}}
        $
    and
    $\typingContextThree_s =
        \set{s[\procA][\procB] \hasQueueType \xi(\procA,\procB)}_{(\procA,\procB) \in \channelsOf{\CSMabb{A}}}
   $
 \item \label{lm:subject-reduction-assumption-4} $R \redto R'$,
\end{enumerate}
then there exist
$\typingContextTwo'$ and
$\typingContextThree'$
with
$\typingContextTwo \typingContextCat \typingContextThree
\redto
\typingContextTwo' \typingContextCat \typingContextThree'$
such that
$\typingContextOne
    \typingContextCat
    \typingContextTwo'
    \typingContextCat
    \typingContextThree'
\types
R'$.
\end{theorem}
\begin{proof}
We do an induction on the depth of the typing derivation
    $\typingContextOne
        \typingContextCat
        \typingContextTwo
        \typingContextCat
        \typingContextThree
    \types
    R$
and do a case analysis on the last applied rule of the derivation.

For the induction base, we consider both rules with depth $0$ and show that there is no $R'$ such that $R \redto R'$.
\begin{itemize}[leftmargin=4ex]
 \item \runtimeTypingZero: \\
        It is trivial that $\zero$ cannot reduce.
 \item \runtimeTypingEmptyQueue: \\
        In this case $R = \queueProc{s}{\emptystring}$ for which none of the reduction rules apply.
\end{itemize}

For the induction step, the induction hypothesis yields that the claim holds for typing derivations of smaller depth.
We do a case analysis on the typing rule that was applied last.

\begin{itemize}[leftmargin=4ex]
 \item \runtimeTypingProcName: \\
    We have that

    \vspace{-3ex}
    { \small
    \begin{mathpar}
        \inferrule*[right=\runtimeTypingProcName]{
            \typingContextOne(\pn{Q}) = L_1,\dots,L_n \\
        }{
            \typingContextOne
                \typingContextCat
                c_1 \hasType L_1,
                \dots,
                c_n \hasType L_n
                \typingContextCat
                \emptyset
                \types
                \pn{Q}[\vec{c}]
        }
    \end{mathpar}
    }

    Thus, we know that $R = \pn{Q}[\vec{c}]$.
    However, none of the reduction rules apply, contradicting~\ref{lm:subject-reduction-assumption-4}.

 \item \runtimeTypingEnd: \\
    We have that

    \vspace{-3ex}
    { \small
    \begin{mathpar}
        \inferrule*[right=\runtimeTypingEnd]{
            \typingContextOne
            \typingContextCat
            \typingContextTwo
            \typingContextCat
            \typingContextThree
            \types R \\
\EndState(q)
}{
            \typingContextOne
            \typingContextCat
            c \hasType q,
            \typingContextTwo
            \typingContextCat
            \typingContextThree
                \types
            R
        }
    \end{mathpar}
    }

    We have to show that

    \vspace{-2ex}
    { \small
    \begin{mathpar}
        \inferrule*[right=\runtimeTypingEnd]{
            \typingContextOne
            \typingContextCat
            \typingContextTwo'
            \typingContextCat
            \typingContextThree'
            \types R' \\
\EndState(q)
}{
            \typingContextOne
            \typingContextCat
            c \hasType q,
            \typingContextTwo'
            \typingContextCat
            \typingContextThree'
                \types
            R'
        }
    \end{mathpar}
    }

    The second premise
    $\EndState(q)$
    is the same for both.
    The first premise follows by the induction hypothesis, which also gives 
    $
    \typingContextTwo \typingContextCat \typingContextThree \redto
    \typingContextTwo' \typingContextCat \typingContextThree'
    $, concluding this case.

 \item \runtimeTypingIntCh: \\
    Inversion on the typing derivation yields that
    $R = \IntCh_{i \in I} s[\procA][\procB_i] ! \labelAndMsg{l_i}{v_i} \seq P_i$.
    None of the reduction rules apply and, thus, there is no $R'$ with $R \redto R'$, contradicting~\ref{lm:subject-reduction-assumption-4}.

 \item \runtimeTypingExtCh: \\
    Inversion on the typing derivation yields that
    $R = \ExtCh_{i \in I} s[\procA][\procB_i] ? \labelAndVar{l_i}{L_i} \seq P_i$.
    None of the reduction rules apply and, thus, there is no $R'$ with $R \redto R'$, contradicting~\ref{lm:subject-reduction-assumption-4}.

 \item \runtimeTypingParallel: \\
    We do inversion on the reduction~\ref{lm:subject-reduction-assumption-4}.
    \begin{itemize}[leftmargin=2ex]
     \item \runtimeReductionProcName: \\
        We have

        \vspace{-2ex}
        { \small
        \begin{mathpar}
            \inferrule*[right=\runtimeTypingParallel]{
                \inferrule*[left=\runtimeTypingProcName]{
                    \typingContextOne(\pn{Q}) = \vec{L}
                }{
                    \typingContextOne
                        \typingContextCat
                        \vec{c} \hasType \vec{L}
                        \typingContextCat
                        \emptyset
                        \types
                        \pn{Q}[\vec{c}] \\
                }
\typingContextOne
                    \typingContextCat
                    \typingContextTwo_2
                    \typingContextCat
                    \typingContextThree_2
                    \types
                    R_2 \\
            }{
                \typingContextOne
                    \typingContextCat
                    \vec{c} \hasType \vec{L},
                    \typingContextTwo_2
                    \typingContextCat
                    \typingContextThree_2
                    \types
                    \pn{Q}[\vec{c}] \parallel R_2
            }
        \end{mathpar}
        }

        By inversion on~\ref{lm:subject-reduction-assumption-4}, we have

        \vspace{-2ex}
        { \small
        \begin{mathpar}
            \inferrule*[right=\runtimeReductionProcName]{
                \Defs(\pn{Q}, \vec{c}) \parallel R_2
                \redto
                R'
            }{
                \pn{Q}[\vec{c}] \parallel R_2
                \redto
                R'
            }
        \end{mathpar}
        }

        By assumption that $\pn{Q}$ is defined in $\Defs$ and by definition of $\Defs$, we have that
        \[
            \Defs = \Defs_1; (\pn{Q}[\vec{x}] = P); \Defs_2
        \]
        for some $\Defs_1$ and $\Defs_2$.

        We claim there is a typing derivation for
        \[
                \typingContextOne
                    \typingContextCat
                    \vec{c} \hasType \vec{L}
                    \types
                    P[\vec{c} / \vec{x}]
        \]
        \ref{lm:subject-reduction-assumption-1} states that $\types \Defs \hasType \typingContextOne$.
        Inversion on
        \ref{lm:subject-reduction-assumption-1} for $\card{\Defs_1}$ times yields
        \[
            \typingContextOne
            \typingContextCat
            \vec{x} \hasType \vec{L}
            \types P
            \enspace .
        \]
        We obtain a typing derivation after $\card{\vec{x}}$ applications of the substitution lemma (\cref{lm:substitution-lemma}). 

        We do a case analysis on the structure of $P$ and simultaneously if $R' = \err$ for second case.

        \begin{itemize}[leftmargin=2ex]
         \item $P =
                \IntCh_{i \in I} x[\procB_i] ! \labelAndMsg{l_i}{x_i} \seq P_i$: \\
                Let us rewrite the typing context:
                $\vec{c} \hasType \vec{L} =
                    (x \hasType q,
                    \set{x_j \hasType L_j}_{j \in I},
                    \typingContextTwo_1)
                    [\vec{c}/\vec{x}]
                $.
                Without loss of generality, let $x[\vec{c}/\vec{x}] = s[\procA]$.
                Then $R' =
                    \procToRuntime{P_k[\vec{c}/\vec{x}]}
                    \parallel
                    \queueProc{s}{\queuecontent[(\procA, \procB_k) \mapsto \vec{m} \cat \labelAndMsg{l_k}{v_k}]}$
                    for some $k \in I$.
                We show there is a typing derivation for any $k \in I$:

                \vspace{-2ex}
                { \scriptsize
                \begin{mathpar}
                \inferrule*[right=\runtimeTypingParallel]{
                    \typingContextOne
                    \typingContextCat
                    (x \hasType q_k,
                        \set{x_j \hasType L_j}_{j \in I \setminus \set{k}},
                        \typingContextTwo_1)
                        [\vec{c}/\vec{x}]
                    \typingContextCat
                    \emptyset
                        \types
                    \procToRuntime{P_k[\vec{c}/\vec{x}]}
                    \\
                    \typingContextOne
                    \typingContextCat
                    (c_k \hasType L_k)[\vec{c}/\vec{x}],
                    \typingContextTwo_2
                    \typingContextCat
                    \hat{\typingContextThree}_2,
                    s[\procA][\procB_k] \hasQueueType \queueType \cat \labelAndType{l_k}{L_k[\vec{c}/\vec{x}]}
                        \types
                    \queueProc{s}{\queuecontent[(\procA, \procB_k) \mapsto \vec{m} \cat \labelAndMsg{l_k}{v_k}]}
                }{
                    \typingContextOne
                    \typingContextCat
                    (x \hasType q_k,
                    \set{x_j \hasType L_j}_{j \in I},
                        \typingContextTwo_1)
                        [\vec{c}/\vec{x}],
                    \typingContextTwo_2
                    \typingContextCat
                    \hat{\typingContextThree}_2,
                    s[\procA][\procB_k] \hasQueueType \queueType \cat \labelAndType{l_k}{L_k}
                        \types
                    R'
}
                \end{mathpar}
                }

                By inversion on
                \[
                    \typingContextOne
                    \typingContextCat
                    (x \hasType q,
                    \set{x_j \hasType L_j}_{j \in I},
                    \typingContextTwo_1)
                    [\vec{c}/\vec{x}]
                    \types
                    (\IntCh_{i \in I} x[\procB_i] ! \labelAndMsg{l_i}{x_i} \seq P_i)[\vec{c}/\vec{x}]
                    \enspace ,
                \]
                we get
                $
                    \typingContextOne
                    \typingContextCat
                    (x \hasType q_i,
                    \set{x_j \hasType L_j}_{j \in I \setminus \set{i}},
                    \typingContextTwo_1)
                    [\vec{c}/\vec{x}]
                        \types
                    P_i[\vec{c}/\vec{x}]
                $
                for every $i \in I$.
                Instantiating $i = k$ and
                applying \cref{lm:typing-proc-to-runtime-configuration} yields the desired premise:
                \[
                    \typingContextOne
                    \typingContextCat
                    (x \hasType q_i, \set{x_j \hasType L_j}_{j \in I \setminus \set{i}},
                    \typingContextTwo_1)
                    [\vec{c}/\vec{x}]
                    \typingContextCat
                    \emptyset
                        \types
                    \procToRuntime{P_i[\vec{c}/\vec{x}]}
                \]

                We show there is a typing derivation

                \vspace{-2ex}
                { \tiny
                \begin{mathpar}
                    \inferrule*[right=\runtimeTypingQueue]{
                        \typingContextOne
                        \typingContextCat
                        \typingContextTwo_2
                        \typingContextCat
                        \hat{\typingContextThree}_2,
                        s[\procA][\procB_k] \hasQueueType \queueType
                            \types
                        \queueProc{s}{\queuecontent[(\procA, \procB_k) \mapsto \vec{m}]}
                    }{
                        \typingContextOne
                        \typingContextCat
                        (c_k \hasType L_k)[\vec{c}/\vec{x}],
                        \typingContextTwo_2
                        \typingContextCat
                        \hat{\typingContextThree}_2,
                        s[\procA][\procB_k] \hasQueueType \queueType \cat \labelAndType{l_k}{L_k[\vec{c}/\vec{x}]}
                            \types
                        \queueProc{s}{\queuecontent[(\procA, \procB_k) \mapsto \vec{m} \cat \labelAndMsg{l_k}{v_k}]}
                    }
                \end{mathpar}
                }

                By inversion on $R \redto R'$, we have that
                $R_2 =
                    \queueProc{s}{\queuecontent[(\procA, \procB_k) \mapsto \vec{m}]}$.
                By inversion on
                $
                    \typingContextOne
                        \typingContextCat
                        \typingContextTwo_2
                        \typingContextCat
                        \typingContextThree_2
                        \types
                        R_2
                $,
                we obtain

                \vspace{-2ex}
                { \small
                \begin{mathpar}
                    \typingContextOne
                    \typingContextCat
                    \typingContextTwo_2
                    \typingContextCat
                    \hat{\typingContextThree}_2,
                    s[\procA][\procB_k] \hasQueueType \queueType
                        \types
                    \queueProc{s}{\queuecontent[(\procA, \procB_k) \mapsto \vec{m}]}
                \end{mathpar}
                }

                which is the desired premise.

                It is straightforward that there is a typing context reduction for the corresponding typing contexts, using \typingReductionIntCh.

         \item $P =
                \ExtCh_{i \in I} x[\procB_i] ? \labelAndVar{l_i}{y_i} \seq P_i$ and $R' \neq \err$: \\
                Let us rewrite the typing context:
                $\vec{c} \hasType \vec{L} =
                    (x \hasType q,
                    \typingContextTwo_1)
                    [\vec{c}/\vec{x}]
                $.
                Without loss of generality, let $x[\vec{c}/\vec{x}] = s[\procA]$.
                Then
                $R' = \procToRuntime{P_k[v_k/y_k]} \parallel \queueProc{s}{[(\procB_k, \procA) \mapsto \vec{m}]}$
                and
                $R_2 = \queueProc{s}{[(\procB_k, \procA) \mapsto \labelAndMsg{l_k}{v_k} \cat \vec{m}]}$.
                We show there is a typing derivation for any $k \in I$:

                \vspace{-2ex}
                { \tiny
                \begin{mathpar}
                \inferrule*[right=\runtimeTypingParallel]{
                    \typingContextOne
                    \typingContextCat
                    v_k \hasType L_k,
                    (x \hasType q_k,
                        \typingContextTwo_1)
                        [\vec{c}/\vec{x}]
                    \typingContextCat
                    \emptyset
                        \types
                    \procToRuntime{P_k[\vec{c}/\vec{x}][v_k/y_k]}
                    \\
                    \typingContextOne
                    \typingContextCat
                    \typingContextTwo_2
                    \typingContextCat
                    \hat{\typingContextThree}_2,
                    s[\procA][\procB_k] \hasQueueType \queueType
                        \types
                    \queueProc{s}{\queuecontent[(\procA, \procB_k) \mapsto \vec{m}]}
                }{
                    \typingContextOne
                    \typingContextCat
                    (x \hasType q_k,
                        \typingContextTwo_1)
                        [\vec{c}/\vec{x}],
                    \typingContextTwo_2
                    \typingContextCat
                    \hat{\typingContextThree}_2,
                    s[\procA][\procB_k] \hasQueueType \queueType
                        \types
                    \procToRuntime{P_k[\vec{c}/\vec{x}][v_k/y_k]}
                    \parallel
                    \queueProc{s}{\queuecontent[(\procA, \procB_k) \mapsto \vec{m}]}
                }
                \end{mathpar}
                }

                By inversion on
                $
                    \typingContextOne
                    \typingContextCat
                    (x \hasType q,
                        \typingContextTwo_1)
                        [\vec{c}/\vec{x}]
                        \types
                    (\ExtCh_{i \in I} x[\procB_i] ? \labelAndVar{l_i}{y_i} \seq P_i)[\vec{c}/\vec{x}]
                $,
                we get
                \[
                    \typingContextOne
                    \typingContextCat
                    (x \hasType q_i,
                        y_i \hasType L_i,
                        \typingContextTwo_1)
                        [\vec{c}/\vec{x}]
                        \types
                    P_i[\vec{c}/\vec{x}]
                \]
                for every $i \in I$.
                Instantiating $i = k$ and applying \cref{lm:typing-proc-to-runtime-configuration} yields the desired premise:
                \[
                    \typingContextOne
                    \typingContextCat
                    v_k \hasType L_k,
                    (x \hasType q_k,
                        \typingContextTwo_1)
                        [\vec{c}/\vec{x}]
                    \typingContextCat
                    \emptyset
                        \types
                    \procToRuntime{P_k[\vec{c}/\vec{x}][v_k/y_k]}
                \]

                The second premise is obtained by inversion on
                \[
                    \typingContextOne
                        \typingContextCat
                        v_k \hasType L_k,
                        \typingContextTwo_2
                        \typingContextCat
                        \hat{\typingContextThree}_2,
                        s[\procA][\procA_k] \hasQueueType \labelAndVar{l_k}{L_k} \cat \queueType
                        \types
                        \queueProc{s}{[(\procB_k, \procA) \mapsto \labelAndMsg{l_k}{v_k} \cat \vec{m}]}
                        \enspace .
                \]

                It is straightforward that there is a typing context reduction for the corresponding typing contexts, using \typingReductionExtCh.

         \item $P =
                \ExtCh_{i \in I} x[\procB_i] ? \labelAndVar{l_i}{y_i} \seq P_i$ and $R' = \err$: \\
                Let us rewrite the typing context:
                $\vec{c} \hasType \vec{L} =
                    (x \hasType q,
                    \typingContextTwo_1)
                    [\vec{c}/\vec{x}]
                $.
                Without loss of generality, let $x[\vec{c}/\vec{x}] = s[\procA]$.
                Then, by inversion on $R \redto R'$, we have
                $R_2 = \queueProc{s}{\queuecontent}$ and
                for every $i \in I$,
                $\queuecontent(\procB_i, \procA) = \labelAndMsg{l}{\_} \cat \vec{m}$ and $l_i \neq l$.
                We claim that there is no typing derivation
                \[
                    \typingContextOne
                    \typingContextCat
                    (x \hasType q,
                        \typingContextTwo_1)
                        [\vec{c}/\vec{x}],
                    \typingContextTwo_2
                    \typingContextCat
                    \typingContextThree_2
                        \types
                    (\ExtCh_{i \in I} x[\procB_i] ? \labelAndVar{l_i}{y_i} \seq P_i)
                    [\vec{c}/\vec{x}]
                    \parallel
                    \queueProc{s}{\queuecontent}
                    \enspace .
                \]
                By inversion, such a typing derivation must have the following shape:

                \vspace{-2ex}
                { \small
                \begin{mathpar}
                 \inferrule*[right=\runtimeTypingParallel]{
                    \inferrule*[right=\runtimeTypingExtCh]{
                        \delta(q) =
                        \set{(\rcv{\procB_i}{\procA}{\labelAndType{l_i}{L_i}}, q_i) \mid i \in I}
                        \\
                        \meta{\forall i \in I \st}
                        \typingContextOne
                            \typingContextCat
                            (x \hasType q_i,
                                \typingContextTwo_1)
                                [\vec{c}/\vec{x}],
                            y_i \hasType L_i,
                            \typingContextCat
                            \emptyset
                        \types
                        P_i \\
                    }{
                        \typingContextOne
                        \typingContextCat
                        (x \hasType q,
                            \typingContextTwo_1)
                            [\vec{c}/\vec{x}]
                        \typingContextCat
                        \emptyset
                            \types
                        (\ExtCh_{i \in I} x[\procB_i] ? \labelAndVar{l_i}{y_i} \seq P_i)
                        [\vec{c}/\vec{x}]
                    }
                    \\
\typingContextOne
                        \typingContextCat
                        \typingContextTwo_2
                        \typingContextCat
                        \typingContextThree_2
                            \types
                        \queueProc{s}{\queuecontent}
}{
                    \typingContextOne
                    \typingContextCat
                    (x \hasType q,
                        \typingContextTwo_1)
                        [\vec{c}/\vec{x}],
                    \typingContextTwo_2
                    \typingContextCat
                    \typingContextThree_2
                        \types
                    (\ExtCh_{i \in I} x[\procB_i] ? \labelAndVar{l_i}{y_i} \seq P_i)
                    [\vec{c}/\vec{x}]
                    \parallel
                    \queueProc{s}{\queuecontent}
                 }
                \end{mathpar}
                }

Let us rewrite the typing and queue typing context:
                \begin{align*}
                    (x \hasType q, \typingContextOne)[\vec{c}/\vec{x}],
                    \typingContextTwo_2
                    & =
                    \hat{\typingContextTwo}, \typingContextTwo_s
                    \\
                    \typingContextThree_2
                    & =
                    \hat{\typingContextThree}, \typingContextThree_s
                \end{align*}

                By assumption, we know that
                    there is
                    $
                    (\vec{q}, \xi) \in \reach(\CSMabb{A})
                    $ such that
                    $\typingContextTwo_s =
                        \set{s[\procA] \hasType q}_{\procA \in \ProcsOf{\CSMabb{A}}}
                        $
                    and
                    $\typingContextThree_s =
                        \set{s[\procA][\procB] \hasQueueType \xi(\procA,\procB)}_{(\procA,\procB) \in \channelsOf{\CSMabb{A}}}
                    $.
                Recall the condition on the reduction semantics:
                $
                        \forall i \in I \st
                        \queuecontent(\procB_i, \procA) = \labelAndMsg{l}{\_} \cat \vec{m}
                        \text{ and }
                        l_i \neq l
                $.
                Thus, with \cref{lm:queue-types-model-queues}, it follows that, for all $i \in I$,
                $\xi(\procB_i,\procA) = \labelAndVar{l'_i}{\_} \cat \_$
                with $l'_i \neq l_i$.
                For the CSM $\CSMabb{A}$, this entails that $\procA$ expects to receive a message from a set of other participants, ranged over by $\procB_i$, but the first message in each channel does not match.
                This yields a contradiction to feasible eventual reception:
                there has been at least one send event to $\procA$ and there is no matching receive event yet;
                because $\procA$ will never proceed, no matching receive event can ever happen.
        \end{itemize}

     \item \runtimeReductionContext: \\
        For the context rule, two cases apply:
        $R \parallel \redContext$
        or
        $\redContext \parallel R$.
        Both cases can be proven analogous, which is why we only prove the first.
        We have that

        \vspace{-2ex}
        { \small
        \begin{mathpar}
            \inferrule*[right=\runtimeTypingParallel]{
                \typingContextOne
                    \typingContextCat
                    \typingContextTwo_1
                    \typingContextCat
                    \typingContextThree_1
                    \types
                    R_1 \\
\typingContextOne
                    \typingContextCat
                    \typingContextTwo_2
                    \typingContextCat
                    \typingContextThree_2
                    \types
                    R_2 \\
            }{
                \typingContextOne
                    \typingContextCat
                    \typingContextTwo_1,
                    \typingContextTwo_2
                    \typingContextCat
                    \typingContextThree_1,
                    \typingContextThree_2
                    \types
                    R_1 \parallel R_2
            }
        \end{mathpar}
        }

        and we want to show that
        
        \vspace{-2ex}
        { \small
        \begin{mathpar}
            \inferrule*[right=\runtimeTypingParallel]{
                \typingContextOne
                    \typingContextCat
                    \typingContextTwo'_1
                    \typingContextCat
                    \typingContextThree'_1
                    \types
                    R'_1 \\
\typingContextOne
                    \typingContextCat
                    \typingContextTwo_2
                    \typingContextCat
                    \typingContextThree_2
                    \types
                    R_2 \\
            }{
                \typingContextOne
                    \typingContextCat
                    \typingContextTwo'_1,
                    \typingContextTwo_2
                    \typingContextCat
                    \typingContextThree'_1,
                    \typingContextThree_2
                    \types
                    R'_1 \parallel R_2
            }
        \end{mathpar}
        }

        The second premise is trivially satisfied.
        The first premise follows from the induction hypothesis, which also yields that
        $
        \typingContextTwo_1 \typingContextCat \typingContextThree_1
        \redto
        \typingContextTwo'_1 \typingContextCat \typingContextThree'_1
        $.
        We can apply
        \cref{lm:typing-reduction-cong}
        to obtain
        $
        \typingContextTwo_1, \typingContextTwo_2
        \typingContextCat
        \typingContextThree_1, \typingContextThree_2
        \redto
        \typingContextTwo'_1, \typingContextTwo_2
        \typingContextCat
        \typingContextThree'_1, \typingContextThree_2
        $, concluding this case.

     \item \runtimeReductionOut: \\
        With three inversions on the typing derivation, we have a typing derivation with the following shape for
        $R = \IntCh_{i \in I} s[\procA][\procB_i] ! \labelAndMsg{l_i}{v_i} \seq P_i
             \parallel
             {\queuecontent[(\procA, \procB_k) \mapsto \vec{m}]} $:

        \vspace{-2ex}
        { \scriptsize
        \begin{mathpar}
        \inferrule*[right=\runtimeTypingParallel]{
            \inferrule*[right=\runtimeTypingIntCh]{
                \delta(q) \sups
                \set{(\snd{\procA}{\procB_i}{\labelAndType{l_i}{L_i}}, q_i) \mid i \in I}
                \\
\meta{\forall i \in I \st}
                \typingContextOne \typingContextCat
                    \hat{\typingContextTwo}_1 ,
                    s[\procA] \hasType q_i,
                    \set{v_j \hasType L_j}_{j \in I\setminus \set{i}}
                    \typingContextCat
                    \typingContextThree_1
                \types P_i
            }{
                \typingContextOne
                \typingContextCat
                \hat{\typingContextTwo}_1,
                s[\procA] \hasType q,
                \set{v_i \hasType L_i}_{i \in I}
                \typingContextCat
                \typingContextThree_1
                    \types
                \IntCh_{i \in I} s[\procA][\procB_i] ! \labelAndMsg{l_i}{v_i} \seq P_i
            }
            \\
\inferrule*[right=\runtimeTypingQueue]{
                \vdots
            }{
                \typingContextOne
                    \typingContextCat
                    \typingContextTwo_2
                    \typingContextCat
                    \hat{\typingContextThree}_2,
                    s[\procA][\procB_k] \hasQueueType \queueType
                    \types
                    \queueProc{s}{\queuecontent[
                        (\procA, \procB) \mapsto \vec{m}
                    ]}
            }
        }{
            \typingContextOne
                \typingContextCat
                \hat{\typingContextTwo}_1,
                s[\procA] \hasType q,
                \set{v_i \hasType L_i}_{i \in I},
                \typingContextTwo_2
                \typingContextCat
                \typingContextThree_1,
                \hat{\typingContextThree}_2,
                s[\procA][\procB_k] \hasQueueType \queueType
                \types
                R
}
        \end{mathpar}
        }

    By inversion, we obtain all premises.
    We show that

        \vspace{-2ex}
        { \tiny
        \begin{mathpar}
        \inferrule*[right=\runtimeTypingParallel]{
            \inferrule*[right=\runtimeTypingIntCh]{
\vdots
            }{
                \typingContextOne
                \typingContextCat
                \hat{\typingContextTwo}_1,
                s[\procA] \hasType q_k,
                \set{v_i \hasType L_i}_{i \in I \setminus \set{k}}
                \typingContextCat
                \typingContextThree_1
                    \types
                \procToRuntime{P_k}
            }
            \\
\inferrule*[right=\runtimeTypingQueue]{
\vdots
            }{
                \typingContextOne
                    \typingContextCat
                    \typingContextTwo_2,
                    v_k \hasType L_k
                    \typingContextCat
                    \hat{\typingContextThree}_2,
                    s[\procA][\procB_k] \hasQueueType \queueType \cat \labelAndMsg{l_k}{L_k}
                    \types
                    \queueProc{s}{\queuecontent[(\procA, \procB_k) \mapsto \vec{m} \cat \labelAndMsg{l_k}{v_k}]}
            }
        }{
            \typingContextOne
                \typingContextCat
                \hat{\typingContextTwo}_1,
                s[\procA] \hasType q_k,
                \set{v_i \hasType L_i}_{i \in I \setminus \set{k}},
                v_k \hasType L_k,
                \typingContextTwo_2
                \typingContextCat
                \typingContextThree_1,
                \hat{\typingContextThree}_2,
                s[\procA][\procB_k] \hasQueueType
                \queueType \cat \labelAndMsg{l_k}{L_k}
                \types
                R'
}
        \end{mathpar}
        }

        for $R' =
                \procToRuntime{P_k}
                \parallel
                {\queuecontent[(\procA, \procB_k) \mapsto \vec{m} \cat \labelAndMsg{l_k}{v_k}]}$.

        First, we show there is a typing derivation for
             \[
                \typingContextOne
                \typingContextCat
                \hat{\typingContextTwo}_1,
                s[\procA] \hasType q_k,
                \set{v_i \hasType L_i}_{i \in I \setminus \set{k}}
                \typingContextCat
                \typingContextThree_1
                    \types
                \procToRuntime{P_k}
            \]
            first.
            We instantiate the premise
            \[
                \meta{\forall i \in I \st}
                \typingContextOne \typingContextCat
                    \hat{\typingContextTwo}_1 ,
                    s[\procA] \hasType q_i,
                    \set{v_j \hasType L_j}_{j \in I\setminus \set{i}}
                    \typingContextCat
                    \typingContextThree_1
                \types P_i
            \]
            for $i = k$ and obtain
            \[
                \typingContextOne \typingContextCat
                    \hat{\typingContextTwo}_1 ,
                    s[\procA] \hasType q_k,
                    \set{v_j \hasType L_j}_{j \in I\setminus \set{k}}
                    \typingContextCat
                    \typingContextThree_1
                \types P_k
            \]
            By
            \cref{lm:proc-typing-queue-context-empty},
            we know that $\typingContextThree_1 = \emptyset$ and, thus,
            \cref{lm:typing-proc-to-runtime-configuration} applies and concludes this case.

        Second, we show there is a typing derivation for
                \[
                \typingContextOne
                    \typingContextCat
                    \typingContextTwo_2,
                    v_k \hasType L_k
                    \typingContextCat
                    \hat{\typingContextThree}_2,
                    s[\procA][\procB_k] \hasQueueType \queueType \cat \labelAndMsg{l_k}{L_k}
                    \types
                    \queueProc{s}{\queuecontent[(\procA, \procB_k) \mapsto \vec{m} \cat \labelAndMsg{l_k}{v_k}]}
                \]
            From inversion of the original typing derivation, we have
            \[
                \typingContextOne
                    \typingContextCat
                    \typingContextTwo_2
                    \typingContextCat
                    \hat{\typingContextThree}_2,
                    s[\procA][\procB_k] \hasQueueType \queueType
                    \types
                    \queueProc{s}{\queuecontent[
                        (\procA, \procB_k) \mapsto \vec{m}
                    ]}
            \]
            With \cref{lm:typing-message-list-reversal}, the claim follows.

        It remains to show that there is a transition for the respective typing contexts:
        \begin{align*}
         \typingContextTwo & \is
            \hat{\typingContextTwo}_1,
            s[\procA] \hasType q,
            \set{v_i \hasType L_i}_{i \in I},
            \typingContextTwo_2
           \\
         \typingContextTwo' & \is
            \hat{\typingContextTwo}_1,
            s[\procA] \hasType q_k,
            \set{v_i \hasType L_i}_{i \in I \setminus \set{k}},
            v_k \hasType L_k,
            \typingContextTwo_2
            \\
         \typingContextThree & \is
            \typingContextThree_1,
            \hat{\typingContextThree}_2,
            s[\procA][\procB_k] \hasQueueType \queueType
            \\
         \typingContextThree' & \is
            \typingContextThree_1,
            \hat{\typingContextThree}_2,
            s[\procA][\procB_k] \hasQueueType
            \queueType \cat \labelAndMsg{l_k}{L_k}
            \\
        \end{align*}
       Note that the change from $\typingContextTwo$ to $\typingContextTwo'$ is solely the type of $s[\procA]$ while $s[\procA][\procB_k]$ is the only change from $\typingContextTwo$ to $\typingContextTwo'$.
       Thus, we can simply apply \typingReductionIntCh to obtain a typing context reduction.

     \item \runtimeReductionIn: \\
        With three inversions on the typing derivation, we have a typing derivation with the following shape for \\
                $R =
                \ExtCh_{i \in I} s[\procA][\procB_i] ? \labelAndVar{l_i}{y_i} \seq P_i
                \parallel
                {\queuecontent[(\procB_k, \procA) \mapsto \labelAndMsg{l_k}{v_k} \cat \vec{m}]}$:

        \vspace{-2ex}
        { \tiny
        \begin{mathpar}
        \inferrule*[right=\runtimeTypingParallel]{
            \inferrule*[right=\runtimeTypingExtCh]{
                \delta(q) =
                \set{(\rcv{\procB_i}{\procA}{\labelAndType{l_i}{L_i}}, q_i) \mid i \in I}
                \\
\meta{\forall i \in I \st}
                \typingContextOne
                    \typingContextCat
                    \hat{\typingContextTwo}_1,
                    y_i \hasType L_i,
                    s[\procA] \hasType q_i
                    \typingContextCat
                    \typingContextThree_1
                \types
                P_i \\
            }{
                \typingContextOne
                \typingContextCat
                \hat{\typingContextTwo}_1,
                s[\procA] \hasType q
                \typingContextCat
                \typingContextThree_1
                    \types
                \ExtCh_{i \in I} s[\procA][\procB_i] ? \labelAndVar{l_i}{y_i} \seq P_i
            }
            \\
\inferrule*[right=\runtimeTypingQueue]{
                \typingContextOne
                    \typingContextCat
                    \hat{\typingContextTwo}_2
                    \typingContextCat
                    \hat{\typingContextThree}_2,
                    s[\procB_k][\procA] \hasQueueType \queueType
                    \types
                    \queueProc{s}{\queuecontent[(\procB_k, \procA) \mapsto \vec{m}]}
            }{
                \typingContextOne
                    \typingContextCat
                    \hat{\typingContextTwo}_2,
                    v_k \hasType L_k,
                    \typingContextCat
                    \hat{\typingContextThree}_2,
                    s[\procB_k][\procA] \hasQueueType \labelAndVar{l_k}{L_k} \cat \queueType
                    \types
                    \queueProc{s}{\queuecontent[(\procB_k, \procA) \mapsto \labelAndMsg{l_k}{v_k} \cat \vec{m}]}
            }
        }{
            \typingContextOne
                \typingContextCat
                \hat{\typingContextTwo}_1,
                s[\procA] \hasType q,
                \hat{\typingContextTwo}_2,
                v_k \hasType L_k
                \typingContextCat
                \typingContextThree_1,
                \hat{\typingContextThree}_2,
                s[\procB_k][\procA] \hasQueueType \labelAndVar{l_k}{L_k} \cat \queueType
                \types
                R
}
        \end{mathpar}
        }

    By inversion, we obtain all the premises.
    We show that there is a typing derivation of shape

        \vspace{-1ex}
        { \scriptsize
        \begin{mathpar}
        \inferrule*[right=\runtimeTypingParallel]{
            \inferrule*[right=\runtimeTypingExtCh]{
                \vdots
            }{
                \typingContextOne
                \typingContextCat
                \hat{\typingContextTwo}_1,
                s[\procA] \hasType q_k,
                v_k \hasType L_k
                \typingContextCat
                \typingContextThree_1
                    \types
                \procToRuntime{P_k[v_k/y_k]}
            }
            \\
\inferrule*[right=\runtimeTypingQueue]{
                \vdots
            }{
                \typingContextOne
                    \typingContextCat
                    \hat{\typingContextTwo}_2
                    \typingContextCat
                    \hat{\typingContextThree}_2,
                    s[\procB_k][\procA] \hasQueueType \queueType
                    \types
                    \queueProc{s}{\queuecontent[(\procB_k, \procA) \mapsto \vec{m}]}
            }
        }{
            \typingContextOne
                \typingContextCat
                \hat{\typingContextTwo}_1,
                s[\procA] \hasType q,
                v_k \hasType L_k,
                \hat{\typingContextTwo}_2
                \typingContextCat
                \typingContextThree_1,
                \hat{\typingContextThree}_2,
                s[\procB_k][\procA] \hasQueueType \queueType
                \types
                \procToRuntime{P_i[v_k/y_k]}
                \parallel
                {\queuecontent[(\procB_k, \procA) \mapsto \vec{m}]}
        }
        \end{mathpar}
        }

        First, we show there is a typing derivation for
        \[
            \typingContextOne
            \typingContextCat
            \hat{\typingContextTwo}_1,
            s[\procA] \hasType q_k,
            v_k \hasType L_k
            \typingContextCat
            \typingContextThree_1
                \types
            \procToRuntime{P_k[v_k/y_k]}
            \enspace .
        \]
        From inversion, we get the following premise from the original typing derivation:
        \[
            \meta{i \in I \st}
            \typingContextOne
                \typingContextCat
                \hat{\typingContextTwo}_1,
                y_i \hasType L_i,
                s[\procA] \hasType q_i
                \typingContextCat
                \typingContextThree_1
            \types
            P_i
        \]
        which we instantiate with $i = k$ to obtain:
        \[
            \typingContextOne
                \typingContextCat
                \hat{\typingContextTwo}_1,
                y_k \hasType L_k,
                s[\procA] \hasType q_k
                \typingContextCat
                \typingContextThree_1
            \types
            P_k
            \enspace .
        \]
        With \cref{lm:substitution-lemma}, we get
        \[
            \typingContextOne
                \typingContextCat
                \hat{\typingContextTwo}_1,
                v_k \hasType L_k,
                s[\procA] \hasType q_k
                \typingContextCat
                \typingContextThree_1
            \types
            P_k[v_k/y_k]
            \enspace .
        \]
        By
        \cref{lm:proc-typing-queue-context-empty},
        we know that $\typingContextThree_1 = \emptyset$ and, thus,
        \cref{lm:typing-proc-to-runtime-configuration} applies and concludes this case.

        Second, there is a typing derivation for
        \[
            \typingContextOne
                \typingContextCat
                \hat{\typingContextTwo}_2
                \typingContextCat
                \hat{\typingContextThree}_2,
                s[\procB_k][\procA] \hasQueueType \queueType
                \types
                \queueProc{s}{\queuecontent[(\procB_k, \procA) \mapsto \vec{m}]}
        \]
        by inversion on the original typing derivation.

        It remains to show that there is a transition for the respective typing contexts:
        \begin{align*}
         \typingContextTwo & \is
            \hat{\typingContextTwo}_1,
            s[\procA] \hasType q,
            \hat{\typingContextTwo}_2,
            v_k \hasType L_k
           \\
         \typingContextTwo' & \is
            \hat{\typingContextTwo}_1,
            s[\procA] \hasType q_k,
            v_k \hasType L_k,
            \hat{\typingContextTwo}_2
            \\
         \typingContextThree & \is
            \typingContextThree_1,
            \hat{\typingContextThree}_2,
            s[\procB_k][\procA] \hasQueueType \labelAndVar{l_k}{L_k} \cat \queueType
            \\
         \typingContextThree' & \is
            \typingContextThree_1,
            \hat{\typingContextThree}_2,
            s[\procB_k][\procA] \hasQueueType
            \queueType
        \end{align*}

       Note that the change from $\typingContextTwo$ to $\typingContextTwo'$ is solely the type of $s[\procA]$ while $s[\procB_k][\procA]$ is the only change from $\typingContextTwo$ to $\typingContextTwo'$.
       Thus, we can simply apply \typingReductionExtCh to obtain a typing context reduction.

     \item \runtimeReductionErrOne: \\
        By assumption, we have a typing derivation for
        $
                \ExtCh_{i \in I} s[\procA][\procB_i] ? \labelAndVar{l_i}{y_i} \seq P_i
                \parallel
                \queueProc{s}{\queuecontent}
        $ and it holds that
        $
                \forall i \in I \st
                \queuecontent(\procB_i, \procA) = \labelAndMsg{l}{\_} \cat \vec{m}
                \text{ and }
                l_i \neq l
        $.
        By inversion and \cref{lm:proc-typing-queue-context-empty}, the typing derivation must have the following shape:

        \vspace{-2ex}
        { \small
        \begin{mathpar}
        \inferrule*[right=\runtimeTypingParallel]{
            \inferrule*[right=\runtimeTypingExtCh]{
                \delta(q) =
                \set{(\rcv{\procB_i}{\procA}{\labelAndType{l_i}{L_i}}, q_i) \mid i \in I} \\
                \meta{\forall i \in I \st}
                \typingContextOne
                    \typingContextCat
                    \hat{\typingContextTwo}_1,
                    y_i \hasType L_i,
                    s[\procA] \hasType q_i
                    \typingContextCat
                    \emptyset
                \types
                P_i \\
            }{
                \typingContextOne
                \typingContextCat
                \hat{\typingContextTwo}_1,
                s[\procA] \hasType q
                \typingContextCat
                \emptyset
                    \types
                \ExtCh_{i \in I} s[\procA][\procB_i] ? \labelAndVar{l_i}{y_i} \seq P_i
            }
                \\
            \inferrule*{
                \vdots
            }{
\typingContextOne
                \typingContextCat
                \typingContextTwo_2
                \typingContextCat
                \typingContextThree_2
                    \types
                \queueProc{s}{\queuecontent}
            }
        }{
                \typingContextOne
                \typingContextCat
                \hat{\typingContextTwo}_1,
                s[\procA] \hasType q,
                \typingContextTwo_2
                \typingContextCat
                \typingContextThree_2
                    \types
                \ExtCh_{i \in I} s[\procA][\procB_i] ? \labelAndVar{l_i}{y_i} \seq P_i
                \parallel
                \queueProc{s}{\queuecontent}
        }
        \end{mathpar}
        }

        Let us rewrite the typing and queue typing context:
        \begin{align*}
            \hat{\typingContextTwo}_1, s[\procA] \hasType q, \typingContextTwo_2
            & =
            \hat{\typingContextTwo}, \typingContextTwo_s
            \\
            \typingContextThree_2
            & =
            \hat{\typingContextThree}, \typingContextThree_s
        \end{align*}
        By assumption, we know that
            there is
            $
            (\vec{q}, \xi) \in \reach(\CSMabb{A})
            $ such that
            $\typingContextTwo_s =
                \set{s[\procA] \hasType q}_{\procA \in \ProcsOf{\CSMabb{A}}}
                $
            and
            $\typingContextThree_s =
                \set{s[\procA][\procB] \hasQueueType \xi(\procA,\procB)}_{(\procA,\procB) \in \channelsOf{\CSMabb{A}}}
            $.
        Recall the condition on the reduction semantics:
        $
                \forall i \in I \st
                \queuecontent(\procB_i, \procA) = \labelAndMsg{l}{\_} \cat \vec{m}
                \text{ and }
                l_i \neq l
        $.
        Thus, with \cref{lm:queue-types-model-queues}, it follows that, for all $i \in I$,
        $\xi(\procB_i,\procA) = \labelAndVar{l'_i}{\_} \cat \_$
        with $l'_i \neq l_i$.
        For the CSM $\CSMabb{A}$, this entails that $\procA$ expects to receive a message from a set of other participants, ranged over by $\procB_i$, but the first message in each channel does not match.
        Thus, none of them will ever be received.
        This yields a contradiction to feasible eventual reception:
        there has been at least one send event to $\procA$ and there is no matching receive event yet;
        because $\procA$ will never proceed, no matching receive event can ever happen.
    \end{itemize}

 \item \runtimeTypingRestr: \\
        By inversion on the typing derivation, there is a typing derivation
        \[
         \typingContextOne
         \typingContextCat
         \typingContextTwo
         \typingContextCat
         \typingContextThree
            \types
         (\restr s \hasType \CSMabb{A}) \, R
         \enspace .
        \]
        We do inversion on~\ref{lm:subject-reduction-assumption-4}, yielding two reduction rules that apply.
        \begin{itemize}[leftmargin=2ex]
         \item \runtimeReductionContext: \\
            We have

            \vspace{-2ex}
            { \small
            \begin{mathpar}
            \inferrule*[right=\runtimeTypingRestr]{
(\vec{q}, \xi) \in \reach(\CSMabb{A})
                \\
                \typingContextTwo_s =
                    \set{s[\procA] \hasType \vec{q}_\procA}_{\procA \in \ProcsOf{\CSMabb{A}}}
                \\
                \typingContextThree_s =
                    \set{s[\procA][\procB] \hasQueueType \xi(\procA,\procB)}_{(\procA,\procB) \in \channelsOf{\CSMabb{A}}}
                \\
\typingContextOne
                    \typingContextCat
                    \typingContextTwo,
                    \typingContextTwo_s
                    \typingContextCat
                    \typingContextThree,
                    \typingContextThree_s
                \types
                R
}{
                \typingContextOne
                    \typingContextCat
                    \typingContextTwo
                    \typingContextCat
                    \typingContextThree
                    \types
                    (\restr s \hasType \CSMabb{A})\, R
            }
         \end{mathpar}
         }

            By inversion, we obtain all premises.
            We show that

            \vspace{-2ex}
            { \small
            \begin{mathpar}
            \inferrule*[right=\runtimeTypingRestr]{
(\pvec{q}', \xi') \in \reach(\CSMabb{A})
                \\
                \typingContextTwo'_s =
                    \set{s[\procA] \hasType q'_\procA}_{\procA \in \ProcsOf{\CSMabb{A}}}
                \\
                \typingContextThree'_s =
                    \set{s[\procA][\procB] \hasQueueType \xi'(\procA,\procB)}_{(\procA,\procB) \in \channelsOf{\CSMabb{A}}}
                \\
\typingContextOne
                    \typingContextCat
                    \typingContextTwo,
                    \typingContextTwo'_s
                    \typingContextCat
                    \typingContextThree,
                    \typingContextThree'_s
                \types
                R'
}{
                \typingContextOne
                    \typingContextCat
                    \typingContextTwo
                    \typingContextCat
                    \typingContextThree
                    \types
                    (\restr s \hasType \CSMabb{A})\, R'
            }
         \end{mathpar}
         }

         The first premise is the same as for the original typing derivation.
         We know that
         $
                \typingContextOne
                    \typingContextCat
                    \typingContextTwo,
                    \typingContextTwo_s
                    \typingContextCat
                    \typingContextThree,
                    \typingContextThree_s
                \types
                R
         $.
         With the induction hypothesis,
         we get
         \[
                \typingContextOne
                    \typingContextCat
                    \typingContextTwo,
                    \typingContextTwo'_s
                    \typingContextCat
                    \typingContextThree,
                    \typingContextThree'_s
                \types
                R'
         \text{ and }
                \typingContextTwo,
                \typingContextTwo_s
                \typingContextCat
                \typingContextThree,
                \typingContextThree_s
                \redto
                \typingContextTwo,
                \typingContextTwo'_s
                \typingContextCat
                \typingContextThree,
                \typingContextThree'_s
         \enspace .
         \]
         The first fact proves the last premise for the new typing derivation.
         For the remaining ones, we apply \cref{lm:typing-reductions-preserve-reachability}, which yields that
            there is $
            (\pvec{q}', \xi') \in \reach(\CSMabb{A})
        $ such that
            $\typingContextTwo'_s =
                \set{s[\procA] \hasType q'}_{\procA \in \ProcsOf{\CSMabb{A}}}
                $
            and
            $\typingContextThree'_s =
                \set{s[\procA][\procB] \hasQueueType \xi'(\procA,\procB)}_{(\procA,\procB) \in \channelsOf{\CSMabb{A}}}
        $.
        These are precisely the remaining premises for the new typing derivation.
        It is obvious that there is a reduction for the typing contexts, which concludes this case.

         \item \runtimeReductionErrTwo: \\
            We have a typing derivation for
            \[
                    \typingContextOne
                    \typingContextCat
                    \typingContextTwo
                    \typingContextCat
                    \typingContextThree
                        \types
                    (\restr s \hasType \CSMabb{A}) \; \queueProc{s}{\queuecontent}
            \]
            and know
                    $\queuecontent(\procA, \procB) \neq \emptystring \text{ for some } \procA, \procB$.
            We do inversion on the typing derivation:

            \vspace{-2ex}
            { \small
            \begin{mathpar}
            \inferrule*[right=\runtimeTypingRestr]{
                \inferrule*[right=\runtimeTypingQueue]{
                    \vdots
                }{
                    \typingContextOne
                        \typingContextCat
                        \typingContextTwo,
                        \typingContextTwo_s
                        \typingContextCat
                        \typingContextThree,
                        \typingContextThree_s
                    \types
                    \queueProc{s}{\queuecontent}
                }
                \\
(\vec{q}, \xi) \in \reach(\CSMabb{A})
                \\
                \typingContextTwo_s =
                    \set{s[\procA] \hasType \vec{q}_\procA}_{\procA \in \ProcsOf{\CSMabb{A}}}
                \\
                \typingContextThree_s =
                    \set{s[\procA][\procB] \hasQueueType \xi(\procA,\procB)}_{(\procA,\procB) \in \channelsOf{\CSMabb{A}}}
            }{
                \typingContextOne
                \typingContextCat
                \typingContextTwo
                \typingContextCat
                \typingContextThree
                    \types
                (\restr s \hasType \CSMabb{A}) \; \queueProc{s}{\queuecontent}
            }
            \end{mathpar}
            }

            By definition of $\typingContextTwo_s$, there is a type $s[\procA] \hasType q$ for every $\procA \in \ProcsOf{\CSMabb{A}}$.
            There is a typing derivation for
            \[
                    \typingContextOne
                        \typingContextCat
                        \typingContextTwo,
                        \set{s[\procA] \hasType \vec{q}_\procA}_{\procA \in \ProcsOf{\CSMabb{A}}}
                        \typingContextCat
                        \typingContextThree,
                        \set{s[\procA][\procB] \hasQueueType \xi(\procA,\procB)}_{(\procA,\procB) \in \channelsOf{\CSMabb{A}}}
                    \types
                    \queueProc{s}{\queuecontent}
                \enspace .
            \]
            The only applicable typing rules (in the whole derivation) are
            \runtimeTypingEnd,
            \runtimeTypingEmptyQueue, and
            \runtimeTypingQueue.
            By our assumption that there is a strict partial order for the CSMs in our system, $s[\hole]$ does not appear in $\queuecontent$.

            Thus, \runtimeTypingEnd needs to be applied to reduce the typing context
            $\set{s[\procA] \hasType \vec{q}_\procA}_{\procA \in \ProcsOf{\CSMabb{A}}}$
            to only contain $\typingContextTwo$, which can then be used to type the queue with \runtimeTypingQueue.
            The premise of \runtimeTypingEnd requires that
            $\EndState(q)$, \ie
            $q$ is a final state
            and has not outgoing receive transition.
            This, however, entails that $(\vec{q}, \xi)$ is a non-final configuration where all participants are in final states
and the channels are not empty, yielding a deadlock.
            This contradicts the fact that $\CSMabb{A}$ is deadlock-free, concluding this case.

        \end{itemize}

 \item \runtimeTypingQueue: \\
        We have the typing derivation

        \vspace{-2ex}
        { \small
        \begin{mathpar}
        \inferrule*[right=\runtimeTypingQueue]{
            \typingContextOne
                \typingContextCat
                \typingContextTwo
                \typingContextCat
                \typingContextThree,
                s[\procA][\procB] \hasQueueType \queueType
                \types
                \queueProc{s}{\queuecontent[
                    (\procA, \procB) \mapsto \vec{m}
                ]}
        }{
            \typingContextOne
                \typingContextCat
                \typingContextTwo,
                v \hasType L
                \typingContextCat
                \typingContextThree,
                s[\procA][\procB] \hasQueueType \labelAndVar{l}{L} \cat \queueType
                \types
                \queueProc{s}{\queuecontent[
                    (\procA, \procB) \mapsto \labelAndMsg{l}{v} \cat \vec{m}
                ]}
        }
        \end{mathpar}
        }

        However, there is not $R'$ such that
                $\queueProc{s}{\queuecontent[
                    (\procA, \procB) \mapsto \labelAndMsg{l}{v} \cat \vec{m}
                ]}
            \redto R'$, contradicting~\ref{lm:subject-reduction-assumption-4}.
\end{itemize}
This concludes the proof of subject reduction. 
\proofEndSymbol
\end{proof}

From subject reduction, \emph{type safety} follows:
if a process can be typed, any runtime configuration that can be reached from this process cannot contain an error.
Note that \cref{cor:safety} in the main text follows from \cref{cor:type-safety}.

\begin{corollary}[Type Safety]
\label{cor:type-safety}
Assume that
$\types \Defs \hasType \typingContextOne$
and
$\typingContextOne
    \typingContextCat
    \emptyset
\types
P$ for some well-annotated process $P$.
If
$\procToRuntime{P} \redto^* R$,
then, $R \neq \err$.
\end{corollary}
\begin{proof}
From \cref{lm:typing-proc-to-runtime-configuration}, we know that
$\typingContextOne
    \typingContextCat
    \emptyset
    \typingContextCat
    \emptyset
\types
\procToRuntime{P}$.
By definition $\redto^* \is \set{\redto^k \mid k \geq 0}$.
We prove a stronger claim:
For all $k \geq 0$,
if
$\procToRuntime{P} \redto^k R$,
then,
\begin{itemize}
 \item $\typingContextOne
           \typingContextCat
           \typingContextTwo
           \typingContextCat
           \typingContextThree
       \types
       R$ with
        $\typingContextTwo = \hat{\typingContextTwo}, \set{\typingContextTwo_s}_{s \in \SessionName}$
        and
        $\typingContextThree = \hat{\typingContextThree}, \set{\typingContextThree_s}_{s \in \SessionName}$, and
 \item for all $s \in \SessionName$, it holds that
       there is $
    (\vec{q}, \xi) \in \reach(\CSMabb{A})
   $ such that
    $\typingContextTwo_s =
        \set{s[\procA] \hasType q}_{\procA \in \ProcsOf{\CSMabb{A}}}
        $
    and
    $\typingContextThree_s =
        \set{s[\procA][\procB] \hasQueueType \xi(\procA,\procB)}_{(\procA,\procB) \in \channelsOf{\CSMabb{A}}}
   $
\end{itemize}
This claim entails that $R \neq \err$ because $\err$ cannot be typed but $R$ can be typed.

We prove the claim by induction on $k$.

For $k = 0$, the claim trivially follows because both the typing and queue typing context is empty, trivially satisfying the conditions.

For the induction step,
we have $\procToRuntime{P} \redto^k R$, the claim holds for $R$, and $R \redto R'$.
With Subject Reduction (\cref{thm:subject-reduction}), we proved precisely what we need to show for $R'$.
\proofEndSymbol
\end{proof}

Subject reduction shows that any step of a runtime configuration can be mimicked by the typing contexts and these can be used to type the new runtime configuration.
Since $\err$ cannot be typed, this shows that a typed runtime configuration can never reduce to $\err$, yielding type safety.
While this is a safety property, \emph{session fidelity} deals with progress.
Roughly speaking, if the typing contexts can take a step, then the runtime configuration can also take a step.
In most MST frameworks, this can only be proven in the presence of a single session.
Thus, we define the following restriction of our type system.

\begin{definition}
We define $\typesSFd$ to be $\types$ but without the rules \procTypingRestr and \runtimeTypingRestr.
Using this, we define $\typesSFs$ for processes as follows:
\begin{mathpar}
\inferrule*[right=\procTypingRestr ']{
    \meta{\forall \procA \in \ProcsOf{\CSMabb{A}} \st}
    \forall c \hasType q \in \typingContextTwo_\procA \st
    \EndState(q)
    \\
\meta{\forall \procA \in \ProcsOf{\CSMabb{A}} \st}
        \typingContextOne
        \typingContextCat
        \typingContextTwo_\procA,
        s[\procA] \hasType \initialState(\CSMabb{A}_\procA)
        \typesSFd
        Q_\procA
}{
    \typingContextOne
        \typingContextCat
        \set{\typingContextTwo_\procA}_{\procA \in \ProcsOf{\CSMabb{A}}}
        \typesSFs
        (\restr s \hasType \CSMabb{A})\,
        (\Parallel_{\procA \in \ProcsOf{\CSMabb{A}}} Q_\procA)
}
\end{mathpar}
We also define $\typesSFs$ for runtime configurations:
\begin{mathpar}
\inferrule*[right=\runtimeTypingRestr ']{
    (\vec{q}, \xi) \in \reach(\CSMabb{A})
    \\
\meta{\forall \procA \in \ProcsOf{\CSMabb{A}} \st}
    \forall c \hasType q' \in \typingContextTwo_\procA \st
    \EndState(q')
    \\
\meta{\forall \procA \in \ProcsOf{\CSMabb{A}} \st}
        \typingContextOne
        \typingContextCat
        \typingContextTwo_\procA,
        s[\procA] \hasType \vec{q}_\procA
        \typesSFd
        Q_\procA
    \\
    \typingContextOne
    \typingContextCat
    \typingContextTwo'
    \typingContextCat
    \set{s[\procA][\procB] \hasQueueType \xi(\procA,\procB)}_{(\procA,\procB) \in \channelsOf{\CSMabb{A}}}
    \typesSFd
    \queueProc{s}{\queuecontent}
}{
    \typingContextOne
        \typingContextCat
        \set{\typingContextTwo_\procA}_{\procA \in \ProcsOf{\CSMabb{A}}},
\typingContextTwo'
        \typingContextCat
        \emptyset
        \typesSFs
        (\restr s \hasType \CSMabb{A})\,
        (\Parallel_{\procA \in \ProcsOf{\CSMabb{A}}} Q_\procA)
        \parallel
        \queueProc{s}{\queuecontent}
}
\end{mathpar}
If
$
    \typingContextOne
        \typingContextCat
        \set{\typingContextTwo_\procA}_{\procA \in \ProcsOf{\CSMabb{A}}},
\typingContextTwo'
        \typingContextCat
        \emptyset
        \typesSFs
        (\restr s \hasType \CSMabb{A})\,
        (\Parallel_{\procA \in \ProcsOf{\CSMabb{A}}} Q_\procA)
        \parallel
        \queueProc{s}{\queuecontent}
$
holds, we know that we can obtain the premises by inversion.
For conciseness, we use the following notation
to refer to the CSM configuration $(\vec{q}, \xi)$:
\[
    \typingContextOne
        \typingContextCat
        \set{\typingContextTwo_\procA}_{\procA \in \ProcsOf{\CSMabb{A}}},
\typingContextTwo'
        \typingContextCat
        \emptyset
        \overset{(\vec{q}, \xi)}{\typesSFs}
        (\restr s \hasType \CSMabb{A})\,
        (\Parallel_{\procA \in \ProcsOf{\CSMabb{A}}} Q_\procA)
        \parallel
        \queueProc{s}{\queuecontent}
    \enspace .
\]
\end{definition}

\begin{proposition}
Let $P$ be a process, $R$ be a runtime configuration and assume
        $\typesSFs \Defs \hasType \typingContextOne$.
If
        $
            \typingContextOne
                \typingContextCat
                \typingContextTwo
                \typesSFs
                (\restr s \hasType \CSMabb{A})\,
                P
        $,
then $P$ is restriction-free.
\\
If
        $
            \typingContextOne
                \typingContextCat
                \typingContextTwo
                \typingContextCat
                \emptyset
                \typesSFs
                (\restr s \hasType \CSMabb{A})\,
                R
        $,
then $R$ is restriction-free.
\end{proposition}

Intuitively, $\typesSFs$ allows us to have one restriction with CSM $\CSMabb{A}$ and requires that all different participants of $\CSMabb{A}$ are played by different processes in parallel.
As argued in the main text, these are standard restrictions for session fidelity (and deadlock freedom).

As for $\types$, we show a correspondence between processes and runtime configurations for $\typesSFd$ and $\typesSFs$.

\begin{lemma}
\label{lm:typingSF-proc-to-runtime-configuration}
Let $P$ be a well-annotated process. \\
If
$\typingContextOne \typingContextCat \typingContextTwo \typesSFd P$,
then
$\typingContextOne \typingContextCat \typingContextTwo \typingContextCat \emptyset \typesSFd \procToRuntime{P}$.
If
$\typingContextOne \typingContextCat \typingContextTwo \typesSFs P$,
then
$\typingContextOne \typingContextCat \typingContextTwo \typingContextCat \emptyset \typesSFs \procToRuntime{P}$.
\end{lemma}
\begin{proof}
We prove the claim $\typesSFd$ by induction on the structure of $P$, as for
\cref{lm:typing-proc-to-runtime-configuration}.
For all except $P = P_1 \parallel P_2$ and $P = (\restr s \hasType \CSMabb{A}) \, P'$, it holds that
$\procToRuntime{P} = P$.
For all typing rules that processes and runtime configurations share, the queue typing context is not changed in the respective runtime configuration typing rule.
Thus,
$\typingContextOne \typingContextCat \typingContextTwo \typingContextCat \emptyset \types \procToRuntime{P}$.
For $P = P_1 \parallel P_2$, the claim follows directly by induction hypothesis.
Since $\typesSFd$ has no rule for restriction, we do not need to consider $P = (\restr s \hasType \CSMabb{A}) \, P'$.

In contrast, $\typesSFs$ only applies to $P$ with shape
$
    (\restr s \hasType \CSMabb{A})\,
    (\Parallel_{\procA \in \ProcsOf{\CSMabb{A}}} Q_\procA)
    \parallel
    \queueProc{s}{\queuecontent}
$,
so we only need to consider such processes.
As for
\cref{lm:typing-proc-to-runtime-configuration},
we can pick the initial states and empty channels.
The remaining premises follow from the first claim.
\proofEndSymbol
\end{proof}

Note that \cref{thm:progress} in the main text is a consequence of \cref{lm:session-fidelity}.

\begin{restatable}[Session fidelity with sink-final FSMs]{theorem}{sessionFidelitySoftDeadlocks}
\label{lm:session-fidelity}
Let $\CSMabb{A}$ be a deadlock-free CSM that satisfies feasible eventual reception and, 
for every $\procA \in \ProcsOf{\CSMabb{A}}$, 
$\CSMabb{A}_\procA$ is sink-final. 
Let $R$ be a runtime configuration. We assume that \begin{enumerate}[label=\textnormal{(\arabic*)}]
 \item \label{sf-cond-1}
        $\typesSFs \Defs \hasType \typingContextOne$,
 \item \label{sf-cond-2}
        $
            \typingContextOne
                \typingContextCat
                \typingContextTwo
\typingContextCat
                \emptyset
                \overset{(\vec{q}, \xi)}{\typesSFs}
                (\restr s \hasType \CSMabb{A})\,
                R
        $, and
 \item \label{sf-cond-3}
       $
       (\vec{q}, \xi)
       \rightarrow
       (\pvec{q}'', \xi'')
       $
       for some
       $\pvec{q}''$ and $\xi''$.
\end{enumerate}
Then, there is
$ (\pvec{q}', \xi') $
with
$
    (\vec{q}, \xi)
    \rightarrow
    (\pvec{q}', \xi')
$
and
$R'$ with $R \redto R'$ such that
        $
            \typingContextOne
                \typingContextCat
                \typingContextTwo
\typingContextCat
                \emptyset
                \overset{(\pvec{q}'\negthinspace, \, \xi')}{\typesSFs}
                (\restr s \hasType \CSMabb{A})\,
                R'
        $.
\end{restatable}
\begin{proof}
By assumption, we know that
$
    (\vec{q}, \xi)
    \xrightarrow{x}
    (\pvec{q}', \xi')
$.
We do a case analysis on the shape of $x$.

First, let $x = \snd{\procC}{\procB}{\labelAndType{l}{L}}$.
We do inversion on
\ref{sf-cond-2}
and rewrite
$\typingContextTwo$
as
$\set{\typingContextTwo_\procA}_{\procA \in \ProcsOf{\CSMabb{A}}},
 \typingContextTwo'$:
\begin{mathpar}
\inferrule*[right=\runtimeTypingRestr ']{
    (\vec{q}, \xi) \in \reach(\CSMabb{A})
    \\
\meta{\forall \procA \in \ProcsOf{\CSMabb{A}} \st}
    \forall c \hasType q' \in \typingContextTwo_\procA \st
    \EndState(q')
    \\
\meta{\forall \procA \in \ProcsOf{\CSMabb{A}} \st}
        \typingContextOne
        \typingContextCat
        \typingContextTwo_\procA,
        s[\procA] \hasType \vec{q}_\procA
        \typesSFd
        Q_\procA
    \\
    \typingContextOne
    \typingContextCat
    \typingContextTwo'
    \typingContextCat
    \set{s[\procA][\procB] \hasQueueType \xi(\procA,\procB)}_{(\procA,\procB) \in \channelsOf{\CSMabb{A}}}
    \typesSFd
    \queueProc{s}{\queuecontent}
}{
    \typingContextOne
        \typingContextCat
        \set{\typingContextTwo_\procA}_{\procA \in \ProcsOf{\CSMabb{A}}},
\typingContextTwo'
        \typingContextCat
        \emptyset
        \typesSFs
        (\restr s \hasType \CSMabb{A})\,
        (\Parallel_{\procA \in \ProcsOf{\CSMabb{A}}} Q_\procA)
        \parallel
        \queueProc{s}{\queuecontent}
}
\end{mathpar}
and obtain all its premises
as well as the fact that
$R =
    (\Parallel_{\procA \in \ProcsOf{\CSMabb{A}}} Q_\procA)
    \parallel
    \queueProc{s}{\queuecontent}
$.
From the fact that $\vec{q}_\procC$ has outgoing transitions, we know that it is not final, which in turn means that $\EndState(\vec{q}_\procC)$ does not hold.
Hence,
two typing rules can apply:
\procTypingProcName and
\procTypingIntCh.

We do a case analysis and show that applying \procTypingProcName will eventually lead to applying \procTypingIntCh as well.
First, we do an inversion:
\begin{mathpar}
\inferrule*[right=\procTypingProcName]{
    \typingContextOne(\pn{Q}) = \vec{L} }{
    \typingContextOne
        \typingContextCat
        \vec{c}  \hasType \vec{L}
\typesSFd
        \pn{Q}[\vec{c}]
}
\end{mathpar}

Assume that we type $\pn{Q}[\vec{x}]$ and
$\pn{Q}[\vec{x}] = P'$.
From
\ref{sf-cond-1},
it follows that
$
    \typingContextOne
    \typingContextCat
    \vec{x} \hasType \vec{L}
    \typesSFd
    P'
$.
By assumption, we know that process definitions are guarded.
Thus, $P'$ needs to be typed with \procTypingIntCh.
The following arguments are very similar from now on.
In fact, there are two differences:
first, we would carry
$\vec{x} \hasType \vec{L}$ around, and
second, we would apply
\runtimeReductionProcName
to prove that $R \redto R'$.
For conciseness, we refrain from doing so and focus on the case without the indirection through a process definition.

Thus, we consider \procTypingIntCh as typing rule for $Q_\procC$.
By inversion, we have
\begin{mathpar}
\inferrule*[right=\procTypingIntCh]{
\delta(q) \sups
    \set{(\snd{\procA}{\procB_i}{\labelAndType{l_i}{L_i}}, q_i) \mid i \in I}
    \\
\meta{\forall i \in I \st}
    \typingContextOne \typingContextCat
        \hat{\typingContextTwo}_\procC,
        s[\procC] \hasType q_i,
        \set{c_j \hasType L_j}_{j \in I\setminus \set{i}}
        \types P_i \\
}{
    \typingContextOne
    \typingContextCat
    \hat{\typingContextTwo}_\procC,
    s[\procC] \hasType q,
    \set{c_i \hasType L_i}_{i \in I}
        \types
    \IntCh_{i \in I} s[\procC][\procB_i] ! \labelAndMsg{l_i}{c_i} \seq P_i
}
\end{mathpar}
and obtain all premises as well as the facts that
$
Q_\procC
=
\IntCh_{i \in I} c[\procB_i] ! \labelAndMsg{l_i}{c_i} \seq P_i
$ and
$
    \typingContextTwo_\procC
    =
    \hat{\typingContextTwo}_\procC,
    \set{c_i \hasType L_i}_{i \in I}
$.

Because of
$
    \delta(q) \sups
    \set{(\snd{\procA}{\procB_i}{\labelAndType{l_i}{L_i}}, q_i) \mid i \in I}
$,
it is possible that
$l \neq l_i$ for all $i \in I$ since $\card{I} > 0$ by definition.
However, we know that there exists at least one label $l_i$ that can be sent and we choose to use this for the witness
$(\pvec{q}', \xi')$.
This is precisely the reason why we cannot ensure that every possible send transition in the CSM can be followed but we can ensure that there is at least one to follow.
Let $k \in I$ such that $l_k = l$.
We choose
\[
R' \is
    (\Parallel_{\procA \in \ProcsOf{\CSMabb{A}} \setminus \set{\procC}} Q_\procA)
        \parallel
    P_k
        \parallel
    \queueProc{s}
    {\queuecontent[(\procA, \procB_k) \mapsto \queuecontent(\procA, \procB_k) \cat \labelAndMsg{l_k}{v_k}]}
\]
Note that
$Q_k$ is restriction by
\ref{sf-cond-2}.
Thus, $P_k$ is restriction-free, which entails that
$\procToRuntime{P_k} = P_k$.
With
\runtimeReductionOut
and
\runtimeReductionContext,
it is straightforward to show that
$R \redto R'$.
It remains to show that
\[
\small
    \typingContextOne
        \typingContextCat
        \set{\typingContextTwo_\procA}_{\procA \in \ProcsOf{\CSMabb{A}}},
\typingContextTwo'
        \typingContextCat
        \emptyset
        \overset{(\pvec{q}'\negthinspace, \, \xi')}{\typesSFs}
        (\restr s \hasType \CSMabb{A})\,
        (\Parallel_{\procA \in \ProcsOf{\CSMabb{A}} \setminus \set{\procC}} Q_\procA)
            \parallel
        P_k
            \parallel
        \queueProc{s}
        {\queuecontent[(\procC, \procB) \mapsto \queuecontent(\procC, \procB) \cat \labelAndMsg{l_k}{v_k}]}
    \enspace
\]
We start building a typing derivation:

\vspace{-2ex}
{ \scriptsize
\begin{mathpar}
\inferrule*[right=\runtimeTypingRestr ']{
    (a):
    (\vec{q}', \xi') \in \reach(\CSMabb{A})
    \\
(b):
    \meta{\forall \procA \in \ProcsOf{\CSMabb{A}} \st}
    \forall c \hasType q' \in \typingContextTwo_\procA \st
    \EndState(q')
    \\
(c):
    \meta{\forall \procA \in \ProcsOf{\CSMabb{A}} \setminus \set{\procC} \st}
        \typingContextOne
        \typingContextCat
        \typingContextTwo_\procA,
        s[\procA] \hasType \pvec{q}'_\procA
            \typesSFd
        Q_\procA
    \\
(d):
    \typingContextOne
    \typingContextCat
    \hat{\typingContextTwo}_\procC,
    \set{c_i \hasType L_i}_{i \in I \setminus \set{k}},
    s[\procC] \hasType \pvec{q}'_\procC
        \typesSFd
    P_k
    \\
(e):
    \typingContextOne
    \typingContextCat
    \typingContextTwo',
    c_k \hasType L_k
    \typingContextCat
    \set{s[\procA][\procB] \hasQueueType \xi(\procA,\procB)}_{(\procA,\procB) \in \channelsOf{\CSMabb{A}}}
    \typesSFd
    \queueProc{s}
    {\queuecontent[(\procA, \procB_k) \mapsto \queuecontent(\procA, \procB_k) \cat \labelAndMsg{l_k}{v_k}]}
}{
    \typingContextOne
        \typingContextCat
\set{\typingContextTwo_\procA}_{\procA \in \ProcsOf{\CSMabb{A}}},
        \typingContextTwo'
\typingContextCat
        \emptyset
        \typesSFs
        (\restr s \hasType \CSMabb{A})\,
        (\Parallel_{\procA \in \ProcsOf{\CSMabb{A}} \setminus \set{\procC}} Q_\procA)
            \parallel
        P_k
            \parallel
        \queueProc{s}
        {\queuecontent[(\procC, \procB) \mapsto \queuecontent(\procC, \procB) \cat \labelAndMsg{l_k}{v_k}]}
}
\end{mathpar}
}

We now argue why $(a)$ to $(e)$ hold.
\begin{itemize}
 \item $(a)$: This trivially holds since $(\vec{q}, \xi)$ was reachable and there is a transition from the latter to $(\pvec{q}', \xi')$.
 \item $(b)$: This is precisely the same premise obtained by inversion of \ref{sf-cond-2}.
 \item $(c)$: This follows from the premise obtained by inversion of \ref{sf-cond-2}. We solely do not need the fact for $Q_\procC$.
 \item $(d)$: This is one of the premises we obtained through inversion of the typing for $Q_\procC$ with typing rule \procTypingIntCh.
 \item $(e)$: For all channels different from $\channel{\procC}{\procB}$, we can build the typing derivation with the premises obtained by inversion of \ref{sf-cond-2}.
 For
 $\channel{\procC}{\procB}$,
 we observe that the message was appended, while the typing rules type message queues from the start.
 However, by applying
 \cref{lm:typing-message-list-reversal},
 we also obtain a respective typing derivation.
\end{itemize}
This concludes this case.

Second, let $x = \rcv{\procB}{\procC}{\labelAndType{l}{L}}$.
The proof is very similar but also differs in some places.
Thus, we spell it out for completeness.
Notably, we can always choose $(\pvec{q}', \xi') = (\pvec{q}'', \xi'')$ in this case since receives always need to be handled.

We do inversion on
\ref{sf-cond-2}:
\begin{mathpar}
\inferrule*[right=\runtimeTypingRestr ']{
    (\vec{q}, \xi) \in \reach(\CSMabb{A})
    \\
\meta{\forall \procA \in \ProcsOf{\CSMabb{A}} \st}
    \forall c \hasType q' \in \typingContextTwo_\procA \st
    \EndState(q')
    \\
\meta{\forall \procA \in \ProcsOf{\CSMabb{A}} \st}
        \typingContextOne
        \typingContextCat
        \typingContextTwo_\procA,
        s[\procA] \hasType \vec{q}_\procA
        \typesSFd
        Q_\procA
    \\
    \typingContextOne
    \typingContextCat
    \typingContextTwo'
    \typingContextCat
    \set{s[\procA][\procB] \hasQueueType \xi(\procA,\procB)}_{(\procA,\procB) \in \channelsOf{\CSMabb{A}}}
    \typesSFd
    \queueProc{s}{\queuecontent}
}{
    \typingContextOne
        \typingContextCat
        \set{\typingContextTwo_\procA}_{\procA \in \ProcsOf{\CSMabb{A}}},
\typingContextTwo'
        \typingContextCat
        \emptyset
        \typesSFs
        (\restr s \hasType \CSMabb{A})\,
        (\Parallel_{\procA \in \ProcsOf{\CSMabb{A}}} Q_\procA)
        \parallel
        \queueProc{s}{\queuecontent}
}
\end{mathpar}
and obtain all its premises
as well as the fact that
$R =
    (\Parallel_{\procA \in \ProcsOf{\CSMabb{A}}} Q_\procA)
    \parallel
    \queueProc{s}{\queuecontent}
$.
From the fact that $\vec{q}_\procC$ has outgoing transitions, we know that it is not final, which in turn means that $\EndState(\vec{q}_\procC)$ does not hold.
Hence,
two typing rules can apply:
\procTypingProcName and
\procTypingExtCh.

The case analysis here is analogous to the first case, eventually leading to an application of \procTypingExtCh of course.

Thus, we consider \procTypingExtCh as typing rule for $Q_\procC$.
By inversion, we have
\begin{mathpar}
\inferrule*[right=\procTypingExtCh]{
    \delta(q) =
    \set{(\rcv{\procB_i}{\procA}{\labelAndType{l_i}{L_i}}, q_i) \mid i \in I}
    \\
\meta{\forall i \in I \st}
    \typingContextOne \typingContextCat
        \hat{\typingContextTwo}_\procC,
        c \hasType q_i,
        y_i \hasType L_i
        \types P_i \\
}{
    \typingContextOne
    \typingContextCat
    \hat{\typingContextTwo}_\procC,
    s[\procC] \hasType q
        \types
    \ExtCh_{i \in I} s[\procC][\procB_i] ? \labelAndVar{l_i}{y_i} \seq P_i
}
\end{mathpar}
and obtain all premises as well as that
$
Q_\procC
=
\ExtCh_{i \in I} s[\procC][\procB_i] ? \labelAndVar{l_i}{y_i} \seq P_i
$ and
$
    \typingContextTwo_\procC
    =
    \hat{\typingContextTwo}_\procC,
    \set{c_i \hasType L_i}_{i \in I}
$.

Because of
$
    \delta(q) =
    \set{(\rcv{\procB_i}{\procA}{\labelAndType{l_i}{L_i}}, q_i) \mid i \in I}
$,
we know that there is $k \in I$ with
$l \neq l_i$.
We choose
\[
R' \is
    (\Parallel_{\procA \in \ProcsOf{\CSMabb{A}} \setminus \set{\procC}} Q_\procA)
        \parallel
    P_k[v_k / y_k]
        \parallel
    \queueProc{s}
    {\queuecontent[(\procB, \procC) \mapsto \vec{m}]}
\]
where
$\queuecontent(\procB, \procC) = \labelAndMsg{l_k}{v_k} \cat \vec{m}$.
Note that
$Q_k$ is restriction by
\ref{sf-cond-2}.
Thus, $P_k$ is restriction-free, which entails that
$\procToRuntime{P_k[v_k / y_k]} = P_k[v_k / y_k]$.
With
\runtimeReductionIn
and
\runtimeReductionContext,
it is straightforward to show that
$R \redto R'$.
It remains to show
\[
    \typingContextOne
        \typingContextCat
        \set{\typingContextTwo_\procA}_{\procA \in \ProcsOf{\CSMabb{A}}},
        \typingContextTwo'
        \typingContextCat
        \emptyset
        \overset{(\pvec{q}'\negthinspace, \, \xi')}{\typesSFs}
        (\restr s \hasType \CSMabb{A})\,
        (\Parallel_{\procA \in \ProcsOf{\CSMabb{A}} \setminus \set{\procC}} Q_\procA)
            \parallel
        P_k[v_k / y_k]
            \parallel
        \queueProc{s}
        {\queuecontent[(\procC, \procB) \mapsto \vec{m}]}
    \enspace
\]
We start building a typing derivation:

\vspace{-2ex}
{ \scriptsize
\begin{mathpar}
\inferrule*[right=\runtimeTypingRestr ']{
    (a):
    (\vec{q}', \xi') \in \reach(\CSMabb{A})
    \\
(b):
    \meta{\forall \procA \in \ProcsOf{\CSMabb{A}} \st}
    \forall c \hasType q' \in \typingContextTwo_\procA \st
    \EndState(q')
    \\
(c):
    \meta{\forall \procA \in \ProcsOf{\CSMabb{A}} \setminus \set{\procC} \st}
        \typingContextOne
        \typingContextCat
        \typingContextTwo_\procA,
        s[\procA] \hasType \pvec{q}'_\procA
            \typesSFd
        Q_\procA
    \\
(d):
    \typingContextOne
    \typingContextCat
    \hat{\typingContextTwo}_\procC,
    v_k \hasType L_k,
    s[\procC] \hasType \pvec{q}'_\procC
        \typesSFd
    P_k[v_k / y_k]
    \\
(e):
    \typingContextOne
    \typingContextCat
    \hat{\typingContextTwo}'
    \typingContextCat
    \set{s[\procA][\procB] \hasQueueType \xi(\procA,\procB)}_{(\procA,\procB) \in \channelsOf{\CSMabb{A}}}
    \typesSFd
    \queueProc{s}
    {\queuecontent[(\procC, \procB) \mapsto \vec{m}]}
}{
    \typingContextOne
        \typingContextCat
\set{\typingContextTwo_\procA}_{\procA \in \ProcsOf{\CSMabb{A}}},
        \typingContextTwo'
\typingContextCat
        \emptyset
        \typesSFs
        (\restr s \hasType \CSMabb{A})\,
        (\Parallel_{\procA \in \ProcsOf{\CSMabb{A}} \setminus \set{\procC}} Q_\procA)
            \parallel
        P_k[v_k / y_k]
            \parallel
        \queueProc{s}
        {\queuecontent[(\procC, \procB) \mapsto \vec{m}]}
}
\end{mathpar}
}

where
$
    \typingContextTwo'
    =
    \hat{\typingContextTwo}',
    v_k \hasType L_k
$,
which is possible as $v_k$ was in the message queue $\queuecontent(\procC, \procB)$ before.

We now argue why $(a)$ to $(e)$ hold.
\begin{itemize}
 \item $(a)$: This trivially holds since $(\vec{q}, \xi)$ was reachable and there is a transition from the latter to $(\pvec{q}', \xi')$.
 \item $(b)$: This is precisely the same premise obtained by inversion of \ref{sf-cond-2}.
 \item $(c)$: This follows from the premise obtained by inversion of \ref{sf-cond-2}. We solely do not need the fact for $Q_\procC$.
 \item $(d)$: This is almost one of the premises we obtained through inversion of the typing for $Q_\procC$ with typing rule \procTypingExtCh.
        We simply need to apply \cref{lm:substitution-lemma} to obtain the respective typing derivation.
 \item $(e)$: For all channels different from $\channel{\procC}{\procB}$, we can build the typing derivation with the premises obtained by inversion of \ref{sf-cond-2}.
 For
 $\channel{\procC}{\procB}$,
 we can use the same and simply skip one step for the message that it is not in the channel anymore.
\end{itemize}
This concludes this case and hence the whole proof.
\proofEndSymbol
\end{proof}

From subject reduction and session fidelity, deadlock freedom would easily follow.

\begin{restatable}[Deadlock freedom with sink-final FSMs]{lemma}{deadlockFreedomNoNonSinkFinal}
Let $\CSMabb{A}$ be a deadlock-free CSM that satisfies feasible eventual reception and, 
for every $\procA \in \ProcsOf{\CSMabb{A}}$, 
$\CSMabb{A}_\procA$ is sink-final. 
Let $P$ be a process. Assume that \begin{itemize}
    \item $\typesSFs \Defs \hasType \typingContextOne$,
    \item $\typingContextOne
            \typingContextCat
            \typingContextTwo
            \typesSFs
            (\restr s \hasType \CSMabb{A}) \, P$,
    \item $\CSMabb{A}_\procA$ is sink-final for every $\procA \in \ProcsOf{\CSMabb{A}}$, and
    \item $\procToRuntime{(\restr s \hasType \CSMabb{A}) \, P} \redto^* R$.
\end{itemize}
Then, it holds that $R \precongr \zero$ or there is $R'$ such that $R \redto R'$.
\end{restatable}
\textit{Proof Sketch.}
First, we claim that for all $k \geq 0$ with
$\procToRuntime{(\restr s \hasType \CSMabb{A}) \, P} \redto^k R$,
it holds that
$
\typingContextOne
    \typingContextCat
    \typingContextTwo
    \typingContextCat
    \emptyset
    \overset{(\vec{q}, \xi)}{\typesSFs}
    R
$
for some
$\vec{q}$ and $\xi$.
This can be shown using subject reduction \cref{thm:subject-reduction}.
We do a case analysis if there is
$
    (\pvec{q}', \xi')
$
with
$
    (\vec{q}, \xi)
        \redto
    (\pvec{q}', \xi')
$.

If so, we know from
\cref{lm:session-fidelity}
that there is $R'$ with $R \redto R'$, which concludes this case.

Suppose that there is no
$
    (\pvec{q}', \xi')
$
with
$
    (\vec{q}, \xi)
        \redto
    (\pvec{q}', \xi')
$.
By the assumption that $\CSMabb{A}$ is deadlock-free and, thus, all $\vec{q}$ are final states (and $\xi$ only has empty channels).

 \section{Additional Material for \cref{sec:reconstructing-global-types}}
\label{app:reconstructing-global-types}

We structure our formalisation in subsections, aligning with the key consequences 
\ref{item:key-consequence-a}
to 
\ref{item:key-consequence-d}, 
explained in the main text. 
Combining these and previous results, we will close this section by proving undecidability of the projectability problem and the strong projectability problem in the presence of mixed choice. 

\subsection{Additional Material for Consequence \ref{item:key-consequence-a}}

We present syntax and semantics of global types following \cite{DBLP:conf/ecoop/Stutz23}.

\begin{definition}[Syntax of global types]
\label{def:global-types}
\emph{Global types} for MSTs are defined by the grammar:
\vspace{-2ex}
    \begin{grammar}
     G \is
       \zero
     | \sum_{i ∈ I} \msgFromTo{\procA_i}{\procB_i}{\val_i} \seq G_i
     | μ t \seq G
     | t
    \end{grammar}
\vspace{-2ex}

The term $\zero$ explicitly represents termination while $\msgFromTo{\procA_i}{\procB_i}{\val_i}$ indicates an interaction where $\procA_i$ sends message $\val_i$ to $\procB_i$.
We assume $\card{I} > 0$ and, if~\mbox{$\card{I} = 1$}, we omit the sum operator.
The operators $\mu t$ and $t$ can be used to encode loops.
We require them to be guarded, i.e., there must be at least one interaction between the binding $\mu t$ and the use of the recursion variable~$t$.
Without loss of generality, all occurrences of recursion variables $t$ are bound and~distinct.
A global type satisfies \emph{mixed choice} if for each
syntactic subterm
$\sum_{i ∈ I} \msgFromTo{\procA_i}{\procB_i}{\val_i} \seq G_i$,
its branches are unique, \ie
$∀ i,j ∈ I.\, i≠j ⇒ \procA_i \neq \procA_j \lor \procB_i \neq \procB_j \lor \val_i ≠ \val_j$;
otherwise it is \emph{non-deterministic}.
Directed choice requires the sender and receiver to be the same but messages to be distinct for branches:
$\A i, j \in I.\, i \neq j ⇒ \procA_i = \procA_j \land \procB_i = \procB_j \land \val_i \neq \val_j$.
In contrast, sender-driven choice requires each receiver-message pair to be distinct:
$\A i,j ∈ I.\, i≠j ⇒ \procA_i = \procA_j \land (\procB_i \neq \procB_j \lor \val_i ≠ \val_j)$.
We may say that a global type is directed or sender-driven.
\end{definition}

\Cref{fig:example1-mst} represents a global type from MSTs.

\begin{definition}[Semantics of global types]
\label{def:language-global-mst}
Let $\GG$ be a global type.
We index every syntactic subterm of $\GG$ with a unique index to distinguish common syntactic subterms, denoted with $[G, k]$ for syntactic subterm $G$ and index $k$.
Without loss of generality, the index for $\GG$ is $1$\emph{:} $[\GG, 1]$.
We define
$\semglobalsync(\GG) = (Q_{\semglobalsync(\GG)}, \AlphSync, δ_{\semglobalsync(\GG)}, q_{0, \semglobalsync(\GG)}, F_{\semglobalsync(\GG)})$ where\begin{itemize}
\item $Q_{\semglobalsync(\GG)}$ is the set of all indexed syntactic subterms $[G, k]$ of $\GG$, 
\item $δ_{\semglobalsync(\GG)}$ is the smallest set containing \\
            $(
            [\sum_{i ∈ I} \msgFromTo{\procA_i}{\procB_i}{\val_i.[G_i, k_i]}, k],
            \msgFromTo{\procA_i}{\procB_i}{\val_i},
            [G_i, k_i]
            )$ for~each $i ∈ I$, \\
            $([μ t. [G',k'_2], k'_1], ε, [G', k'_2])$ and $([t, k'_3], ε, [μ t. [G', k'_2], k'_1])$, \item $q_{0, \semglobalsync(\GG)} = [\GG, 1]$, and
$F_{\semglobalsync(\GG)} = \set{[\zero, k] \mid k \text{ is an index for subterm } \zero}$.
\end{itemize}
We obtain the semantics using $\interswap$:
$
 \lang(\GG)
    \is
 \interswaplang(\lang(\semglobalsync(\GG)))
$.
\end{definition}

For every global type $\GG$, $\semglobalsync(\GG)$, when viewed as a PSM, satisfies a number of properties, which were defined in 
\cite[Def.\,3.5]{DBLP:conf/ecoop/Stutz23}.

\begin{definition}[Ancestor-recursive, non-merging, intermediate recursion, etc.]
\label{def:anc-rec-non-merging-etc}
Let $A = (Q, \Delta, \delta, q_{0}, F)$ be a finite state machine.
For convenience, we write $q \rightarrow q'$ if $q \xrightarrow{x} q'$ for some $x \in \Delta$.
We say that $A$ is \emph{ancestor-recursive} if there is a function $\levelfunc \from Q \to \Nat$ such that, for every transition $q \xrightarrow{x} q' \in \delta$, one of the two holds:
\begin{enumerate}[labelindent=0pt,labelwidth=\widthof{(ii)},label=\textnormal{(\roman*)},itemindent=0em,leftmargin=!]
 \item $\levelfunc(q) > \levelfunc(q')$, or
 \item $x = \emptystring$ and there is a run from the initial state $q_0$ (without going through $q$) to $q'$ which can be completed to reach $q$:
$q_0 \rightarrow \ldots \rightarrow q_n$ is a run with $q_n = q'$ and $q \neq q_i$ for every $0 \leq i \leq n$, and
 the run can be extended to
 $q_0 \rightarrow \ldots \rightarrow q_n \rightarrow \ldots \rightarrow q_{n+m}$ with $q_{n+m} = q$.
Then, the state $q'$ is called \emph{ancestor} of $q$.
\end{enumerate}
We call the first \textnormal{(i)} kind of transition \emph{forward transition} while the second \textnormal{(ii)} kind is a \emph{backward transition}.
The state machine $A$ is said to be free from \emph{intermediate recursion} if every state $q$ with more than one outgoing transition, i.e.,
$
\card{\set{q' \mid q \rightarrow q' \in \delta}} > 1
$,
has only forward transitions.
We say that $A$ is \emph{non-merging} if every state only has one incoming edge with greater level, i.e., for every state $q'$, $\set{q \mid q \rightarrow q' \in \delta \land \levelfunc(q) > \levelfunc(q')} \leq 1$.
\end{definition}

\citet[Prop.\,3.6]{DBLP:conf/ecoop/Stutz23} show that state machines for sender-driven global types satisfy the above properties.
It is straightforward that this also holds for mixed-choice and non-deterministic global types.

\begin{proposition}
\label{prop:global-type-as-psm}
Let $\GG$ be a global type.
Then, $\semglobalsync(\GG)$ is a \sinkfinal, ancestor-recursive, non-merging, dense \sumOnePSM without intermediate recursion.
If $\GG$ is non-deterministic, mixed-choice, sender-driven or directed, so is $\semglobalsync(\GG)$.
\end{proposition}

\subsection{Additional Material for Consequence 
\ref{item:key-consequence-b}}
\label{app:tree-transformation}

Every global type's state machine is a \sinkfinal tree-like \sumOnePSM, when viewed as a PSM.
This raises an obvious question.
For which kind of PSMs can we have global types that have the same (core) language?
And can we preserve the various restrictions on choice, \eg sender-driven choice, for those?
Such preservation is particularly interesting in the light of our undecidability result for projectability of mixed-choice \sumOnePSMs.
It is immediate that we can only achieve this for \sumOnePSMs because we consider the core language and not the semantics.

\paragraph{Anti-patterns.}
Visually, we want to transform a \sumOnePSM with an arbitrary structure into a tree-like structure where recursion only happens at leaves and to ancestors.
There are several anti-patterns one needs to consider: \eg mutual recursion, intermediate recursion, and merging.
As is standard for a tree-like shape, we assume states have different levels:
the initial state has the highest level and the level usually decreases when taking a transition.
If not, it is considered recursion.
First, recursion is supposed to lead to an ancestor, \ie a state from which the state itself can be reached without increasing the level again; if this is not the case, we call this \emph{mutual recursion}.
Second, recursion should happen at a leaf, \ie there is no outgoing transition to a state with smaller level;
if there is such a transition, there is \emph{intermediate recursion}.
Third, every state ought to have at most one incoming transition;
if not, we call this~\emph{merging}.

\paragraph{Naive approach breaks choice restrictions.}
It is rather easy to remove these anti-patterns by duplicating various parts of the \sumOnePSM.
This will easily introduce non-determinism though, defeating the goal of preserving restrictions on choice.
With this in mind, the problem becomes significantly more challenging.

\paragraph{Overview of our workflow.}
We develop a workflow that transforms \sumOnePSMs to tree-like \sumOnePSMs with recursion at leaves and to ancestors, which are easily turned into global types.
Let us first give a very high-level overview and report on technical challenges.
The key insight to establish the desirable properties is the use of regular expressions as intermediate representations.
It is well-known that Arden's Lemma \cite{DBLP:conf/focs/Arden61} can be used to transform an FSM into a regular expression,
but it produces one regular expression for every final state.
We flip Arden's lemma, prove it correct, and use it to produce a regular expression for the (only) initial state.
This is only sound as we solely consider \sinkfinal \sumOnePSMs.
(Our results can only preserve the restrictions of choice for \sinkfinal \sumOnePSMs, which is reasonable because global types as PSMs are always \sinkfinal.)
To prove the preservation of choice restrictions, we also define these for regular expressions over $\AlphAsync$, inspired by deterministic regular expressions by \citet{DBLP:journals/iandc/Bruggemann-KleinW98a}.
Let us now explain how regular expressions help to establish the desirable properties.
Intuitively, one can traverse a regular expression bottom-up and generate an FSM for the same language.
The expression for alternative $\regex[1] + \regex[2]$ becomes a branch to the two respective FSMs for $\regex[1]$ and $\regex[2]$.
For concatenation $\regex[1] \cat \regex[2]$, we simply connect both FSMs, and for Kleene star $\regex^*$ we make the initial state final and add transitions from the final to the initial state.
While this gives an idea of our approach, such treatment still introduces (undesirable) non-determinism to connect different FSMs.
To avoid this, we employ Brzozowski derivatives \cite{DBLP:journals/jacm/Brzozowski64}, but adapt them to PSMs.
They allow us to pull the first event out so we can use labelled transitions to connect the FSMs.
Of course, we also prove that these PSM derivatives preserve the restrictions on choice.

\noindent
\begin{minipage}{\textwidth}
Our workflow comprises the following steps:
\begin{enumerate}[label=\textnormal{(\arabic*)}]
 \addtocounter{enumi}{-1}
 \item \label{wf:psm-sink-state-iff-final-state}
        make the PSM \sinkfinal for the price of introducing non-determinism \item \label{wf:psm-to-regex}
        compute a regular expression for the initial state of the \sinkfinal PSM
 \item \label{wf:regex-to-arnmdirf-psm}
        convert regular expression to a PSM that is ancestor-recursive, non-merging, dense, and intermediate-recursion-free 
 \item \label{wf:arnmdfir-psm-to-global-type}
        if the original PSM is a \sumOnePSM,
        transform the result from the previous step to a global type
\end{enumerate}
\end{minipage}

\medskip
Without loss of generality, we assume that every sink state is final: 
any state, for which this is not the case, can simply be removed while preserving the core language and semantics of a PSM.

For the last step, we only consider \sumOnePSMs because global types always jointly specify send and receive events.

\begin{remark}Our constructions do also apply to FSMs over other alphabets.
For these, the reasoning about preserving sender-driven choice often translates to preserving determinism.
Thus, it shows that the above structural conditions do not change expressivity for \sinkfinal deterministic FSMs.
For FSMs for participants of protocols, this establishes a connection to local types, as shown later.
\end{remark}

\subsubsection*{Step \ref{wf:psm-sink-state-iff-final-state}: Sink State iff Final State.}
\label{sec:psm-sink-state-iff-final-state}

This can be considered to be a preprocessing step for PSMs that are not \sinkfinal, making the workflow more general.
This transformation step simply introduces a new final sink state to which transitions can lead non-deterministically.

We give a construction with a single fresh final state, which can be (non-determinis\-tically) reached instead of any previous final state.

\begin{procedure}[PSM: Sink State iff Final State]
\label{proc:sink-state-iff-final-state}
Let $\PSM = (Q, \AlphAsync, \delta, q_{0}, F)$ be a PSM with $\emptystring \notin \semantics(\PSM)$.
We define a function that turns $\PSM$ into a \sinkfinal PSM:
\[
\psmToSinkFinalFunc(\PSM) \is (Q \dunion \set{q_f}, \AlphAsync, \delta', q_0, \set{q_f})
\]
where
$(q_1, x, q_2) \in \delta'$ if $(q_1, x, q_2) \in \delta$ as well as
$(q_1, x, q_f) \in \delta'$ if $(q_1, x, q_2) \in \delta$ and $q_2 \in F$.
\end{procedure}

The condition that $\emptystring \notin \semantics(\PSM)$ ensures that there is a predecessor state for every final state to which we can add the transition.

\begin{proposition}
\label{prop:psmToSinkFinalFuncCorrect}
Let \sumOnePSM be a $\PSM$ such that $\emptystring \notin \semantics(\PSM)$.
Then, the PSM $\psmToSinkFinalFunc(\PSM)$ is \sinkfinal.
\end{proposition}

It is obvious that this construction introduces and, thus, does not preserve sender-driven choice.

\subsubsection*{Step \ref{wf:psm-to-regex}: From Sink-final PSMs to Regular Expressions.}
\label{sec:psm-to-regex}

This transformation step translates a sink-final PSM to a regular expression over~$\AlphAsync$ that specifies the same core language.
It is well-known that this can be done using Arden's Lemma~\cite{DBLP:conf/focs/Arden61}.
We cannot apply the standard technique though, as it would produce as many regular expressions as final states.
Such treatment makes it very hard to argue about the preservation of sender-driven choice.
Instead, we exploit the fact that the PSM is \sinkfinal and produce a single regular expression for the initial state.
This also enables the treatment of infinite words, which solely require an infinite run that necessarily does not end in a final state.

We define regular expressions and include infinite words in their semantics.

\begin{definition}[Regular Expressions]
Let $\StdAlphabet$ be an alphabet.
Regular expressions (REs) over $\StdAlphabet$ are inductively defined by the following grammar where $a \in \StdAlphabet$:
\begin{grammar}
    \regex
        \is
    \emptystring |
    a |
    \regex + \regex |
    \regex \cat \regex |
    \regex^*
\end{grammar}
The concatenation operator $\cat$ has precedence over $+$.
We define
$\regexfinlang{a} = \set{a}$,
$\regexfinlang{\regex[1] + \regex[2]} =
    \regexfinlang{\regex[1]} \union \regexfinlang{\regex[2]}$,
$\regexfinlang{\regex[1] \cat \regex[2]} =
    \set{w_1 \cat w_2 \mid
        w_1 \in \regexfinlang{\regex[1]},
        w_2 \in \regexfinlang{\regex[2]}}$,
and
$\regexfinlang{\regex^*} =
    \set{w_1 \ldots w_n \mid n \in \Nat, \forall i \leq n
    \st w_i \in \regexfinlang{\regex}}$.
The infinite language $\regexinflang{\regex}$ is defined as
$\set{w \in \StdAlphabet^\omega \mid \A w' \in \pref(w) \st w' \in \pref(\regexfinlang{\regex})}$.
The language $\lang(\regex)$ is the union of $\regexfinlang{\regex}$ and $\regexinflang{\regex}$.
The function $\symbolsFunc(\regex)$ is the set of all letters in $\regex$, \ie the smallest subset $\StdAlphabet' \subseteq \StdAlphabet$ such that $\lang(\regex) \subseteq (\StdAlphabet')^\infty$.
We denote the set of all regular expressions over $\StdAlphabet$ with $\regexs_\StdAlphabet$.
\end{definition}

Instead of constructing the regular expressions for final states, as is standard with Arden's Lemma, we construct one for the initial state.
This is sound because a state is a sink if and only if it is final.
It also lets us handle infinite words.
For a state machine, an infinite word is part of its semantics if there is an infinite run.
Here, we mimic this: an infinite word is in the semantics of a regular expression if every prefix of the word is a prefix of a word in the finite~semantics.

\begin{restatable}[Arden's Lemma -- swapped]{lemma}{ardenSwapped}
\label{lm:Arden-swapped}
Let $\regex[1]$ and $\regex[2]$ be two regular expressions over an alphabet~$\StdAlphabet$.
If $\regex[1]$ does not contain the empty string,
\ie $\emptystring \not\in \regexfinlang{\regex[1]}$, then
$\regex[3] = \regex[2] + (\regex[1]\cat\regex[3])$ has a unique solution that is $\regex[3] = \regex[1]^*\cat\regex[2]$.
\end{restatable}
\begin{proof}
The proof is analogous to the original one:
\begin{align*}
 \regex[3] & = \regex[2] + \regex[1] \cat \regex[3] \\
           & = \regex[2] + \regex[1] \cat (\regex[2] + \regex[1] \cat \regex[3]) \\
           & = \regex[2] + \regex[1] \cat \regex[2] + \regex[1] \cat \regex[1] \cat \regex[3] \\
           & = \ldots \\
           & = \regex[2] + \regex[1] \cat \regex[2] + \regex[1]^2 \cat \regex[3]+ \regex[1]^3 \cat \regex[3] + \ldots  \\
           & = (\emptystring + \regex[1] + \regex[1]^2 + \regex[1]^3 + \ldots) \cat \regex[2]  \\
           & = \regex[1]^* \cat \regex[2]
\end{align*}
\proofEndSymbol
\end{proof}

For sender-driven PSMs, we want to show that sender-driven choice is preserved.
Therefore, we need a notion of sender-driven choice for regular expressions.
We define this following work on deterministic regular expressions \cite{DBLP:journals/iandc/Bruggemann-KleinW98a}.

\newcommand{\markFunc}{\texttt{\upshape mark}}
\newcommand{\unmarkFunc}{\texttt{\upshape unmark}}
\begin{definition}[Marking and unmarking regular expressions]
Let $\StdAlphabet$ be an alphabet and $\regex \in \regexs_\StdAlphabet$ be a regular expression.
We define a function $\markFunc(\regex)$ that simply subscripts every letter in $\regex$ with a distinct index and the inverse function $\unmarkFunc(\regex)$, which is also defined for words over $\StdAlphabet$.
\end{definition}

\begin{definition}[Mixed-choice, sender-driven and directed regular expressions]
Let $\regex \in \regexs_{\AlphAsync}$.
We say that $\regex$ is a \emph{sender-driven} regular expression if the following holds:
for every $u \in \AlphAsync^*$ and $x, y \in \AlphAsync$,
if $ux \in \pref(\lang(\markFunc(\regex)))$, $uy \in \pref(\lang(\markFunc(\regex)))$ and $x \neq y$, then
$\unmarkFunc(x) \neq \unmarkFunc(y)$
as well as
$\unmarkFunc(x) \in \set{\snd{\procA}{\procB}{\_} \mid \procB \in \Procs}$
and
$\unmarkFunc(y) \in \set{\snd{\procA}{\procB}{\_} \mid \procB \in \Procs}$
for some $\procA \in \Procs$.
For \emph{directed choice}, we also require $\procB$ to be the same for both $x$ and $y$ and, for \emph{mixed choice}, we solely require $\unmarkFunc(x) \neq \unmarkFunc(y)$.
\end{definition}

Compared to deterministic regular expressions, our definition requires the special alphabet $\AlphAsync$ (and adds a condition for sender-driven and directed choice).

\begin{proposition}
Every mixed-choice, sender-driven or directed RE is a deterministic~RE.
\end{proposition}

\newcommand{\firstFunc}{\texttt{\upshape first}}
\newcommand{\followFunc}{\texttt{\upshape follow}}

\begin{definition}
Let $\StdAlphabet$ be an alphabet and $L \subseteq \StdAlphabet^\infty$.
We define a function that collects all first letters of $L$:
$\firstFunc(L) \is \pref(L) \inters \StdAlphabet$.
The function $\followFunc(L, a)$ collects all letters that can occur after $a$ in $L$:
$\followFunc(L, a) \is \set{b \mid wab \in \pref(L)}$
and
$\followFunc(L, \emptystring) \is \firstFunc(L)$.
\end{definition}

The following lemma follows from a straightforward adaption of Lemma 2.2 by
\citet{DBLP:journals/iandc/Bruggemann-KleinW98a}.

\begin{lemma}
An RE $\regex \in \regexs_{\AlphAsync}$ is a  sender-driven RE if and only if,
for every $z \in \symbolsFunc(\markFunc(\regex)) \dunion \set{\emptystring}$ and
every $x, y \in \followFunc(\markFunc(\regex), z)$,
if $x \neq y$,
then $\unmarkFunc(x) \neq \unmarkFunc(y)$,
as well as
$\unmarkFunc(x) \in \lang(\AlphAsync_\procA)$
and
$\unmarkFunc(y) \in \lang(\AlphAsync_\procA)$
for some $\procA \in \Procs$.
\end{lemma}

Intuitively, one can check if an RE over $\AlphAsync$ is an sender-driven RE as follows.
For every subexpression of the form $\regex[1] + \regex[2]$ and $\regex[1]^* \cat \regex[2]$,
the REs $\regex[1]$ and $\regex[2]$ should not share any first letters and the union of their first letters belongs to the same participant.
It suffices to consider these operators as these are the only ones where lookahead to take a decision about the path in the RE is needed.

\begin{procedure}[PSM to RE]
Let \,$\PSM = (Q, \AlphAsync, \delta, q_{0}, F)$ be a \sinkfinal PSM.
We generate a system of equations.
For every $q_1 \in Q$, we introduce $\regex[\negthinspace q_1]$ as follows:
\[
    \regex[\negthinspace q_1] =
        \sum\limits_{(q_1, x, q_2) \in \delta} x \cat \regex[\negthinspace q_2]
\]
Given the initial state $q_{0}$, we can solve the system of equations for $\regex[\negthinspace q_0]$ with \cref{lm:Arden-swapped}, yielding a regular expression $\psmToRegexFunc(\PSM)$.
\end{procedure}

The following lemma states the correctness of the previous procedure.

\begin{restatable}{lemma}{psmToRegexFuncCorrect}
\label{lm:psmToRegexFuncCorrect}
For every \sinkfinal PSM $\PSM$,
it holds that $\lang(\psmToRegexFunc(\PSM)) = \lang(\PSM)$.
If\, $\PSM$ is a sender-driven PSM, then $\psmToRegexFunc(\PSM)$ is a sender-driven RE.
If $\emptystring \notin \semantics(\PSM)$,
then $\emptystring$ does not occur in
$\psmToRegexFunc(\PSM)$.
\end{restatable}
\begin{proof}
With the assumption that $\PSM$ is \sinkfinal, the first claim easily follows from \cref{lm:Arden-swapped}.
For the second claim, let us investigate how the system of equations, from which $\psmToRegexFunc(\PSM)$ is obtained, is solved.
We observe that every equation is guarded, \ie there is a letter from $\AlphSync$ before an occurrence of $\regex[q]$ for some state $q$.
Solving the system of equations for the initial state $\regex[q_0]$ can only involve substitution and the application of \cref{lm:Arden-swapped}.
For both, sender-driven choice of $\PSM$ is preserved for the RE across all equations, yielding a sender-driven RE for~$\regex[q_0]$.
For the third claim, it suffices to observe that no $\emptystring$ is introduced in the system of equations.
\proofEndSymbol
\end{proof}

\subsubsection*{Step \ref{wf:regex-to-arnmdirf-psm}: From Regular Expressions to Ancestor-recursive Non-merging Dense Intermediate-recursion-free PSMs.}
\label{sec:regex-to-arnmdirf-psm}

\mbox{After} the transformation, we want the PSM to be ancestor-recursive, non-merging, dense and intermediate-recursion-free by construction.
We need to carefully design this transformation because the standard approach introduces non-determinism, for instance for union.
Resolving this non-determinism would easily break the desired structural properties, making the whole workflow pointless.
We apply the idea of derivatives in order not to introduce non-determinism.
To preserve sender-driven choice, we also ensure that sender-driven regular expressions are closed under Brzozowski Derivatives.
Given a regular expression $\regex$ and a letter $a$, we can use them to construct a regular expression that specifies the language of words in the semantics of $\regex$ which start with $a$ and omits $a$.
We apply a similar idea to PSMs in order not to introduce non-determinism when constructing PSMs from regular expressions.

\newcommand{\regDer}{\operatorname{brz-deriv}}
\newcommand{\psmDer}{\operatorname{psm-deriv}}

\begin{definition}[\cite{DBLP:journals/jacm/Brzozowski64}]
Let $\StdAlphabet$ be an alphabet.
We define the Brzozowski derivative
$\regDer \from \StdAlphabet \times \regexs_\StdAlphabet \to \regexs_\StdAlphabet$
as follows:
{
\small
\[
    \regDer(a, \regex) \is
        \begin{cases}
            \emptystring
                & \text{if } \regex = a \\
            \regDer(a, \regex[1]) + \regDer(a, \regex[2])
                & \text{if } \regex = \regex[1] + \regex[2] \\
            \regDer(a, \regex[1]) \cat \regex[2]
                & \text{if } \regex = \regex[1] \cat \regex[2] \land \emptystring \notin \lang(\regex[1]) \\
            \regDer(a, \regex[1]) \cat \regex[2] +
            \regDer(a, \regex[2])
                & \text{if } \regex = \regex[1] \cat \regex[2] \land \emptystring \in \lang(\regex[1]) \\
            \regDer(a, \regex[1]) \cat \regex[1]^*
                & \text{if } \regex = \regex[1]^* \\
            \text{undefined}
                & \text{otherwise}
        \end{cases}
\]
}
\end{definition}

\begin{lemma}[Correctness of Brzozowski Derivatives \cite{DBLP:journals/jacm/Brzozowski64}]
\label{lm:correctness-brz-deriv}
Let $\regex$ be a regular expression over an alphabet $\StdAlphabet$ and $a \in \StdAlphabet$ be a letter.
If \;$\regDer(a, \regex)$ is defined, it holds that
\[
    \regexfinlang{\regDer(a, \regex)} =
    \set{w \mid aw \in \regexfinlang{\regex}}
    \enspace .
\]
If \,$\regDer(a, \regex)$ is not defined, it holds that $\set{w \mid aw \in \regexfinlang{\regex}} = \emptyset$.
\end{lemma}

We extend this result to infinite words.

\begin{lemma}[Brzozowski Derivatives for infinite words]
\label{lm:correctness-brz-deriv-inf}
Let $\regex$ be a regular expression over an alphabet $\StdAlphabet$ and $a \in \StdAlphabet$ be a letter.
If \;$\regDer(a, \regex)$ is defined, it holds that
\[
    \regexinflang{\regDer(a, \regex)} =
    \set{w \mid aw \in \regexinflang{\regex}}
    \enspace .
\]
If \,$\regDer(a, \regex)$ is not defined, it holds that $\set{w \mid aw \in \regexinflang{\regex}} = \emptyset$.
\end{lemma}
\begin{proof}
For the first claim,
we consider infinite words.
By definition, an infinite word is in a language if all its prefixes are a prefix of some word in the finite language.
Thus, it suffices to show that
$
    \pref(\regexfinlang{\regDer(a, \regex)})
    =
    \pref(\set{w \mid aw \in \regexinflang{\regex}} \inters \StdAlphabet^*)
$.
By definition, the prefixes of infinite words and finite words are the same for a regular expression.
Thus, it remains to show that
\[
    \pref(\regexfinlang{\regDer(a, \regex)})
    =
    \pref(\set{w \mid aw \in \regexfinlang{\regex}} \inters \StdAlphabet^*)
    \enspace .
\]
This follows from
\cref{lm:correctness-brz-deriv}.

The second claim simply follows from
\cref{lm:correctness-brz-deriv} and the definition of $\regexinflang{\hole}$, which requires $\regexfinlang{\hole}$ to be non-empty.
\proofEndSymbol
\end{proof}

From the correctness of Brzozowski derivatives, this observation follows immediately.

\begin{corollary}
\label{cor:correctness-brz-deriv}
Let $\regex$ be a regular expression over an alphabet $\StdAlphabet$ and $D \subseteq \StdAlphabet$ be the first letters in words in~$\regex$, \ie
$D \is \set{a_1 \mid a_1 \ldots a_n \in \regexfinlang{\regex}}$.
Then, it holds that
\[
    \regexfinlang{\regex} = \Dunion_{a \in D}
    \set{a \cat w \mid w \in \regexfinlang{\regDer(a, \regex)}}
    \enspace .
\]
\end{corollary}

Intuitively, we pull out every first letter for union and concatenation to avoid the introduction of $\emptystring$-transitions.
For this to work, we need to introduce a PSM derivative (function).
If we used the Brzozowski Derivative, we could not apply structural induction to prove equivalence of the regular expression and the PSM.
Still, we show sender-driven choice is preserved by the Brzozowski~Derivative.

\begin{restatable}{lemma}{regexWithChoiceClosedUnderBrzDerivative}
\label{lm:regex-with-choice-closed-under-brz-derivative}
Let $\regex$ be an sender-driven RE
and $x \in \firstFunc(\lang(\regex))$.
Then, it holds that $\regDer(x, \regex)$ is a sender-driven RE.
\end{restatable}
\begin{proof}
\citet{DBLP:journals/iandc/Bruggemann-KleinW98a} show that deterministic REs, which they call \mbox{1-unambiguous}, are closed under the Brzozowski derivative.
Their result generalises to sender-driven REs.
They define star normal form for regular expressions \cite[Def.\,3.3]{DBLP:journals/iandc/Bruggemann-KleinW98a}.
They recall that deterministic REs can always be specified by an RE in star normal form~\cite{DBLP:conf/latin/Bruggemann-Klein92}.
With \cite[Thm.\,B]{DBLP:journals/iandc/Bruggemann-KleinW98a}, they show that the Brzozowski derivative of a deterministic RE in star normal form is again deterministic and in star normal form.
The conditions on sender-driven choice for REs do not restrict representability in star normal form and, thus, the result generalises to sender-driven REs.
\proofEndSymbol
\end{proof}

The last ingredient for our transformation is a procedure that applies the derivative to PSMs, preserving the properties of interest.

\begin{lemma}[PSM for derivative]
\label{lm:psm-for-derivative}
Let $\PSM = (Q, \AlphAsync, \delta, q_{0}, F)$ be a PSM that is
ancestor-recursive, non-merging, dense, intermediate-recursion-free, and \sinkfinal.
and let $a \in \firstFunc(\lang(\PSM))$.
Then, there is a PSM, $\psmDer(a, \PSM)$, such that
$\lang(\psmDer(a, \PSM')) = \regDer(a, \lang(\PSM))$,
which is ancestor-recursive, non-merging, dense, intermediate-recursion-free, and \sinkfinal.
If $\PSM$ is sender-driven, $\psmDer(a, \PSM)$ is sender-driven.
\end{lemma}
\begin{proof}
Let $q_0$ be the initial state of $\PSM$,
$Q_1$ be the states with incoming transitions from $q_0$,
and $Q_2$ be the states with outgoing transition to $q_0$.
(Without loss of generality, we can assume that these are $\emptystring$-transitions.)
Formally
$Q_1 \is \set{q_1 \mid (q_0, \_, q_1) \in \delta}$ and
$Q_2 \is \set{q_2 \mid (q_2, \emptystring, q_0) \in \delta}$.
By assumption that
$a \in \firstFunc(\lang(\PSM))$,
we know that there is $q_1 \in Q_1$ such that $(q_0, a, q_1) \in \delta$.
We construct $\psmDer(a, \PSM)$ as follows.
We take $q_1$ as its initial state and do only keep the states for which $q_1$ is an ancestor.
For every state $q_2$ in $Q_2$, which was not deleted, we copy the original PSM~$\PSM$, remove the state $q_2$ and replace it by the initial state from the copy.
By assumption that the original PSM $\PSM$ is ancestor-recursive, non-merging, dense, intermediate-recursion-free, and \sinkfinal, this construction yields a PSM $\psmDer(a, \PSM)$ with the same properties and $\lang(\psmDer(a, \PSM)) = \regDer(a, \lang(\PSM))$.
By construction, $\psmDer(a, \PSM)$ is sender-driven if $\PSM$ is sender-driven.
\proofEndSymbol
\end{proof}

With this, we can provide the procedure that translates REs to PSMs.

\begin{procedure}[RE to PSM]
Given a regular expression $\regex$ without $\emptystring$ over $\AlphAsync$, we inductively construct the PSM $\regexToPsmFunc(\regex)$:
    \begin{enumerate}
\item $a$: one initial state with one transition labelled $a$ to one final state;
    \item $\regex[1] + \regex[2]$:
        We add one initial state.
        We do the same for both $\regex[1]$.
        We describe it for $\regex[1]$: we compute
        $\psmDer(a, \regexToPsmFunc(\regex[1]))$
        for every first letter $a$
        and add a transition labelled with the letter from the initial state to the initial state of the PSM.
\item $\regex[1] \cat \regex[2]$:
        We construct the FSM for $\regex[1]$.
        For $\regex[2]$, we apply the derivative idea again:
        we compute $\psmDer(a, \regexToPsmFunc(\regex[2]))$ for every first letter $a$ and
        copy each FSM as often as there are final states in the FSM for $\regex[1]$ and add transitions from each such final state.
    \item $\regex^*$:
        For Kleene Star, we construct the FSM for the inner regex and connect the final state(s) with the initial one by an $\emptystring$-transition and make the initial one final.
        (These are backward transitions and, thus, should to be labelled $\emptystring$ for the PSM to be dense.)
    \end{enumerate}
\end{procedure}

We prove the previous procedure to be correct.

\begin{restatable}{lemma}{regexToPsmFuncCorrect}
\label{lm:regexToPsmFuncCorrect}
Let \,$\regex$ be a regular expression over $\AlphAsync$ without $\emptystring$ and $\regexToPsmFunc(\regex)$ be a PSM.
Then, the core language of both are the same:
$\lang(\regex) = \lang(\regexToPsmFunc(\regex))$.
Also,
$\regexToPsmFunc(\regex)$
is
ancestor-recursive, non-merging, intermediate-recursion-free, dense, and \sinkfinal.
If \,$\regex$ is a sender-driven RE, then $\regexToPsmFunc(\regex)$ is a sender-driven PSM.
\end{restatable}
\begin{proof}
We prove the claims by induction on the structure of the regular expression $\regex$.

The case for a single letter $\regex = a$ is obvious.

Let $\regex = \regex[1] \cat \regex[2]$.
We first show that
$\lang(\regex[1] \cat \regex[2]) = \lang(\regexToPsmFunc(\regex[1] \cat \regex[2]))$.
The PSM construction applies the PSM derivative
$\psmDer(a, \regexToPsmFunc(\regex[2]))$
for every $a \in \firstFunc(\lang(\regexToPsmFunc(\regex[2]))$ and copies the resulting PSM for every final state of
$\regexToPsmFunc(\regex[1])$ and adds a transition with label $a$.
Thus, for every word $w$ in $\regexToPsmFunc(\regex[1] \cat \regex[2])$,
we have that
\begin{itemize}
 \item $w \in \lang(\regexToPsmFunc(\regex[1])) \inters \AlphAsync^\omega$ or
 \item $w = u \cat a \cat v$ with
        $u \in \lang(\regexToPsmFunc(\regex[1]))$,
        $a \in \firstFunc(\lang(\regexToPsmFunc(\regex[2])))$,
        and $v \in \lang(\psmDer(a, \regexToPsmFunc(\regex[2])))$.
\end{itemize}
By induction hypothesis, we have that
$\lang(\regex[1]) = \lang(\regexToPsmFunc(\regex[1]))$
and
$\lang(\regex[2]) = \lang(\regexToPsmFunc(\regex[2]))$.
By \cref{lm:psm-for-derivative},
$
\lang(\psmDer(a, \regex[2]))
=
\lang(\regDer(a, \regex[2]))
$.
Hence, we obtain:
\begin{itemize}
 \item $w \in \lang(\regex[1]) \inters \AlphAsync^\omega$ or
 \item $w = u \cat a \cat v$ with
        $u \in \lang(\regex[1])$,
        $a \in \firstFunc(\lang(\regex[2]))$,
        and $v \in \lang(\regDer(a, \regex[2]))$.
\end{itemize}

By the semantics of regular expressions and
\cref{lm:correctness-brz-deriv,lm:correctness-brz-deriv-inf},
it follows that
$w \in \lang(\regex[1] \cat \regex[2])$, which shows language equality.

By induction hypothesis, we know that
$\regexToPsmFunc(\regex[1])$ and
$\regexToPsmFunc(\regex[2])$
are ancestor-recursive, non-merging, dense, intermediate-recursion-free, and \sinkfinal.
By \cref{lm:psm-for-derivative},
for every
$a \in \firstFunc(\lang(\regexToPsmFunc(\regex[2])))$,
it holds that
$\psmDer(a, \regexToPsmFunc(\regex[2]))$
is
ancestor-recursive, non-merging, intermediate-recursion-free, dense, \sinkfinal and sender-driven if $\regexToPsmFunc(\regex[2])$ is.
Thus, by construction, the PSM
$\regexToPsmFunc(\regex[1] \cat \regex[2])$
is
ancestor-recursive, non-merging, dense, intermediate-recursion-free, and \sinkfinal,
where the multiple copies ensure ancestor-recursiveness and non-merging property and the derivatives preserve density;
and sender-driven choice is also preserved.

For $\regex[1] + \regex[2]$, the construction applies the PSM derivative to avoid introducing non-determinism (as is common in standard constructions for FSMs from REs) if it was not present before.
If it was there before, it preserves it to avoid subsequent merging and, thus, avoids introducing non-sink final states.
For sender-driven choice, the assumption for the regular expression yields that the first letters are pair-wise distinct and, thus, the newly introduced branching satisfies the sender-driven choice condition for PSMs.
The remaining reasoning is very similar to the previous case for concatenation and, hence, omitted.
Here, both $\regex[1]$ and $\regex[2]$ are treated the same, like the second part of concatenation.

For $\regex[1]^*$, we simply introduce a backward transition, which ought to be labelled by $\emptystring$ and it is.
In fact, these are the only $\emptystring$-transition, ensuring that the PSM is dense.
Note that we construct a PSM without forward transitions that are labelled with $\emptystring$.
This is different from state machines for global types where every subterm of shape $\mu t \seq G$ has only one incoming backward and one outgoing forward transition labelled by $\emptystring$.
In this construction, we basically merge the states for $\mu t \seq G$ and $G$.
It is also the place where recursion is introduced, ensuring ancestor-recursion and intermediate recursion freedom.
For sender-driven choice, analogously, the assumption for the regular expression yields that the first letters are pair-wise distinct and, thus, the branch that decides whether to start or repeat with~$\regex[1]$ or continue with the next regular expression satisfies the sender-driven choice condition for the PSMs.
\proofEndSymbol
\end{proof}

\subsubsection*{Step \ref{wf:arnmdfir-psm-to-global-type}:
From Ancestor-recursive Non-merging Dense Intermediate-recursion-free PSMs to Global Types.}
\label{sec:wf:arnmdfir-psm-to-global-type}

While the previous steps apply to arbitrary (\sinkfinal) PSMs, this one only applies for \sumOnePSMs since global types specify send and receive events together.
This transformation is rather straightforward.
The global type can be constructed via a traversal of the \sumOnePSM.

\begin{procedure}[\sumOnePSM to Global Type]
Let $\PSM$ be
ancestor-recursive, dense, non-merging, and inter\-mediate-recursion-free \sumOnePSM.
As a preprocessing step, we merge asynchronous events and assume $\PSM$ works on the alphabet of synchronous events~$\AlphSync_\Procs$.
We start with an empty global type and start the traversal from initial state:
\begin{itemize}
 \item If the state is final: add $0$ and return;
 \item if the state has an incoming transition: \\
        add $\mu t.$ for fresh $t$ and store $t$ for this state;
 \item if the state has an outgoing transition to previously seen state: \\
 add $t$ for the destination of the outgoing transition and return;
 \item if the state has outgoing transitions to unseen states: \\
 add $\Sum_{i \in I}$ with a fresh index set $I$ (for $\card{I}$ branches) with one branch for each next state with according transition label and recurse for each of next states.
\end{itemize}
Note that a state can have an incoming transition and more than one outgoing transitions, in contrast to the state machine of a global type where every subterm of shape $\mu t \seq G$ has only one incoming and one outgoing transition labelled by~$\emptystring$.
In this construction, the states for $\mu t \seq G$ and $G$ are merged.
We denote the result of this procedure with $\psmToGlobalTypeFunc(\PSM)$.
\end{procedure}

The following lemma states the correctness of the previous procedure.

\begin{restatable}{lemma}{psmToGlobalTypeFuncCorrect}
\label{lm:psmToGlobalTypeFuncCorrect}
Let $\PSM$ be
ancestor-recursive, non-merging, dense, intermediate-recursion-free, and \sinkfinal \sumOnePSM and $\psmToGlobalTypeFunc(\PSM)$ be the global type constructed from $\PSM$.
Then, their core languages are equal:
$\lang(\PSM) = \lang(\psmToGlobalTypeFunc(\PSM))$.
If $\PSM$ is a sender-driven PSM, then $\psmToGlobalTypeFunc(\PSM)$ is a sender-driven global~type.
\end{restatable}
\begin{proof}
The assumptions guarantee that the traversal does not revisit states and only sink states are final.
The preprocessing simplifies the translation to the corresponding terms of a global type.
The claim then follows easily by construction.
We sketch how to formalise it.
One can define a formalism that jointly/recursively represents languages starting from states in an FSM and (partial) global types.
The construction iteratively refines this representation,
preserving the specified language and sender-driven choice if given.
\proofEndSymbol
\end{proof}

\subsubsection*{Wrapping Up: From PSMs to Global Types.}
\label{sec:summary-transformation-psm-mst}

Let us first observe that part of this workflow can be applied when using the more general alphabet $\AlphAsync$ where send and receive events may not happen next to each other, yielding the following results. 

\begin{lemma}
For every PSM $\PSM$ with $\emptystring \notin \semantics(\PSM)$, there is an
ancestor-recursive, non-merging, dense, and intermediate-recursion-free PSM $\PSM'$ with the same core language.
If $\PSM$ is \sinkfinal and satisfies mixed choice (sender-driven choice, or directed choice respectively),
then $\PSM'$ is \sinkfinal and satisfies mixed choice (sender-driven choice, or directed choice respectively).
If $\PSM$ is not \sinkfinal, restrictions on choice are not preserved.
\end{lemma}
\begin{proof}
Let $\PSM$ be a \sinkfinal PSM with $\emptystring \notin \semantics(\PSM)$.
We do a case analysis if $\PSM$ is \sinkfinal.
If $\PSM$ is \sinkfinal,
\[
    \regexToPsmFunc(\psmToRegexFunc(\PSM))
\]
is such a PSM and any restriction on choice is preserved by
\cref{lm:psmToRegexFuncCorrect,lm:regexToPsmFuncCorrect}.
If $\PSM$ is not \sinkfinal,
\[
    \regexToPsmFunc(\psmToRegexFunc(\psmToSinkFinalFunc(\PSM)))
\]
is such a PSM by
\cref{prop:psmToSinkFinalFuncCorrect,lm:psmToRegexFuncCorrect,lm:regexToPsmFuncCorrect}.
\proofEndSymbol
\end{proof}

For the special case of \sumOnePSMs, we can convert such PSMs to global types, proving the main result in this section.

\sumOnePSMasGlobalType*

\begin{proof}
For the first claim, let $\PSM$ be a \sinkfinal \sumOnePSM.
We do a case analysis if $\emptystring \in \semantics(\PSM)$.

If so, we know that $\semantics(\PSM) = \set{\emptystring}$ because $\PSM$ is \sinkfinal and no transition is labelled by $\emptystring$.
Then, $G = 0$ has the same core language.

If not, we can apply our workflow and construct the global type
\[
    \psmToGlobalTypeFunc(\regexToPsmFunc(\psmToRegexFunc(\PSM)))
\]
which represents the same core language and any restriction on choice is preserved by
\cref{lm:psmToRegexFuncCorrect,lm:regexToPsmFuncCorrect,lm:psmToGlobalTypeFuncCorrect}.

\end{proof}

With the optional step (0), our workflow can also transform non-\sinkfinal \sumOnePSMs into global types; however, at the price of adding non-determinism.

\begin{theorem}
\label{lm:sumOnePSMasNonDetGlobalType}
Every \sumOnePSM $\PSM$ with $\emptystring \notin \semantics(\PSM)$ can be represented as a non-deterministic global type with the same core language.
\end{theorem}
\begin{proof}
Let $\PSM$ be a non-\sinkfinal \sumOnePSM with $\emptystring \notin \semantics(\PSM)$.
Then,
\[
    \psmToGlobalTypeFunc(\regexToPsmFunc(\psmToRegexFunc(\psmToSinkFinalFunc(\PSM))))
\]
is a non-deterministic global type that represents the same core language by
\cref{prop:psmToSinkFinalFuncCorrect,lm:psmToRegexFuncCorrect,lm:regexToPsmFuncCorrect,lm:psmToGlobalTypeFuncCorrect}.
\proofEndSymbol
\end{proof}

\subsection{Additional Material for Consequence \ref{item:key-consequence-c}}

Local types are defined analogously to global types.
We present local types which allow participants to send and receive from different participants, as defined in~\cite{DBLP:conf/concur/MajumdarMSZ21}.

\begin{definition}[Local Types]
\label{def:local-types}
The \emph{local types} for a role $\procA$ are defined as:
    \begin{grammar}
     L \is 0
         | \IntCh_{i ∈ I} \snd{}{\procB_i}{\val_i}.L_i
         | \ExtCh_{i ∈ I} \rcv{\procB_{i}}{}{\val_i}.L_i
         | μ X. L
         | X
    \end{grammar}
where the internal choice $(\IntCh)$ and  external choice $(\ExtCh)$ both respect
$∀ i,j ∈ I.\; i≠j ⇒ (\procB_i, \val_i) ≠ (\procB_j, \val_j)$.
As for global types,
we assume every recursion variable is bound, each recursion
operator~($μ$) uses a different identifier $t$,
We assume $\card{I} > 0$ and, if~$\card{I} = 1$, we omit $\IntCh$ and $\ExtCh$.
The semantics is defined analogously to global types, using a state machine $\semlocal(\hole)$, and is omitted.
\end{definition}

The workflow of \cref{app:tree-transformation} can obtain a local type from a \sinkfinal FSM over $\AlphAsync_\procA$ for a participant $\procA$.
If a deterministic FSM is not \sinkfinal, there is no \sinkfinal deterministic FSM for the same language, making a non-deterministic local type the most one can achieve.
If the FSM has no mixed-choice states, the transformation yields a local type.
If it has, the structure still resembles the one of local types but requires the simultaneous specification of receiving and sending. 

\localTypesEquiExpressive*

\subsection{Additional Material for \ref{item:key-consequence-d}}
\label{app:projection-without-soft-deadlocks}

In this section, we provide the formalisation to prove \cref{thm:soft-impl-sink-final-red}. 

When using local types, final configurations are always sink-state configurations, \ie where each participant is in a sink state.
For our setting, this is not the case and has repercussions on the semantics and meaning of deadlocks:
we can have final configurations where some participants are not in a final state.
If there is no next transition for such a configuration,
we call it a soft deadlock.
We define the notion of soft deadlocks and recall the definition of strong projectability. 

\begin{definition}[Soft Deadlocks, Strong Projectability]
\label{def:soft-implementability-problem}
Let $\CSM{A}$ be a CSM.
A configuration $(\vec{q}, \xi)$ is a \emph{soft deadlock}
if there is no $(\pvec{q'}, \xi')$ with
$(\vec{q}, \xi) \to (\pvec{q'}, \xi')$
and
$(\vec{q}, \xi)$ is no final sink-state configuration.
We say $\CSM{A}$ is free from soft deadlocks if every reachable configuration is no soft deadlock.
A language $L \subseteq \AlphAsync^\omega$ is said to be \emph{strongly projectable} if there exists a CSM $\CSM{B}$ such that
$\CSM{B}$ is free from soft deadlocks
 \emph{(soft deadlock freedom)},
 and
$L$ is the language of $\CSM{B}$
 \emph{(protocol fidelity)}.
We say that $\CSM{B}$ is a strong projection of $L$.
\end{definition}

With this notion of deadlock, intuitively, a configuration can only be considered actually final if no state machine has an outgoing transition.
It is obvious that every deadlock is a soft deadlock.
Different applications call for different notions of deadlock freedom.
In distributed computing, it is fine if a server keeps listening for incoming requests while, in embedded computing, it can be essential that all participants eventually stop.
We believe this is a design~choice.

Intuitively, if one aims for soft deadlock freedom, no state with outgoing transitions needs to be final because,
for soft deadlocks, only final sink-state configurations matter.
One could require a projection to be \sinkfinal.
However, while this is a sufficient, it is not obviously necessary:
any final non-sink-state configuration could simply not get stuck, never exploiting the fact it is final.
We show that this is not the case if one of two conditions hold for CSMs.

We use the subset construction from~\cite[Def.\,5.4]{DBLP:conf/cav/LiSWZ23}, denoted by $\subsetcons{\GG}{\procA}$ for participant $\procA$, which simply projects the global type's state machine onto the participant alphabet and determinises the result. 

\begin{definition}[Subset construction]
\label{def:subset-construction}
Let $\GG$ be a global type and $\procA$ be a participant.
Then, the \emph{subset construction} for $\procA$ is defined as
\[
\subsetcons{\GG}{\procA} =
\bigl(
Q_{\procA},
\Alphabet_{\!\procA},
\delta_{\procA},
s_{0, \procA},
F_{\procA}
\bigr)
\text{ where }
\]
\begin{itemize}
\item $ \delta(s, a) \is
	\set{q' \in Q_{\GG}
		\mid
		\exists q \in s,
		q \xrightarrow{a} \xrightarrow{\emptystring}\mathrel{\vphantom{\to}^*} q' \in \projerasuretrans
	}
	$
	for
	every
	$s \subseteq Q_{\GG}$ and
	$a \in \Alphabet_{\procA}$,
	\item $s_{0, \procA} \is
	\set{q \in Q_{\GG} \mid
		q_{0, \GG} \xrightarrow{\emptystring} \mathrel{\vphantom{\to}^*} q \in \projerasuretrans}$,
	\item $Q_{\procA} \is \lfp_{\set{s_{0,\procA}}}^\subseteq \lambda Q.\, Q \union \set{ \delta(s,a) \mid s \in Q \land a \in \Alphabet_{\!\procA}} \setminus \set{\emptyset}$,
\item $\delta_{\procA} \is \restrict{\delta}{Q_{\procA} \times \Alphabet_{\!\procA}}$, and
	\item $F_{\procA} \is
	\set{s \in Q_{\procA}
		\mid s \inters F_{\GG} \neq \emptyset}$.
\end{itemize}
\end{definition}

To prove \cref{thm:soft-impl-sink-final-red}, 
we add two more equivalent statements and prove four implications, yielding a cycle and equivalence of all statements.
It is trivial that
\cref{thm:soft-impl-sink-final-red}
follows from
\cref{thm:soft-impl-sink-final}.

\begin{restatable}{theorem}{softImplSinkFinal}
\label{thm:soft-impl-sink-final}
Let $\GG$ be a projectable global type.
Then, the following statements are equivalent, where $\subsetcons{\GG}{\procA}$ is the subset construction of $\GG$ onto $\procA$:
\begin{enumerate}[labelindent=0pt,labelwidth=\widthof{(iii)},label=\textnormal{(\roman*)},itemindent=0em,leftmargin=!]
 \item \label{thm:soft-impl-sink-final-1}
    There is a \sinkfinal CSM that is a projection of $\GG$ and satisfies feasible eventual reception or every of its state machines is deterministic.
    $\GG$ can be implemented by a \sinkfinal CSM that satisfies feasible eventual reception or every of its state machines is deterministic.
 \item \label{thm:soft-impl-sink-final-2}
    The subset construction $\subsetcons{\GG}{\procA}$ is \sinkfinal for every participant $\procA$.
 \item \label{thm:soft-impl-sink-final-3}
    All reachable final configurations of $\CSMl{\subsetcons{\GG}{\procA}}$ are sink-state. \item \label{thm:soft-impl-sink-final-4}
    There is a CSM that is a strong projection of $\GG$ and this CSM satisfies feasible eventual reception or every of its state machines is deterministic.
\end{enumerate}
\end{restatable}
\begin{proof}
Note that the subset projection is basically the subset construction but checks validity conditions \cite[Sec.\,5]{DBLP:conf/cav/LiSWZ23}.
Together with completeness of their approach \cite[Thm.\,7.1]{DBLP:conf/cav/LiSWZ23}, the Send Validity condition enforces that any final state in the subset projection, and thus construction of a projectable global type, cannot have outgoing send transitions.

Proof that
\ref{thm:soft-impl-sink-final-1}
implies
\ref{thm:soft-impl-sink-final-2}: \\
Let $\CSM{A}$ be a \sinkfinal CSM that satisfies feasible eventual reception or every of its state machines is deterministic, and is a projection of $\GG$.
Towards a contradiction, assume that the subset construction is not \sinkfinal.
Without loss of generality, let $\subsetcons{\GG}{\procA}$ be the subset construction with at least one final non-sink state and let $s$ be one of the states that is final and has outgoing transitions.

By the fact that $s$ is final, there is $0 \in s$.
By the fact that $s$ has outgoing transitions, there is $G' \in s$ with $G' \xrightarrow{x} G''$ for $x \wproj_{\AlphAsync_{\!\procA}} \in \AlphAsync_{\!\procA}$ and some~$G''$.
Because of Send Validity \cite[Def.\,5.4]{DBLP:conf/cav/LiSWZ23}, we have $x = \msgFromTo{\procB}{\procA}{\val}$ for some participant $\procB$ and message $\val$.

In the subset construction, two subterms $G_1$ and $G_2$ do only occur in the same state $\vec{s}_\procA$ of $\subsetcons{\GG}{\procA}$ if there are two runs
$
G \xrightarrow{w}\mathrel{\vphantom{\to}^*} G_1
$
and
$
G \xrightarrow{w'}\mathrel{\vphantom{\to}^*} G_2
$
with
$w \in \AlphAsync^*$,
$w' \in \AlphAsync^*$,
and $w \wproj_{\Alphabet_{\!\procA}} = w' \wproj_{\Alphabet_{\!\procA}}$.
Here, we use choose $G_1 = 0$ and $G_2 = G'$.
Thus, we have $w \in \lang(\GG)$.

Let
$(\vec{s}, \xi)$
be the configuration of the
subset construction $\CSMl{\subsetcons{\GG}{\procA}}$
that is reached after processing $w$.
Because of protocol fidelity, determinacy of the subset construction and $w \in \lang(\GG)$, it holds that $(\vec{s}, \xi)$ is final.
Recall there is a transition from $G'$ to~$G''$ labelled with $\msgFromTo{\procB}{\procA}{\val}$, so we can extend $w'$ to obtain $w'' \is w' \cat \snd{\procA}{\procB}{\val}$
for which
$w \wproj_{\Alphabet_{\!\procA}} = w'' \wproj_{\Alphabet_{\!\procA}}$.
Let
$(\pvec{s}'', \xi'')$
be the configuration of
$\CSMl{\subsetcons{\GG}{\procA}}$
after processing~$w''$.
By construction of the traces, the channels are empty after processing~$w''$.
Hence, we have $\xi''(\procB, \procA) = \val$ with the additional send event in $w''$.
Because of
$w \wproj_{\Alphabet_{\!\procA}} = w'' \wproj_{\Alphabet_{\!\procA}}$, it holds that
$\vec{s}_\procA = \pvec{s}''_\procA$.
Let us consider the two configurations of
$\CSM{A}$
that are reached with $w$ and $w''$.
By
\cite[Lm.\,20]{DBLP:conf/concur/MajumdarMSZ21},
they will have the same channel contents as the subset construction respectively.
Let
$(\vec{t}, \xi)$
and
$(\pvec{t}'', \xi'')$
be the configurations of
$\CSM{A}$
after processing $w$ and $w''$ respectively.
By
$w \wproj_{\AlphAsync_{\!\procA}}
=
w'' \wproj_{\AlphAsync_{\!\procA}}$,
it is possible and we assume that
$\vec{t}_\procA = \pvec{t}''_\procA$.
Note that we do not assume determinacy of $\CSM{A}$ but we can assume that both runs end in the same state for $\procA$ since all ways of non-determinism need to be accounted for.
We show that $\vec{t}_\procA$ is not \sinkfinal: it is final but has at least one outgoing~transition.

To start, we show that $\vec{t}_\procA$ is final.
It suffices to show that
$(\vec{t}, \xi)$
is a final configuration.
By the fact that
$(\vec{s}, \xi)$
is final, $\xi$ has only empty channels.
Towards a contradiction, assume that $\vec{t}_\procC$ is not final for some participant $\procC$.
Then, $w$ is not in $\lang(\CSM{A'})$ if there is no other run for $w$ but, if there were, $\CSM{A'}$ still deadlocks in the configuration $(\vec{t}, \xi)$, contradicting deadlock freedom.
Hence,
$(\vec{t}, \xi)$ is~final.

It remains to show that $\vec{t}_\procA$ has an outgoing transition.
We do a case analysis on the side condition for
$\CSM{A}$.

First, we assume that $\CSM{A}$ satisfies feasible eventual reception.
We use the second configuration
$(\pvec{t}'', \xi'')$
where $\procA$ is in the same state.
We know that $\val$ is the first message in $\xi''(\procB, \procA)$.
It was sent and must be received.
Thus, $\pvec{t}''_\procA$ has at least one outgoing~transition.

Second, assume that $A_\procC$ is deterministic for every participant $\procC$.
Again, we use the second configuration
$(\pvec{t}'', \xi'')$
where $\procA$ is in the same state.
By the semantics of global types, there is an extension $w'''$ of $w''$ with
$w''' \in \lang(\GG)$ that contains the receive event $\rcv{\procB}{\procA}{\val}$ for the enqueued message $\val$.
If $w'''$ is finite, it is straightforward that $\procA$ needs to be able to receive the message $\val$ from $\procB$ to satisfy protocol fidelity, ensuring an outgoing transition.
If $w''$ is infinite, towards a contradiction, assume that $\vec{t}_\procA$ has no outgoing transition.
The semantics require that every prefix $u \in \pref(w''')$ is in $\pref(\CSM{A})$.
However, if $\vec{t}_\procA$ could not receive $\val$, because of determinacy of $A_\procA$, there is a prefix of $w'''$ that is not.
This contradicts the assumption that $\CSM{A}$ is a projection of $\GG$.

Proof that
\ref{thm:soft-impl-sink-final-2}
implies
\ref{thm:soft-impl-sink-final-3}: \\
By assumption, all final configurations are sink-state configurations.
Hence, all reachable final configurations are sink-state configurations.

Proof that
\ref{thm:soft-impl-sink-final-3}
implies
\ref{thm:soft-impl-sink-final-4}: \\
We claim that
$\CSMl{\subsetcons{\GG}{\procA}}$
is such a CSM.
Because of soundness \cite[Thm.\,6.1]{DBLP:conf/cav/LiSWZ23},
we know that
$\CSMl{\subsetcons{\GG}{\procA}}$
satisfies protocol fidelity and is deadlock-free.
Hence, it suffices to show
$\CSMl{\subsetcons{\GG}{\procA}}$
is free from soft deadlocks.
Since
$\CSMl{\subsetcons{\GG}{\procA}}$
is deadlock-free,
the only way there could be soft deadlocks is that there is a reachable stuck final configuration that is no sink-state configuration.
This is impossible because every reachable final configuration is a sink-state configuration by assumption.
It remains to show that
$\CSMl{\subsetcons{\GG}{\procA}}$
satisfies one of both side conditions.
Both properties hold.
First, the subset construction determinises by definition.
Second, every global type satisfies feasible eventual reception by construction and a projection preserves this property.

Proof that
\ref{thm:soft-impl-sink-final-4}
implies
\ref{thm:soft-impl-sink-final-1}: \\
We know there is a strong projection for $\GG$ (with certain properties) and want to show there is a \sinkfinal projection (with the same properties).
From the assumption, let $\CSM{A}$ be a strong projection of $\GG$.
By definition, every strong projection is also a projection of $\GG$ because soft deadlock freedom implies deadlock freedom.
We assume it is not \sinkfinal for some participant $\procA$ as the claim follows trivially otherwise.
This means that $A_\procA$ has final states that have outgoing transitions.
We will show that these states do not need to be final for any such~$\procA$.
Hence, we can turn $\CSM{A}$ to a \sinkfinal CSM $\CSM{A'}$ by inductively applying this to all participants whose state machines are not \sinkfinal.

We claim that, for a strong projection $\CSM{A}$, whether a state is final or not does only matter for soft deadlock freedom and not protocol fidelity.
Assume it was relevant for protocol fidelity.
Then, there is a reachable non-sink-state configuration that is final, which means that all participants are in final states and the channels are empty.
We established earlier that no final state can have an outgoing send transition.
Thus, this would constitute a soft deadlock, contradicting the fact that $\CSM{A}$ is a strong projection.

From soft deadlock freedom, we know that there is no stuck reachable final non-sink-state configuration.
In other words, all the stuck configurations are final sink-state configurations.
Thus, any non-sink final state does not need to be final.
This proves the claim that the existence of $\CSM{A}$ always implies the existence of a \sinkfinal CSM $\CSM{A'}$ which is our witness for projectability.
The side conditions for both statements are the same and not affected by our construction so they simply carry over.
\proofEndSymbol
\end{proof}

For our result, it is fine to assume projectability of $\GG$ as it is a prerequisite for strong projectability.
If we aim for a strong projection, we can construct the global type's subset projection (the subset construction with validity conditions) and check if it is \sinkfinal.
If it provides one that is not, there is no strong projection of it.
If this is undesirable, the protocol needs redesigning.
To obtain local types from a CSM of \sinkfinal state machines, we can use \cref{lm:local-types-equi-expressive}. This requires no mixed-choice states, which is always the case for the subset projection. 

\begin{corollary}
Let $\GG$ be a strongly projectable global type.
Then, there is a local type $L_\procA$ for every $\procA \in \Procs$ such that
$\CSMl{\semlocal(L_\procA)}$ is a strong projection of~$\GG$.
\end{corollary}

\subsection{Projectability and Strong Projectability \\ for Global Types with Mixed Choice is Undecidable}
\label{app:mixed-choice-global-types-undec}

In \cref{sec:mixed-choice-yields-undecidable-projection}, we constructed a \sinkfinal mixed-choice \sumOnePSM to prove undecidability of the respective projectability problem.
With our results from
\cref{app:tree-transformation},
we can transform this encoding into a mixed-choice global type.

\checkingImplMixedChoiceGlobalTypesUndec*
\begin{proof}
From \cref{thm:checking-implementability-sinkfinal-mixed-choice-PSM-undecidable}, we know that the projectability problem is undecidable for \sinkfinal mixed-choice \sumOnePSMs in general.
From \cref{thm:sumOnePSMasGlobalType}, we know that such PSMs can be transformed into a mixed-choice global type, which proves the claim for the projectability problem.
By construction, all participants are informed when the protocol of the encoding ends, making the projectability and soft projectability problem equivalent.
This proves the claim for the soft projectability problem.
\proofEndSymbol
\end{proof}
 }
{
}

\end{document}